\documentclass[11pt]{amsart}
\usepackage[utf8]{inputenc}

\usepackage{listings}
\usepackage[dvipsnames]{xcolor} 
\usepackage{mathrsfs}

\usepackage{amsthm}
\usepackage{mathrsfs}
\usepackage{braket}
\usepackage{wasysym}
\usepackage{amsmath}
\usepackage{mathtools}
\usepackage{hyperref}
\numberwithin{equation}{subsection}

\usepackage{chngcntr}
\counterwithin{table}{subsection}
\counterwithin{figure}{subsection}

\hypersetup{
    colorlinks,
    citecolor=NavyBlue,
    filecolor=black,
    linkcolor=Green,
    urlcolor=Orchid,
    linktocpage=true,
}
\usepackage[margin=1.0in]{geometry}

\usepackage{tikz}
\usetikzlibrary{calc,positioning,shadows.blur,decorations.pathreplacing, decorations.markings}
\usetikzlibrary{tqft}
\usepackage{etoolbox}  
\usepackage{cleveref} 
\usepackage[makeroom]{cancel}
\usepackage{subcaption}
\usepackage{placeins}

\usepackage{booktabs}

\usepackage{listings}
\usepackage{xcolor}
\usepackage{amsfonts}
\usepackage{lipsum}
\usepackage{hyperref}
\usepackage{slashed}
\usepackage{simplewick}
\usepackage[squaren]{SIunits}
\usepackage{sidecap} 
\usepackage{tikz-cd}
\usepackage[titletoc]{appendix}
\usepackage{algorithm}
\usepackage{algpseudocode}

\usepackage{etoolbox}

\usepackage[compat=1.1.0]{tikz-feynman}
\usepackage{contour}
\usepackage{multicol}
\usepackage{scalerel,stackengine}

\patchcmd{\part}{\thispagestyle{plain}}{\thispagestyle{part}}{}{}

\definecolor{grey}{rgb}{0.98,0.98,0.98}
\definecolor{codeRed}{rgb}{0.5, 0.027, 0.02}

\DeclareMathOperator{\Der}{Der}
\DeclareMathOperator{\Sym}{Sym}
\DeclareMathOperator{\Id}{id}
\DeclareMathOperator{\Dens}{Dens}
\DeclareMathOperator{\Maps}{Maps}

\DeclareMathOperator{\loops}{loops}
\DeclareMathOperator{\Aut}{Aut}

\DeclareMathOperator{\Grad}{Gr}
\DeclareMathOperator{\Spec}{Spec}
\DeclareMathOperator{\pt}{pt}
\DeclareMathOperator{\Jet}{Jet}
\DeclareMathOperator{\Sp}{Sp}

\DeclareMathOperator{\pbw}{pbw}

\usepackage[
backend=biber,
style=alphabetic,
sorting=nyt,giveninits=true,maxbibnames=99
]{biblatex}
\theoremstyle{plain}
\newtheorem*{theoremn}{Theorem}

\theoremstyle{plain}
\theoremstyle{definition}
\newtheorem{defn}{Definition}[subsection]
\newtheorem{assump}[defn]{Assumption}
\newtheorem{notat}[defn]{Notation}
\newtheorem{exm}[defn]{Example}
\theoremstyle{plain}
\newtheorem{thm}[defn]{Theorem} 
\newtheorem{lmm}[defn]{Lemma}
\newtheorem{prp}[defn]{Proposition} 
\theoremstyle{remark}
\newtheorem{rmk}[defn]{Remark}

\newcommand{\beq}{\begin{equation}}
\newcommand{\eeq}{\end{equation}}

\newcommand{\Tr}{\mathrm{Tr}}
\newcommand{\F}{\mathcal{F}}

\stackMath
\newcommand\reallywidehat[1]{%
\savestack{\tmpbox}{\stretchto{%
  \scaleto{%
    \scalerel*[\widthof{\ensuremath{#1}}]{\kern-.6pt\bigwedge\kern-.6pt}%
    {\rule[-\textheight/2]{1ex}{\textheight}}
  }{\textheight}%
}{0.5ex}}%
\stackon[1pt]{#1}{\tmpbox}%
}
\tikzfeynmanset{
  every fermion={black},
  every anti fermion={black},
}

\newcommand{\Sc}[0]{\mathcal{S}}

\newcommand{\g}[0]{\mathfrak{g}}

\newcommand{\W}[0]{\mathcal{W}}
\newcommand{\CS}[0]{\reallywidehat{\Sym}^\bullet(T^{\vee 1,0}M)}
\newcommand{\Chr}[3]{\Gamma^{#1}_{#2#3}}
\newcommand{\Riem}[4]{R^{#1}_{#2#3#4}}
\newcommand{\Linf}[0]{L_{\infty}}

\newcommand{\MdR}[0]{{\Sigma_3}_{,\, \mathrm{dR}}}
\newcommand{\Xd}[0]{M_{\Bar{\partial}}}
\newcommand{\bon}[0]{^{\partial}_{\partial \Sigma_3}}

\newcommand{\sur}[0]{_{\Sigma_3}}
\newcommand{\surM}[0]{_{\partial M}}

\newcommand{\surg}[0]{_{\Sigma_3, x}}
\newcommand{\surgR}[0]{_{\Sigma_3, x, R}}
\newcommand{\surgS}[0]{^{\, \text{s}}_{\Sigma_3, x}}
\newcommand{\surgSR}[0]{^{\, \text{s}}_{\Sigma_3, x, R}}
\newcommand{\surgSRbar}[0]{^{\, \text{s}}_{\Sigma_3, x, \bar{R}}}
\newcommand{\surgSP}[0]{^{\, \text{s},\, \mathcal{P}}_{\Sigma_3, x}}

\newcommand{\surgSB}[0]{^{\, \text{s},\, \partial}_{\partial\Sigma_3, x}}

\newcommand{\surgSBP}[0]{^{\, \text{s},\, \partial,\, \mathcal{P}}_{\partial\Sigma_3, x}}
\newcommand{\bonP}[0]{^{\mathcal{P}}_{\partial \Sigma_3}}
\newcommand{\Gr}[0]{\mathcal{D}_{\mathrm{G}}}
\newcommand{\Fe}[0]{\mathcal{D}_{\mathrm{F}}}
\newcommand{\FDg}[0]{\mathcal{F}_{\mathrm{CLL}}}
\newcommand{\A}[0]{\mathcal{A}}
\newcommand{\derham}[0]{\Omega^\bullet(M)}
\newcommand{\cinfty}[0]{\mathcal{C}^\infty(M)}

\newcommand{\bea}{\begin{align*}}
\newcommand{\eea}{\end{align*}}

\tikzset{->-/.style={decoration={
  markings,
  mark=at position #1 with {\arrow{>}}},postaction={decorate}}}
  \tikzset{-<-/.style={decoration={
  markings,
  mark=at position #1 with {\arrow{<}}},postaction={decorate}}}

  \tikzset{
    ncbar angle/.initial=90,
    ncbar/.style={
        to path=(\tikztostart)
        -- ($(\tikztostart)!#1!\pgfkeysvalueof{/tikz/ncbar angle}:(\tikztotarget)$)
        -- ($(\tikztotarget)!($(\tikztostart)!#1!\pgfkeysvalueof{/tikz/ncbar angle}:(\tikztotarget)$)!\pgfkeysvalueof{/tikz/ncbar angle}:(\tikztostart)$)
        -- (\tikztotarget)
    },
    ncbar/.default=0.5cm,
    }

\tikzset{square left brace/.style={ncbar=0.5cm}}
\tikzset{square right brace/.style={ncbar=-0.5cm}}

\tikzset{round left paren/.style={ncbar=0.5cm,out=120,in=-120}}
\tikzset{round right paren/.style={ncbar=0.5cm,out=60,in=-60}}

\title[Formal Global Perturbative Quantization of the RW Model]{Formal Global Perturbative Quantization of the Rozansky--Witten Model in the BV-BFV Formalism}
\author[N. Moshayedi]{Nima Moshayedi}
\author[D. Saccardo]{Davide Saccardo}
\address{Department of Mathematics, University of California, Berkeley California 94305, USA}
\email[N.~Moshayedi]{nmosha@math.berkeley.edu}
\address{Institut f\"ur Theoretische Physik, ETH Z\"urich, Wolfgang-Pauli-Strasse 27, CH-8093 Z\"urich}
\email[D.~Saccardo]{davide.saccardo14@gmail.com}

\begin{document}

\maketitle

\begin{abstract}
We describe a globalization construction for the Rozansky--Witten model in the BV-BFV formalism for a source manifold with and without boundary in the classical and quantum case. After having introduced the necessary background, we define an AKSZ sigma model, which, upon globalization through notions of formal geometry extended appropriately to our case, is shown to reduce to the Rozansky--Witten model. The relations with other relevant constructions in the literature are discussed. 
Moreover, we split the model as a $BF$-like theory and we construct a perturbative quantization of the model in the quantum BV-BFV framework. In this context, we are able to prove the modified differential Quantum Master Equation and the flatness of the quantum Grothendieck BFV operator. Additionally, we provide a construction of the BFV boundary operator in some cases. 

\end{abstract}

\tableofcontents
\section{Introduction}

\subsection{Overview and motivation}
An important class of field theories in physics is represented by \textit{gauge theories}. These are theories containing a redundant number of degrees of freedom which causes physical quantities to be invariant under certain local transformations, called \textit{gauge symmetries}. Indeed the presence of gauge symmetries lead to challenging problems from the definition of path integral to the general problem of understanding the perturbative quantization of a gauge theory. Since the physical information about a classical field theory is encoded in the set of solutions of the Euler--Lagrange equations (the \textit{critical locus}), a possible solution to deal with such problems is to consider the critical locus modulo the gauge symmetries. The fields are then constructed as functions on this quotient. However, this is not feasible since it turns out that these quotients are, in general, singular. Batalin and Vilkovisky introduced a method, which is known today as the \textit{BV formalism} \cite{BV77,BV81,BV83}, that employs symplectic (co)homological tools \cite{KT79} to treat these field theories, in particular it overcomes difficulties connected to the singularity of the quotient by taking homological resolution of the critical locus. A crucial observation in the BV formalism is also that gauge-fixing then corresponds to the choice of a Lagrangian submanifold.
Another method developed around the same time is the \textit{BFV formalism} by  Batalin, Fradkin and Vilkovisky \cite{FV77,BF83,BF86}, which deals with gauge theories in the Hamiltonian setting, while the BV construction is formulated in the Lagrangian approach. 

Recently, the study of gauge theories on spacetime manifolds with boundary lead Cattaneo, Mnev and Reshetikhin \cite{CMR11,CMR14} to relate these two formulations in order to develop the \textit{BV-BFV formalism}. Their idea was that, under certain conditions, BV theories in the bulk can induce a BFV theory on the boundary. This approach was successfully applied to a large number of physical theories such as e.g. electrodynamics, Yang-Mills theory, scalar field theory and $BF$-theories \cite{CMR14}. In particular, the \textit{AKSZ construction}, developed in \cite{AKSZ97}, produces naturally a large variety of theories which satisfy automatically the BV-BFV axioms as it was shown in \cite{CMR14}. This is quite remarkable since many theories of interest are actually of AKSZ-type, such as e.g. Chern--Simons (CS) theory, $BF$-theory and the Poisson sigma model (PSM) \cite{CMR14}. 

In \cite{CMR17}, a perturbative quantization scheme for gauge theories in the BV-BFV framework was introduced, which was called \textit{quantum BV-BFV formalism}. The importance of this method relies on its compatibility with cutting and gluing in the sense of \textit{topological quantum field theories} (TQFTs). The quantum BV-BFV formalism has been applied successfully in various physically relevant theories such as e.g. $BF$-theory and the PSM \cite{CMR17}, split CS theory \cite{CMnW17} and CS theory \cite{CMnW21}, the relational symplectic groupoid\footnote{The relational symplectic groupoid was first defined in \cite{CC15}.} \cite{CMoW17} and 2D Yang--Mills theory on manifolds with corners \cite{Ir18,IM19}. 

An important effort has been spent to study TQFT within the quantum BV-BFV framework. Indeed, the method has been introduced to accomplish the goal of constructing perturbative topological invariants of manifolds with boundary compatible with cutting and gluing for topological field theories. During the years, two prominent TQFTs have been studied in detail: CS theory \cite{AS91,AS94} in \cite{CMnW17,We18} and the PSM \cite{SS94,Ik94} in \cite{CMoW20}. 



In \cite{CMoW19} a globalized version of the (quantum) BV-BFV formalism in the context of nonlinear split AKSZ sigma models on manifolds with and without boundary by using methods of formal geometry à la Bott \cite{Bo11}, Gelfand and Kazhdan \cite{GK71} (see also \cite{BCM12} for an application of the globalization procedure for the PSM in the context of a closed source manifold) was developed. Their construction is able to detect changes of the quantum state when one modifies the constant map around which the perturbation is developed. This required them to formulate a ``differential" version of the \textit{(modified) Classical Master Equation} and the \textit{(modified) Quantum Master Equation}, which are the two key equations in the BV(-BFV) formalism. As an example, this procedure was applied to the PSM on manifolds with boundary and extended to the case of corners in \cite{CMoW20}.

In this paper, we continue the effort in analyzing TQFTs within the quantum BV-BFV formalism by studying the \textit{Rozansky--Witten (RW) theory}. 
The RW model is a topological sigma model with a source 3-dimensional manifold $\Sigma_3$, which was introduced by Rozansky and Witten in \cite{RW96} through a \textit{topological twist} of a 6-dimensional supersymmetric sigma model with target a hyperK{\"a}hler manifold $M$. Of particular interest is the perturbative expansion of the RW partition function. Rozansky and Witten obtained this expansion as a combinatorial sum in terms of Feynman diagrams $\Gamma$, which are shown to be trivalent graphs:
\begin{equation}
    Z_M(\Sigma_3)=\sum_\Gamma b_\Gamma(M)I_\Gamma(\Sigma_3),
\end{equation}
the $b_\Gamma (M)$ are complex valued functions on trivalent graphs constructed from the target manifold, while $I_\Gamma(\Sigma_3)$ contains the integral over the propagators of the theory and depends on the source manifold. There are evidences which suggest that $I_\Gamma(\Sigma_3)$ are the \textit{LMO invariants} of Le, Murakami and Ohtsuki \cite{LMO98}. On the other hand, Rozansky and Witten showed that $b_\Gamma(M)$ satisfy the famous AS (which is reflected in the absence of tadpoles diagrams) and IHX relations. As a result, $b_\Gamma(M)$ constitute the \textit{Rozansky--Witten weight system} for the \textit{graph homology}, the space of linear combinations of equivalence classes of trivalent graphs (modulo the AS and IHX relations). This means that the RW weights can be used to construct new finite type topological invariants for 3-dimensional manifolds \cite{B95}. 



The RW theory opened up a new branch of research which was undertaken by many mathematicians and physicists (e.g. \cite{HT99,Th00}). Shortly after the original paper, Kontsevich understood that the RW invariants could be obtained by the characteristic classes of foliations and Gelfand-Fuks cohomology \cite{Ko99}. Inspired by the work of Kontsevich, Kapranov reformulated the weight system in cohomological terms (instead of using differential forms) in \cite{KA99}. This idea relies on the fact that one can replace the Riemann curvature tensor by the Atiyah class \cite{At57}, which is the obstruction to the existence of a global holomorphic connection. As a consequence of Kontsevich's and Kapranov's approaches, the RW weights were understood to be invariant under the hyperK{\"a}hler metric on $M$: in fact, the model could be constructed more generally with target a holomorphic symplectic manifold. In this way, the RW weights were also called \textit{RW invariants}\footnote{The terminology is unfortunate as in reality the proper invariants should be the products of the weights with $I_\Gamma(\Sigma_3)$.} of $M$ (see \cite{Sa04} for a detailed exposition). On the other hand, the possibility to consider as target manifold a holomorphic symplectic manifold was later interpreted in the context of topological sigma models by Rozansky and Witten in the appendix of \cite{RW96}. 

In the last 20 years, the RW model has been the focus of intense research in order to formulate it as an extended TQFT (see \cite{Sa00,RS02}), in order to investigate its boundary conditions and defects \cite{KRS09, KR09}, and in order to construct its globalization formulation \cite{QZ10,KQZ13,CLL17}. 

\subsection{Our contribution}
The main contribution of this paper is to add the RW theory to the list of TQFTs which have been studied successfully within the globalized version of the quantum BV-BFV framework \cite{CMoW19}. This will be a step towards the higher codimension quantization of RW theory, which will possibly lead to new insights towards the 3-dimensional correspondence between CS theory \cite{Wit89} and the Reshetikhin--Turaev construction \cite{Reshetikhin1991} from the point of view of (perturbative) extended field theories described by Baez--Dolan \cite{BD95} and Lurie \cite{Lu09}. Moreover this could also help in understanding (generalizations of a globalized version of the) Berezin--Toeplitz quantization (star product) \cite{Schlichenmaier2010} through field-theoretic methods using cutting and gluing similarly as it was done for Kontsevich's star product \cite{Ko03} in the case of the PSM in \cite{CMoW20}.

We construct the BV-BFV extension of an AKSZ model having a 3-dimensional manifold $\Sigma_3$ (possibly with boundary) as source and a holomorphic symplectic manifold $M$ as target with holomorphic symplectic form $\Omega$. Following \cite{KA99}, we define a formal holomorphic exponential map $\varphi$. This is used to linearize the space of fields of our model obtaining
\begin{equation}
   \Tilde{\mathcal{F}}_{\Sigma_3, x}=\Omega^{\bullet}(\Sigma_3)\otimes T^{1,0}_{x}M,
\end{equation}
where $\Omega^{\bullet}(\Sigma_3)$ denotes the complex of de Rham forms on the source manifold and $T^{1,0}_{x}M$ is the holomorphic tangent space on the target. In order to vary the constant solution around which we perturb, we define a classical Grothendieck connection which can be seen as a complex extension of the Grothendieck connection used in \cite{CMoW19,CMoW20}. In this way, we construct a \textit{formal global action} for our model, i.e.
\begin{equation}
    \Tilde{\Sc}\surg\coloneqq\varint_{\Sigma_3}\bigg(\frac{1}{2}\Omega_{ij}\hat{\mathbf{X}}^id\hat{\mathbf{X}}^j+\Big(\hat{R}^i\sur\Big)_j(x; \hat{\mathbf{X}})\Omega_{il}\hat{\mathbf{X}}^l dx^j+\Big(\hat{R}^i\sur\Big)_{\Bar{j}}(x; \hat{\mathbf{X}})\Omega_{il}\hat{\mathbf{X}}^l dx^{\Bar{j}}\bigg)
\end{equation}
with $\hat{\mathbf{X}}^i$ the coordinates of the spaces of fields $\Tilde{\mathcal{F}}_{\Sigma_3, x}$ organized as superfields, $x$ is the constant map over which we expand, $\Big(\hat{R}^i\sur\Big)_j$ and $\Big(\hat{R}^i\sur\Big)_{\Bar{j}}$ the components of the Grothendieck connection given by
\begin{equation}
    \begin{split}
        &R^i_j(x;y)dx^j:=-\bigg[\bigg(\frac{\partial \varphi}{\partial y}\bigg)^{-1}\bigg]^i_p\frac{\partial \varphi^p}{\partial x^j}dx^j,\\
        &R^i_{\Bar{j}}(x;y)dx^{\Bar{j}}:=-\bigg[\bigg(\frac{\partial \varphi}{\partial y}\bigg)^{-1}\bigg]^i_p\frac{\partial \varphi^p}{\partial x^{\Bar{j}}}dx^{\Bar{j}},
    \end{split}
\end{equation}
where $\{y^i\}$ are the generators of the fiber of $\reallywidehat{\Sym}^\bullet(T^{\vee1,0}M)$. The formal action is such that the \textit{differential Classical Master Equation} (dCME) is satisfied, namely 
\begin{equation}
    d_M\Tilde{\mathcal{S}}\surg+\frac{1}{2}(\Tilde{\mathcal{S}}\surg,\Tilde{\mathcal{S}}\surg)=0,
\end{equation}
with $d_M=d_x+d_{\Bar{x}}$ the sum of holomorphic and antiholomorphic Dolbeault differentials on $M$. The dCME presented here is different from the one presented in e.g. \cite{BCM12,CMoW19,CMoW20} since there $d_M$ was the de Rham differential on the body of the target manifold.  

The globalized model is then shown to be a globalization of the RW model \cite{RW96}, which reduces to the RW model itself in the appropriate limits. Our globalization of the RW model is compared with other globalization constructions as the one developed in \cite{CLL17} for a closed source manifold by using Costello's approach \cite{Co11a,Co11b} to \textit{derived geometry} \cite{To06, To14, PTTV13}, the procedure in \cite{Ste17} which extends the work of \cite{CLL17} to manifolds with boundary and the procedure in \cite{QZ10,KQZ13}. In general, our model is compatible with all these apparently different views. In particular, we give a detailed account of the similarities between our method and the one in \cite{CLL17}, thus confirming the claim in Remark 3.6 in \cite{CMoW19} about the equivalence between Costello's approach and ours. 

In order to quantize the theory according to the quantum BV-BFV formalism, we formulate a split version of our globalized RW model. Since the globalization is controlled by an $\Linf$-algebra, following \cite{Ste17} and inspired by the work of Cattaneo, Mnev and Wernli for CS theory \cite{CMnW17}, we assume that we can split the $\Linf$-algebra in two isotropic subspaces. The action of the \textit{globalized split RW model} is then
 \begin{equation}
    \Tilde{\Sc}\surgS=
    \braket{\hat{\mathbf{B}}, D\hat{\mathbf{A}}}+\Big\langle\Big(\hat{R}\sur\Big)_j(x; \hat{\mathbf{A}}+\hat{\mathbf{B}})dx^j,\hat{\mathbf{A}}+\hat{\mathbf{B}}\Big\rangle+\Big\langle\Big(\hat{R}\sur\Big)_{\bar{j}}(x; \hat{\mathbf{A}}+\hat{\mathbf{B}})dx^{\bar{j}},\hat{\mathbf{A}}+\hat{\mathbf{B}}\Big\rangle,
\end{equation} 
where $\braket{-,-}$ denotes the BV symplectic form on the space of fields $\Tilde{\mathcal{F}}\surgS$ with values in the Dolbeault complex of $M$, $\hat{\mathbf{A}}^i$ and $\hat{\mathbf{B}}_i$ are the fields found from the splitting of the field $\hat{\mathbf{X}}^i$, and $D$ denotes the superdifferential. Note that $d$ is the de Rham differential on the target, not on the source.

Finally, we quantize the globalized split RW model within the quantum BV-BFV formalism framework. Here, we obtained the following two theorems.
\begin{theoremn}[Flatness of the qGBFV operator (Theorem \ref{thm:flatness})]
The \emph{quantum Grothendieck BFV (qGBFV) operator} $\nabla_{\textup{G}}$ for the anomaly-free globalized split RW model squares to zero, i.e. 
\begin{equation}
\label{flatness_GBFV}
    (\nabla_{\textup{G}})^2\equiv0,
\end{equation}
where 
\begin{equation}
    \nabla_{\textup{G}}=d_{M}-i\hbar \Delta_{\mathcal{V}_{\Sigma_3, x}}+\frac{i}{\hbar}\boldsymbol{\Omega}_{\partial \Sigma_3}=d_x+d_{\Bar{x}}-i\hbar \Delta_{\mathcal{V}_{\Sigma_3, x}}+\frac{i}{\hbar}\boldsymbol{\Omega}_{\partial \Sigma_3},
\end{equation}
with $d_M$ the sum of the holomorphic and antiholomorphic Dolbeault differentials on the target $M$, $\Delta_{\mathcal{V}_{\Sigma_3, x}}$ the BV Laplacian and $\boldsymbol{\Omega}_{\partial \Sigma_3}$ the full BFV boundary operator.
\end{theoremn}
\begin{theoremn}[mdQME for anomaly-free globalized split RW model (Theorem \ref{thm:mdQME})]
 Consider the full covariant perturbative state $\hat{\psi}_{\Sigma_3,x}$ as a quantization of the anomaly-free globalized split RW model. Then 
   \begin{equation}
   \label{mdqme_thm}
        \bigg(d_M-i\hbar \Delta_{\mathcal{V}_{\Sigma_3, x}}+\frac{i}{\hbar}\boldsymbol{\Omega}_{\partial \Sigma_3}\bigg)\boldsymbol{\hat{\psi}}\surgR=0.
    \end{equation}
 \end{theoremn}
The proof of both the theorems is very similar to the ones exhibited in \cite{CMoW19} for non linear split AKSZ sigma models. Hence, we refer to \cite{CMoW19} when the procedure is the same whereas we remark when there are differences (which are related to the presence of the sum of the holomorphic and antiholomorphic Dolbeault differentials in the quantum Grothendieck BFV operator instead of the de Rham differential as in \cite{CMoW19}). 

We provide an explicit expression for the BFV boundary operator up to one bulk vertices in the $\mathbb{B}$-representation by adapting to our case the degree counting techniques of \cite{CMoW19}. Unfortunately, due to some complications related to the number of Feynman rules, we are not able to provide an explicit expression of the BFV boundary operator in the $\mathbb{B}$-representation in the case of a higher number of bulk vertices. See \cite{Sac21} for a limited example of graphs that appear when there are three bulk vertices. 

This paper is structured as follows:
\begin{itemize}
    \item In Section \ref{sec:BV-BFV} we introduce the most important notions of the classical an quantum BV-BFV formalism. Moreover, we give an overview of \textit{AKSZ theories}.
    \item In Section \ref{sec:RW_model} we introduce the necessary preliminaries to understand the RW model.
    \item In Section \ref{sec:Classical_Theory} we define an AKSZ model which upon globalization can be reduced to the RW model.
    \item In Section \ref{sec:comp_orig_RW} we compare our construction to  the original construction by Rozansky and Witten.
    \item In Section \ref{sec:comp_other} we compare our globalization construction with other globalization constructions of the RW model. 
    \item In Section \ref{sec:BF-like_formulation} we give a $BF$-like formulation by a splitting of the fields of the RW model in order to to be able to give a suitable description of its quantization.
    \item In Section \ref{sec:pert_quant_RW} we quantize the globalized split RW model according to the quantum BV-BFV formalism introduced in Section \ref{sec:BV-BFV}.
    \item In Section \ref{sec:mdQME} we introduce the \emph{quantum Grothendieck BFV operator} for the globalized split RW model, we prove that it is flat and, in the end, we use it to prove the modified differential Quantum Master Equation.
    \item Finally, in Section \ref{sec:outlook} we present some possible future directions.
\end{itemize}

\textbf{Notation}
Throughout the whole paper, we will keep the following conventions:
\begin{itemize}
    \item we will drop the wedge product wherever its presence would make the expressions too cumbersome;
    \item we will employ the Einstein summation convention, meaning that expressions of the form $A^iB_i$ should be interpreted as $\sum_i A^iB_i$;
    \item we will denote the dual of a vector space $V$ as $V^\vee$.
\end{itemize}

\textbf{Acknowledgements}
We thank A. S. Cattaneo, I. Contreras, P. Mnev and T. Willwacher for comments and remarks.
D. S. wants to thank P. Steffens for sending him his master thesis and for valuable comments.
This research was (partially) supported by the NCCR SwissMAP, funded by the Swiss National Science Foundation, and by the SNF grant No. 200020\_192080.
This paper is based on the master thesis of D. S. \cite{Sac21}.

\section{The BV-BFV formalism}
\label{sec:BV-BFV}

\subsection{Classical BV-BFV formalism}
In this section, we will recall the classical BV-BFV formalism as in \cite{CMR14}. We also refer to \cite{Mnev2019} for an excellent introduction to the BV formalism and to \cite{CattMosh1} for a more detailed exposition on the BV-BFV formalism.
\begin{defn}[BV manifold]
\label{BVmfld}
A \textit{BV manifold} is a triple $(\F, \omega, \mathcal{S})$ where $\F$ is a supermanifold\footnote{See \cite{BerLei75} for an original reference and \cite{CattaneoSchaetz2011,CattMosh1} for a concise introduction.} with $\mathbb{Z}$-grading, $\omega$ an odd symplectic form of degree $-1$ and $\mathcal{S}$ an even function of degree zero on $\F$, satisfying the \textit{Classical Master
Equation} (CME):
\begin{equation}
\label{cms}
    (\mathcal{S},\mathcal{S})=0.
\end{equation}
 We denote the odd Poisson bracket associated to $\omega$ with round brackets $(-,-)$. Usually we refer to $(-,-)$ as the \textit{BV bracket} (or \emph{anti bracket}).
\end{defn}
\begin{rmk}
Note that we have two different gradings, the $\mathbb{Z}_2$-grading from
the supermanifold structure and an additional $\mathbb{Z}$-grading. In a physics context, the $\mathbb{Z}$-grading corresponds to the ghost number $gh$ and the $\mathbb{Z}_2$-grading corresponds to the ``parity", which distinguishes bosonic and fermionic particles.
\end{rmk}

Equivalently, one may introduce the \textit{Hamiltonian vector field} $Q$ of $\mathcal{S}$, with ghost degree 1, defined by
\begin{equation}
\label{cms2}
        \iota_{Q}\omega=\delta \mathcal{S}  
 \end{equation}
   Furthermore, we require $Q$ to be a \textit{cohomological vector field}, i.e. 
   \begin{equation}
   \label{Q}
       [Q,Q]=0,
   \end{equation}
    where $[-,-]$ denotes the Lie bracket of vector fields.
    
\begin{rmk}
 Using symplectic geometry tools, one can show that Eq. \eqref{Q} and Eq. (\ref{cms2}) together imply $(S,S)=\text{constant}$, which for degree reasons reduces to Eq. (\ref{cms}).
This allows us then to write an equivalent definition of a BV manifold as a quadruple collection of data $(\F, \omega, \mathcal{S}, Q)$ satisfying Eq. \eqref{cms2} and Eq. (\ref{Q}).
\end{rmk}

 \begin{defn}[BV theory] A $d$-dimensional \emph{BV theory} on a closed manifold $M$ is the association of a BV manifold to every closed $d$ manifold $M$:
 \begin{equation}
     M \mapsto (\F_M,\omega_M,\mathcal{S}_M,Q_M)
 \end{equation}
 \end{defn}
 The natural question one might ask is how these definitions extend to the case of a manifold with boundary, which will also be the relevant case for our work.
 
 \begin{defn}[BV extension] A BV theory $(\F_M,\omega_M,\mathcal{S}_M,Q_M)$ is a \textit{BV extension} of a local field theory $M \mapsto (F_M, S_M)$ if for all $d$-manifolds $M$, the degree zero part $(\F_M)_0$ of $\F_M$ satisfies $(\F_M)_0=F_M$ and $\mathcal{S}_M\vert_{(\mathcal{F}_M)_0}=S_M$. In addition we want $\omega_M, \mathcal{S}_M$ and $Q_M$ to be local.
 \end{defn}
 Extending the BV formalism to manifolds with boundary amounts to considering its Hamiltonian counterpart, namely the BFV formalism.
 \begin{defn}[BFV manifold] A \emph{BFV manifold} is a triple $(\F^\partial, \omega^\partial,Q^\partial)$, where similarly as in Definition \ref{BVmfld}, $\F^\partial$ is a graded manifold, $\omega$ an even symplectic form of degree zero, and $Q^\partial$ a degree 1 cohomological vector field on $\F^\partial$. Moreover, if $\omega^\partial=\delta \alpha^\partial$, i.e. exact, the BFV manifold is called \textit{exact}.
 \end{defn}
 The result of merging BV and BFV formalism is encapsulated in the following definition:
 \begin{defn}[BV-BFV manifold, \cite{CMR14}]
 A \emph{BV-BFV manifold} over a given exact BFV manifold $\F^\partial=(\F^\partial,\omega^\partial=\delta\alpha^\partial,Q^\partial)$
is a quintuple
\begin{equation}
    \F=(\F,\omega,\mathcal{S},Q,\pi)
\end{equation}
 with $\pi:\F \rightarrow \F^\partial$ a surjective submersion obeying
 \begin{equation}
 \label{cme_boundary}
 \iota_Q\omega=\delta \mathcal{S}+\pi^*\alpha^\partial
 \end{equation}
 and $Q^\partial=\delta \pi Q$, where $\delta \pi$ is the differential of $\pi$.  
 \end{defn}
 \begin{rmk}
 Note that if $\F^\partial$ is a point then $(\F,\omega,\mathcal{S})$ is a BV manifold. 
 \end{rmk}
 We will adopt the short notation $\pi:\F \rightarrow \F^\partial$ for a BV-BFV manifold.
 
 We can now formulate a generalization of a BV theory:
 \begin{defn}[BV-BFV theory]
 A $d$-dimensional \textit{BV-BFV theory} associates to every closed
$(d-1)$-dimensional manifold $\Sigma$ a BFV manifold $\F^\partial_\Sigma$, and to a $d$-dimensional manifold $M$ with
boundary $\partial M$ a BV-BFV manifold $\pi_M:\F_M \rightarrow \F^\partial_{\partial M}$.
\end{defn}
\begin{rmk}
 For $Q$ a Hamiltonian vector field of $\mathcal{S}$, one can formally write
 \[(\mathcal{S},\mathcal{S})=\iota_Q\iota_Q \omega = Q(\mathcal{S}).\] In the case of a BV-BFV theory for a manifold $M$ with boundary $\partial M$, we have
 \begin{equation}
 Q(\mathcal{S})=\pi^*(2\mathcal{S}^\partial-\iota_{Q^\partial}\alpha^\partial).
 \end{equation}
 This can be phrased equivalently as
 \begin{equation}
     \iota_Q\iota_Q\omega=2\pi^*\mathcal{S}^\partial,
 \end{equation}
 which we will refer to as the \textit{modified Classical Master Equation} (mCME).
\end{rmk}
An important as well as classical example of a BV-BFV theory are the \textit{$BF$-like theories}.
\begin{defn}[$BF$-like theory]A BV-BFV theory is called \emph{$BF$-like}  if 
\begin{equation}
    \begin{split}
        \F_M&=\big(\Omega^\bullet(M) \otimes V[1]\big) \oplus \big(\Omega^\bullet (M)\otimes V^\vee[d-2]\big),\\
   \mathcal{S}_M&=\varint_M\big(\braket{B,dA} + \mathcal{V}(A,B)\big),
    \end{split}
 \end{equation}
 with $V$ a graded vector space, $\langle -, - \rangle$ a pairing between $V^\vee$ and $V$, and $\mathcal{V}$ a density-valued function of the fields $A$ and $B$, such that $\mathcal{S}_M$ satisfies the CME for $M$ without boundary.
\end{defn}
\begin{rmk}
\label{bv-bfv:rmk_bf_like}
 Equivalently, by picking up a graded basis $e^i$ for $V$ and $e_i$ for $V^\vee$, we may define a $BF$-like theory as a BV-BFV theory with 
\begin{equation}
    \begin{split}
        \F_M&=\big(\Omega^\bullet(M) \otimes V[k_i]\big) \oplus \big(\Omega^\bullet (M)\otimes V^\vee[d-k_i-1]\big),\\
   \mathcal{S}_M&=\varint_M\big(B_idA^i + \mathcal{V}(A,B)\big).
    \end{split}
 \end{equation}
 To pass from one definition to the other, it is sufficient to set $k_i=1-|e_i|$, where $|e_i|$ is the degree of $e_i$.
\end{rmk}

\subsection{Quantum BV-BFV formalism}
\label{sec_qbvbfv}
In this section, we introduce a perturbative quantization method for BV-BFV theories compatible with cutting and gluing. Originally, this procedure was proposed in \cite{CMR17} under the name of \textit{quantum BV-BFV formalism}. We start by defining what is a quantum BV-BFV theory and then we will explain how to produce such a theory by quantizing perturbatively a classical BV-BFV theory.

\begin{defn}[Quantum BV-BFV theory]
\label{def_Quantum BV-BFV theory}
Given a BV-BFV theory\footnote{The perturbative quantization scheme goes through if certain conditions are satisfied. In the following, we will be interested to $BF$-like theories for which this method works smoothly.}, a $d$-dimensional \textit{quantum BV-BFV theory} associates
\begin{itemize}
    \item To every closed $(d-1)$-dimension manifold $\Sigma$ a graded $\mathbb{C}[\![\hbar]\!]$-module $\mathcal{H}_{\Sigma}$, called the \textit{space of states}. 
    \item To every $d$-dimensional manifold (possibly with boundary) $M$:
    \begin{itemize}
        \item a degree 1 coboundary operator $\Omega\surM$ on $\mathcal{H}\surM$ called \textit{quantum BFV operator}. We call $\mathcal{H}\surM$ \textit{space of boundary states}
        \item a finite-dimensional BV manifold $\mathcal{V}_M$, called \textit{space of residual fields}. 
        \item a homogeneous element\footnote{Usually, the quantum state $\hat{\psi}_M$ will have degree 0. This is always the case when the gauge-fixing Lagrangian has degree 0, which is true for all the examples considered in this paper.}, the \textit{quantum state} $ \hat{\psi}_M\in \hat{\mathcal{H}}_M$. By denoting the space of half-densities on $\mathcal{V}_M$ as $\Dens^{\frac12}(\mathcal{V}_M)$, we define $\hat{\mathcal{H}}_M$ as $\hat{\mathcal{H}}_M\coloneqq\Dens^{\frac12}(\mathcal{V}_M)\otimes \mathcal{H}\surM$. It is a graded vector space endowed with two commuting boundary operators 
        \begin{equation}
            \hat\Omega_{\partial M}\coloneqq\Id \otimes \Omega\surM\quad \text{and}\quad \hat\Delta_{\mathcal{V}_M}\coloneqq\Delta_{\mathcal{V}_M}\otimes\Id,
        \end{equation}
        where $\Delta_{\mathcal{V}_M}$ is the canonical BV Laplacian on half-densities on residual fields. However, by abuse of notation, we will still write $\Omega_{\partial M}$ whenever we actually mean $\hat{\Omega}_{\partial M}$. The same is done for the BV Laplacian.\\
        We require the state to satisfy the \textit{modified Quantum Master Equation (mQME)}:
        \begin{equation}
        \label{ov:bv-bfv_mqme}
         (\hbar^2\Delta_{\mathcal{V}_M}+\Omega\surM)\hat{\psi}_M=0
        \end{equation}
    \end{itemize}
\end{itemize}
\end{defn}

In the following, we will refer to a quantum BV-BFV theory with the shorthand notation 
\[M\mapsto (\hat{\mathcal{H}}_M, \hat{\psi}_M,\Delta_{\mathcal{V}_M}, \Omega\surM).\]

\begin{rmk}
Since $\Delta_{\mathcal{V}_M}^2=0$, $\Omega_M$ and $\Delta_M$ endow $\hat{\mathcal{H}}_M$ with the structure of a bicomplex.
\end{rmk}

\begin{rmk}
Here we would like to precise the terminology used in Definition \ref{def_Quantum BV-BFV theory} by relating it to the literature. First of all, we call $\mathcal{H}_\Sigma$ space of fields because it is constructed by quantizing the symplectic manifold of boundary fields (as we will see below). An element of this space is thus called state. It is produced by integrating over bulk fields. However, following Wilson's ideas, it is useful to split the contribution of bulk fields into ``low energy" (or ``slow") fields, which we refer to as residual fields, and a complement (usually called ``high energy" or ``fluctuation" fields) on which we integrate over. Hence, our state will depend on both residual fields and boundary contribution. We have the following cases:
\begin{enumerate}
    \item in absence of residual fields, $\hat{\psi}_M$ is referred as state in \cite{Wit89},
    \item when $M$ is a cylinder, $\hat{\psi}_M$ is an evolution operator,
    \item in absence of boundaries and residual fields, $\hat{\psi}_M$ is referred as partition function (see \eqref{ov.tft:part_funct_1}),
    \item in the presence of both boundaries and residual fields, $\hat{\psi}_M$ will be a proper state only after we have integrated out the residual fields. We note that this is actually not always possible (see e.g. \cite{Mo20b} and references therein).
\end{enumerate}
Keeping in mind these possibilities, we still prefer to refer to $\psi_M$ as state. 
\end{rmk}

\begin{defn}[Equivalence]
\label{bv-bfv:def_equiv}
Two quantum BV-BFV theories $(\hat{\mathcal{H}}_M, \hat{\psi}_M,\Delta_{\mathcal{V}_M}, \Omega\surM)$\\ and $(\hat{\mathcal{H'}}_M, \hat{\psi'}_M,\Delta_{\mathcal{V'}_M}, \Omega'\surM)$ are \textit{equivalent} if for every manifold $M$ with boundary $\partial M$ there is a quasi-isomorphism of bicomplexes
\begin{equation}
    I_M: (\hat{\mathcal{H}}_M,\Delta_{\mathcal{V}_M}, \Omega\surM)\rightarrow (\hat{\mathcal{H'}}_M,\Delta_{\mathcal{V'}_M}, \Omega'\surM)
\end{equation}
such that $I_M(\hat{\psi}_M)=\hat{\psi'}_M$.
\end{defn}

\begin{defn}[Change of data]
\label{bv-bfv:def_change_data}
Two quantum BV-BFV theories $(\hat{\mathcal{H}}_M, \hat{\psi}_M,\Delta_{\mathcal{V}_M}, \Omega\surM)$ and $(\hat{\mathcal{H'}}_M, \hat\psi'_M,\Delta_{\mathcal{V'}_M}, \Omega'\surM)$ are related by change of data if there is an operator $\tau$ of degree $0$ on $\mathcal{H}\surM$ and an element $\chi\in\hat{\mathcal{H}}_M$ with $\deg(\chi)=\deg(\psi)-1$ such that
\begin{equation}
    \begin{split}
        \Omega'\surM&=[\Omega\surM, \tau],\\
        \hat\psi'_M&=(\hbar^2\Delta_{\mathcal{V}_M}+\Omega\surM)\chi-\hat{\tau}\hat{\psi}_M,
    \end{split}
\end{equation}
where $\hat{\tau}=\Id\otimes \tau$ is the extension of $\tau$ to $\hat{\mathcal{H}}_M$.
\end{defn}

\subsubsection{BV pushforward}
Let $(\mathcal{M}_1,\omega_1)$ and $(\mathcal{M}_2,\omega_2)$ two graded manifolds with odd symplectic forms $\omega_1$ and $\omega_2$ and canonical Laplacians $\Delta_1$ and $\Delta_2$, respectively. Consider $\mathcal{M}=\mathcal{M}_1 \times \mathcal{M}_2$ with symplectic form $\omega=\omega_1+\omega_2$ and canonical Laplacian $\Delta$. The space of half-densities on $\mathcal{M}$ factorizes as
\begin{equation}
    \Dens^{\frac12}(\mathcal{M})=\Dens^{\frac12}(\mathcal{M}_1)\hat{\otimes}\Dens^{\frac12}(\mathcal{M}_2).
\end{equation}
If we do a BV integration in the second factor, over the Lagrangian submanifold $\mathcal{L} \subset \mathcal{M}_2$, we are able to define a BV pushforward map on half-densities
\begin{equation}
    \varint_{\mathcal{L}}: \Dens^{\frac12}(\mathcal{M})\xrightarrow{\Id\otimes \varint_{\mathcal{L}}}\Dens^{\frac12}(\mathcal{M}_1).
\end{equation}
This map is also called \textit{fiber BV integral}, its properties are defined by the following theorem.

\begin{thm}[Batalin--Vilkovisky--Schwarz]
\label{bv_push}
Let $(\mathcal{M}_1,\omega_1)$ and $(\mathcal{M}_2,\omega_2)$ two graded manifolds with odd symplectic forms $\omega_1$ and $\omega_2$ and canonical Laplacian $\Delta_1$ and $\Delta_2$, respectively. Consider $\mathcal{M}=\mathcal{M}_1 \times \mathcal{M}_2$ with product symplectic form $\omega$ and canonical Laplacian $\Delta$ and let $\mathcal{L}, \mathcal{L}' \subset \mathcal{M}_2$ be any two Lagrangian submanifolds which can be deformed into each-other. For any half-density $f \in \mathrm{Dens}^{\frac12}(\mathcal{M})$ one has:
\begin{enumerate}
    \item $\varint_{\mathcal{L}}\Delta f=\Delta_1 \varint_{\mathcal{L}}f$
    \item $\varint_{\mathcal{L}}f-\varint_{\mathcal{L'}}f=\Delta_1 \xi$ for some $\xi \in \mathrm{Dens}^{\frac12}(\mathcal{M}_1)$, if $\Delta f=0$.
\end{enumerate}
\end{thm}

\subsubsection{Summary}
Let us explain here how to construct a quantum BV-BFV theory. Consider a classical BV-BFV theory $\pi: \mathcal{F}_M\rightarrow\mathcal{F}^{\partial}_{\partial M}$. Note that from now on we will assume $\mathcal{F}_M$ and $\mathcal{F}^{\partial}_{\partial M}$ to be vector spaces. This will be the case when we will quantize the \textit{globalized split RW theory}. 

The main steps can be summarized as follows
\begin{enumerate}
\item[(i)]\textbf{(Geometric Quantization)} Given a $(d-1)$-manifold $\Sigma$, the BV-BFV theory associates to it a symplectic manifold $(\mathcal{F}^{\partial}_\Sigma, \omega^{\partial}_\Sigma, Q^{\partial}_\Sigma$). The idea here is to construct the space of states $\mathcal{H}_\Sigma$ and the quantum BFV operator $\Omega_\Sigma$ as a \textit{geometric quantization}\footnote{For an introduction to geometric quantization see e.g. \cite{BW97}.} of this symplectic vector space\footnote{In fact, for the case when $\Sigma$ is given by the boundary of another manifold, the BFV operator is constructed by the methods of \emph{deformation quantization} as it was also pointed out in \cite{Moshayedi2021}.}.
In order to accomplish such task, we require the data of a polarization $\mathcal{P}$ on this symplectic vector space, in particular, we consider real fibrating polarizations. Then, it is sufficient to split $\mathcal{F}^{\partial}_\Sigma$ into Lagrangian subspaces as
\begin{equation}
    \mathcal{F}^{\partial}_\Sigma\cong \mathcal{B}^{\mathcal{P}}_{\Sigma}\times \mathcal{K}^{\mathcal{P}}_\Sigma,
\end{equation}
with $\mathcal{K}^{\mathcal{P}}_\Sigma$ thought as a Lagrangian distribution on $\mathcal{F}^{\partial}_\Sigma$ and $\mathcal{B}^{\mathcal{P}}_{\Sigma}$ identified with the leaf space of the polarization, i.e. $\mathcal{B}^{\mathcal{P}}_{\Sigma}=\mathcal{F}^{\partial}_\Sigma/\mathcal{P}$. If we assume the 1-form $\alpha^{\partial}_\Sigma$ to vanish along $\mathcal{P}$ and in the case of real polarization, the space of states $\mathcal{H}_{\Sigma}$ is modeled as a space of complex-valued functionals on $\mathcal{B}^{\mathcal{P}}_{\Sigma}$ (or more generally $\mathcal{H}_{\Sigma}$ is the space of polarized sections of the trivial ``prequantum" line bundle over $\mathcal{F}^{\partial}_{\Sigma}$). This means that the space of states is obtained as a geometric quantization of the space of boundary fields as we preannounced above. 
On the other hand, when $\alpha^{\partial}_\Sigma   \Big|_{\mathcal{P}}\neq 0$, we can use a gauge transformation and modify $\alpha^{\partial}_\Sigma$ by an exact term $\delta f^{\mathcal{P}}_\Sigma$, with $f^{\mathcal{P}}_\Sigma$ a local functional. Consequently, if we assume from now on $\Sigma=\partial M$, to preserve Eq. (\ref{cme_boundary}), we change $\Sc$ by a boundary term obtaining $\Sc^{\mathcal{P}}$. In this case, with $\Sc^{\mathcal{P}}$ and $\alpha^{\mathcal{P}}_{\partial M}$, we have a new BV-BFV manifold.

   
    \item[(ii)] \textbf{(Extraction of boundary fields)} The aim is to split bulk and boundary field contributions in the space of fields $\mathcal{F}_M$. We proceed as follows: consider the projection $\mathcal{F}^{\partial}_{\partial M}\xrightarrow{p^{\mathcal{P}}_{\partial M}} B^{\mathcal{P}}_{\partial M}$, we have a surjective submersion
    \begin{equation}
        p^{\mathcal{P}}_{\partial M}\circ \pi:\mathcal{F}_M\rightarrow B^{\mathcal{P}}_{\partial M}.
    \end{equation}
    Assume we can choose a section $\sigma$ on $\mathcal{F}_M\rightarrow B^{\mathcal{P}}_{\partial M}$ such that we can split
    \begin{equation}
    \label{ov:bv-bfv_splitting}
        \mathcal{F}_M\cong \sigma(B^{\mathcal{P}}_{\partial M}) \otimes \mathcal{Y}.
    \end{equation}
    The space $\sigma(B^{\mathcal{P}}_{\partial M})$ is a bulk extension of $B^{\mathcal{P}}_{\partial M}$ which we denote as $\tilde{B}^{\mathcal{P}}_{\partial M}$.
    This splitting is subjected to the following assumption\footnote{In order to accomplish this assumption, we are forced to choose singular extensions of boundary fields. The boundary fields are thus extended by 0 to the bulk.}:
    \begin{assump}
    \label{ov:bv-bfv_ass1}
    There is a weakly symplectic form $\omega_{\mathcal{Y}}$ on $\mathcal{Y}$ such that $\omega_M$ is the extension of $\omega_{\mathcal{Y}}$ to $\mathcal{F}_M$.
    \end{assump}
    In the splitting (\ref{ov:bv-bfv_splitting}), the space $\mathcal{Y}$ is a complement of $\Tilde{B}^{\mathcal{P}}_{\partial M}$ which is interpreted as the space of \textit{bulk fields} (while $\Tilde{B}^{\mathcal{P}}_{\partial M}$ is thought of as the space of \textit{boundary fields} extended to the bulk). 
    \item[(iii)] \textbf{(Construction of $\Omega_{\partial M}$)} As a result of the geometric quantization procedure, $\mathcal{H}_{\partial M}$ is a cochain complex. Following the same line of thought, we construct the coboundary operator $\Omega_{\partial M}$ as quantization of the boundary actions $\Sc^{\partial}_{\partial M}$. We can proceed as follows. Assume we have Darboux coordinates $(q,p)$ on $\mathcal{F}^{\partial}_{\partial M}$. In particular $q$ are coordinates on $\mathcal{B}^{\mathcal{P}}_{\partial M}$ and $p$ are coordinates on the fiber $p^{\mathcal{P}}_{\partial M}:\mathcal{F}^{\partial}_{\partial M}\rightarrow B^{\mathcal{P}}_{\partial M}$, which is still part of $\mathcal{Y}$. We define $\Omega_{\partial M}$ as the standard-ordering quantization of $\Sc^{\partial}_{\partial M}$:
    \begin{equation}
    \label{omega_std_ordering}
        \Omega_{\partial M}\coloneqq \Sc^{\partial}_{\partial M}\bigg(q, -i\hbar \frac{\partial}{\partial q}\bigg),
    \end{equation}
    where all the derivatives are positioned on the right. 
    \item[(iv)] \textbf{(Choice of residual fields)} 
    We further split the bulk contributions in $\mathcal{Y}$ into residual fields and a complement $\mathcal{Y}'$, which represents the space of \textit{fluctuation fields} (also called ``high-energy" or ``fast" fields). This means, we choose a splitting
    \begin{equation}
        \mathcal{Y}\cong \mathcal{V}^{\mathcal{P}}_M\times \mathcal{Y}'
    \end{equation}
    which depends on the boundary polarization and satisfies the following assumption:
    \begin{assump}
    \label{ov:bv-bfv_ass2}
    The following holds:
    \begin{enumerate}
        \item[(1)] $\mathcal{V}^{\mathcal{P}}_M$ and $\mathcal{Y}'$ are BV manifolds,
        \item[(2)] $\mathcal{V}^{\mathcal{P}}_M$ is finite-dimensional,
        \item[(3)] The symplectic form splits as $\omega_{\mathcal{Y}}=\omega_{\mathcal{V}^{\mathcal{P}}_M}+\omega_{\mathcal{Y}'}$.
    \end{enumerate}
    \end{assump}
    Usually, the space $\mathcal{V}^{\mathcal{P}}_M$ is chosen as the space of solutions of $\delta \Sc^0_M=0$ modulo gauge transformations, where $\Sc^0_M$ is the quadratic part of the action $\Sc_M$. This is called minimal choice, and we refer to this space as the space of \textit{zero modes}. Other choices are possible and they are all related by the equivalence relations defined above (see Definition \ref{bv-bfv:def_equiv} and Definition \ref{bv-bfv:def_change_data}). Finally, we sum up the last two bullet points with the following definition:
    \begin{defn}[Good splitting, \cite{CMoW19}]
    \label{bv-bfv_good_split}
    A splitting 
    \begin{equation}
        \mathcal{F}_M\cong \mathcal{B}^{\mathcal{P}}_{\partial M}\times \mathcal{V}^{\mathcal{P}}_M\times \mathcal{Y}'
    \end{equation}
    is called \textit{good} if it satisfies Assumptions \ref{ov:bv-bfv_ass1} and \ref{ov:bv-bfv_ass2}.
    \end{defn}
    
    Given a good splitting, an element $\mathbf{X}$ of $\mathcal{F}_M$ is written accordingly as $\mathbf{X}=\mathbb{X}+\mathsf{x}+\xi$.
    \item[(v)] \textbf{(The state)} When we have a good splitting, the gauge-fixing consists of choosing a Lagrangian $\mathcal{L}\subset \mathcal{Y}'$. Set $\mathcal{Z}_M=\mathcal{B}^{\mathcal{P}}_{\partial M}\times \mathcal{V}^{\mathcal{P}}_M$ (\textit{bundle of residual fields} over $\mathcal{B}^{\mathcal{P}}_{\partial M}$). Then, we define $\hat{\mathcal{H}}^{\mathcal{P}}_M\coloneqq\Dens^{\frac12}(\mathcal{Z}_M)= \Dens^{\frac12}(\mathcal{V}^{\mathcal{P}}_M)\times \Dens^{\frac12}(\mathcal{B}^{\mathcal{P}}_{\partial M})$ and the BV Laplacian $\hat{\Delta}_{\mathcal{V}^{\mathcal{P}}_M}\coloneqq\Id\otimes \Delta_{\mathcal{V}^{\mathcal{P}}_M}$ (as before the hat will be omitted in the following). 
    The state is then defined as a BV pushforward of the exponential of the bulk action
    \begin{equation}
    \label{ov.bv-bfv_state}
        \hat{\psi}_M(\mathbb{X},\mathsf{x})\coloneqq\varint_{\xi\in \mathcal{L}}e^{\frac{i}{\hbar}\Sc_M(\mathbb{X}+\mathsf{x}+\xi)}.
    \end{equation}
    Moreover, if $\Delta_{\mathcal{Y}}\Sc^{\mathcal{P}}_M=0$, as a consequence of Theorem \ref{bv_push}:
    \begin{itemize}
        \item $\hat{\psi}_M$ is closed under the coboundary operator $\hbar^2\Delta_{\mathcal{V}^{\mathcal{P}}_M}+\Omega\surM$, i.e. Eq. (\ref{ov:bv-bfv_mqme}) holds,
        \item the state does not change under smooth deformation of the gauge-fixing Lagrangian $\mathcal{L}$ used in the BV pushforward up to $(\hbar^2\Delta_{\mathcal{V}^{\mathcal{P}}_M}+\Omega\surM)$-exact terms. 
    \end{itemize}
    \item[(vi)] \textbf{(Perturbative expansion)} The procedure detailed so far is valid for finite-dimensional situations. However, the space of fields $\mathcal{F}_M$ is usually infinite-dimensional, since, for example, it can contain the de Rham complex of differential forms over $\partial M$. As a result, the integral in Eq. (\ref{ov.bv-bfv_state}) is ill-defined. To fix this problem, we \textit{define} the integral perturbatively, i.e. as formal power series in $\hbar$ with coefficients given by sums of Feynman diagrams. For the perturbative expansion to be well-defined, we need the following assumption to be satisfied:
    \begin{assump}
    \label{ov.bv-bfv_ass3}
    The restriction of the action $\Sc^{\mathcal{P}}_M$ to $\mathcal{L}$ has isolated critical points.
    \end{assump}
    We note that this does not hold for every Lagrangian.
    
    \begin{rmk}
    It is important to highlight that, for Assumption \ref{ov.bv-bfv_ass3} to be satisfied, we need to choose carefully the residual fields. The problem here is represented by \textit{zero modes} $\mathcal{V}^0_M$, which can be present in the quadratic part of the bulk action. The zero modes are bulk fields configurations that are annihilated by the kinetic operator and correspond to the tangent directions to the Euler--Lagrange moduli space (solutions of $\delta \Sc^0_M=0$ modulo gauge transformation). Hence, their presence implies non-isolated critical points in the action: the perturbative expansion is obstructed. To solve this situation, we need the space of residual fields to at least contain the space of zero modes, i.e. $\mathcal{V}^0_M\subseteq \mathcal{V}_M$. In this way, we can obtain a good gauge-fixing Lagrangian, which satisfies Assumption \ref{ov.bv-bfv_ass3}. We call it \textit{minimal choice} (or \textit{minimal realization} of the state)\footnote{We have a non-minimal realization when $\mathcal{V}^0_M\subset \mathcal{V}_M$. In that case, we can pass from a non-minimal realization to a smaller one by BV pushforward, which can be interpreted as a sort of \textit{renormalization group flow} \cite{Ir18}.}, when $\mathcal{V}_M\cong\mathcal{V}^0_M$. 
    \end{rmk}
    
    \noindent When we pass to the infinite case, another problem arises: the BV Laplacian is ill-defined. Therefore, every equation containing it is only \textit{formal}. In this regard, Theorem \ref{bv_push} has only been proven in the finite-dimensional setting. Hence, we can not conclude that the mQME is satisfied even if the action is formally annihilated by the Laplacian. The mQME has to be verified for each theory at the level of Feynman diagrams.  
    In this paper, we add the globalized RW theory to the class of $BF$-like theories for which the mQME has been proven in the infinite-dimensional perturbative setting. The proof relies on Stokes' Theorem for integrals over compactified configuration spaces.
\end{enumerate}

\subsection{Quantum states in $BF$-like theories} 
\label{bv-bfv_sub_qs_bf_like}
In $BF$-like theories one can define the quantum state in a perturbative way using  Feynman graphs via integrals defined on the configuration space of these graphs. Two convenient choices of polarizations in $BF$-like theories are the $\frac{\delta}{\delta \mathbf{A}}$- and the $\frac{\delta}{\delta \mathbf{B}}$-polarization. Concretely, we fix a polarization  by splitting the boundary $\partial M$ into two parts $\partial_1 M$ and $\partial_2M$, where we choose the polarization  $\frac{\delta}{\delta \mathbf{B}}$ on $\partial_1M$ and $\frac{\delta}{\delta \mathbf{A}}$ on $\partial_2M$. The associated space of leaves for the $A$-leaf and $B$-leaf are denoted by $\mathbb{A} \in \mathcal{B}_{\partial M}^{\frac{\delta}{\delta \mathbf{B}}}$ and $\mathbb{B} \in \mathcal{B}_{\partial M}^{\frac{\delta}{\delta \mathbf{A} }}$ respectively.

For $BF$-like theories, the first splitting determined by the polarization is
\begin{equation}
\begin{split}
    \mathcal{B}^{\mathcal{P}}_{\partial M}&=\big(\Omega^\bullet (\partial_1M)\otimes V[1]\big) \oplus \big(\Omega^\bullet(\partial_2M)\otimes V^{\vee}[d-2]\big), \\
    \mathcal{Y}&=\big(\Omega^\bullet(M,\partial_1M) \otimes V[1]\big) \oplus \big(\Omega^\bullet(M,\partial_2M)\otimes V^{\vee}[d-2]\big).
\end{split}
\end{equation}
The minimal space of residual fields is
\begin{equation}
    \mathcal{V}^{\mathcal{P}}_M \cong \big(H^\bullet (M,\partial_1M) \otimes V[1]\big) \oplus (H^\bullet \big(M,\partial_2M)\otimes V^{*}[d-2]\big)  
\end{equation}
for $V$ some graded vector space. One way to get a good splitting is then by considering a splitting of the complex of de Rham forms with relative boundary conditions into a subspace $\mathcal{V}^{\mathcal{P}}_M$ isomorphic to cohomology and a complementary space $\mathcal{Y}'$ in a way compatible with the symplectic structure. This can be done by using a Riemannian metric and embed the cohomology as harmonic forms. 
As a result, the space of fields $\F_M$ splits as $\F_M=\mathcal{B}^{\mathcal{P}}_{\partial M} \times \mathcal{V}^{\mathcal{P}}_M \times \mathcal{Y}'$, where an element $(\mathbf{A}, \mathbf{B}) \in \F_M$ is given by
\begin{equation}
\begin{split}
    \mathbf{A}&=\mathbb{A} + \underline{\mathbf{A}}=\mathbb{A} +\mathsf{a} + \alpha,\\
    \mathbf{B}&=\mathbb{B}+\underline{\mathbf{B}}=\mathbb{B} + \mathsf{b} + \beta.
\end{split}
\end{equation}

There is one last ingredient that we need to introduce before defining the quantum state, namely the \textit{composite fields}. We denote them by square brackets $[  \hspace{2mm}]$, i.e. for a boundary field $\mathbb{A}$ we have $[\mathbb{A}^{i_1} \dots \mathbb{A}^{i_k}]$. One can think of them as a \textit{regularization}  of higher functional derivatives, in the sense 
that a higher functional $\frac{\delta^k}{\delta \mathbb{A}^{i_1}\dots \delta\mathbb{A}^{i_k}}$ is replaced by a first order functional derivative $\frac{\delta^k}{[\delta \mathbb{A}^{i_1}\dots \delta\mathbb{A}^{i_k}]}$. For further details see \cite{CMR17}.


\begin{defn}[Regular functional]
\label{regularfct}
A \textit{regular functional} on the space of base boundary fields is a linear combination of expressions of the form
\begin{equation}
\label{ov:bv-bfv.reg_func}
    \begin{split}
        \varint_{\mathrm{C}_{m_1}(\partial_1M)\times\mathrm{C}_{m_2}(\partial_2M)} L^{J^1_1\dots J^{l_1}_1\dots J^1_2\dots J^{l_2}_2\dots}_{I^1_1\dots I^{r_1}_1\dots I^1_2\dots I^{r_2}_2\dots}\wedge \pi^*_1\prod^{r_1}_{j=1}\bigg[\mathbb{A}^{I^j_1}\bigg]\wedge \dots&\wedge \pi^*_{m_1}\prod^{r_{m_1}}_{j=1}\bigg[\mathbb{A}^{I^j_{m_1}}\bigg]\wedge \dots \\
        &\wedge \pi^*_1\prod^{l_1}_{j=1}\bigg[\mathbb{B}_{J^j_1}\bigg]\wedge \dots \wedge \pi^*_{m_1}\prod^{l_{m_2}}_{j=1}\bigg[\mathbb{B}_{J^j_{m_2}}\bigg],
    \end{split}
\end{equation}
where $I^j_i$ and $J^j_i$ are (target) multi-indices and $L^{J^1_1\dots J^{l_1}_1\dots J^1_2\dots J^{l_2}_2\dots}_{I^1_1\dots I^{r_1}_1\dots I^1_2\dots I^{r_2}_2\dots}$ is a smooth differential form on the direct product of compactified configuration spaces  $\mathrm{C}_{m_1}(\partial_1M)\times \mathrm{C}_{m_2}(\partial_2M)$, which depends on the residual fields. A regular functional is called \textit{principal} if all multi-indices have length 1. For more details on configuration spaces and configuration space integrals we refer to \cite{Kontsevich1994,BC,CamposIdrissiLambrechtsWillwacher2018} (see also Remark \ref{rem:conf}).
\end{defn}

\begin{defn}[Full space of boundary states]
The \textit{full space of boundary states} $\mathcal{H}^{\mathcal{P}}_{\partial M}$ is given by the linear combination of regular functionals of the form (\ref{ov:bv-bfv.reg_func}).
\end{defn}
\begin{defn}[Principal space of boundary states] The \textit{principal space of boundary
states} $\mathcal{H}_{\partial M}^{\mathcal{P}, \text{princ}}$ is defined as the subspace of $H^{\mathcal {P}}_{\partial M}$
, where we only consider principal regular functionals. 
\end{defn}
We use Feynman rules and graphs to define the state. Let us elaborate them in the BV-BFV setting (for perturbations of abelian $BF$-theory).
\begin{defn}[($BF$) Feynman graph]
A \textit{(BF) Feynman graph} is an oriented graph with three types of vertices $V(\Gamma)=V_{\text{bulk}}(\Gamma)\sqcup V_{\partial_1}\sqcup V_{\partial_2}$, called bulk vertices and type 1 and 2 boundary vertices, such that
\begin{itemize}
    \item bulk vertices can have any valence,
    \item type 1 boundary vertices carry any number of incoming half-edges (and no outgoing half-edges),
    \item type 2 boundary vertices carry any number of outgoing half-edges (and no incoming half-edges),
    \item multiple edges and loose half-edges (leaves) are allowed.
\end{itemize}
\end{defn}
A labeling of a Feynman graph is a function from the set of half-edges to $\{1,\dots,\dim V\}$.
\begin{defn}[Principal graph]A Feynman graph is called \textit{principal} if all boundary vertices (type 1 and type 2) are univalent or zero valent.
\end{defn}
 Let $\Gamma$ be a Feynman graph and $M$ a manifold with boundary $\partial M=\partial_1 M \sqcup \partial_2M$ and define
\begin{equation}
    \mathrm{Conf}_\Gamma(M)\coloneqq \mathrm{Conf}_{V_{\text{bulk}}}(M) \times \mathrm{Conf}_{V_{\partial_1}}(\partial_1 M) \times \mathrm{Conf}_{V_{\partial_2}}(\partial_2 M).
\end{equation}
The Feynman rules are given by a map associating to a Feynman graph $\Gamma$ a differential form $\omega_\Gamma \in \Omega^\bullet(\mathrm{Conf}_\Gamma(M))$.
\begin{defn}[($BF$) Feynman rules] Let $\Gamma$ be a labeled Feynman graph. We choose a configuration $\iota:V(\Gamma) \rightarrow \mathrm{Conf} (\Gamma)$, such that decompositions are respected. Then, we \textit{decorate} the graph according to the following rules, namely, the \textit{Feynman rules}:
\begin{itemize}
    \item Bulk vertices in $M$ decorated by ``vertex tensors"
    \begin{equation}
        \mathcal{V}^{i_1\dots i_s}_{j_1\dots j_t} \coloneqq \frac{\partial^{s+t}}{\partial \underline{\mathbf{A}}^{i_1}\dots\partial \underline{\mathbf{A}}^{i_s} \partial \underline{\mathbf{B}}_{j_1}\dots \partial \underline{\mathbf{B}}_{j_t}} \bigg|_{\underline{\mathbf{A}}=\underline{\mathbf{B}}=0} \mathcal{V}(\underline{\mathbf{A}},\underline{\mathbf{B}}),
    \end{equation}
    where $s, t$ are the out- and in-valencies of the vertex and $i_1, \dots, i_s$ and $j_1, \dots, j_t$ are the labels of the out- and in- oriented half-edges and $\mathcal{V}(\underline{\mathbf{A}},\underline{\mathbf{B}})$ is the interaction term in a $BF$-like theory.

\item Boundary vertices $v \in V_{\partial_1}(\Gamma)$ with incoming half-edges labeled $i_1, \dots, i_k$ and no out-going  half-edges are decorated by a composite field $[\mathbb{A}^{i_1} \dots \mathbb{A}^{i_k}]$ evaluated at the point (vertex location) $\iota(v)$ on $\partial_1 M$.
\item Boundary vertices $v \in V_{\partial_2}$ on $\partial_2 M$ with outgoing half-edges labeled $j_1 \dots j_l$ and no in-going  half-edges are decorated by $[\mathbb{B}_{j_1} \dots \mathbb{B}_{j_l}]$ evaluated at the point on $\partial_2 M$. 
\item Edges between vertices $v_1, v_2$ are decorated with the propagator $\eta (\iota(v_1),\iota(v_2))\cdot \delta^i_j$, with $\eta$ the propagator induced by $\mathcal{L} \subset \mathcal{Y}'$, the gauge-fixing Lagrangian.
\item Loose half-edges (leaves) attached to a vertex $v$ and labeled $i$ are decorated with the residual
fields $\mathsf{a}_i$ (for out-orientation), $\mathsf{b}^i$
(for in-orientation) evaluated at the point $\iota(v)$.
\end{itemize}

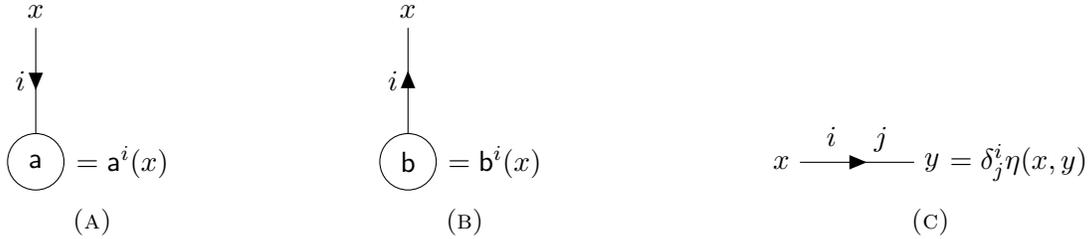
\begin{figure}[hbt!]
    \centering
    \begin{subfigure}[b]{0.25\textwidth}
    \centering
\begin{tikzpicture}
  \begin{feynman}[every blob={/tikz/fill=white!30,/tikz/inner sep=1pt}]
  \vertex[blob] (m) at (0,-2) {$\mathsf{a}$};
  \vertex (a) at (0,0) {$x$} ;
 \diagram*{
      (a) -- [fermion, edge label' = $i$] (m)
    };
  \vertex [right=3em of m] {\(=\mathsf{a}^i(x)\)};
  \end{feynman}
\end{tikzpicture}
\caption{}
        \label{}
    \end{subfigure}%
    \hfill
\begin{subfigure}[b]{0.25\textwidth}
\centering
    \begin{tikzpicture}
  \begin{feynman}[every blob={/tikz/fill=white!30,/tikz/inner sep=0.5pt}]
  \vertex[blob] (m) at (0,-2) {$\mathsf{b}$};
  \vertex (a) at (0,0) {$x$} ;
 \diagram*{
      (m) -- [fermion, edge label = $i$] (a)
    };
     \vertex [right=3em of m] {\(=\mathsf{b}^i(x)\)};
  \end{feynman}
\end{tikzpicture}
  \caption{}
    \label{}
\end{subfigure}%
\hfill
\begin{subfigure}[b]{0.4\textwidth}
\centering
    \begin{tikzpicture}
  \begin{feynman}[every blob={/tikz/fill=white!30,/tikz/inner sep=0.5pt}]
  \vertex (a) at (0,0) {$x$};
  \vertex (b) at (2,0) {$y$} ;
 \diagram*{
      (a) -- [fermion, edge label = $i\hspace{5mm} j$] (b)
    };
     \vertex [right=3em of b] {\(=\delta^i_j\eta(x,y)\)};
  \end{feynman}
\end{tikzpicture}
  \caption{}
    \label{}
\end{subfigure}
    \caption{Feynman rules for residual fields and propagator.}
    \label{bv-bfv:fig_feyn_rules1}
\end{figure}

\begin{figure}[hbt!]
    \centering
    \begin{subfigure}[b]{0.3\textwidth}
    \centering
       \begin{tikzpicture}
  \begin{feynman}[every blob={/tikz/fill=white!30,/tikz/inner sep=0.5pt}]
  \vertex (a) at (-1,0);
  \vertex (b) at (1,0);
  \vertex (m1) at (0, 1);
 
 \diagram*{
      (a) -- m [dot] -- (b),
      (m1) -- [fermion] m
    };
  \vertex [below=0.75em of m] {\(\mathbb{B}\)};
  \vertex [left=0.25em of a] {\(\partial_2 \Sigma_3\)};
  \end{feynman}
\end{tikzpicture}
\caption{}
        \label{}
    \end{subfigure}%
    \quad
\begin{subfigure}[b]{0.3\textwidth}
\centering
    \begin{tikzpicture}
  \begin{feynman}[every blob={/tikz/fill=white!30,/tikz/inner sep=0.5pt}]
  \vertex (a) at (-1,0);
  \vertex (b) at (1,0);
  \vertex (m1) at (0, -1);
 
 \diagram*{
      (a) -- m [dot] -- (b),
      m -- [fermion] (m1)
    };
  \vertex [above=0.75em of m] {\(\mathbb{A}\)};
  \vertex [left=0.25em of a] {\(\partial_1 \Sigma_3\)};
  \end{feynman}
\end{tikzpicture}
  \caption{}
    \label{}
\end{subfigure}%
\hfill
\begin{subfigure}[b]{0.42\textwidth}
\centering
        \begin{tikzpicture}
  \begin{feynman}[every blob={/tikz/fill=white!30,/tikz/inner sep=0.5pt}]
  \vertex (a) at (-1.2, 0.35);
  \vertex (b) at (-0.3, 0.79);
  \vertex (c) at (+1.2, 0.4);
  \vertex (e) at (-1.2, -0.35);
  \vertex (f) at (-0.3, -0.79);
  \vertex (g) at (+1.2, -0.4);
  \diagram*{
  (a) -- [fermion] m [dot],
  (b) -- [fermion] m [dot],
  (c) -- [fermion] m [dot],
  (e) -- [anti fermion] m [dot],
  (f) -- [anti fermion] m [dot],
  (g) -- [anti fermion] m [dot],
  }; 
 \vertex [right=7em of m] {\(=\mathcal{V}^{i_1\dots i_s}_{j_1\dots j_t}\)};
 
 \vertex [above=0.2em of m, label=80:\(\dots\)] {}; 
 \vertex [] at (-1.35,0.4) {\(j_1\)};
  \vertex [] at (-0.45,1) {\(j_2\)};
 \vertex [] at (1.4, 0.4) {\(j_t\)};
  \vertex [] at (-1.35,-0.4) {\(i_1\)};
  \vertex [] at (-0.45,-1) {\(i_2\)};
 \vertex [] at (1.4, -0.4) {\(i_s\)};
  \end{feynman}
  \end{tikzpicture}
  \caption{}
  \label{}
    \end{subfigure}%
    \caption{Feynman rules for boundary fields and interaction vertices.}
    \label{bv-bfv:fig_feyn_rules2}
\end{figure}
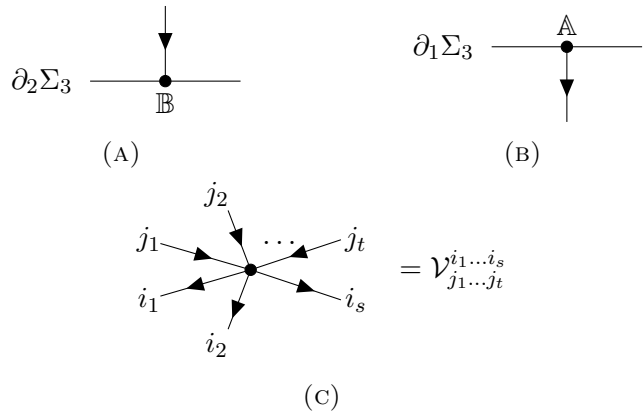

\begin{figure}[hbt!]
    \centering
    \begin{subfigure}[b]{0.3\textwidth}
    \centering
       \begin{tikzpicture}
  \begin{feynman}[every blob={/tikz/fill=white!30,/tikz/inner sep=0.5pt}]
  \vertex (a) at (-2,0);
  \vertex (b) at (2,0);
  \vertex (d) at (0.4, 0.75);
  \vertex (e) at (1, 0.75);
  \vertex (f) at (-1, 0.75);
 \diagram*{
      (a) -- m [dot] -- (b),
      (d) -- [anti fermion] m,
      (e) -- [anti fermion] m,
      (f) -- [anti fermion] m
    };
  \vertex [below=0.75em of m] {\([\mathbb{A}^{i_1}\dots \mathbb{A}^{i_k}]\)};
  \vertex [left=0.25em of a] {\(\partial_1 \Sigma_3\)};
  \node at (-0.2, 0.5) {$\dots$};
  \node at (-1.1, 0.6) {$i_k$};
  \node at (0.5, 1) {$i_2$};
  \node at (1.2, 0.9) {$i_1$};
  \end{feynman}%
\end{tikzpicture}
\caption{}
        \label{fig:comp_fields}
    \end{subfigure}%
    \qquad \begin{subfigure}[b]{0.3\textwidth}
    \centering
       \begin{tikzpicture}
  \begin{feynman}[every blob={/tikz/fill=white!30,/tikz/inner sep=0.5pt}]
  \vertex (a) at (-2,0);
  \vertex (b) at (2,0);
  \vertex (d) at (0.4, 0.75);
  \vertex (e) at (1, 0.75);
  \vertex (f) at (-1, 0.75);
 \diagram*{
      (a) -- m [dot] -- (b),
      (d) -- [fermion] m,
      (e) -- [fermion] m,
      (f) -- [fermion] m
    };
  \vertex [below=0.75em of m] {\([\mathbb{B}_{j_1}\dots \mathbb{B}_{j_l}]\)};
  \vertex [left=0.25em of a] {\(\partial_2 \Sigma_3\)};
  \node at (-0.2, 0.5) {$\dots$};
  \node at (-1.1, 0.6) {$j_l$};
  \node at (0.5, 1) {$j_2$};
  \node at (1.2, 0.9) {$j_1$};
  \end{feynman}
\end{tikzpicture}
\caption{}
        \label{}
    \end{subfigure}%
    \caption{Feynman rules for the composite fields.}
    \label{fig:comp_fields_old}
\end{figure}
The differential forms given by the decorations are denoted collectively by $\omega_d$. The differential form $\omega_\Gamma$ at $\iota$ is then defined by the product of all decorations and summing over all labels:
\begin{equation}
    \omega_\Gamma=\sum_{\text{labellings of $\Gamma$}}\;\,\, \prod_{{\text{decorations $d$ of $\Gamma$}}} \omega_d.
\end{equation}
The Feynman rules are represented in Figs. \ref{bv-bfv:fig_feyn_rules1}, \ref{bv-bfv:fig_feyn_rules2} and \ref{fig:comp_fields_old}.
\end{defn}
\begin{rmk}[Configuration spaces]\label{rem:conf}
We will exploit the Fulton--MacPherson/Axelrod--Singer compactification of configuration spaces on manifolds with boundary (FMAS compactification \cite{FM94,AS94}). Axelrod and Singer proved via non-trivial analytic tools, that the propagator, a priori defined only on the open configuration space $\mathrm{Conf}_2 (M)$, extends to the compactification $\mathrm{C}_2(M)$. This implies that also $\omega_\Gamma$, for all Feynman graphs $\Gamma$, extends to the compactification $\mathrm{C}_\Gamma (M)$ of $\mathrm{Conf}_\Gamma(M)$. But, adding strata of lower codimension leaves the integrals unchanged, and hence this proves that the integrals appearing in Eq. \eqref{principalqs} below are finite. In addition, one can exploit the combinatorics of the stratification for various computations using Stokes' theorem.
\end{rmk}
\begin{defn}[Principal quantum state]
Let $M$ be a manifold (with boundary). Given a $BF$-like  BV-BFV theory $\pi_M:\F_M \rightarrow \F^\partial_{\partial M}$, a polarization $\mathcal{P}$ on $\F^\partial_{\partial M}$, a good splitting $\F_M=\mathcal{B}^{\mathcal{P}}_{\partial M} \times \mathcal{V}^{\mathcal{P}}_M \times \mathcal{Y}'$ and $\mathcal{L} \subset \mathcal{Y}'$, the gauge-fixing Lagrangian,  we can define the \textit{principal part of the quantum state} by the formal power series
\begin{equation}
\label{principalqs}
    \hat{\psi}_{M}(\mathbb{A},\mathbb{B}; \mathsf{a},\mathsf{b}) \coloneqq T_{M}\exp\bigg(\frac{i}{\hbar}\sum_{\Gamma}\frac{(-i\hbar)^{\loops(\Gamma)}}{|\Aut(\Gamma)|}\varint_{\text{C}_\Gamma(M)}\omega_{\Gamma}(\mathbb{A}, \mathbb{B}; \mathsf{a}, \mathsf{b})
    \bigg),
\end{equation}
where for an element $(\mathbf{A}, \mathbf{B}) \in \F_M$, we denote the split by
\begin{equation}
\begin{split}
    \mathbf{A}&=\mathbb{A} + \mathsf{a} + \alpha,\\
    \mathbf{B}&=\mathbb{B} + \mathsf{b} + \beta.
\end{split}
\end{equation}
The sum runs over all connected, oriented, principal $BF$ Feynman graphs $\Gamma$, $\Aut(\Gamma)$ denotes the set of all automorphisms of $\Gamma$, and $\loops(\Gamma)$ denotes the number of all loops of $\Gamma$. The coefficient $T_M$ is related to the Reidemeister torsion of $M$
; its exact expression is not needed in our context.
\end{defn}
\begin{rmk}
The formal power series in  (\ref{principalqs}) is the definition of the formal perturbative expansion
of the BV integral
\begin{equation}
    \hat{\psi}_M(\mathbb{A},\mathbb{B}; \mathsf{a},\mathsf{b})=\varint_{\mathcal{L}\subset \mathcal{Y}'}e^{\frac{i}{\hbar}\mathcal{S}_M(\mathbf{A},\mathbf{B})} \in \hat{\mathcal{H}}_M^{\mathcal{P}} \coloneqq \hat{\mathcal{H}}_{\partial M}^{\mathcal{P}} \otimes \Dens^{\frac12}(\mathcal{V}_M^{\mathcal{P}}).
\end{equation}
\end{rmk}
Given a good splitting (see Definition \ref{bv-bfv_good_split}), the action can be decomposed as (\cite{CMR17})
\begin{equation}
    \Sc_M=\hat{\Sc}_M+\hat{\Sc}^{\, \text{pert}}_M+\Sc^{\, \text{res}}+\Sc^{\, \text{source}},
\end{equation}
where
\begin{equation}
    \begin{split}
        \hat{\Sc}_M&=\varint_{\Sigma_3}\beta_id\alpha^i,\\
        \hat{\Sc}^{\, \text{pert}}_M&=\varint\sur \mathcal{V}(\mathsf{a}+\alpha, \mathsf{b}+\beta),\\
        \Sc^{\, \text{res}}&= - \bigg(\varint_{\partial_2\Sigma_3}\mathbb{B}_i\mathsf{a}^i-\varint_{\partial_1\Sigma_3}\mathsf{b}_i\mathbb{A}^i\bigg),\\
        \Sc^{\, \text{source}}&= - \bigg(\varint_{\partial_2\Sigma_3}\mathbb{B}_i\alpha^i-\varint_{\partial_1\Sigma_3}\beta_i\mathbb{A}^i\bigg).\\
    \end{split}
\end{equation}
\begin{rmk}
We can rewrite 

\begin{equation}
    \hat{\psi}_M(\mathbb{A},\mathbb{B}; \mathsf{a},\mathsf{b})=T_{M}\big\langle e^{\frac{i}{\hbar}(\Sc^{\text{res}}+\Sc^{\text{source})}}\big\rangle=T_{M}e^{\frac{i}{\hbar}\Sc^{\text{eff}}},
\end{equation}
\end{rmk}
where $\langle\ \rangle$ denotes the expectation value with respect to the bulk theory $\hat{\Sc}+\hat{\Sc}^{\text{pert}}$ and
\begin{equation}
    \Sc^{\, \text{eff}}=- \bigg(\varint_{\partial_2M}\mathbb{B}_i\mathsf{a}^i-\varint_{\partial_1M}\mathsf{b}_i\mathbb{A}^i\bigg)+\varint_{\partial_2M\times\partial_1M}\pi^*_1\mathbb{B}_i\eta^i_{\; j} \pi^*_2\mathbb{A}^j.
\end{equation}
Note that the \textit{effective action} manifests as we sum over connected graphs. 

We are now interested in constructing a product on the full state space using composite fields. We define the \textit{bullet product}:
\begin{equation}
\label{bulletproof}
\begin{split}
    &\varint_{\partial_1M} u_i \wedge \mathbb{A}^i \bullet \varint_{\partial_1M} v_j \wedge \mathbb{A}^j \coloneqq \\
   &(-1)^{\mid \mathbb{A}^i \mid (d-1+\mid v_j\mid)+\mid u_i\mid(d-1)} \bigg( \varint_{\mathrm{C}_2(\partial_1M)} \pi_1^*u_i \wedge \pi_2^*v_j \wedge \pi_1^*\mathbb{A}^i \wedge \pi_2^*\mathbb{A}^j+\varint_{\partial_1M}u_i \wedge v_j \wedge [\mathbb{A}^i \mathbb{A}^j] \bigg ),
\end{split}
\end{equation}
with $u,v$ smooth differential forms depending on the bulk and residual fields.
\begin{rmk}
Consider the operator $\varint_{\partial_1M} F^{ij} \frac{\delta^2}{\delta \mathbb{A}^i \delta \mathbb{A}^j}$. It can be interpreted as $\varint_{\partial_1M} F^{ij} \frac{\delta}{\delta[\mathbb{A}^i\mathbb{A}^j]}$, and therefore, we have
\begin{equation}
    \varint_{\partial_1M}F^{ij}\frac{\delta^2}{\delta \mathbb{A}^i \delta \mathbb{A}^j}\bigg (\varint_{\partial_1M} u_i \wedge \mathbb{A}^i \bullet \varint_{\partial_1M}v_j \wedge \mathbb{A}^j \bigg )= \varint_{\partial_1M} u_iv_jF^{ij},
\end{equation}
which matches our prediction.
\end{rmk}

\begin{defn}[Full quantum state]
\label{full_quantum_state}
Let $M$ be a manifold (with boundary). Given a $BF$-like  BV-BFV theory $\pi_M:\F_M \rightarrow \F^\partial_{\partial M}$, a polarization $\mathcal{P}$ on $\F^\partial_{\partial M}$, a good splitting $\F_M=\mathcal{B}_{\partial M}^{\mathcal{P}} \times \mathcal{V}_M^{\mathcal{P}} \times \mathcal{Y}'$ and $\mathcal{L} \subset \mathcal{Y}'$, the gauge-fixing Lagrangian,  we can define the \textit{full quantum state} by the formal power series
\begin{equation}
    \boldsymbol{\hat{\psi}}_M(\mathbb{A},\mathbb{B};\mathsf{a},\mathsf{b})=T_{M}\exp\bigg(\frac{i}{\hbar}\sum_{\Gamma}\frac{(-i\hbar)^{\text{loops}(\Gamma)}}{|\Aut(\Gamma)|}\varint_{\text{C}_\Gamma(M)}\omega_{\Gamma}(\mathbb{A}, \mathbb{B}; \mathsf{a}, \mathsf{b})
    \bigg).
\end{equation}
\end{defn}

\begin{rmk}Exploiting the bullet product in \eqref{bulletproof}, we can write the full quantum state as the expectation value
\begin{equation}
    \boldsymbol{\hat{\psi}}_M(\mathbb{A},\mathbb{B};\mathsf{a},\mathsf{b})=T_{M}\big\langle e^{\frac{i}{\hbar}(\Sc^{\text{res}}+\Sc^{\text{source}})}_{\bullet}\big\rangle=T_{M}e^{\frac{i}{\hbar}\Sc^{\text{eff}}}_{\bullet},
\end{equation}
with $e_\bullet$ the exponential with respect to the bullet product.
\end{rmk}

\subsubsection{The BFV boundary operator}
 peOur next ingredient is the quantum BFV boundary operator for $BF$-like theories \cite{CMR17}. We will follow the same procedure as with the state, writing firstly its principle part and then extending to a regularization using the composite fields. One obtains the quantum BFV boundary operator via the quantization of the BFV action such that Theorem \ref{bv-bfv_thm_mqme} is satisfied.
\begin{defn}[Principal part of the BFV boundary operator] The \textit{principal part} of the BFV boundary operator is given by
\begin{equation}
    \Omega^{\text{princ}}=\underbrace{\Omega_0^{\mathbb{A}}+\Omega_0^{\mathbb{B}}}_{\coloneqq \Omega_0}+\underbrace{\Omega_{\text{pert}}^{\mathbb{A}}+\Omega_{\text{pert}}^{\mathbb{B}}}_{\coloneqq \Omega_{\text{pert}}^{\text{princ}}},
\end{equation}
where
\begin{equation}
    \begin{split}
        \Omega_0^{\mathbb{A}} &\coloneqq (-1)^d i\hbar \varint_{\partial_1M} \bigg(d\mathbb{A} \frac{\delta}{\delta \mathbb{A}} \bigg),\\
    \Omega_0^{\mathbb{B}} &\coloneqq (-1)^d i\hbar \varint_{\partial_2M} \bigg(d\mathbb{B} \frac{\delta}{\delta \mathbb{B}} \bigg),
    \end{split}
\end{equation}
\begin{equation}
   \begin{split}
        \Omega_{\text{pert}}^{\mathbb{A}}&\coloneqq \sum_{n,k \geq 0}\sum_{\Gamma_1'}\frac{(i \hbar)^{\text{loops($\Gamma_1')$}}}{\mid \Aut(\Gamma_1')\mid} \varint_{\partial_1M}\bigg(\sigma_{\Gamma_1'}\bigg)_{i_1\dots i_n}^{j_1\dots j_k}\wedge \mathbb{A}^{i_1}\wedge \dots \wedge \mathbb{A}^{i_n}\bigg((-1)^d i\hbar\frac{\delta}{\delta \mathbb{A}^{j_1}}\bigg)\dots \bigg((-1)^d i\hbar \frac{\delta}{\delta \mathbb{A}^{j_k}} \bigg),\\
    \Omega_{\text{pert}}^{\mathbb{B}}&\coloneqq \sum_{n,k \geq 0}\sum_{\Gamma_2'}\frac{(i \hbar)^{\text{loops($\Gamma_2')$}}}{\mid \Aut(\Gamma_2')\mid} \varint_{\partial_2M}\bigg(\sigma_{\Gamma_2'}\bigg)_{i_1\dots i_n}^{j_1\dots j_k}\wedge \mathbb{B}_{j_1}\wedge \dots \wedge \mathbb{B}_{j_k}\bigg((-1)^d i\hbar\frac{\delta}{\delta \mathbb{B}_{i_1}}\bigg)\dots \bigg((-1)^d i\hbar \frac{\delta}{\delta \mathbb{B}_{i_n}} \bigg),
   \end{split}
\end{equation}
where, for $\mathbb{F}_1=\mathbb{A}$, $\mathbb{F}_2=\mathbb{B}$ and $l \in \{1,2\}$, $\Gamma_l'$ runs over graphs with  
\begin{itemize}
    \item $n$ vertices on $\partial_lM$ of valence 1 with adjacent half-edges oriented inwards and decorated with boundary fields $\mathbb{F}_l^{i_1},\dots,\mathbb{F}_l^{i_n}$ all evaluated at the point of collapse $p\in \partial_lM$,
    \item $k$ outward leaves if $l = 1$ and $k$ inward leaves if $l = 2$, decorated with variational derivatives
in boundary fields
\begin{equation}
    (-1)^d i\hbar \frac{\delta}{\delta \mathbb{F}_l^{j_1}},\dots,(-1)^di\hbar\frac{\delta}{\delta\mathbb{F}_l^{j_k}} 
\end{equation}
at the point of collapse, 
\item  no outward leaves if $l = 2$ and no inward leaves if $l = 1$ (graphs with them do not contribute).
\end{itemize}
\end{defn}
The form $\sigma_{\Gamma_l'}$ is obtained as the integral over the compactification $\tilde{\mathrm{C}}_{\Gamma_l'}(\mathbb{H}^d)$ of the  open configuration space modulo scaling and translation, with $\mathbb{H}^d$ the $d$-dimensional upper half-space:
\begin{equation}
    \sigma_{\Gamma_l'}=\varint_{\tilde{\mathrm{C}}_{\Gamma_l'(\mathbb{H}^d)}} \omega_{\Gamma_l'},
\end{equation}
where $\omega_{\Gamma_l'}$ is the product of limiting propagators at the point $p$ of collapse and vertex tensors.

Our goal now is to describe the BFV boundary operator with composite fields. For this, we introduce the following auxiliary concept.

Consider the regular functional in (\ref{regularfct}), where we get a term $L$ replaced by $dL$ plus all the terms corresponding to the boundary of the configuration space.  Since $L$ is smooth, its restriction to the boundary is smooth as well, and can be integrated on the fibers giving rise to a smooth form on the base configuration space. For example
\begin{equation}
    \Omega_0 \varint_{\partial_1M}L_{IJ} \wedge [\mathbb{A}^I] \wedge [\mathbb{A}^{J}]= \pm i \hbar \varint_{\partial_1M}dL_{IJ}\wedge [\mathbb{A}^I]\wedge [\mathbb{A}^J],
\end{equation}
\begin{equation}
    \Omega_0\varint_{\mathrm{C}_2(\partial_1M)}L_{IJK}\wedge \pi_1^*([\mathbb{A}^I]\wedge[\mathbb{A}^J]) \wedge \pi_2^*[\mathbb{A}^K]\pm i\hbar\varint_{\partial_1M}\underline{L_{IJK}}\wedge [\mathbb{A}^I]\wedge [\mathbb{A}^J] [\mathbb{A}^K],
\end{equation}
with $\underline{L_{IJK}}=\pi_*^{\partial}$, where $\pi^\partial : \partial \mathrm{C}_2(\partial_1M)\rightarrow\partial_1M$ is the canonical projection. \newline
For any two regular functionals $S
_1$ and $S_2$ we can write
\begin{equation}
    \Omega_0(S_1 \bullet S_2)=\Omega_0(S_1) \bullet S_2 \pm S_1 \bullet \Omega_0(S_2).
\end{equation}
The rest of the allowed generators are products of expressions in the following shape:
\begin{equation}
    \varint_{\partial_1M}L^J_{I^1 \dots I^r}[\mathbb{A}^{I_1}]\wedge \dots \wedge [\mathbb{A}^{I_r}] \frac{\delta^{\mid J \mid}}{\delta [\mathbb{A}^I]},
\end{equation}
\begin{equation}
    \varint_{\partial_2M}L^{J^1\dots J^l}_{I}[\mathbb{B}_{J_1}]\wedge \dots \wedge [\mathbb{B}_{J_l}] \frac{\delta^{\mid I \mid}}{\delta [\mathbb{B}_J]}.
\end{equation}
\begin{defn}[Full BFV boundary operator] 
\label{full_bfv_def}
The \textit{full BFV boundary operator} is
\begin{equation}
    \boldsymbol{\Omega}_{\partial M}\coloneqq\Omega_0 + \underbrace{\boldsymbol{\Omega}_{\text{pert}}^{\mathbb{A}} + \boldsymbol{\Omega}_{\text{pert}}^{\mathbb{B}}}_{\boldsymbol{\Omega}_{\text{pert}}},
\end{equation}
with 
\begin{equation}
    \begin{split}
        \boldsymbol{\Omega}_{\text{pert}}^{\mathbb{A}}&\coloneqq \sum_{n,k \geq 0}\sum_{\Gamma_1'}\frac{(i \hbar)^{\text{loops($\Gamma_1')$}}}{\mid \Aut(\Gamma_1')\mid} \varint_{\partial_1M}\bigg(\sigma_{\Gamma_1'}\bigg)_{I_1\dots I_n}^{J_1\dots J_k}\wedge \mathbb{A}^{I_1}\wedge \dots \wedge \mathbb{A}^{I_n}\bigg((-1)^{kd} (i\hbar)^{k}\frac{\delta^{\mid J_1\mid + \dots +\mid J_k\mid}}{\delta [\mathbb{A}^{J_1} \dots \mathbb{A}^{J_k}]}\bigg),\\
    \boldsymbol{\Omega}_{\text{pert}}^{\mathbb{B}}&\coloneqq \sum_{n,k \geq 0}\sum_{\Gamma_2'}\frac{(i \hbar)^{\text{loops($\Gamma_2')$}}}{\mid \Aut(\Gamma_2')\mid} \varint_{\partial_2M}\bigg(\sigma_{\Gamma_2'}\bigg)^{I_1\dots I_n}_{J_1\dots J_k}\wedge \mathbb{B}_{I_1}\wedge \dots \wedge \mathbb{B}_{I_n}\bigg((-1)^{kd} (i\hbar)^{k}\frac{\delta^{\mid J_1\mid + \dots +\mid J_k\mid}}{\delta [\mathbb{B}_{J_1} \dots \mathbb{B}_{J_k}]}\bigg),
    \end{split}
\end{equation}     
where, for $\mathbb{F}_1=\mathbb{A}$, $\mathbb{F}_2=\mathbb{B}$ and $l \in \{1,2\}$, $\Gamma_l'$ runs over graphs with  
\begin{itemize}
    \item $n$ vertices on $\partial_lM$, where vertex $s$ has valence $| I_s| \geq 1$, with adjacent half-edges oriented inwards and decorated with boundary fields $[\mathbb{F}_l^{I_1}],\dots,[\mathbb{F}_l^{I_n}]$ all evaluated at the point of collapse $p\in \partial_lM$,
    \item $| J_1 |+ \dots+| J_k|$ outward leaves if $l = 1$ and $| J_1|+ \dots+| J_k |$ inward leaves if $l = 2$, decorated with variational derivatives
in boundary fields
\begin{equation}
    (-1)^d i\hbar \frac{\delta}{\delta [\mathbb{F}_l^{J_1}]},\dots,(-1)^di\hbar\frac{\delta}{\delta[\mathbb{F}_l^{J_k}]} 
\end{equation}
at the point of collapse, 
\item  no outward leaves if $l = 2$ and no inward leaves if $l = 1$ (graphs with them do not contribute).
\end{itemize}
\end{defn}
Similarly as before, the form $\sigma_{\Gamma_l'} $
can be obtained as the integral over the compactified configuration space $\tilde{\mathrm{C}}_{\Gamma_l'}(\mathbb{H}^d)$, given by
\begin{equation}
    \sigma_{\Gamma_l'}=\varint_{\tilde{\mathrm{C}}_{\Gamma_l'(\mathbb{H}^d)}} \omega_{\Gamma_l'},
\end{equation}
where $\omega_{\Gamma_l'}$ is the product of limiting propagators at the point $p$ of collapse and vertex tensors.

\begin{thm}[\cite{CMR17}]
\label{bv-bfv_thm_mqme}
Let $M$ be a smooth manifold (possibly with boundary). Then the following statements are satisfied:
\begin{enumerate}
    \item  The full covariant state $\boldsymbol{\hat{\psi}}_M$ satisfies the \textit{modified Quantum Master Equation} (mQME):
    \begin{equation}
        (\hbar^2 \Delta_{\mathcal{V}_M} + \boldsymbol{\Omega}_{\partial M}) \boldsymbol{\hat{\psi}}_M=0.
    \end{equation}
    \item The full BFV boundary operator $\boldsymbol{\Omega}_{\partial M}$ squares to zero:
    \begin{equation}
       (\boldsymbol{\Omega}_{\partial M})^2=0.
    \end{equation}
    \item A change of propagator or residual fields leads to a theory related by change of data as in  Definition \ref{bv-bfv:def_change_data}.
\end{enumerate}
\end{thm}

\subsection{AKSZ theories}
\label{subsec:AKSZ}
In \cite{AKSZ97}, Alexandrov, Kontsevich, Schwarz, and Zaboronsky presented a class of local field theories that are compatible with the BV construction called \textit{AKSZ theories}. The compatibility here means that the constructed local actions are solutions to the CME. These theories thus form a subclass of BV theories. We describe here the essential concepts needed for the future sextions\footnote{Maybe it is appropriate here to give a short warning about the notation. In the previous chapters, we have mostly denoted by $M$ the source, whereas now we will denote the target by $M$. Moreover, before the letter $\Sigma$ was reserved for a manifold with one dimension less than the target. From now on it will be mostly used for the source.}.
\begin{defn}[Differential graded symplectic manifold]  A \textit{differential graded
symplectic manifold} of degree $k$ is a triple
\begin{equation}
    (M,\Theta_M,\omega_M=d_M\alpha_M)
\end{equation}
with $M$ a $\mathbb{Z}$-graded manifold, $\Theta_M \in \mathcal{C}^{\infty}(M)$ a function on $M$ of degree $k+1$, $d_M$ the de Rham differential on $M$ and $\omega_M \in \Omega^2(M)$, an exact symplectic form of degree $k$ with primitive 1-form $\alpha_M \in \Omega^1(M)$, such that   
\begin{equation}
    (\Theta_M, \Theta_M)_{\omega_M}=0,
\end{equation}
where $(-,-)_{\omega_M}$ is the odd Poisson bracket induced by the symplectic form $\omega_M$.
\end{defn}
\begin{rmk}
 We denote by $Q_M$ the Hamiltonian vector field of $\Theta_M$. Then the quadruple $(M,Q_M,\Theta_M,\omega_M=d_M\alpha_M)$ is also called a \textit{Hamiltonian $Q$-manifold}.
\end{rmk}

\subsubsection{AKSZ sigma models}
Let $\Sigma_d$  be a $d$-dimensional compact, oriented manifold and let $T[1]\Sigma_d$ be its shifted tangent bundle. We fix a Hamiltonian $Q$-manifold
\begin{equation}
    (M,Q_M,\Theta_M,\omega_M=d_M \alpha_M)
\end{equation}
 of degree $d-1$ for $d \geq 0$. The space of fields can be defined as the mapping space space of graded manifolds:
 \begin{equation}
     \F_{\Sigma_d} \coloneqq \Maps(T[1]\Sigma_d,M)
 \end{equation}
 where $\Maps$ denotes the mapping space. Our goal is to endow $\F_{\Sigma_d}$ with a $Q$-manifold structure, and to do this we consider the lifts of the de Rham differential $d_{\Sigma_d}$ on $\Sigma_d$ and of the cohomological vector field $Q_M$ on the target $M$ to the mapping space. Therefore, we get the following cohomological vector field
 \begin{equation}
     Q_{\Sigma_d} \coloneqq \hat{d}_{\Sigma_d}+\hat{Q}_M,
 \end{equation}
with $\hat{d}_{\Sigma_d}$ and $\hat{Q}_M$ the corresponding lifts to the mapping space. We remark that we can see $d_{\Sigma_d}$ as the cohomological vector field on $T[1]\Sigma_d$. Consider the push-pull diagram
\begin{equation}
    \F_{\Sigma_d} \xleftarrow{\text{p}} \F_{\Sigma_d} \times T[1]\Sigma_d \xrightarrow {\text{ev}} M,
\end{equation}
with $\mathrm{p}$ and $\mathrm{ev}$ the projection and evaluation map respectively. One can construct a \textit{transgression} map
\begin{equation}
    \mathcal{T}_{\Sigma_d} \coloneqq \mathrm{p}_*\mathrm{ev}^*:\Omega^\bullet(M) \rightarrow \Omega^\bullet (\F_{\Sigma_d}).
\end{equation}
Note that map $p_*$ is given by fiber integration on $T[1]\Sigma_d$. As a next step, we will endow the space of fields with a symplectic structure $\omega_{\Sigma_d}$ defined as:
\begin{equation}
    \omega_{\Sigma_d} \coloneqq (-1)^d \mathcal{T}_{\Sigma_d}(\omega_M) \in \Omega^2(\F_{\Sigma_d}).
\end{equation}
Remarkably, we get a solution $\mathcal{S}_{\Sigma_d}$ of the CME, namely the BV action functional 
\begin{equation}
    \mathcal{S}_{\Sigma_d} \coloneqq \underbrace{ \iota_{{\hat{d}_{\Sigma_d}}}\mathcal{T}_{\Sigma_d}(\alpha_M)}_{\coloneqq \mathcal{S}_{\Sigma_d}^{\text{kin}}} \underbrace{ +\mathcal{T}_{\Sigma_d}(\Theta_M)}_{\coloneqq \mathcal{S}_{\Sigma_d}^{\text{target}}}\in \mathcal{C}^\infty(\F_{\Sigma_d}).
\end{equation}
We can indeed check that 
\begin{equation}
    (\mathcal{S}_{\Sigma_d},\mathcal{S}_{\Sigma_d})_{\omega_{\Sigma_d}}=0.
\end{equation}

Note that the symplectic form $\omega_{\Sigma_d}$ has degree $(d-1)-d=-1$ as predicted and the action $\mathcal{S}_{\Sigma_d}$ has degree 0. Hence, this setting induces a BV manifold $(\F_{\Sigma_d}, \mathcal{S}_{\Sigma_d}, \omega_{\Sigma_d})$. Let $\{x^\mu\}$ and $\{u^i\}$ $(1 \leq i \leq d)$, be  local coordinates on $M$  and $\Sigma_d$ respectively. We will denote the odd fiber  coordinates of degree $+1$ on $T[1]\Sigma_d$ by $\theta^i=d_{\Sigma_d}u^i$. For a field $\mathbf{X} \in \F_{\Sigma_d}$ we then have the following local expression
\begin{equation}
\begin{split}
    \mathbf{X}^\mu(u,\theta)=\sum_{l=0}^d\,\, \underbrace{\sum_{1 \leq i_1<\dots<i_l\leq d}\mathbf{X}^\mu_{i_1\dots i_l}(u) \theta^{i_1} \wedge \dots\wedge \theta^{i_l}}_{\mathbf{X}^\mu_{(l)}(u,\theta)} \in \bigoplus_{l=0}^d\mathcal{C}^\infty(\Sigma_d) \otimes \bigwedge\nolimits^lT^\vee \Sigma_d. 
\end{split}
\end{equation}
The functions $\mathbf{X}^\mu_{i_1\dots i_l} \in \mathcal{C}^\infty(\Sigma_d)$ have degree $\deg(x^\mu)-l$ on $\F_{\Sigma_d}$. The symplectic form $\omega_M$ and its primitive 1-form $\alpha_M$ on $M$ are given by  
\begin{equation}
    \begin{split}
        \alpha_M&=\alpha_\mu(x)d_Mx^\mu \in \Omega^1(M),\\
    \omega_M&=\frac{1}{2}\omega_{\mu_1 \mu_2}(x)d_Mx^{\mu_1} d_Mx^{\mu_2} \in \Omega^2(M).
    \end{split}
\end{equation}
Using the above equations, we locally get the following  expressions for the BV symplectic form, its primitive 1-form and the BV action functional:
\begin{equation}
   \begin{split}
        \alpha_{\Sigma_d}&=\varint_{\Sigma_d}\alpha_\mu(\mathbf{X})\delta\mathbf{X}^\mu \in \Omega^1(\F_{\Sigma_d}),\\
   \omega_{\Sigma_d}&=(-1)^d\frac{1}{2}\varint_{\Sigma_d}\omega_{\mu_1\mu_2}(\mathbf{X})\delta \mathbf{X}^{\mu_1}  \delta \mathbf{X}^{\mu_2} \in \Omega^2(\F_{\Sigma_d}),\\
    \mathcal{S}_{\Sigma_d}&=\varint_{\Sigma_d} \alpha_\mu(\mathbf{X})d_{\Sigma_d}\mathbf{X}^\mu + \varint_{\Sigma_d} \Theta_M(\mathbf{X}) \in \mathcal{C}^\infty(\F_{\Sigma_d}).
   \end{split}
\end{equation}
We have denoted by $\delta$ the de Rham differential on $\F_{\Sigma_d}$. Using Darboux coordinates on $M$, we can write
\begin{equation}
    \omega_M=\frac{1}{2}\omega_{\mu_1\mu_2}d_Mx^{\mu_1} d_Mx^{\mu_2},
\end{equation}
with $\omega_{\mu_1 \mu_2}$ constant, implying that $\alpha_M=\frac{1}{2}x^{\mu_1}\omega_{\mu_1\mu_2}d_Mx^{\mu_2}$. We get the BV symplectic form
\begin{equation}
\begin{split}
    \omega_{\Sigma_d}&=\frac{1}{2} \varint_{T[1]\Sigma_d}\mu_{\Sigma_d}(\omega_{\mu_1\mu_2} \delta \mathbf{X}^{\mu_1} \delta\mathbf{X}^{\mu_2}) \\
    &=\frac{1}{2}\varint_{\Sigma_d}(\omega_{\mu_1 \mu_2}\delta \mathbf{X}^{\mu_1} \delta \mathbf{X}^{\mu_2})^{\text{top}}.
\end{split}
\end{equation}
The (master) action is 
\begin{equation}
    \mathcal{S}_{\Sigma_d}=\varint_{T[1]\Sigma_d}\mu_{\Sigma_d}\bigg(\frac{1}{2} \mathbf{X}^{\mu_1} \omega_{\mu_1 \mu_2} D\mathbf{X}^{\mu 2}\bigg)+(-1)^d \varint_{T[1]\Sigma_d}\mu_
    {\Sigma_d}\mathbf{X}^*\Theta_M,
\end{equation}
with $\mu_{\Sigma_d}$ a canonical measure on $T[1]\Sigma_d$ and $D=\theta^j\frac{\delta}{\delta u_j}$, the superdifferential n $T[1]\Sigma_d$.

\section{The Rozansky--Witten model}
\label{sec:RW_model}
The RW model is a 3-dimensional topological sigma model. It was originally discovered with target a hyperK{\"a}hler manifold in \cite{RW96} as a result of a topological twist of 3-dimensional $N=4$ super Yang-Mills theory. However, shortly after, Kapranov \cite{KA99} and Kontsevich \cite{Ko99} showed that the model required less structure than originally thought: the target manifold does not have to be hyperk{\"a}hler, but, more generally, it can have a holomorphic symplectic structure. Since we will focus on this latter case, here we will present how this generalization proposed by Kapranov and Kontsevich was understood in the context of topological sigma models by Rozansky and Witten. After the work of Kapranov and Kontsevich, Rozansky and Witten added an appendix to \cite{RW96}, where they explained how to extend their formulation of the model to the case of a holomorphic symplectic target manifold.

\begin{notat}
Except for the name of the manifolds, which we adapt to the notation we will use in Section \ref{sec:Classical_Theory}, the notation will be the same as in \cite{RW96}.
\end{notat}

\subsection{First definitions}
Let $\Sigma_3$ be the source 3-dimensional manifold and $(M, \omega)$ the target holomorphic symplectic manifold. The fields are the following:
\begin{itemize}
    \item \textit{bosonic fields} described by the smooth maps $\phi:\Sigma_3\rightarrow M$, in local coordinates we have $\phi^I(x^\mu)$ and $\Bar{\phi}^{\Bar{I}}(x^\mu)$,
    \item \textit{fermionic (or Grassman) fields} $\eta\in\Gamma(\Sigma_3,\phi^*T^{0,1}M)$ and $\chi\in\Gamma(\Sigma_3,\Omega^1(M)\otimes\phi^*T^{1,0}M)$, with $T^{1,0}(M)$ and $T^{0,1}M$, the holomorphic and anti-holomorphic tangent bundle, respectively. In local coordinates we can write them as $\eta^{\Bar{I}}_\mu(x^\mu)$ and $\chi^I_\mu(x^\mu)$.
\end{itemize}

Consider a single fermionic symmetry on these fields, which we will denote by $\Bar{Q}$. Its action is:
\begin{alignat}{2}
         &\delta\phi^I=0,\qquad  &&\delta\Bar{\phi}^{\Bar{I}}=\eta^{\Bar{I}},\\
         &\delta\eta^{\bar{I}}=0,\qquad  &&\delta\chi_\mu^I=-\partial_\mu\phi^I.
\end{alignat}

To introduce the Lagrangian density of the theory, we add some extra structure to the target manifold. Let $\Gamma^I_{JK}$ be a symmetric connection in the holomorphic tangent bundle of $M$, i.e. $\Gamma^I_{JK}=\Gamma^I_{KJ}$. The $(1,1)$-part of the curvature, related to $\Gamma_{JK}^I$ represents the \textit{Atiyah class}\footnote{The Atiyah class is defined to be the obstruction to the existence of a global holomorphic connection.} of $M$ \cite{At57}:
\begin{equation}
    \Riem{I}{J}{K}{\Bar{L}}=\frac{\partial \Chr{I}{J}{K}}{\partial \Bar{\phi}^{\Bar{L}}}.
\end{equation}

In \cite{RW96}, it is noted that the connection does not have to be compatible with the holomorphic symplectic form\footnote{This compatibility condition will be assumed for the RW model when we compare it with the model we will develop in Section \ref{sec:Classical_Theory} (see Section \ref{sec:comp_orig_RW}).} $\omega_{IJ}$. We require $\omega_{IJ}$ to be non-degenerate and closed, i.e.
\begin{equation}
    \frac{\partial \omega_{IJ}}{\partial \Bar{\phi}^{\Bar{K}}}=0,\qquad \frac{\partial \omega_{IJ}}{\partial \phi^K}+\frac{\partial \omega_{KI}}{\partial \phi^J}+\frac{\partial \omega_{JK}}{\partial \phi^I}=0.
\end{equation}

Rozansky and Witten define a $\Bar{Q}$-invariant Lagrangian density $\mathscr{L}$ to be $\mathscr{L}\coloneqq \mathscr{L}_2+\mathscr{L}_1$. The Lagrangian $\mathscr{L}_2$ is given by
     \begin{equation}
         \mathscr{L}_2=\frac12\frac{1}{\sqrt{h}}\epsilon^{\mu\nu\rho}\bigg(\omega_{IJ}\chi_\mu^I\nabla_\nu\chi_\rho^J-\frac{1}{3}\omega_{IJ}R^J_{KL\Bar{M}}\chi_\mu^I\chi_\nu^K\chi_\rho^L\eta^{\Bar{M}}+\frac{1}{3}(\nabla_L\Omega_{IK})(\partial_\mu\phi^I)\chi_\nu^K\chi_\rho^L \bigg),
    \end{equation}
where $\nabla_\mu$ is a covariant derivative with respect to the pullback of the connection $\Chr{I}{J}{K}$:
\begin{equation}
    \nabla_\mu\chi^I_\nu=\partial_\mu\chi^I_\nu+(\partial_\mu\phi^J)\Chr{I}{J}{K}\chi^K_\nu.
\end{equation}

In order to construct the $\Bar{Q}$-exact Lagrangian $\mathscr{L}_1$, we need to choose an Hermitian metric $g_{i\Bar{j}}$ on $M$. We define by
\begin{equation}
    \Tilde{\Gamma}^{\Bar{I}}_{\Bar{J}\Bar{K}}\coloneqq\frac12g^{\Bar{I}L}\bigg(\frac{\partial g_{L\Bar{J}}}{\partial \Bar{\phi}^{\Bar{K}}}+\frac{\partial g_{L\Bar{K}}}{\partial \Bar{\phi}^{\Bar{J}}}\bigg),\quad \Tilde{T}^{\Bar{I}}_{\Bar{J}\Bar{K}}\coloneqq\frac12g^{\Bar{I}L}\bigg(\frac{\partial g_{L\Bar{J}}}{\partial \Bar{\phi}^{\Bar{K}}}-\frac{\partial g_{L\Bar{K}}}{\partial \Bar{\phi}^{\Bar{J}}}\bigg),
\end{equation}
the \textit{symmetric connection} and the \textit{torsion} associated with $g_{I\Bar{J}}$. Then $\mathscr{L}_1$ is defined by
\begin{equation}
    \mathscr{L}_1\coloneqq\Bar{Q}\bigg(g_{I\Bar{J}}\chi_\mu^I(\partial_\mu\Bar{\phi}^{\Bar{J}})\bigg)=g_{I\Bar{J}}\partial_\mu\phi^I\partial_\mu\Bar{\phi}^{\Bar{J}}+g_{I\Bar{J}}\chi^I_\mu\Tilde{\nabla}_\mu\eta^{\Bar{J}},
\end{equation}
where $\Tilde{\nabla}_\mu$ is a covariant derivative with respect to the connection  $\Tilde{\Gamma}^{\Bar{I}}_{\Bar{J}\Bar{K}}+\Tilde{T}^{\Bar{I}}_{\Bar{J}\Bar{K}}$, i.e. we have
\begin{equation}
    \Tilde{\nabla}_\mu\eta^{\Bar{I}}=\partial_\mu \eta^{\Bar{I}}+(\partial_\mu\Bar{\phi}^{\Bar{J}})(\Tilde{\Gamma}^{\Bar{I}}_{\Bar{J}\Bar{K}}+\Tilde{T}^{\Bar{I}}_{\Bar{J}\Bar{K}})\eta^{\Bar{K}}.
\end{equation}
Moreover, if $g_{I\Bar{J}}$ is a K{\"a}hler metric, then $\Tilde{T}^{\Bar{I}}_{\Bar{J}\Bar{K}}=0$.

\subsection{Perturbative quantization}
\label{sec:pert_exp}
The partition function of the RW model is 
\begin{equation}
    Z_{M}(\Sigma_3)\coloneqq\varint e^{\frac{i}{\hbar}S} \mathscr{D}[\phi^i]\mathscr{D}[\eta^I]\mathscr{D}[\chi^I_\mu],
\end{equation}
where $S\coloneqq\varint_{\Sigma_3} \mathscr{L}$ and $\mathscr{D}$ is a formal measure.

As mentioned in \cite{RW96} (see also \cite{Th99,HT99}), in order to do a perturbative expansion around critical points of the action (which are constant maps from $\Sigma_3$ to $M$), we need to deal with the \textit{zero modes}: 
\begin{itemize}
    \item \textit{bosonic zero modes}: they are constant modes of $\phi$;
    \item \textit{fermionic zero modes}: here we should distinguish two cases
    \begin{itemize}
        \item if $M$ is a rational homology sphere (i.e. the first Betti number $b_1=0$), the fermionic zero modes are the constant modes of $\eta$. There are $2n$ zero modes if $\dim M=4n$.
        \item if $M$ is not a rational homology sphere but the first Betti number $b_1>0$, then there are $2nb_1$ zero modes of $\chi_\mu$ additionally.
    \end{itemize}
\end{itemize}

Taking into account the zero modes, one can decompose $\phi^i=\phi^i_0+\phi^i_\perp$, where $\phi^i_0$ are the constant maps and $\phi^i_\perp$ are required to be orthogonal to $\phi^i_0$. Similarly, the $\eta^I$ are also decomposed as $\eta^I=\eta^I_0+\eta^I_\perp$, where $\eta^I_0$ are harmonic 0-forms with coefficients in the fiber $V_{\phi_0}$ of the $\Sp(n)$-bundle $V\rightarrow X$ and $\eta^I_\perp$ are orthogonal to the harmonic part.

For our purposes, we will only consider the Lagrangian $\mathscr{L}_2$, in light of these decompositions we can rewrite it as
\begin{equation}
\label{RW:L_2}
    \mathscr{L}_2=\frac12\bigg(\omega_{IJ}(\phi_0)\chi^Id\chi^J+\frac13\Riem{}{IJ}{K}{\Bar{L}}(\phi_0)\chi^I\chi^J\chi^K\eta^{\Bar{L}}_0+\frac{1}{3}(\nabla\Omega_{IK})(\partial\phi^I_\perp)\chi^K\chi^L\bigg).
\end{equation}
  
As a result of an analysis on the absorption of fermionic zero modes by the Feynman diagrams, Rozansky and Witten concluded that only diagrams with trivalent vertices contribute. Moreover, these trivalent vertices have to be exactly $2n$ to saturate the $2n$ zero modes of $\eta$. They call these diagrams ``minimal". The Lagrangian $L_2$ contains the following vertex with the needed properties
\begin{equation}
    V=\frac16\Riem{}{IJ}{K}{\Bar{L}}(\phi_0)\chi^I\chi^J\chi^K\eta^{\Bar{L}}_0.
\end{equation}
Here we should think of $\eta_0$ as a ``coupling constant", in fact, we should focus on the order of the $\eta^I_0$ during the perturbative expansion.

Since all the fields $\eta$ are used to absorb zero modes, we only need the propagators for the fields $\phi^i$ and $\chi^I_{mu}$. According to \cite{RW96}, these are
\begin{equation}
\label{propagators_RW}
    \begin{split}
        &\braket{\chi^I_\mu(x_1),\chi^J_\nu(x_2)}=\hbar\omega^{IJ}G^{\chi}_{\mu\nu}(x_1,x_2),\\
        &\braket{\phi^i(x_1),\phi^j(x_2)}=-\hbar g^{ij}G^{\phi}(x_1,x_2),
    \end{split}
\end{equation}
with $G^{\chi}_{\mu\nu}(x_1,x_2)$ and $G^{\phi}(x_1,x_2)$ Green's functions. We refer to \cite{RW96} for a detailed description of the Green's functions.

The Feynman diagrams participating in the calculation of the partition function depend only on the dimension of the target manifold $M$ and on the first Betti number $b_1$ of the source 3-manifold $\Sigma_3$. The former causes the number of vertices of the graphs to be equal to $2n$. The latter has consequences on the valence of the vertices. We have the following cases:
\begin{itemize}
    \item($b_1=0$) There are no $\chi$ zero modes to absorb. Hence, all the Feynman diagrams are closed graphs with $2n$ trivalent vertices. This is the case when $\Sigma_3$ is a rational homotopy sphere.
    \item($b_1=1$) There are $2n$ $\chi$ zero modes coming from a harmonic 1-form. As a consequence, each vertex absorbs exactly one zero mode of $\chi$, and thus all the Feynman diagrams are closed graphs with $2n$ bivalent vertices.
    \item($b_1=2$) There are $4n$ $\chi$ zero modes coming from a two harmonic 1-forms on $\Sigma_3$. As a consequence, each vertex absorbs exactly two zero modes of $\chi$, one for each harmonic 1-form, and thus all the Feynman diagrams are closed graphs with $2n$ univalent vertices.
    \item($b_1=3$) There are $6n$ $\chi$ zero modes coming from three harmonic 1-forms on $\Sigma_3$. As a consequence, each vertex absorbs exactly three zero modes of $\chi$, one for each harmonic 1-form, and thus all the Feynman diagrams are a collection of $2n$ totally disconnected vertices with no edges.
    \item($b_1\geq 4$) The $\chi$ zero modes become too many and they can not be integrated out by the $\chi$ present in the vertices  (at most three), so the RW partition function vanishes.
\end{itemize}

Let us denote by $\Gamma_{n,m}$ the set of all closed graphs with $2n$ $m$-valent vertices and as $Z_{M, \Gamma}(\Sigma_3; \phi^i_0)$ the sum of all the contributions of the minimal Feynman diagrams corresponding to a given graph $\Gamma$. The total contribution of Feynman diagrams is
\begin{equation}
   Z_{M}(\Sigma_3; \phi^i_0)=\sum_{\Gamma\in\Gamma_{n,3-b_1(\Sigma_3)}}Z_{M, \Gamma}(\Sigma_3; \phi^i_0),
\end{equation}
where each $Z_{M, \Gamma}(\Sigma_3; \phi^i_0)$ can actually be written as a sum of two factors
\begin{equation}
   Z_{M, \Gamma}(\Sigma_3; \phi^i_0)=\text{W}_\Gamma(M;\phi_0^i)\sum_a I_{\Gamma,a}(\Sigma_3).
\end{equation}
We will explain each factor one-by-one. First, $I_{\Gamma,a}(\Sigma_3)$ includes the integral over $\Sigma_3$ of the propagators $G^{\chi}_{\mu\nu}(x_1,x_2)$ and $G^{\phi}(x_1,x_2)$ as well as $\chi$ zero modes. The sum is over all the possible ways to contract the fields of $\Gamma$ with the propagators (\ref{propagators_RW}). On the other hand, the factor $\text{W}_{\Gamma}(M;\phi^i_0)$ is a product of tensors $\Riem{}{IJ}{K}{L}$ coming from the vertices $V_1$ and $V_2$, which are contracted by the $\omega^{IJ}$ contained in the propagators. After antisymmetrizing over the anti-holomorphic indices (coming from the zero modes' contributions), we obtain a $\Bar{\partial}$-closed $(0,2n)$-form on $M$. In other words, we have a map
\begin{equation}
        \Gamma_{n,3}\rightarrow H^{0,2n}(M),
\end{equation}  
where $H^{0,2n}(X)$ is the Dolbeault cohomology. This corresponds to a weight system, the \textit{Rozansky--Witten weight system}. By definition a function on $\Gamma_{n,3}$ is called a weight if it satisfies the AS and IHX relations (see also \cite{B95}).

The AS relation means that $\Gamma_{n,3}$ is antisymmetric under the permutation of legs at a vertex. For RW, this is not valid on the nose since the curvature tensor is completely symmetric. However, we can prove the vanishing of tadpole diagrams (i.e. diagrams with a loop centered at a vertex), which is consistent with the AS relation. The proof follows simply because the loop is constructed by contracting two indices of the symmetric tensor $\Riem{}{IJ}{K}{\Bar{L}}\eta^{\Bar{L}}_0$ with $\omega^{IJ}$, which is antisymmetric.

On the other hand, the IHX relation means that the sum over all possible (three) ways of collapsing a propagator such that we obtain a graph with a four-valent vertex, while the other vertices are trivalent, vanishes (see Fig. \ref{fig:IHX}).

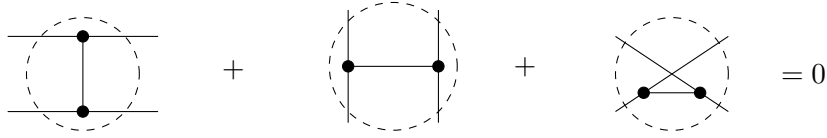
\begin{figure}[H]
\centering
\begin{subfigure}[b]{0.33\textwidth}
\centering
\begin{tikzpicture}
  \begin{feynman}[every blob={/tikz/fill=white!30,/tikz/inner sep=0.5pt}]
  \diagram*{
  a -- m [dot] -- b,
  c -- e [dot] -- d,
  m -- e
  }; 
  \end{feynman}
\draw [dashed] (1.75,-0.5) arc[x radius=0.75, y radius=0.75, start angle=0, end angle=360];
\node[] at (3, -0.45) {$+$};
  \end{tikzpicture}
  \end{subfigure}\hspace{-1.5cm}
  \begin{subfigure}[b]{0.33\textwidth}
  \centering
  \begin{tikzpicture}
  \begin{feynman}
   \vertex (c) at (-0.6,0.25);
  \node[dot] (m) at (-0.6, -0.5);
   \node[dot] (e) at (0.6, -0.5);
   \vertex (b) at (-0.6, -1.25);
   \vertex (a) at (0.6, -1.25);
   \vertex (f) at (0.6, 0.25);
  \diagram*{
  (c) -- (m) -- (b),
  (f) -- (e) -- (a),
   (m) -- (e)
  }; 
  \end{feynman}
\draw [dashed] (0.85,-0.5) arc[x radius=0.85, y radius=0.85, start angle=0, end angle=360];
\node at (1.75,-0.5) {$+$};
  \end{tikzpicture}  
  \end{subfigure}\hspace{-1.75cm}
  \begin{subfigure}[b]{0.33\textwidth}
  \centering
  \begin{tikzpicture}
  \begin{feynman}
   \vertex (c) at (-0.75,-1);
  \node[dot] (m) at (-0.375, -0.75);
   \node[dot] (e) at (0.375, -0.75);
   \vertex (b) at (0.75, -1);
   \vertex (a) at (-0.75, 0);
   \vertex (f) at (0.75, 0);
  \diagram*{
  (c) -- (m) -- (f),
  (b) -- (e) -- (a),
   (m) -- (e)
  }; 
  \end{feynman}
\draw [dashed] (0.75,-0.5) arc[x radius=0.75, y radius=0.75, start angle=0, end angle=360];
\node at (1.75,-0.5) {$=0$};
  \end{tikzpicture}  
  \end{subfigure}
    \caption{The IHX relation. It holds whenever the graphs are identical outside the dashed circle.} \label{fig:IHX}
\end{figure}
Explicitly, the sum of the three contributions is equal to the expression
\begin{equation}
    \eta^{\Bar{L}}_0\eta^{\Bar{L'}}_0\omega^{IJ}\bigg(\Riem{}{IJ}{K}{\Bar{L}}\Riem{}{I'J'}{K'}{\Bar{L'}}+\Riem{}{IJ}{K'}{\Bar{L}}\Riem{}{I'J'}{K}{\Bar{L'}}+\Riem{}{IJ}{J'}{\Bar{L}}\Riem{}{I'K}{K'}{\Bar{L'}}\bigg)- \Bar{L}\leftrightarrow \Bar{L'},
\end{equation}
with the notation $\Bar{L}\leftrightarrow \Bar{L'}$ we mean that we are subtracting the same quantity with the indices $\Bar{L}$ and $\Bar{L}'$ switched, in this way the expression vanishes. In other words, the IHX relation follows as a result of the Bianchi identity for the curvature tensor $R$. The validity of the IHX relation ensures we are obtaining topological invariants of 3-manifolds in the perturbative expansion of the partition function \cite{Sa04}. 

At this point, we can take the product of the $(0,2n)$-form which is the image of the graph $\Gamma$ with the $(2n,0)$-form $\omega^n\in H^{2n,0}(M)$ and integrate the resulting $(2n,2n)$-form over $M$. In this way we obtain the weights $b_\Gamma(M)$, which are numbers called \textit{Rozansky--Witten invariants} studied by Sawon in \cite{Sa04}. More explicitly, we have
\begin{equation}
    b_\Gamma(M)= \frac{1}{(2\pi)^{2n}} \varint_M \text{W}_\Gamma(M;\phi^i_0)\omega^n.
\end{equation}

Finally, the RW partition function is shown to be \cite{RW96}
\begin{equation}
    Z_{M}(\Sigma_3)=\Big|H_1(\Sigma_3, \mathbb{Z})\Big|'\sum_{\Gamma\in\Gamma_{n,3-b_1(\Sigma_3)}}b_\Gamma(M)\sum_aI_{\Gamma,a},
\end{equation}
where $\Big|H_1(\Sigma_3, \mathbb{Z})\Big|'$ is the number of torsion elements in $H_1(\Sigma_3, \mathbb{Z})$ (see \cite{FG91}). 

\subsection{Comparison with Chern--Simons theory}
In this section, we are going to explore briefly the similarities between CS theory and the RW theory as exhibited in \cite{RW96}. The main message is that RW is a kind of ``Grassmann odd version" of CS theory. Let us make this more precise. 

Recall the CS Lagrangian
\begin{equation}
    \mathscr{L}_{\text{CS}}=\Tr\bigg(AdA+\frac{2}{3}A^3\bigg).
\end{equation}
Let us compare it with the RW Lagrangian in \eqref{RW:L_2}. As we can see from Table \ref{Tab:RW_CS} (where we denote by $T_a$ the generators of the Lie algebra and $f_{abc}$ the structure constants), there is almost a direct match. We are using the word ``almost" because the symmetry properties of the various objects in the table are reversed: $\Tr T_aT_b$ is symmetric in its arguments while the holomorphic symplectic form $\omega_{IJ}$ is antisymmetric, $f_{abc}$ is totally antisymmetric whereas $\Riem{}{IJ}{K}{L}\eta^L_0$ is totally symmetric. However, this should not come as a surprise since by definition $A^a$ is an anti-commuting object, while $\chi^I$ is commuting.
\begin{table}[h]
\centering
\begin{tabular}{ |cc| } 
 \toprule
 CS & RW \\
 \midrule
 $A^a$ & $\chi^I$ \\ 
 $\Tr T_aT_b$ & $\omega_{IJ}$ \\
 $f_{abc}$ & $\Riem{}{IJ}{K}{L}\eta^L_0$ \\
 \bottomrule
\end{tabular}
\caption{Comparison between CS theory and RW theory.}
\label{Tab:RW_CS}
\end{table}

By doing the associations in the table, the vertex in CS is the same as the vertex in RW. It follows that the diagrams of the two theories coincide. Consequently, the partition function differs only for the weight factors since for RW they are proportional to the curvature tensor of $M$ rather than the structure constants of a Lie group. 

Other similarities come at the level of gauge-fixing. We refer the interested reader to \cite{RW96,HT99} for a detailed discussion.


\begin{rmk}
\label{rmk_diff_rw_cs}
There is an important difference between CS and RW theories. In the RW model, the vertex carries an odd Grassmann odd harmonic mode $\eta^I_0$, hence it can never appear more than $2n$ times in any diagram. This corresponds to a natural \textit{cut-off} of the perturbative expansion of the RW model.
\end{rmk}

\section{Classical formal globalization}
\label{sec:Classical_Theory}
The idea is to construct a 3-dimensional topological sigma model, which, when globalized, reduces to the original RW model. In particular, we are interested in the formulation of the RW model with target a holomorphic symplectic manifold\footnote{This construction reduces to the RW model considered in the bulk of \cite{RW96}, when we consider as target a hyperK{\"a}hler manifold.} (i.e. a complex symplectic manifold with a holomorphic symplectic form, see also the appendix in \cite{RW96}). Hence, let $M$ be a holomorphic symplectic manifold endowed with coordinates $X^i$ and $X^{\Bar{i}}$, and with a holomorphic symplectic form $\Omega=\Omega_{ij}\delta X^i\delta X^j$, i.e. a closed, non degenerate $(2,0)$ form. Moreover, consider a 3-dimensional manifold $\Sigma_3$ and construct an AKSZ sigma model\footnote{See Section \ref{subsec:AKSZ} for an introduction.} with source $T[1]\Sigma_3$ and target $M$. In this case, the space of maps is
\begin{equation}
    \mathcal{F}\sur\coloneqq \Maps\bigg(T[1]\Sigma_3, M\bigg).
\end{equation}

On the source manifold, we choose bosonic coordinates $\{u\}$ (ghost degree 0) on $\Sigma_3$ and fermionic odd coordinates $\{\theta\}$ (ghost degree 1) on the fibers of $T[1]\Sigma_3$. Moreover, by picking up local coordinates $X^i$ on $M$, maps in $\mathcal{F}\sur$ can be described by a superfield $\mathbf{X}$, whose components are chosen as:
\begin{equation}
    \mathbf{X}^i=X^i(u)+\theta^\mu X^i_\mu(u)+\frac{1}{2!}\theta^\mu\theta^\nu X^i_{\mu\nu}(u)+\frac{1}{3!}\theta^\mu\theta^\nu\theta^\chi X^i_{\mu\nu\chi}(u),
\end{equation}
where $X^i$ is a 0-form, $X^i_\mu$ is a 1-form etc. 
To these maps maps $X^i, X^i_\mu, \dots$ we assign ghost degrees such that the ghost degree of $\mathbf{X}$ is equal to the one of $X^i$ (that is 0), for example $X^i_{\mu\nu}$ has form degree 2 and ghost degree $-2$.



Now, we can define a symplectic form for the space of fields. Since it should have ghost degree $-1$, we assign ghost degree\footnote{More precisely, we add a formal parameter $q$ of ghost degree 2 in front of $\Omega_{ij}$. The parameter is immediately suppressed from the notation.} 2 to $\Omega_{ij}$ (in this way the target manifold has degree 2 and the AKSZ construction can be done without problems) and define
\begin{equation}
    \omega\sur=\varint_{T[1]\Sigma_3}\mu\sur\bigg(\frac12\Omega_{ij}\delta \mathbf{X}^i\delta \mathbf{X}^j\bigg),
\end{equation}
where by $\delta$ we denote the de Rham differential on the space of fields. Since we have a canonical Berezinian $\mu_{\Sigma_3}$ on $T[1]\Sigma_3$ of degree $-3$, the symplectic form has degree $-1$ as we desired. Hence, the space of fields is equipped with an odd Poisson bracket $(-,-)$.
We have an associated AKSZ action given by
\begin{equation}
        \Sc\sur=\varint_{T[1]\Sigma_3}\mu_{\Sigma_3}\bigg(\frac{1}{2}\Omega_{ij}\mathbf{X}^iD\mathbf{X}^j\bigg),
\end{equation}
where $D=\theta^\mu\frac{\partial}{\partial u^\mu}$ is the differential on $T[1]\Sigma_3$. When $\Sigma_3$ is a closed manifold, the action $\Sc\sur$ satisfies the CME
\begin{equation}
    (\Sc\sur,\Sc\sur)=0.
\end{equation}
Equivalently, we can introduce a cohomological Hamiltonian vector field $Q\sur$ on $\mathcal{F}\sur$ defined by
\begin{equation}
    \iota_{Q\sur}\omega\sur=\delta S\sur.
\end{equation}
This vector field has the following form
\begin{equation}
    Q\sur=\varint_{T[1]\Sigma_3}\mu_{\Sigma_3}\bigg( D\mathbf{X}^i\frac{\delta}{\delta \mathbf{X}^i}\bigg).
\end{equation}

The above can be restated by saying that $(\mathcal{F}\sur, \omega\sur, \Sc\sur)$ is a BV manifold.

In the presence of boundaries, the model can be extended to a BV-BFV theory by associating the BV-BFV manifold
\begin{equation}
    (\mathcal{F}\bon, \omega\bon=\delta \alpha\bon, \Sc\bon, Q\bon)
\end{equation}
over the BV manifold $(\mathcal{F}_{\Sigma_3}, \omega_{\Sigma_3}, Q_{\Sigma_3})$, with the following set of data
\begin{equation}
    \begin{split}
        &\mathcal{F}\bon=\Maps (T[1]\partial \Sigma_3, X), \\
        &\Sc\bon=\varint_{T[1]\partial\Sigma_3}\mu_{\partial \Sigma_3}\bigg(\frac{1}{2}\Omega_{ij}\mathbf{X}^iD\mathbf{X}^j\bigg),\\
        &\alpha\bon=\varint_{T[1]\partial\Sigma_3}\mu_{\partial \Sigma_3}\bigg(\frac12\Omega_{ij}\mathbf{X}^i\delta \mathbf{X}^j\bigg),\\
        &Q\bon=\varint_{T[1]\partial\Sigma_3}\mu_{\partial \Sigma_3}\bigg(D\mathbf{X}^i\frac{\delta}{\delta \mathbf{X}^i}\bigg),
    \end{split}
\end{equation}
with $\mu_{\partial \Sigma_3}$ the Berezinian on the boundary $\partial \Sigma_3$ of degree $-2$. The data is such that
\begin{equation}
    \iota_{Q\sur}\omega\sur=\delta \Sc\sur+\pi^*\alpha\bon.
\end{equation}

\begin{rmk}
A possible modification of the model consists in coupling the target manifold with $\mathfrak{g}^{\vee}[1]\otimes \mathfrak{g}[1]$ or $\mathfrak{g}[1]$, forming thus the ``$BF$-RW" model and the ``CS-RW" model \cite{KQZ13}, respectively. In this way, after globalization, one should get an extension of the results obtained by K{\"a}llén, Qiu and Zabzine \cite{KQZ13}.
\end{rmk}


\subsection{Globalization}
\label{sec:globalization}
In the last section, we introduced a very simple AKSZ sigma model. Here we globalize that construction using methods of formal geometry \cite{GK71,Bo11} (see Appendix \ref{app:formal_geometry} for an introduction) following \cite{CMoW19}. First, we expand around critical points of the kinetic part of the action. The Euler--Lagrange equations for our model are simply $d\mathbf{X}^i=0$, which means that the component of $\mathbf{X}^i$ of ghost degree 0 is a constant map: we denote it by $x^i$ and we think of it as a \textit{background field} \cite{Moshayedi2021}. Moreover, since we want to vary $x$ itself, we lift the fields as the pullback of a formal exponential map at $x$. We also note that the fields $\mathbf{X}^{\Bar{i}}$ are just \emph{spectators}, which means that they do not contribute to the action, hence we can think of taking constant maps also in the antiholomorphic direction.


The above allows to linearize the space of fields $\mathcal{F}\sur$ by working in the formal neighbourhoods of the constant map $x\in M$. We define the following \textit{holomorphic formal exponential map}
\begin{align}
\begin{split}
\varphi:\ T^{1,0}M &\rightarrow M\\
(x,y)&\mapsto \varphi^{i}(x,y)=x^i+y^i+\frac{1}{2}\varphi^i_{jk}(x^i,x^{\Bar{i}})y^jy^k+\dots
\end{split}
\end{align}
\begin{rmk}
We think about the holomorphic formal exponential map here defined as an extension to the complex case of the formal exponential map used in e.g. \cite{CF01}. This notion should correspond to the ``canonical coordinates'' introduced in \cite{BCOV94} and the holomorphic exponential map applied by Kapranov to the RW case in \cite{KA99}. 
\end{rmk}

The formal exponential map lifts $\mathcal{F}\sur$ to
\begin{align}
\begin{split}
    \Tilde{\varphi}_{x}:\ \Tilde{\mathcal{F}}_{\Sigma_3, x}:=\Maps (T[1]\Sigma_3, T^{1,0}_{x}M)&\rightarrow \Maps (T[1]\Sigma_3, M)\\
    \hat{\mathbf{X}}&\mapsto \mathbf{X}
\end{split}
\end{align}
which is given by precomposition with $\varphi^{-1}_{x}$, i.e. $\Tilde{\mathcal{F}}_{\Sigma_3, x}=\varphi^{-1}_{x}\circ\mathcal{F}\sur$ and $\mathbf{X}=\varphi_{x}(\hat{\mathbf{X}})$. Now, since the target is linear, we can write the space of fields as
\begin{equation}
   \Tilde{\mathcal{F}}_{\Sigma_3, x}=\Omega^{\bullet}(\Sigma_3)\otimes T^{1,0}_{x}M.
\end{equation}
Consequently, we lift the BV action, the BV 2-form and the primitive 1-form obtaining:
\begin{equation}
    \begin{split}
        &\Sc\surg:=\mathrm{T}\Tilde{\varphi}_{x}^*\Sc\sur=\varint_{T[1]\Sigma_3}\mu\sur\, \bigg(\frac12\Omega_{ij}\hat{\mathbf{X}}^iD\hat{\mathbf{X}}^j\bigg),\\
        &\omega\surg\coloneqq\Tilde{\varphi}^*_{x}\omega\sur=\varint_{T[1]\Sigma_3}\mu\sur\,\bigg(\frac12\Omega_{ij}\delta \hat{\mathbf{X}}^i\delta \hat{\mathbf{X}}^j\bigg),\\
        &\alpha^{\partial}_{\partial \Sigma_3, x}\coloneqq \Tilde{\varphi}^*_{x}\alpha\bon=\varint_{T[1]\partial\Sigma_3}\mu_{\partial\Sigma_3}\,\bigg(\frac12\Omega_{ij}\hat{\mathbf{X}}^i\delta \hat{\mathbf{X}}^j\bigg),
    \end{split}
\end{equation}
where $\mathrm{T}$ denotes the Taylor expansion around the fiber coordinates $\{y\}$ at zero.
This set of data satisfies the mCME for any $x\in M$:
\begin{equation}
    \iota_{Q\surg}\omega\surg=\delta \Sc\surg+\pi^*\alpha^{\partial}_{\partial \Sigma_3, x},
\end{equation}
with $Q\surg=\varint_{T[1]\Sigma_3}\mu_{\Sigma_3}\, \Big(D\hat{\mathbf{X}}^i\frac{\delta}{\delta \hat{\mathbf{X}}^i}\Big)$. Hence, we have a BV-BFV manifold associated to the space of fields $\Tilde{\mathcal{F}}\surg$.

The next remark introduces an important ingredient to write down the globalized action.
\begin{rmk}
\label{rmk_R}
The constant map $x:T[1]\Sigma_3\rightarrow M$ in $\mathcal{F}\sur$ can be thought of as an element in $M$. Hence, we have a natural inclusion $M\hookrightarrow \mathcal{F}\sur$.  We exploit this fact by defining, for a constant field $x$ and $\mathbf{X}\in \mathcal{F}\sur$, a 1-form:
\begin{equation}
    R\sur=\Big(R\sur\Big)_j(x;\mathbf{X})dx^j+\Big(R\sur\Big)_{\Bar{j}}(x;\mathbf{X})dx^{\Bar{j}}\in \Omega^1\Big(M, \Der\Big(\reallywidehat{\Sym}^\bullet(T^{\vee1,0}M)\Big)\Big).
\end{equation}
As before, we lift this 1-form to $\Tilde{\mathcal{F}}\surg$. This lift, denoted by $\hat{R}\sur$, is locally written as:
\begin{equation}
     \hat{R}\sur=\Big(\hat{R}\sur\Big)_j(x;\hat{\mathbf{X}})dx^j+\Big(\hat{R}\sur\Big)_{\Bar{j}}(x;\mathbf{X})dx^{\Bar{j}}.
\end{equation}
\end{rmk}

\subsection{Variation of the classical background}
\label{class:sec:var_class_back}
So far, the classical background $x$ has been fixed. However, our aim is to vary $x$ and construct a global formulation of the action. Hence, we understand the collection $\{\Sc\surg\}_{x\in M}$ as a map $\hat{\Sc}\sur$ to be given by $\hat{\Sc}\sur: x\mapsto S\surg$ and we compute how it changes over $M$. In order to accomplish this task, inspired by \cite{CMoW19,CMoW20, KQZ13, BCM12}, choosing a background field $x\in M$, we define
\begin{equation}
\label{class:glob_terms}
    \Sc\surgR\coloneqq\varint_{\Sigma_3}\bigg(\Big(\hat{R}^i\sur\Big)_j(x; \hat{\mathbf{X}})\Omega_{il}\hat{\mathbf{X}}^l dx^j+\Big(\hat{R}^i\sur\Big)_{\Bar{j}}(x; \hat{\mathbf{X}})\Omega_{il}\hat{\mathbf{X}}^l dx^{\Bar{j}}\bigg)=\mathcal{S}_R+\Sc_{\bar{R}}.
\end{equation}
The integrand is a well defined term of degree 3, since we assigned degree 2 to the symplectic form and $\hat{R}\sur$ is a 1-form on $M$. After integration, $\Sc\surgR$ is then of total degree 0.

The term $\hat{R}\sur$ has been introduced in Remark \ref{rmk_R}. However, its connection with the globalization procedure is not clear. To explain it, we introduce the \textit{classical Grothendieck connection} adapted to our case (see Appendix \cite{CF01}).
    
\begin{defn}[Classical Grothendieck connection]
    Given a holomorphic formal exponential map $\varphi$, we can define the associated \textit{classical Grothendieck connection} on $\reallywidehat{\Sym}^\bullet(T^{\vee1,0}M)$, given by $\Gr \coloneqq d_M+R$, where $d_M$ is the sum of the holomorphic and antiholomorphic Dolbeault differentials on $M$ and $R\in \Omega^1\Big(M, \Der\Big(\reallywidehat{\Sym}^\bullet(T^{\vee1,0}M)\Big)\Big)$. By using local coordinates $\{x\}$ on the basis and $\{y\}$ on the fibers, we have $R=R_j(x;y)dx^j+R_{\bar{j}}(x;y)dx^{\bar{j}}$, where $R_j=R^i_j(x;y)\frac{\partial}{\partial y}$ and $R_{\bar{j}}=R^i_{\bar{j}}(x;y)\frac{\partial}{\partial y}$ with
    \begin{equation}
    \begin{split}
        &R^i_j(x;y)dx^j:=-\bigg[\bigg(\frac{\partial \varphi}{\partial y}\bigg)^{-1}\bigg]^i_p\frac{\partial \varphi^p}{\partial x^j}dx^j,\\
        &R^i_{\Bar{j}}(x;y)dx^{\Bar{j}}:=-\bigg[\bigg(\frac{\partial \varphi}{\partial y}\bigg)^{-1}\bigg]^i_p\frac{\partial \varphi^p}{\partial x^{\Bar{j}}}dx^{\Bar{j}}.
    \end{split}
\end{equation}
    \end{defn}

Note that $R^i_j(x;y)$ and $R^i_{\bar{j}}(x;y)$ are formal power series in the second argument, namely
\begin{equation}
    \begin{split}
        R^i_j(x;y)&=\sum^{\infty}_{k=0} R^i_{j;j_1,\dots,j_k}(x)y^{j_1}\dots y^{j_k},\\
        R^i_{\bar{j}}(x;y)&=\sum^{\infty}_{k=0} R^i_{{\bar{j}};j_1,\dots,j_k}(x)y^{j_1}\dots y^{j_k}.
    \end{split}
\end{equation}

\begin{rmk}
\label{class:RMK_Grothendieck_prop}
The classical Grothendieck connection has a couple of important properties:
\begin{itemize}
    \item It is flat, which can be rephrased by saying that the following equation is satisfied
    \begin{equation}
    \label{class:flatnessR}
        d_MR+\frac12[R,R]=0.
    \end{equation}
    \item A section $\sigma$ is closed under $\mathcal{D}_{\text{G}}$ i.e. $\mathcal{D}_{\text{G}}\sigma=0$ if and only if $\sigma=\mathrm{T}\varphi^*_xf$, where $f\in \mathcal{C}^\infty(M)$. 
\end{itemize}
In more down-to-Earth terms, the second property says that the classical Grothendieck connection selects those sections which are global.
\end{rmk}

Finally, we can clarify the relation between $\hat{R}\sur$ and the Grothendieck connection. The components $\Big(\hat{R}\sur^i\Big)_j(x;\hat{\mathbf{X}})$ and $\Big(\hat{R}\sur^i\Big)_{\bar{j}}(x;\hat{\mathbf{X}})$ are given by the components of the classical Grothendieck connection $R^i_j(x;y)$ and $R^i_{\bar{j}}(x;y)$ evaluated in the second argument at $\hat{\mathbf{X}}$. 

Having set up all the necessary tools, we can compute how $\hat{\Sc}\sur$ varies when we change the background $x\in M$. On a closed manifold, we have
\begin{equation}
\label{dx_eq}
    d_M \hat{\Sc}\sur=-(\Sc\surgR, \hat{\Sc}\sur),
\end{equation}
which follows from the Grothendieck connection and that $\Sc\surg=\mathrm{T}\varphi^*_x\Sc\sur$.

The above identities can be collected in a nicer way via  the following definition.

\begin{defn}(Formal global action) The \textit{formal global action} for the model is defined by
\begin{equation}
\label{glob_action}
    \begin{split}
        \Tilde{\Sc}\surg&\coloneqq\varint_{\Sigma_3}\bigg(\frac{1}{2}\Omega_{ij}\hat{\mathbf{X}}^id\hat{\mathbf{X}}^j+\Big(\hat{R}^i\sur\Big)_j(x; \hat{\mathbf{X}})\Omega_{il}\hat{\mathbf{X}}^l dx^j+\Big(\hat{R}^i\sur\Big)_{\Bar{j}}(x; \hat{\mathbf{X}})\Omega_{il}\hat{\mathbf{X}}^l dx^{\Bar{j}}\bigg)\\
    &=\hat{\Sc}\sur+\underbrace{\Sc_R+\Sc_{\bar{R}}}_{\coloneqq\Sc\surgR}.
    \end{split}
\end{equation}
\end{defn}

By using the formal global action, the \textit{differential Classical Master Equation} (dCME) is satisfied 
\begin{equation}
\label{class:dCME}
    d_M\Tilde{\Sc}\surg+\frac{1}{2}(\Tilde{\Sc}\surg,\Tilde{\Sc}\surg)=0.
\end{equation}

\begin{rmk}
Note that $\Tilde{\Sc}\surg$ is an inhomogeneous form over $M$, where $\hat{S}\sur$ is a 0-form and $S\surgR$ is a 1-form. Therefore, Eq. \eqref{class:dCME} has a 0-form, a 1-form and a 2-form part. Specifically, the 0-form part
\begin{equation}
    (\hat{\Sc}\sur,\hat{\Sc}\sur)=0,
\end{equation}
is the usual CME. The 1-form part:
\begin{equation}
    d_M\hat{\Sc}\sur+(\Sc\surgR,\hat{\Sc}\sur)=0,
\end{equation}
means that $\hat{\Sc}\sur$ is a global object (see Remark \ref{class:RMK_Grothendieck_prop}). The 2-form part
\begin{equation}
    d_M\Sc\surgR+\frac{1}{2}(\Sc\surgR,\Sc\surgR)=0,
\end{equation}
means that the operator $\mathcal{D}_G$ is flat connection (see Eq. (\ref{class:flatnessR})). Explicitly, we have
\begin{align}
        &d_x\Sc_R+\frac12(\Sc_R,\Sc_R)=0 \label{dcme_1},\\
        &d_x\Sc_{\bar{R}}+\frac12(\Sc_R,\Sc_{\bar{R}})=0\label{dcme_2},\\
        &d_{\bar{x}}\Sc_R+\frac12(\Sc_{\bar{R}},\Sc_R)=0 \label{dcme_3},\\ 
        &d_{\bar{x}}\Sc_{\bar{R}}+\frac12(\Sc_{\bar{R}},\Sc_{\bar{R}})=0. \label{dcme_4}
\end{align}
\end{rmk}

Let $\Sigma_3$ be (again) a manifold with boundary. The BV-BFV theory on $\Tilde{\mathcal{F}}\surg$ furnishes the cohomological vector field $Q\surg$. Moreover, by using the lift of $\hat{R}\sur$, we can define 
\begin{equation}
    \Tilde{Q}\surg=Q\surg+\hat{R}\sur.
\end{equation}
Then, the \textit{modified differential Classical Master Equation} (mdCME) is satisfied:
\begin{equation}
\label{mdcme}
    \iota_{\Tilde{Q}\surg}\omega\surg=\delta \Tilde{\Sc}\surg+\pi^*\alpha^{\partial}_{\partial \Sigma_3, x},
\end{equation}
where
\begin{equation}
    \begin{split}
        \Tilde{Q}\surg=\varint\sur\bigg(-d\hat{\mathbf{X}}\frac{\delta}{\delta \hat{\mathbf{X}}}&-\Omega^{pq}\frac{\delta \Big(\hat{R}^i\sur\Big)_j(x; \hat{\mathbf{X}})}{\delta \hat{\mathbf{X}}^p}\Omega_{il}\hat{\mathbf{X}}^l  dx^j\frac{\delta}{\delta \hat{\mathbf{X}}^q}- 
    \Big(\hat{R}^i\sur\Big)_j(x; \hat{\mathbf{X}})dx^j\Omega_{ip}\frac{\delta}{\delta \hat{\mathbf{X}}^p}\\
    &-\Omega^{pq}\frac{\delta \Big(\hat{R}^i\sur\Big)_{\bar{j}}(x; \hat{\mathbf{X}})}{\delta \hat{\mathbf{X}}^p}\Omega_{il}\hat{\mathbf{X}}^l  dx^{\Bar{j}}\frac{\delta}{\delta \hat{\mathbf{X}}^q}-
    \Big(\hat{R}^i\sur\Big)_{\bar{j}}(x; \hat{\mathbf{X}})dx^{\Bar{j}}\Omega_{ip}\frac{\delta}{\delta \hat{\mathbf{X}}^p}\bigg).
    \end{split}
\end{equation}

In preparation for the comparisons we will draw in the following section, we redefine the components $\Big(\hat{R}^i\sur\Big)_j$ and $\Big(\hat{R}^i\sur\Big)_{\bar{j}}$ by a multiplicative factor $1/k!$  as
\begin{equation}
\label{redef_R}
    \begin{split}
        \Big(\hat{R}^i\sur\Big)_j(x;\hat{\mathbf{X}})&=\sum^{\infty}_{k=0} \frac{1}{(k+1)!}\hat{R}^i_{j;j_1,\dots,j_k}(x)\hat{\mathbf{X}}^{j_1}\dots \hat{\mathbf{X}}^{j_k},\\
       \Big(\hat{R}^i\sur\Big)_{\bar{j}}(x;\hat{\mathbf{X}})&=\sum^{\infty}_{k=0}\frac{1}{(k+1)!} \hat{R}^i_{{\bar{j}};j_1,\dots,j_k}(x)\hat{\mathbf{X}}^{j_1}\dots \hat{\mathbf{X}}^{j_k}.
    \end{split}
\end{equation}

\section{Comparison with the original Rozansky--Witten model}
\label{sec:comp_orig_RW}
 
In this section, we show that the globalized model we have just constructed reduces to the RW model (Section \ref{sec:RW_model}) and, moreover it provides a globalization of the former. 

In order to compare effectively these models, we need to be more explicit about the terms involved in the classical Grothendieck connection. First, we discuss the choice of holomorphic formal exponential map in more detail. Since our target is a symplectic manifold, we choose the formal exponential map which preserves the symplectic form considered in \cite{QZ15} and we adapt it to our case, i.e. 
\begin{equation}
\label{comp:exp_map_zq}
    \varphi^{i}=x^i+y^i-\frac{1}{2}\Chr{i}{j}{k}y^jy^k+\bigg\{-\frac{1}{6}\partial_c\Chr{i}{j}{k}+\frac{1}{3}\Chr{i}{m}{c}\Chr{m}{j}{k}-\frac{1}{24}\Riem{i}{c}{j}{k}\bigg\}y^cy^jy^k+O(y^4),
\end{equation}
where $\Riem{i}{c}{j}{k}=(\Omega^{-1})^{bi}R^{\hspace{2mm} a}_{bc\; k}\Omega_{aj}$.

The Grothendieck connection is then
\begin{equation}
\label{Grothendieck_con}
    \Gr=dx^i\frac{\partial}{\partial x^i}+dx^{\bar{i}}\frac{\partial}{\partial x^{\bar{i}}}+
    dx^j\Big(R\sur\Big)_j +dx^{\Bar{j}}\Big(R\sur\Big)_{\Bar{j}},
\end{equation}
where the third term on the right hand side was computed in \cite{QZ15}, 
\begin{equation}
    \Big(\hat{R}^i\sur\Big)_{j} dx^j=-\bigg[\bigg(\frac{\partial \varphi}{\partial y}\bigg)^{-1}\bigg]^i_p\frac{\partial \varphi^p}{\partial x^j}dx^j=-\bigg[dx^j\bigg(\delta^i_j+\Chr{i}{k}{j}y^k-\bigg(\frac{1}{8}R^{\hspace{1.5mm} i}_{j\; ks}+\frac{1}{4}R^{\hspace{2.5mm} i}_{jk\; s}\bigg)y^ky^s\bigg)\frac{\partial}{\partial y^i}+\dots\bigg],
\end{equation}
whereas the fourth term is
\begin{equation}
     \begin{split}
        \Big(\hat{R}^i\sur\Big)_{\bar{j}}dx^{\Bar{j}}&=-\bigg[\bigg(\frac{\partial \varphi}{\partial y}\bigg)^{-1}\bigg]^i_p\frac{\partial \varphi^p}{\partial x^{\Bar{j}}}dx^{\Bar{j}} =\bigg[-\frac{1}{2}\Chr{p}{ab}{,\Bar{j}}y^ay^bdx^{\Bar{j}}\bigg][\delta^i_p+\dots]       =-\bigg[-\frac{1}{2}\Chr{i}{ab}{,\Bar{j}}y^ay^bdx^{\Bar{j}}+\dots\bigg]\\
         &=\Riem{i}{a}{b}{\bar{j}}y^ay^bdx^{\bar{j}}-\dots .
     \end{split}
\end{equation}

Considering the terms coming from the classical Grothendieck connection and the redefinition \eqref{redef_R}, we can re-write the formal global action \eqref{glob_action} as 
\begin{equation}
\label{class:expl_S_global}
    \Tilde{\Sc}\surg=\varint\sur\bigg(\frac12 \Omega_{ij}\hat{\mathbf{X}}^id\hat{\mathbf{X}}^j-\frac12\Chr{i}{j}{k}\hat{\mathbf{X}}^k\Omega_{il}\hat{\mathbf{X}}^ldx^j-\delta^i_j\Omega_{il}\hat{\mathbf{X}}^ldx^j+\dots+\frac{1}{3!}\Riem{i}{k}{s}{\bar{j}}\hat{\mathbf{X}}^k\hat{\mathbf{X}}^s\Omega_{il}\hat{\mathbf{X}}^ldx^{\bar{j}}+\dots\bigg).
\end{equation}

For convenience, we recall the RW action \cite{RW96}
 \begin{equation}
 \label{class:S_RW_original}
         S_{\text{RW}}=\varint\sur\frac12\frac{1}{\sqrt{h}}\epsilon^{\mu\nu\rho}\bigg(\Omega_{IJ}\chi_\mu^I\nabla_\nu\chi_\rho^J-\frac13\Omega_{IJ}R^J_{KL\Bar{M}}\chi_\mu^I\chi_\nu^K\chi_\rho^L\eta^{\Bar{M}}+\frac13(\nabla_L\Omega_{IK})(\partial_\mu\phi^I_\perp)\chi_\nu^K\chi_\rho^L \bigg).
\end{equation}
If we assume that the connection is compatible with the symplectic form, the third term in the RW action \eqref{class:S_RW_original} drops. We are left with the first two terms. By associating $\hat{\mathbf{X}}^i\leftrightarrow \chi^I$ and $dx^{\bar{j}}\leftrightarrow \eta^{\bar{M}}$, we can sum up the comparison in Table \ref{class:Tab.model_RW}. 
\begin{table}[h!]
    \centering
    \begin{tabular}{|lcc|}
    \toprule
         & Kinetic term &  Interaction term \\
         \midrule
    Original RW model & $\frac12\frac{1}{\sqrt{h}}\epsilon^{\mu\nu\rho}\Omega_{IJ}\chi_\mu^I\nabla_\nu\chi_\rho^J$     &  $-\frac{1}{3!}\Omega_{IJ}R^J_{KL\Bar{M}}\chi_\mu^I\chi_\nu^K\chi_\rho^L\eta^{\Bar{M}}$ \\
    Our model & $\frac12 \Omega_{ij}\hat{\mathbf{X}}^id\hat{\mathbf{X}}^j-\frac12\Chr{i}{j}{k}\hat{\mathbf{X}}^k\Omega_{il}\hat{\mathbf{X}}^ldx^j$     &  $\frac{1}{3!}\Riem{i}{k}{s}{\bar{j}}\hat{\mathbf{X}}^k\hat{\mathbf{X}}^s\Omega_{il}\hat{\mathbf{X}}^ldx^{\bar{j}}$ \\
    \bottomrule
    \end{tabular}
    \caption{Comparison between kinetic term and interaction term for the RW theory and our model.}
    \label{class:Tab.model_RW}
\end{table}

The sign discrepancy comes from having defined the connection as $\nabla=d-\Gamma$ which gives a negative sign in front of the $\Chr{i}{j}{k}$ (see Eq. \eqref{comp:exp_map_zq}).

Moreover, when the curvature misses the $(2,0)$-part (which could happen when we have a Hermitian metric), the remaining terms in our model are just the perturbative expansion of $\Riem{i}{k}{s}{\bar{j}}$ around $x$. If we cut off the expansion at the first order, we are left with the original RW model.

\section{Comparison with other globalization constructions}
\label{sec:comp_other}
In the next sections, we are going to compare our globalization model with other three constructions: the first by \cite{CLL17}
uses tools of derived geometry to linearize the space of fields in the neighbourhood of a constant map as well as the Fedosov connection \cite{Fe49}, the second \cite{Ste17} is an extension to a manifold with boundary of the first procedure, while the third \cite{QZ09,QZ10,KQZ13} uses an approach similar to ours. 

\subsection{Comparison with the CLL construction}
\label{comp:sec:summary}
We compare our model with the formulation of the RW model constructed in \cite{CLL17} in the setting of derived geometry (see Appendix \ref{app:derived_geometry}). 

Let $\Sigma_3$ be a closed $3$ dimensional manifold and $M$ be a holomorphic symplectic manifold with a non-degenerate holomorphic $2$-form $\omega$. 

To determine fields we use the language of $\Linf$-spaces (see \cite{Co11a,Co11b} for an introduction) and we define the the space of fields as 
\begin{equation}
\Maps(\MdR,M_{\Bar{\partial}}),
\end{equation}
where $\MdR$ is the elliptic ringed space equipped with a sheaf of differential forms over $\Sigma_3$, i.e. 
$\Omega^{\bullet}(\Sigma_3)$ and $M_{\Bar{\partial}}=(M,\g_M)$ is a sheaf of $\Linf$-algebras, where $\g_M=\Omega^{\bullet,\bullet}(M)\otimes T^{1,0}M[-1]$ with $T^{1,0}M$ the holomorphic tangent bundle.
Since the critical points of the action functional are constant maps from $\Sigma_3$ to $M$, we are going to study $\Maps(\MdR,M_{\Bar{\partial}})$ in the neighbourhood of a constant map $x\in\Maps(\MdR,\Xd)$, namely 
\begin{equation}
    \mathcal{F}_{\text{CLL}}\coloneqq \reallywidehat{\Maps}(\MdR,\Xd)=\Omega^{\bullet}(\Sigma_3)\otimes \g_M[1],
\end{equation}
with $\reallywidehat{\Maps}(\MdR,\Xd)$ defined as in \cite{Co11a,Co11b}. 

Having specified the space of fields, the shifted symplectic structure is given by
\begin{equation}
\label{class:sympl_cll}
\begin{split}
\braket{-,-}:\ &\FDg\otimes_{\Omega^{\bullet,\bullet}(M)}\FDg \rightarrow \Omega^{\bullet,\bullet}(M)[-1]\\
&\braket{\alpha\otimes g_1,\beta\otimes g_2}:=\underbrace{\omega(g_1,g_2)}_{\text{sympl. struct. on $M$}}\,\,\,\varint\sur\alpha\wedge \beta,
\end{split}
\end{equation}
where $\Omega^{\bullet,\bullet}(M)=\Gamma\left(\bigwedge^{\bullet}T^{\vee}M\right)$ is a section of the cotangent bundle.

Since $C^{\bullet}(\g_M):=\reallywidehat{\Sym}^\bullet_{\Omega^{\bullet,\bullet}(M)}(\g^{\vee}_M[1])=\Omega^{\bullet,\bullet}(M)\otimes_{\mathcal{C}^{\infty}(M)}\reallywidehat{\Sym}^\bullet_{\mathcal{C}^{\infty}(M)}(T^{\vee1,0}M)$, to construct the action functional and to find our $L_{\infty}$-algebra we can use a procedure similar to the Fedosov's construction of a connection on a symplectic manifold \cite{F94}.
Let us denote the sections of the \textit{holomorphic Weyl bundle} on $M$ by
\begin{equation}
    \mathcal{W}=\Omega^{\bullet,\bullet}(M)\otimes_{\mathcal{C}^{\infty}(M)}\CS[\![\hbar ]\!]
\end{equation}
where $\CS[\![\hbar ]\!]$ is the completed symmetric algebra over $T^{\vee1,0}M$, the holomorphic cotangent bundle which has a local basis $\{y^i\}$ with respect to the local holomorphic coordinates $\{x^i\}$.  We call the sub-bundle $\Omega^{p,q}(M)\otimes \Sym^r(T^{\vee1,0}M)$ of $\W$ its $(p,q,r)$ component, in particular we refer to $r$ as \textit{weight}. To $\hbar$ is assigned a weight of 2.

\begin{prp}[\cite{CLL17}]
\label{prp_flatness}
There is a connection on the holomorphic Weyl bundle of the following form
\begin{equation}
\label{Fedosov_conn}
    \mathcal{D}_\mathrm{F}=\nabla-\delta+\frac{1}{\hbar}[I,-]_{\W},
\end{equation}
which is flat modulo $\hbar$ and $[-,-]_{\mathcal{W}}$ is defined as in \cite{CLL17}. Here $I$ is a $1$-form valued section of the Weyl bundle of weight $\geq 3$, i.e. $I\in \bigoplus_{r\geq 3}\Gamma(\Sym^r(T^{\vee 1,0}M)\otimes T^\vee M)$, $\nabla$ is the extension to $\mathcal{W}$ of a connection on $T^{1,0}M$ which is compatible with the complex structure as well as with the holomorphic symplectic form and torsion free, $\delta=dx^i\wedge \frac{\partial}{\partial y^i}$ is an operator on $\mathcal{W}$.
\end{prp}

The connection $\Fe$ is called \textit{Fedosov connection} and it provides the $L_\infty$-structure on $\mathfrak{g}_M$. In these terms 
the action can be written as
\begin{equation}
\label{class:dg_action_2}
    \Sc_{\text{CLL}}=\frac{1}{2}\braket{d\sur\alpha,\alpha}+\sum^{\infty}_{k=0}\frac{1}{(k+1)!}\braket{\ell_k(\alpha^{\otimes k}),\alpha}
\end{equation}
with $\alpha \in \mathcal{F}_{\text{CLL}}$, $\braket{-,-}$ defined as in (\ref{class:sympl_cll}) and $\ell_k$ are the higher brackets in the $\Linf$-algebra, and $d_{\Sigma_3}$ the de Rham differential on the source $\Sigma_3$. We can read $\ell_0$ from the Fedosov connection in \eqref{Fedosov_conn}, i.e.
\begin{equation}
        \ell_0=-dx^i\frac{\partial}{\partial y^i}
\end{equation}

The $\Linf$-products $\ell_1$ and $\ell_2$ are computed in the next section, when we compare the Fedosov connection with the classical Grothendieck connection.

\begin{rmk}
The action in \eqref{class:dg_action_2} satisfies the CME $(\Sc_{\text{CLL}},\Sc_{\text{CLL}})=0$ (see \cite[Proposition 2.16]{CLL17}). Moreover, in \cite{CLL17} it was observed that this construction is the formal version of the original RW model in the case when the $(2,0)$-part of the curvature is zero (see \cite[Section 2.3]{CLL17}).  
\end{rmk}

\subsubsection{Comparison between the Fedosov connection and the classical Grothendieck connection}
The sufficient condition for the flatness of $\Fe$ (see the proof of Proposition \ref{prp_flatness} in \cite{CLL17})  implies that $I$ satisfies
\begin{equation}
\label{eqI}
    I=\delta^{-1}(R+\nabla I)+\frac{1}{\hbar}\delta^{-1}I^2,
\end{equation}
where $\delta^{-1}=y^i\cdot \iota_{\partial_{x^i}}$ (up to a normalization factor) is another operator on $\mathcal{W}$ and $R$ is the curvature tensor.

\begin{rmk}
Since $I$ is a $1$-form valued section of $\W$, we can decompose it into its holomorphic and antiholomorphic component respectively. In particular the antiholomoprhic part component is the Taylor expansion of the Atiyah class as noted in \cite{CLL17}. In the case $R^{2,0}=0$, the $\Linf$-algebra is fully encoded by Taylor expansion of the Atiyah class as first noted by Kapranov in \cite{KA99}.  
\end{rmk}

Since the operator $\delta^{-1}$ increases the weight by $1$, while $\nabla$ preserves the weight and $I$ has at least weight $3$, we can find a solution of the above equation with the following leading term (cubic term)\footnote{See also \cite{F00,GLL17}.}:
\begin{equation}
\label{deltaR}
    \delta^{-1}R=\frac18\big[-\Chr{}{ij}{k,r}+\Chr{}{si}{r}\Chr{}{pj}{k}\Omega^{sp}\big]y^i y^j y^r dx^k+\frac16\Riem{}{\Bar{k}r}{i}{j} y^i y^j y^r dx^{\Bar{k}}=\delta^{-1}R_t+\delta^{-1}\bar{R}.
\end{equation}
Since the Fedosov connection requires the computation of $\frac{1}{\hbar}[I,-]_{\W}$, we compute this commutator for the leading order term of $I$, which is the cubic term we have just found. For the first term on the right hand side of Eq. (\ref{deltaR}) we have
\begin{equation}
    \begin{split}
        \frac{1}{\hbar}[\delta^{-1}R_t,-]_{\mathcal{W}}&=\bigg[\frac{1}{8}\bigg(-\Chr{}{rj}{k,q}+\Chr{}{s}{rq}\Chr{}{p}{jk}\Omega^{sp}\bigg)+\frac{1}{4}\bigg(-\Chr{}{qj}{k,r}+\Chr{}{sq}{k}\Chr{}{pj}{r}\Omega^{sp}\bigg)\bigg]\Omega^{qi}y^j y^r dx^k\frac{\partial}{\partial y^i}\\
        &=\bigg[\frac18\bigg(-\Omega_{mr}\Chr{m}{j}{k,q}+\Omega_{mr}\Chr{m}{s}{k}\Chr{}{p}{jk}\Omega^{sp}\bigg)+\frac14\bigg(-\Omega_{mq}\Chr{m}{j}{k,r}+\Omega_{mq}\Chr{m}{s}{k}\Chr{}{pj}{r}\Omega^{sp}\bigg)\bigg]\times\\
        &\hspace{10cm} \times\Omega^{qi}y^j y^r dx^k\frac{\partial}{\partial y^i}\\
        &=\bigg[\frac18\Riem{\hspace{1.75mm} i}{k\ }{r}{j}+\frac14\Riem{\hspace{2.9mm}i}{k}{r\, }{j}\bigg]y^j y^r dx^k\frac{\partial}{\partial y^i}.
    \end{split}
\end{equation}
For the second term we have
\begin{equation}
    \frac{1}{\hbar}\big[\delta^{-1}\bar{R},-\big]_{\W}=\frac12\Chr{i}{j}{k,\Bar{r}}y^j y^k dz^{\Bar{r}}\frac{\partial}{\partial y^i}.
\end{equation}
After renaming some indices, the Fedosov connection is then
\begin{equation}
\label{class:Fed_conn_expl}
    \Fe=d_x+d_{\bar{x}}-dx^j\frac{\partial}{\partial y^j}-dx^j\Chr{i}{k}{j}y^k\frac{\partial}{\partial y^i}+dx^j\bigg(\frac18R^{\hspace{1.5mm} i}_{j\; ks}+\frac14R^{\hspace{2.5mm} i}_{jk\; s}\bigg)y^ky^s\frac{\partial}{\partial y^i}+\frac12dx^{\Bar{j}}\Chr{i}{ks}{,\Bar{j}}y^ky^s\frac{\partial}{\partial y^i}+\dots.
\end{equation}
More explicitly, 
\begin{equation}
\label{class:linf_prod}
    \begin{split}
        \ell_1&=-dx^j\Chr{i}{k}{j}y^k\frac{\partial}{\partial y^i},\\
        \ell_2&=dx^j\bigg(\frac18R^{\hspace{1.5mm} i}_{j\; ks}+\frac14R^{\hspace{2.5mm} i}_{jk\; s}\bigg)y^ky^s\frac{\partial}{\partial y^i}+\frac12dx^{\Bar{j}}\Chr{i}{ks}{,\Bar{j}}y^ky^s\frac{\partial}{\partial y^i}.
    \end{split}
\end{equation}
\begin{rmk}
The first terms in the Fedosov connection, explicitly written in (\ref{class:Fed_conn_expl}) coincide with the first terms for the classical Grothendieck connection (\ref{Grothendieck_con}). Furthermore, by substituting the explicit expressions of $\ell_1$ and $\ell_2$ in the action $\Sc_{\text{CLL}}$ (\ref{class:dg_action_2}), we can see that it coincides with the action $\Tilde{\Sc}\surg$ (\ref{class:expl_S_global}).
\end{rmk}

\subsubsection{Comparison between the CLL space of fields and globalization space of fields}

By rephrasing the argument of \cite[Section 6.1]{Mo20} to our context, we can extend the classical Grothendieck connection $\Gr$ to the complex 
\begin{equation}
\label{extensioncomplex}
    \Gamma\bigg(\bigwedge\nolimits^{\bullet}T^\vee M\otimes   \CS \bigg).
\end{equation}
which is the algebra of functions on the formal graded manifold
\begin{equation}
    T[1]M\bigoplus T^{1,0}M.
\end{equation}
This graded manifold is turned into a differential graded manifold by the classical Grothendieck connection $\Gr$. Moreover, since $\Gr$ vanishes on the body of the graded manifold, we can linearize at $x\in M$ and we get
\begin{equation}
    T_x[1]M\bigoplus T_x^{1,0}M.
\end{equation}
On this graded manifold, we have a curved $\Linf$-structure (which is the same as $\mathfrak{g}_M[1]$) and Eq. (\ref{extensioncomplex}) can be interpreted as the Chevalley--Eilenberg complex of the aforementioned $\Linf$-algebra. Then, the space of fields for the globalized theory can be rewritten as 
\begin{equation}
\label{extendedsof}
    \Tilde{\mathcal{F}}\surg=\Omega^{\bullet}(\Sigma_3)\otimes \Omega^{\bullet,\bullet}(M)\otimes T^{1,0}_xM
\end{equation}
which is the same as $\FDg$ by linearizing at $x\in M$ the holomorphic tangent bundle as $\mathcal{D}_\text{F}$ vanishes on $M$. 

\begin{rmk}
The idea that the classical Grothendieck connection and the Fedosov connection coincide is not new, in particular see \cite[Remark 3.6]{CMoW19} and \cite[Section 2.3]{CLL17}.
\end{rmk}
\begin{rmk}
Finally note that in \cite{CLL17} the source manifold $\Sigma_3$ was considered to be a closed manifold. As explained above (see Section \ref{class:sec:var_class_back}) our construction is valid also when $\partial\Sigma_3\neq \emptyset$. In the next section, we tackle this last setting by comparing our approach with \cite{Ste17}, where the derived geometric framework was implemented for manifolds with boundary.
\end{rmk}

\subsection{Comparison with Steffens' construction}
In \cite{Ste17}, Steffens applied the same derived geometry approach we have seen in the last section to what he calls \textit{AKSZ theories of Chern--Simons type}: CS theory and RW theories. In particular, his BV formulation of the RW model is completely analogue to the one in \cite{CLL17}: same space of fields, $\Linf$-algebra, action, etc. 
However, he takes a step further. He proves a \textit{formal AKSZ theorem} \cite[Theorem 2.4.1]{Ste17} in the context of derived geometry. His RW model is then shown to be an AKSZ theory by attaching degree 2 to the holomorphic symplectic form (as we did ourselves in Section \ref{sec:Classical_Theory}). Consequently, he provides a BV-BFV formulation for the RW model. The BFV action found in \cite{Ste17} is analogous to the action in (\ref{class:dg_action_2}) in one dimension less (as it is customary with AKSZ theories). Even if the $\Linf$ products are not explicit in his construction, by using the ones in (\ref{class:linf_prod}), his BV-BFV formulation of the RW model is visibly  identical to ours.

\subsection{Comparison with the (K)QZ construction}
Let $\Sigma_3$ be a 3-dimensional manifold and $M$ a hyperK{\"a}hler manifold with holomorphic symplectic form $\Omega$. Consider the symplectic graded manifold $\mathcal{M}\coloneqq T^{\vee0,1}[2]T^{\vee0,1}[1]M$ constructed out of $M$. It has the following coordinates: $X^i, X^{\Bar{i}}$ of degree 0 parametrizing $M$, $V^{\Bar{i}}$ of degree 1 parametrizing the fiber $T^{0,1}M$  and dual coordinates $P_{\Bar{i}}, Q_{\Bar{i}}$ of degree 2  and 1, respectively. The symplectic form is
\begin{equation}
\label{class:QZ_sympl_form}
    \omega_{\mathcal{M}}=dP_{\Bar{i}} \wedge dX^{\Bar{i}}+dQ_{\Bar{i}}\wedge dV^{\Bar{i}}+\frac{1}{2}\Omega_{ij}dX^i\wedge dX^j.
\end{equation} 
In order to have a ghost degree 2 symplectic form, the authors assign degree 2 to $\Omega$. With this setup, in \cite{QZ10,QZ09, KQZ13}, K{\"a}llén, Qiu and Zabzine construct an AKSZ model
\begin{align}
    \mathcal{F}_{\text{QZ}}&\coloneqq  \Maps(T[1]\Sigma_3, T^{\vee0,1}[2]T^{\vee0,1}[1]M)\\
    \label{comp:action_qz}
    \Sc_{\text{QZ}}&=\varint_{T[1]\Sigma_3}d^3zd^3\theta\bigg(\mathbf{P}_{\Bar{i}}D\mathbf{X}^{\Bar{i}}+\mathbf{Q}_{\Bar{i}}D\mathbf{V}^{\Bar{i}}+\frac{1}{2}\Omega_{ij}\mathbf{X}^iD\mathbf{X}^j+\mathbf{P}_{\Bar{i}}\mathbf{V}^{\Bar{i}}\bigg)
\end{align}
endowed with a cohomological vector field
\begin{equation}
    Q=\varint_{T[1]\Sigma_3}d^3zd^3\theta \bigg(D\mathbf{P}_{\Bar{i}}\frac{\partial}{\partial \mathbf{P}_{\Bar{i}}}+D\mathbf{Q}_{\Bar{i}}\frac{\partial}{\partial \mathbf{Q}_{\Bar{i}}}+D\mathbf{V}^{\Bar{i}}\frac{\partial}{\partial \mathbf{V}^{\Bar{i}}}+    D\mathbf{X}^i\frac{\partial}{\partial \mathbf{X}^i}+D\mathbf{X}^{\Bar{i}}\frac{\partial}{\partial \mathbf{X}^{\Bar{i}}}+\mathbf{P}_{\Bar{i}}\frac{\partial}{\partial \mathbf{Q}_{\Bar{i}}}+\mathbf{V}^{\Bar{i}}\frac{\partial}{\partial \mathbf{X}^{\Bar{i}}}\bigg),
\end{equation}
where to the source manifold $T[1]\Sigma_3$, we assign coordinates $\{z^i\}$ of ghost degree 0 and coordinates $\{\theta^i\}$ of degree 1.

\begin{rmk}
\label{rmk_comp_zq_rw}
With a suitable gauge-fixing consisting on a particular choice of Lagrangian submanifolds, the action $\Sc_{\text{QZ}}$ reduces to the RW model up to a factor of $\hbar$ (see \cite[Section 4]{QZ09}):
\begin{equation}
    \Sc_{\text{QZ}}\bigg|_{\text{GF}}=\frac{1}{2}\varint d^3z \bigg(\Omega_{ij} X^i_{(1)}\wedge d^{\nabla} X^{j}_{(1)}-\frac{1}{3}R_{k\Bar{k}j}^i X^k_{(1)}\wedge \Omega_{li}X^l_{(1)}\wedge X^j_{(1)}V^{\bar{k}}_{(0)}\bigg),
\end{equation}
with $d^{\nabla}X^i_{(1)}=dX^i_{(1)}+\Chr{i}{j}{k}dX^j_{(0)}X^k_{(1)}$. Note that the only fields left are the even scalar $X^i_{(0)}$, the odd 1-form $X^i_{(1)}$ and the odd scalar $V^{\bar{k}}_{(0)}$. A quick glance to our expression for the RW model in (\ref{class:expl_S_global}) (assume again the $(2,0)$ part of the curvature is zero as well as we cut off the perturbative expansion of the $(1,1)$ part at the $\Riem{i}{k}{s}{\bar{j}}$) suggests the association $V^{\bar{k}}_{(0)}\Leftrightarrow dx^{\bar{k}}$. 
We will comment more on this later.
\end{rmk}

By expanding $\mathbf{X}^i$ through the geodesic exponential map and by pulling back $\omega_{\mathcal{M}}$ as well as $S_{\text{QZ}}$ through it, the authors find 
\begin{align}
        \exp^*\omega_{\mathcal{M}}&=dP_{\Bar{i}} \wedge dX^{\Bar{i}}+dQ_{\Bar{i}}\wedge dV^{\Bar{i}}+\frac{1}{2}\Omega_{ij}(x)dy^i\wedge dy^j-\delta X^{\bar{i}}\delta \Theta_{\bar{i}}\\
        \exp^*\Sc_{\text{QZ}}\bigg|_{\tilde{\mathbf{P}}}&=\varint_{T[1]\Sigma_3}d^3zd^3\theta \bigg(\tilde{{\mathbf{P}}}_{\Bar{i}}D \mathbf{X}^{\Bar{i}}+\mathbf{Q}_{\Bar{i}}D\mathbf{V}^{\Bar{i}}+\frac{1}{2}\Omega_{ij}\mathbf{y}^i D\mathbf{y}^j-\Tilde{\mathbf{P}}_{\Bar{i}}\mathbf{V}^{\Bar{i}}+\Theta_{\Bar{i}}(x;\mathbf{y})\mathbf{V}^{\Bar{i}}\bigg)
\end{align}
where $\Theta_{\Bar{i}}$ is of degree 2 and given by 
\begin{equation}
\label{class:qz_theta}
    \Theta_{\Bar{i}}(x;y)=\sum^{\infty}_{n=3}\frac{1}{n!}\nabla_{l_4}\dots \nabla_{l_n}R^{\hspace{3mm}k}_{\Bar{i}l_1\ l_3}\Omega_{kl_2}(x)y^{l_1}\dots y^{l_n}
\end{equation} 
and $\Tilde{P}_{\Bar{i}}:=P_{\Bar{i}}+\Theta_{\Bar{i}}$.


After removing the \emph{spectator fields} (see \cite{QZ10,KQZ13}), the action becomes
\begin{equation}
    \varint_{T[1]\Sigma_3}d^3zd^3\theta\bigg(\frac12\Omega_{ij}\mathbf{y}^iD\mathbf{y}^j+\Theta_{\bar{i}}(x;\mathbf{y})\mathbf{V}^{\bar{i}}\bigg),
\end{equation}
which further reduces to
\begin{equation}
\label{class:qz_action}
    \varint_{T[1]\Sigma_3}d^3zd^3\theta\bigg(\frac12\Omega_{ij}\mathbf{y}^iD\mathbf{y}^j+\Theta_{\bar{i}}(x;\mathbf{y})V^{\bar{i}}_{(0)}\bigg)
\end{equation}
for degree reasons ($V^{\bar{i}}_{(0)}$ is an odd scalar). This action fails the CME by a $\bar{\partial}$-exact term due to $\Theta$ satisfying the Maurer--Cartan equation
\begin{equation}
\label{class:theta_MC}
    \bar{\partial}_{[\bar{i}}\Theta_{\bar{j}]}=-(\Theta_{\bar{i}}, \Theta_{\bar{j}}),
\end{equation}
where $\bar{\partial}$ is the Dolbeault differential and $[\bar{i}\bar{j}]$ denotes antisymmetrization over the indices $\bar{i}$ and $\bar{j}$.

The hyperK{\"a}hler structure is then relaxed. A new connection which still preserves $\Omega$ (crucial for the perturbative approach through the exponential map above) is found. However, since the connection is not Hermitian, the curvature of $\Gamma$ exhibits also a $(2,0)$-component. This complicates the exponential map which can not be worked out at all orders as in (\ref{class:qz_theta}). In \cite{KQZ13}, the authors argue that a solution to this problem should originate from principles related to the globalization issues discussed in \cite{BLN02} and the application of Fedosov connection in order to deal with perturbation theory on curved manifold \cite{CF01}. In the realm of this paper, we furnish an affirmative answer to both their ideas. In particular, as we have seen in Section \ref{comp:sec:summary}, the Fedosov connection allowed to compute the terms in the $\Linf$-algebra and thus to work out the exponential map. In Section \ref{class:sec:var_class_back}, we have seen the Grothendieck connection to accomplish the same in the context of formal geometry. 
\begin{rmk}
We can compare the procedure above with our globalization construction by associating $V^{\bar{i}}_{(0)}$ with $dx^{\bar{i}}$. First, note that $\Big(R\sur\Big)_{\bar{i}}$ in Eq. (\ref{class:qz_theta}) matches with the second term in Eq. \eqref{class:glob_terms}. Second, the action in \eqref{class:qz_action} coincides with our globalized action in \eqref{glob_action} if we ``forget" the $(2,0)$-part of the curvature. In particular, by associating $\bar{\partial}$ with $dx^{\bar{i}}\frac{\partial}{\partial x^{\bar{i}}}$, we can interpret the failure of \eqref{class:qz_action} to satisfy the CME due to the term \eqref{class:theta_MC} as a consequence of the action satisfying the $(1,1)$-part of the dCME (Eq. (\ref{class:dCME})).
\end{rmk}

We reserve the last remark of the section to precise the association between $V^{\bar{i}}_{(0)}$ and $dx^{\bar{i}}$ as well as their ``meaning" as we promised in Remark \ref{rmk_comp_zq_rw}.
\begin{rmk}
\label{rmk_parameters}
As we have seen above, $V^{\bar{i}}_{(0)}$ and $dx^{\bar{i}}$ arise in two different contexts: the first is an odd scalar coordinate parametrizing the fibers of $T^{0,1}M$, while the second is introduced through the classical Grothendieck connection as well as the perturbative expansion.

Nevertheless, the association makes sense considering that $V^{\bar{i}}_{(0)}$ is interpreted as an odd harmonic zero mode in \cite{KQZ13}. In fact, recall from Section \ref{sec:globalization}, $x$ is the zero mode obtained from the Euler--Lagrange equation $D\mathbf{X}=0$. If we enlarge the complex (see Eq. \eqref{extensioncomplex}), the space of fields becomes (\ref{extendedsof}) meaning that $dx^{\bar{i}}\in T^{\vee 0,1}M$, i.e. an odd zero mode. This association was pointed out first by Qiu and Zabzine in \cite{QZ12}.

The presence of these quantities has been known in the literature since the early days of the RW model and has deep consequences. Since they are odd, there  can be as many as the dimension of $M$. As such, the perturbative expansion can not be infinite, but it can only stop at a certain order. This is a crucial difference between the CS and the RW theory, which was originally spotted in \cite{RW96} and attributed to the need for the RW theory to saturate the zero modes. According to Kontsevich in \cite{Ko99}, as a result the RW model can be understood as an AKSZ model with ``parameters" (these parameters are $V^{\bar{i}}_{(0)}$ or $dx^{\bar{i}}$). In the same article, he presented a different perspective on this subject by pointing out that the RW invariants come from characteristic classes of holomorphic connections. 
\end{rmk}


\section{$BF$-like formulation of the Rozansky--Witten model}
\label{sec:BF-like_formulation}
In order to quantize our globalized version of the RW model in the quantum BV-BFV framework \cite{CMR17}, we need to formulate the model as a $BF$-like theory. This can be done by exploiting the similarities between the RW theory and the CS theory. These similarities have also been crucial in the construction of \cite{Ste17}. There it was argued that RW could be split following a similar approach to the one of Cattaneo, Mnev and Wernli for the CS theory in \cite{CMnW17} (see also \cite{We18} for a more detailed exposition). 

As shown in \cite{CLL17} (see Eq. (\ref{class:sympl_cll})), we have a pairing on $ \Tilde{\mathcal{F}}\surg$ given by the BV symplectic form which can be defined on homogeneous elements $\hat{\mathbf{Y}}\otimes g_1$ and $\hat{\mathbf{Z}}\otimes g_2$ as
\begin{equation}
    \begin{split}
        \braket{-,-}:\ & \Tilde{\mathcal{F}}\surg\otimes  \Tilde{\mathcal{F}}\surg \rightarrow \Omega^{\bullet,\bullet}(M),\\
&\braket{\hat{\mathbf{Y}}\otimes g_1,\hat{\mathbf{Z}}\otimes g_2}:=\underbrace{\Omega(g_1,g_2)}_{\text{sympl. struct. on $M$}}\,\,\varint_{T[1]\Sigma_3}\mu_{\Sigma_3}\bigg(\hat{\mathbf{Y}}\wedge \hat{\mathbf{Z}}\bigg).
    \end{split}
\end{equation}
By expanding $\hat{\mathbf{X}}\in \Tilde{\mathcal{F}}\surg$ as $\hat{\mathbf{X}}=\hat{\mathbf{X}}^ie_i$, we have
\begin{equation}
    \braket{\hat{\mathbf{X}},\hat{\mathbf{X}}}=\Omega(e_i,e_j)\varint_{T[1]\Sigma_3}\mu_{\Sigma_3}\bigg(\hat{\mathbf{X}}^i\wedge \hat{\mathbf{X}}^j\bigg)=\varint_{T[1]\Sigma_3}\mu_{\Sigma_3}\bigg(\Omega_{ij}\hat{\mathbf{X}}^i\wedge \hat{\mathbf{X}}^j\bigg).
\end{equation}
We can rewrite the globalized action (\ref{glob_action}) in the same way as in \cite{CLL17} (see the action in \eqref{class:dg_action_2}), we have
\begin{equation}
    \Tilde{\Sc}\surg=\frac12\Big\langle\hat{\mathbf{X}}, D\hat{\mathbf{X}}\Big\rangle+\Big\langle\Big(\hat{R}\sur\Big)_j(x; \hat{\mathbf{X}})dx^j,\hat{\mathbf{X}}\Big\rangle+\Big\langle\Big(\hat{R}\sur\Big)_{\bar{j}}(x; \hat{\mathbf{X}})dx^{\bar{j}},\hat{\mathbf{X}}\Big\rangle,
\end{equation} 
with
\begin{equation}
\label{class:coeff}
    \begin{split}
        \Big(\hat{R}\sur\Big)_j(x,\hat{\mathbf{X}})&=\sum^{\infty}_{k=0}\frac{1}{(k+1)!}\Big(\hat{R}_k\Big)_j(\hat{\mathbf{X}}^{\otimes k}),\\
        \Big(\hat{R}\sur\Big)_{\bar{j}}(x,\hat{\mathbf{X}})&=\sum^{\infty}_{k=2}\frac{1}{(k+1)!}\Big(\hat{R}_k\Big)_{\bar{j}}(\hat{\mathbf{X}}^{\otimes k}).
    \end{split}
\end{equation}
Now, similarly to the approach in \cite{CMnW17}, we assume that we can split the $\Linf$-algebra as 
\begin{equation}
\mathfrak{g}[1]=\Omega^{\bullet,\bullet}(M)\otimes T^{\vee1,0}M=\Omega^{\bullet,\bullet}(M)\otimes V\oplus \Omega^{\bullet,\bullet}(M)\otimes W,
\end{equation}
with $V$ and $W$ two isotropic subspaces. We identify $W\cong V^\vee$ via the pairing (in particular thanks to the holomorphic symplectic form). Consequently, the superfield splits as $\hat{\mathbf{X}}=\hat{\mathbf{A}}+\hat{\mathbf{B}}=\hat{\mathbf{A}}^i\xi_i+\xi^i\hat{\mathbf{B}}_i$ with $\xi_i\in V$ and $\xi^i\in W$. Concerning the assignment of degrees, we make the following choices. Since $\Omega$ has ghost degree $2$ (and as such $\Omega^{-1}$ has ghost degree $-2$), we assign total degree 0 to $\mathbf{A}^i$ and $\xi_i$, total degree $2$ to $B_i$ and total degree $-2$ to $\xi^i$. We refer to Table \ref{class:Tab_degrees_split} for an explanation of the ghost degrees for the components of the superfields $\hat{\mathbf{A}}^i$ and $\hat{\mathbf{B}}_i$. Then $\hat{\mathbf{A}}^i\oplus \hat{\mathbf{B}}_i\in \Omega^{\bullet}(\Sigma_3)\oplus\Omega^{\bullet}(\Sigma_3)[2]$, which is a $BF$-like theory.

\begin{table}[hbt!]
    \centering
    \begin{tabular}{|lcc|}
    \toprule
         & Form degree &  Ghost degree \\
         \midrule
    $A^i_{(0)}$ & 0     &  0 \\
    $A^i_{(1)}$ & 1     &  $-1$ \\
    $A^i_{(2)}$ & 2     &  $-2$\\
    $A^i_{(3)}$ & 3     &  $-3$\\
    \midrule 
    $B_{(0)i}$ & 0     &  2 \\
    $B_{(1)i}$ & 1     &  1 \\
    $B_{(2)i}$ & 2     &  $0$\\
    $B_{(3)i}$ & 3     &  $-2$\\
    \bottomrule
    \end{tabular}
    \caption{Explanation for the form degree and ghost degree for the components of the superfields $\hat{\mathbf{A}}^i$ and $\hat{\mathbf{B}}_i$.}
    \label{class:Tab_degrees_split}
\end{table}

\begin{rmk}
As explained in \cite[Remark 4.2.2]{Ste17}, the splitting of the target $T^{\vee1,0}_xM$ into two transversal holomorphic Lagrangian subbundles is not possible when $M$ is a K3 surface. Instead, it is possible when $M=T^\vee Y$, with $Y$ any complex manifold. In this case $M$ with the standard holomorphic symplectic form will have a vertical as well as a horizontal polarization.
\end{rmk}
To sum up, the  space of fields is split as
\begin{equation}
\label{space_of_fields_gs}
     \Tilde{\mathcal{F}}\surgS=\Omega^{\bullet}(\Sigma_3)\otimes\Omega^{\bullet,\bullet}(M)\otimes V\oplus \Omega^{\bullet}(\Sigma_3)[2]\otimes\Omega^{\bullet,\bullet}(M)\otimes W.
\end{equation}
 \begin{defn}[Globalized split RW action]
 The \textit{globalized split RW action} is defined as
 \begin{equation}
\label{split_global_action}
    \begin{split}
        \Tilde{\Sc}\surgS&\coloneqq
    \Big\langle\hat{\mathbf{B}}, D\hat{\mathbf{A}}\Big\rangle+\Big\langle\Big(\hat{R}\sur\Big)_j(x; \hat{\mathbf{A}}+\hat{\mathbf{B}})dx^j,\hat{\mathbf{A}}+\hat{\mathbf{B}}\Big\rangle+\Big\langle\Big(\hat{R}\sur\Big)_{\bar{j}}(x; \hat{\mathbf{A}}+\hat{\mathbf{B}})dx^{\bar{j}},\hat{\mathbf{A}}+\hat{\mathbf{B}}\Big\rangle\\
    &=\hat{\mathcal{S}}\surgS+\mathcal{S}\surgSR+\mathcal{S}\surgSRbar.
    \end{split}
\end{equation} 
with
\begin{equation}
\label{class:R_1}
    \begin{split}
    \bigg\langle\Big(\hat{R}\sur\Big)_j(x,\hat{\mathbf{A}}+\hat{\mathbf{B}})dx^{j},\hat{\mathbf{A}}+\hat{\mathbf{B}}\bigg\rangle&=\varint_{T[1]\Sigma_3}\mu_{\Sigma_3}\bigg(\sum^{\infty}_{k=0}\frac{1}{(k+1)!}\Big(\hat{R}_k\Big)^i_j((\Omega_{il}\hat{\mathbf{A}}^l+\hat{\mathbf{B}}_i)^{\otimes k})(\hat{\mathbf{A}}+\hat{\mathbf{B}})dx^{j}\bigg)\\
        &=\varint_{T[1]\Sigma_3}\mu_{\Sigma_3}\Bigg\{\bigg(\Big(\hat{R}_0\Big)^i_j\Omega_{il}\hat{\mathbf{A}}^l+\Big(\hat{R}_0\Big)^i_j\hat{\mathbf{B}}_i+\frac{1}{2}\Big(\hat{R}_1\Big)^i_j(\xi_s)\Omega_{il}\hat{\mathbf{A}}^s\hat{\mathbf{A}}^l\\
        &\quad+\Big(\hat{R}_1\Big)^i_j(\xi^s)\Omega_{il}\hat{\mathbf{A}}^l\hat{\mathbf{B}}_s
        +\frac{1}{2}\Big(\hat{R}_1\Big)^i_j(\xi^s)\hat{\mathbf{B}}_i\hat{\mathbf{B}}_s\\
        &\quad+\frac{1}{6}\Big(\hat{R}_2\Big)^i_j(\xi_s\xi_m)\Omega_{il}\hat{\mathbf{A}}^l\hat{\mathbf{A}}^s\hat{\mathbf{A}}^m+\frac{1}{2}\Big(\hat{R}_2\Big)^i_j(\xi_s\xi^m)\Omega_{il}\hat{\mathbf{A}}^l\hat{\mathbf{A}}^s\hat{\mathbf{B}}_m\\
        &\quad+\frac{1}{2}\Big(\hat{R}_2\Big)^i_j(\xi^s\xi^m)\Omega_{il}\hat{\mathbf{A}}^l\hat{\mathbf{B}}_s\hat{\mathbf{B}}_m+\frac{1}{6}\Big(\hat{R}_2\Big)^i_j(\xi^s\xi^m)\hat{\mathbf{B}}_i\hat{\mathbf{B}}_s\hat{\mathbf{B}}_m+\dots\bigg)dx^{j}\Bigg\}\\
        &=\varint_{T[1]\Sigma_3}\mu_{\Sigma_3}\Bigg\{\bigg(\Big(\hat{R}_0\Big)^i_{j}\Omega_{il}\hat{\mathbf{A}}^i+\Big(\hat{R}_0\Big)^i_j\hat{\mathbf{B}}_i+\frac{1}{2}\Big(\hat{R}_1\Big)^i_{j;s}\Omega_{il}\hat{\mathbf{A}}^l\hat{\mathbf{A}}^s\\
        &\quad+\Big(\hat{R}_1\Big)^{is}_{j}\Omega_{il}\hat{\mathbf{A}}^l\hat{\mathbf{B}}_s+\frac{1}{2}\Big(\hat{R}_1\Big)^{is}_j\hat{\mathbf{B}}_i\hat{\mathbf{B}}_s\\
        &\quad+\frac{1}{6}\Big(\hat{R}_2\Big)^i_{j;sm}\Omega_{il}\hat{\mathbf{A}}^l\hat{\mathbf{A}}^s\hat{\mathbf{A}}^m+\frac{1}{2}\Big(\hat{R}_2\Big)^{im}_{j;s}\Omega_{il}\hat{\mathbf{A}}^l\hat{\mathbf{A}}^s\hat{\mathbf{B}}_m\\
        &\quad+\frac{1}{2}\Big(\hat{R}_2\Big)^{ism}_{j}\Omega_{il}\hat{\mathbf{A}}^l\hat{\mathbf{B}}_s\hat{\mathbf{B}}_m+\frac{1}{6}\Big(\hat{R}_2\Big)^{ism}_j\hat{\mathbf{B}}_i\hat{\mathbf{B}}_s\hat{\mathbf{B}}_m+\dots\bigg)dx^{j}\Bigg\}
    \end{split}
\end{equation}
and
\begin{equation}
\label{class:R_2}
    \begin{split}
    \Big\langle\Big(\hat{R}\sur\Big)_{\bar{j}}(x,\hat{\mathbf{A}}+\hat{\mathbf{B}})dx^{\bar{j}},\hat{\mathbf{A}}+\hat{\mathbf{B}}\Big\rangle&=\varint_{T[1]\Sigma_3}\mu_{\Sigma_3}\bigg(\sum^{\infty}_{k=2}\frac{1}{(k+1)!}\Big(\hat{R}_k\Big)^i_{\bar{j}}((\Omega_{il}\hat{\mathbf{A}}^l+\hat{\mathbf{B}}_i)^{\otimes k})(\hat{\mathbf{A}}+\hat{\mathbf{B}})dx^{\bar{j}}\bigg)\\
        &=\varint_{T[1]\Sigma_3}\mu_{\Sigma_3}\Bigg\{\bigg(\frac{1}{6}\Big(\hat{R}_2\Big)^i_{\bar{j}}\Omega_{il}(\xi_s\xi_m)\hat{\mathbf{A}}^l\hat{\mathbf{A}}^s\hat{\mathbf{A}}^m+\frac{1}{2}\Big(\hat{R}_2\Big)^i_{\bar{j}}\Omega_{il}(\xi_s\xi^m)\hat{\mathbf{A}}^l\hat{\mathbf{A}}^s\hat{\mathbf{B}}_m\\
        &\quad+\frac{1}{2}\Big(\hat{R}_2\Big)^i_{\bar{j}}\Omega_{il}(\xi^s\xi^m)\hat{\mathbf{A}}^l\hat{\mathbf{B}}_s\hat{\mathbf{B}}_m+\frac{1}{6}\Big(\hat{R}_2\Big)^i_{\bar{j}}(\xi^s\xi^m)\hat{\mathbf{B}}_i\hat{\mathbf{B}}_s\hat{\mathbf{B}}_m+\dots\bigg)dx^{\bar{j}}\Bigg\}\\
        &=\varint_{T[1]\Sigma_3}\mu_{\Sigma_3}\Bigg\{\bigg(\frac{1}{6}\Big(\hat{R}_2\Big)^i_{\bar{j};sm}\Omega_{il}\hat{\mathbf{A}}^l\hat{\mathbf{A}}^s\hat{\mathbf{A}}^m+\frac{1}{2}\Big(\hat{R}_2\Big)^{im}_{\bar{j};s}\Omega_{il}\hat{\mathbf{A}}^l\hat{\mathbf{A}}^s\hat{\mathbf{B}}_m\\
        &\quad+\frac{1}{2}\Big(\hat{R}_2\Big)^{ism}_{\bar{j}}\Omega_{il}\hat{\mathbf{A}}^l\hat{\mathbf{B}}_s\hat{\mathbf{B}}_m+\frac{1}{6}\Big(\hat{R}_2\Big)^{ism}_{\bar{j}}\hat{\mathbf{B}}_i\hat{\mathbf{B}}_s\hat{\mathbf{B}}_m+\dots\bigg)dx^{\bar{j}}\Bigg\}.
    \end{split}
\end{equation}
We call the model associated with the action \eqref{split_global_action}, \textit{globalized split RW model}.
  \end{defn}

We present in Table \ref{class:Tab_coeff_split} the explicit expression as well as total degree for the component of $\Big(\hat{R}_k\Big)_j$ and $\Big(\hat{R}_k\Big)_{\bar{j}}$ in (\ref{class:R_1}) and (\ref{class:R_2}), respectively.

\begin{table}[h!]
    \centering
    \begin{tabular}{|lcc|}
    \toprule
    Operator     & Explicit expression &  Total degree \\
         \midrule
    $(\hat{R}_0\big)^i_{j}\Omega_{il}$ & -$\delta^{i}_j\Omega_{il}$     &  2 \\[2mm]
    $(\hat{R}_0\big)^i_{j}$ & $-\delta^i_j$     &  0 \\[2mm]
    $\Big(\hat{R}_1\Big)^i_{j;s}\Omega_{il}$ & $-\Chr{i}{s}{j}\Omega_{il}$     &  2\\[2mm]
    $\Big(\hat{R}_1\Big)^{is}_{j}\Omega_{il}$ & $-\Chr{i}{q}{j}\Omega_{il}\Big(\Omega^{-1}\Big)^{qs}$     &  0\\[2mm]
    $\Big(\hat{R}_1\Big)^{is}_{j}$ & $-\Chr{i}{j}{q}\Big(\Omega^{-1}\Big)^{qs}$      &  $-2$ \\[2mm]
    $\Big(\hat{R}_2\Big)^i_{j;sm}\Omega_{il}$ & $\bigg(\frac18R^{\hspace{1.5mm} i}_{j\; ms}+\frac14R^{\hspace{4mm} i}_{jm\; s}\bigg)\Omega_{il}$   &  2 \\[2mm]
    $\Big(\hat{R}_2\Big)^{im}_{j;s}\Omega_{il}$ & $\bigg(\frac18R^{\hspace{1.5mm} i}_{j\; ps}+\frac14R^{\hspace{3mm} i}_{jp\; s}\bigg)\Omega_{il}\Big(\Omega^{-1}\Big)^{pm}$    &  0\\[2mm]
    $\Big(\hat{R}_2\Big)^{ism}_{j}\Omega_{il}$ & $\bigg(\frac18R^{\hspace{1.5mm} i}_{j\; pn}+\frac14R^{\hspace{3mm} i}_{jp\; n}\bigg)\Omega_{il}\Big(\Omega^{-1}\Big)^{ns}\Big(\Omega^{-1}\Big)^{pm}$   &  $-2$\\[2mm]
    $\Big(\hat{R}_2\Big)^{ism}_{j}$ & $\bigg(\frac18R^{\hspace{1.5mm} m}_{j\hspace{2.5mm} kl}+\frac14R^{\hspace{2.5mm} m}_{jk\hspace{2.5mm} l}\bigg)\Big(\Omega^{-1}  \Big)^{ls}\Big(\Omega^{-1}\Big)^{ki}$ & $-4$ \\
    $\Big(\hat{R}_2\Big)^i_{\bar{j};sm}\Omega_{il}$ & $\bigg(\frac12\Riem{i}{\bar{j}}{m}{s}\bigg)\Omega_{il}$   &  2 \\[2mm]
    $\Big(\hat{R}_2\Big)^{im}_{\bar{j};s}\Omega_{il}$ & $\bigg(\frac12\Riem{i}{\bar{j}}{n}{s}\bigg)\Omega_{il}\Big(\Omega^{-1}\Big)^{nm}$    &  0\\[2mm]
    $\Big(\hat{R}_2\Big)^{ism}_{\bar{j}}\Omega_{il}$ & $\bigg(\frac12\Riem{i}{\bar{j}}{n}{p}\bigg)\Omega_{il}\Big(\Omega^{-1}\Big)^{ns}\Big(\Omega^{-1}\Big)^{pm}$   &  $-2$\\[2mm]
    $\Big(\hat{R}_2\Big)^{ism}_{\bar{j}}$ & $\bigg(\frac12\Riem{m}{\bar{j}}{k}{l}\bigg)\Big(\Omega^{-1}  \Big)^{ls}\Big(\Omega^{-1}\Big)^{ki}$ & $ -4$ \\
    \bottomrule
    \end{tabular}
    \caption{Explicit expression and total degree of the coefficients in \eqref{class:R_1} and \eqref{class:R_2}.}
    \label{class:Tab_coeff_split}
\end{table}
If $\Sigma_3$ is a closed manifold, the globalized split RW action satisfies the dCME:
\[
d_M\tilde{\Sc}\surgS+\frac12(\tilde{\Sc}\surgS,\tilde{\Sc}\surgS)=0,
\]
with $d_M=d_x+d_{\bar{x}}$ the sum of the holomorphic and antiholomorphic Dolbeault differentials on the target manifold $M$.
In the presence of boundary, the globalized split action satisfies the mdCME:
\begin{equation}
\label{iotasplit}
    \iota_{\Tilde{Q}\surgS}\omega\surgS=\delta \tilde{\Sc}\surgS+\pi^*\alpha\surgSB,
\end{equation}
with 
\begin{align}
       \begin{split}
        \Tilde{Q}\surgS{}&=\varint_{T[1]\Sigma_3}\mu_{\Sigma_3}\bigg( -D\hat{\mathbf{A}}^i\frac{\delta}{\delta \hat{\mathbf{A}}^i}-D\hat{\mathbf{B}}_i\frac{\delta}{\delta \hat{\mathbf{B}}_i}+\sum^{\infty}_{k=0}\frac{1}{k!}\Big(\hat{R}_k\Big)^i_j((\hat{\mathbf{A}}+\hat{\mathbf{B}})^{\otimes k})dx^{j}\frac{\delta}{\delta \hat{\mathbf{A}}^i}\\
        &\quad-\sum^{\infty}_{k=0}\frac{1}{k!}\Big(\hat{R}_{k}\Big)^l_j((\hat{\mathbf{A}}+\hat{\mathbf{B}})^{\otimes k})dx^{j}\Omega_{li}\frac{\delta}{\delta \hat{\mathbf{B}}_i}+\sum^{\infty}_{k=0}\frac{1}{k!}\Big(\hat{R}_k\Big)^i_{\bar{j}}((\hat{\mathbf{A}}+\hat{\mathbf{B}})^{\otimes k})dx^{{\bar{j}}}\frac{\delta}{\delta \hat{\mathbf{A}}^i},\\
        &\quad-\sum^{\infty}_{k=0}\frac{1}{k!}\Big(\hat{R}_{k}\Big)^l _{\bar{j}}((\hat{\mathbf{A}}+\hat{\mathbf{B}})^{\otimes k})dx^{{\bar{j}}}\Omega_{li}\frac{\delta}{\delta  \hat{\mathbf{B}}_i}\bigg),
        \end{split}
\end{align}
\begin{align}
    \omega\surgS{}&=\varint_{T[1]\Sigma_3}\mu_{\Sigma_3}\bigg( \delta\hat{\mathbf{B}}_i\delta \hat{\mathbf{A}}^i\bigg),\\
    \alpha\surgSB{}&=\varint_{T[1]\partial \Sigma_3}\mu_{\partial \Sigma_3}\bigg(\hat{\mathbf{B}}_i\delta \hat{\mathbf{A}}^i\bigg).
\end{align}

\section{Perturbative quantization of the globalized split Rozansky--Witten model}
\label{sec:pert_quant_RW}

In the last section, we have formulated our globalized RW model as a $BF$-like theory. This allows us to quantize perturbatively the newly constructed globalized split RW model according to the Quantum BV-BFV framework \cite{CMR17} (see Section \ref{sec_qbvbfv} for an introduction). The quantization of the kinetic part of the action is analogous to the example of section 3 in \cite{CMR17}, since the theory reduces to the abelian $BF$ theory.. Hence we will be rather quick in the exposition referring to \cite{CMR17} for further details. We will focus our attention to the interacting part of the action (in our case this is actually just the globalization term), which has a rich, as well as complicated, structure. In particular, we will draw some comparison with the PSM, which has been considered in \cite{CMoW19}. 


\subsection{Polarization}
The recipe to perturbatively quantize a $BF$-like theory according to the quantum BV-BFV formalism starts by requiring the data of a polarization. 

Following the result of Section \ref{sec:BF-like_formulation}, in the globalized split RW theory, the space of boundary fields splits as
\begin{equation}
\label{space_boundary_fields_sg}
    \Tilde{\mathcal{F}}\surgSB=\Omega^\bullet(\partial\Sigma_3)\otimes\Omega^{\bullet,\bullet}(M)\otimes V\oplus\Omega^\bullet(\partial\Sigma_3)[2]\otimes\Omega^{\bullet,\bullet}(M)\otimes W.
\end{equation}
Since we split $T^{1,0}M$ into isotropic subspaces, by the isotropy condition the subspaces are, in particular, Lagrangian. Therefore, either of them can be used as a base or fiber of the polarization. 

\begin{notat} 
From now on we will drop the hat from the notation of the ``globalized" superfields (e.g. $\hat{\mathbf{A}}^i$). Moreover, we will denote the coordinates on the base of the polarization by $\mathbb{A}^i$ or $\mathbb{B}^i$ and refer to this choice as $\mathbb{A}$- or $\mathbb{B}$-representation.
\end{notat}

Let us choose a decomposition of the boundary $\partial \Sigma_3=\partial_1\Sigma_3 \sqcup\partial_2\Sigma_3$, where $\partial_1\Sigma_3$ and $\partial_2\Sigma_3$ are two compact manifolds. Here, we can define a polarization $\mathcal{P}$ by choosing the $\mathbb{A}$-representation on $\partial_1\Sigma_3$ and the $\mathbb{B}$-representation on $\partial_2\Sigma_3$. 
The space of leaves of the associated foliations are  $\mathcal{B}_1:=\Omega^{\bullet}(\partial_1\Sigma_3)$ and $\mathcal{B}_2:=\Omega^{\bullet}(\partial_2\Sigma_3)[2]$, respectively. The space of boundary fields is $\mathcal{B}\bonP=\mathcal{B}_1\times \mathcal{B}_2\ni (\mathbb{A}^i, \mathbb{B}_i)$. 

The BFV 1-form
is 
\begin{equation}
     \alpha\surgSBP=\varint_{\partial_1 \Sigma_3} \mathbf{B}_i\delta \mathbf{A}^i-\varint_{\partial_2 \Sigma_3} \delta \mathbf{B}_i \mathbf{A}^i
\end{equation}
and the quadratic part of the action \eqref{split_global_action} is
\begin{equation}
    \hat{\Sc}\surgSP=\varint_{\Sigma_3}\mathbf{B}_id\mathbf{A}^i-\varint_{\partial_2\Sigma_3}\mathbf{B}_i\mathbf{A}^i.
\end{equation}

\subsection{Extraction of boundary fields}
We split the space of fields as 
\begin{align}
\begin{split}
\tilde{\mathcal{F}}\surgS &\rightarrow \tilde{\mathcal{B}}\bonP\oplus \mathcal{Y}\\ \label{quantiz:split_fields_1}
(\mathbf{A}^i, \mathbf{B}_i)&\mapsto (\Tilde{\mathbb{A}}^i,\tilde{\mathbb{B}}_i)\oplus (\underline{\mathbf{A}}^i,\underline{\mathbf{B}}_i),
\end{split}
\end{align}
where $\tilde{\mathcal{B}}\bonP$ denotes the bulk extension of $\mathcal{B}\bonP$ to $\tilde{\mathcal{F}}\surgS$ with $\Tilde{{\mathbb{A}}}^i$ and $\Tilde{{\mathbb{B}}}_i$ the extensions of the boundary fields $\mathbb{A}^i$ and $\mathbb{B}$ to the bulk space of fields $\Tilde{\mathcal{F}}\surgS$; $\underline{\mathbf{A}}^i$ and $\underline{\mathbf{B}}_i$ are the bulk fields, which are required to restrict to zero on $\partial_1\Sigma_3$ and $\partial_2\Sigma_3$, respectively. Here, the extensions are chosen to be singular: $\Tilde{\mathbb{A}}^i$ and $\tilde{\mathbb{B}}_i$ are required to restrict to zero outside the boundary (a choice pointed out first in \cite{CMR17}). The action reduces to
\begin{equation}
    \hat{\Sc}\surgSP=\varint\sur \underline{\mathbf{B}}_id\underline{\mathbf{A}}^i-\bigg(\varint_{\partial_2\Sigma_3} \mathbb{B}_i\underline{\mathbf{A}}^i-\varint_{\partial_1\Sigma_3}\underline{\mathbf{B}}_i\mathbb{A}^i\bigg).
\end{equation}

\subsection{Construction of \texorpdfstring{$\Omega_0$}{O}}
At this point, we can construct the coboundary operator $\Omega_0$ by canonical quantization: we consider the boundary action and we replace any $\underline{\mathbf{B}}_i$ by $(-i\hbar \frac{\delta}{\delta \mathbb{A}^i})$ on $\partial_1\Sigma_3$, any $\underline{\mathbf{A}}^i$ by $(-i\hbar \frac{\delta}{\delta \mathbb{B}}_i)$ on $\partial_2\Sigma_3$. We obtain
\begin{equation}
    \Omega_0=-i\hbar \bigg(\varint_{\partial_2\Sigma_3}d\mathbb{B}_i\frac{\delta}{\delta \mathbb{B}_i}+\varint_{\partial_1\Sigma_3}d\mathbb{A}^i\frac{\delta}{\delta \mathbb{A}^i}
    \bigg).
\end{equation}

\subsection{Choice of residual fields}
The bulk contribution in the space of fields $\mathcal{Y}$ is further split into the space of residual fields $\mathcal{V}_{\Sigma_3}$ and a complement, the space of fluctuations fields $\mathcal{Y}'$, namely
\begin{alignat}{2}
&\mathcal{Y} &&\rightarrow \mathcal{V}_{\Sigma_3}\oplus \mathcal{Y}'\\ \label{quantiz:split_fields_1}
(\underline{\mathbf{A}}^i, &\underline{\mathbf{B}}_i)&&\mapsto (\mathsf{a}^i,\mathsf{b}_i)\oplus (\alpha^i,\beta_i),
\end{alignat}
where $\mathsf{a}^i$ and $\mathsf{b}_i$ are the residual fields, whereas $\alpha^i$ and $\beta_i$ are the fluctuations. Note that the fluctuation $\alpha^i$ is required to restrict to zero on $\partial_1\Sigma_3$ while  $\beta_i$ is required to restrict to zero on $\partial_2\Sigma_3$.
In our case, the minimal space of residual fields is 
\begin{equation}
    \mathcal{V}_{\Sigma_3}=H^{\bullet}(\Sigma_3, \partial_1\Sigma_3)[0]\oplus H^{\bullet}(\Sigma_3, \partial_2\Sigma_3)[2]\ni (\mathsf{a}^i,\mathsf{b}_i).
\end{equation}

Here we can also define the BV Laplacian. To do it, pick a basis $\{[\chi_i]\}$ of $H^{\bullet}(\Sigma_3, \partial_1\Sigma_3)$ and its dual basis $\{[\chi^i]\}$ of $H^{\bullet}(\Sigma_3, \partial_2\Sigma_3)$ with representatives $\chi_i$ in $\Omega^{\bullet}(\Sigma_3, \partial_1\Sigma_3)$ and $\chi^i$ in $\Omega^{\bullet}(\Sigma_3, \partial_2\Sigma_3)$, with $\varint_{\Sigma_3}\chi_i\chi^j=\delta^j_i$. We can write the residual fields in a basis as
\begin{equation}
    \begin{split}
        \mathsf{a}^i&=\sum_k(z^{k}\chi_k)^i,\\
        \mathsf{b}_i&=\sum_k (z^+_k\chi^{k})_i,
    \end{split}
\end{equation}
where $\{z^k,z^+_k\}$ are canonical coordinates on $\mathcal{V}_{\Sigma_3}$ with BV symplectic form
\begin{equation}
    \omega_{\mathcal{V}_{\Sigma_3}}=\sum_i(-1)^{\deg z^k}\delta z^+_k\delta z^k.
\end{equation}
Finally, the BV Laplacian on $\mathcal{V}_{\Sigma_3}$ is
\begin{equation}
    \Delta_{\mathcal{V}_{\Sigma_3}}=\sum_i(-1)^{\deg z^k}\frac{\partial}{\partial z^k}\frac{\partial}{\partial z^+_k}.
\end{equation}

\subsection{Gauge-fixing and propagator}
We now have to fix a Lagrangian subspace $\mathcal{L}$ of $\mathcal{Y}'$. In the case of abelian BF theory, in \cite{CMR17}, the authors proved that such Lagrangian can be obtained from a \textit{contracting triple} $(\iota,p,K)$ 
for the complex $\Omega^\bullet_{\underline{D}}(\Sigma_3)$. 




In particular, the integral kernel of $K$ is the propagator, which we call $\eta$.
Since $K$ is actually the inverse of an elliptic operator (as shown in \cite{CMR17}), the propagator is singular on the diagonal of $\Sigma_3\times \Sigma_3$. Hence, we will define it as follows. Let 
\begin{equation}
    \mathrm{Conf}_2(\Sigma_3)=\{(x_1,x_2)\in \Sigma_3\mid x_1\neq x_2\},
\end{equation}
and let $\iota_{\mathfrak{D}}$ be the inclusion of 
\begin{equation}
    \mathfrak{D}\coloneqq\{x_1\times x_2\in (\partial_1\Sigma_3\times \Sigma_3)\cup(\Sigma_3\times \partial_2\Sigma_3)\mid x_1\neq x_2\}
\end{equation}
into $\mathrm{Conf}_2(\Sigma_3)$. Then the propagator is the 2-form $\eta\in \Omega^2(\mathrm{Conf}_2(\Sigma_3),\mathfrak{D})$, where 
\begin{equation}
    \Omega^\bullet(\mathrm{Conf}_2(\Sigma_3),\mathfrak{D})=\{\gamma\in \Omega^\bullet(\mathrm{Conf}_2(\Sigma_3))\mid \iota^{*}_{\mathfrak{D}}\gamma=0\}.
\end{equation}
Explicitly,
\begin{equation}
\label{quantiz:prop}
    \eta(x_1,x_2)=\frac{1}{T_{\Sigma_{3}}}\frac{1}{i\hbar}\varint_{\mathcal{L}}e^{\frac{i}{\hbar}\hat{\Sc}\surgSP}\pi^*_1\alpha^i(x_1)\pi^*_2\beta_i(x_2),
\end{equation}
with $\pi_1,\pi_2$ the projections from $M\times M$ to its first and second factor.
The coefficient $T_{\Sigma_3}$ is related to the Reidemeister torsion on $\Sigma_3$ as shown in \cite{CMR17}. However, its precise nature is irrelevant for the purposes of the present paper.


\subsection{The quantum state}
We can sum up the splittings we have made so far as
\begin{align}
\begin{split}
\tilde{\mathcal{F}}\surgS &\rightarrow \mathcal{B}\bonP\times \mathcal{V}^{\mathcal{P}}_{\Sigma_3}\times \mathcal{Y}'\\
(\mathbf{A}^i, \mathbf{B}_i)&\mapsto (\mathbb{A}^i,\mathbb{B}_i)+ (\mathsf{a}^i,\mathsf{b}_i)+ (\alpha^i,\beta_i).
\end{split}
\end{align}

\begin{rmk}
As a result of the procedure detailed in \cite{CMR17}, this is referred to as \textit{good splitting}.
\end{rmk}

According to the splitting of the space of fields, the action decomposes as 
\begin{equation}
    \Sc\surgSP=\hat{\Sc}\surgSP+\hat{\Sc}^{\, \text{pert}}+\Sc^{\, \text{res}}+\Sc^{\, \text{source}},
\end{equation}
where
\begin{align}
    \hat{\Sc}\surgSP&=\varint_{\Sigma_3}\beta_id\alpha^i,\\
    \hat{\Sc}^{\, \text{pert}}&=\varint\sur \mathcal{V}(\underline{\mathbf{A}}, \underline{\mathbf{B}}),\\
    \Sc^{\, \text{res}}&= - \bigg(\varint_{\partial_2\Sigma_3}\mathbb{B}_i\mathsf{a}^i-\varint_{\partial_1\Sigma_3}\mathsf{b}_i\mathbb{A}^i\bigg),\\
    \Sc^{\, \text{source}}&= - \bigg(\varint_{\partial_2\Sigma_3}\mathbb{B}_i\alpha^i-\varint_{\partial_1\Sigma_3}\beta_i\mathbb{A}^i\bigg),
\end{align}
where $\hat{\Sc}^{\, \text{pert}}$ is an interacting term made up by a density-valued function $\mathcal{V}$ which depends on the fields but not on their derivatives (by assumption). 

The state is given by:
\begin{equation}
\begin{split}
        \hat{\psi}_{\Sigma_3}(\mathbb{A}, \mathbb{B}, \mathsf{a}, \mathsf{b})&=\varint_{(\alpha, \beta)\in \mathcal{L}}e^{\frac{i}{\hbar}\Sc\surgSP (\mathbb{A}+\mathsf{a}+\alpha,\mathbb{B}+\mathsf{b}+\beta)}\mathscr{D}[\alpha]\mathscr{D}[\beta]\\
        &=e^{\frac{i}{\hbar}\Sc^{\, \text{res}}}\varint_{\mathcal{L}}e^{\frac{i}{\hbar}\hat{\Sc}\surgSP}e^{\frac{i}{\hbar}\hat{\Sc}^{\, \text{pert}}}e^{\frac{i}{\hbar}\Sc^{\text{source}}},
\end{split}
\end{equation}
where we denote by $\mathscr{D}$ a formal measure on $\mathcal{L}$.
The idea here is to compute the integral through a perturbative expansion, hence let us expand the exponentials as
\begin{equation}
    \begin{split}
        \hat{\psi}_{\Sigma_3}(\mathbb{A}, \mathbb{B}, \mathsf{a}, \mathsf{b})&=\sum_{k,l,m}\frac{1}{k!l!m!}(-1)^{k+m}\bigg(\varint_{\partial_2\Sigma_3}\mathbb{B}_i\mathsf{a}^i-\varint_{\partial_1\Sigma_3}\mathsf{b}_i\mathbb{A}^i\bigg)^k\times\\
        &\quad\times\varint_{\mathcal{L}}e^{\frac{i}{\hbar}\hat{\Sc}\surgSP}\bigg(\varint\sur \mathcal{V}(\underline{\mathbf{A}}, \underline{\mathbf{B}})\bigg)^l
        \bigg(\varint_{\partial_2\Sigma_3}\mathbb{B}_i\alpha^i-\varint_{\partial_1\Sigma_3}\beta_i\mathbb{A}^i\bigg)^m.
    \end{split}
\end{equation}
In the globalized split RW model, the interaction term is actually given by the globalization terms (the second and third terms in the action \ref{split_global_action}).
After having expanded the globalization terms in in residual fields and in fluctuations, the integration over $\mathcal{L}$ can be solved by using the \textit{Wick theorem}.



\subsection{Feynman rules}
In this section, we are going to introduce the Feynman rules needed to define precisely the quantum state of our theory. 


Since our aim is to prove the mdQME for the globalized split RW model, we will need to take care of the quantum Grothendieck BFV operator. This is a coboundary operator in which higher functional derivatives may appear (and as we will see they will indeed be present). As explained in \cite{CMR17}, higher functional derivatives requires a sort of ``regularization". This is provided by the composite fields, which we denote by square brackets $[\enspace]$ (e.g. for the boundary field $\mathbb{B}$, we will write $[\mathbb{B}_{i_1}\dots \mathbb{B}_{i_k}]$)\footnote{See Section \label{qs_bflike} for a short introduction.}. 


\begin{defn}[Globalized split RW Feynman graph]
A \textit{globalized split RW Feynman graph} is an oriented graph with three types of vertices $V(\Gamma)=V_{\text{bulk}}(\Gamma)\sqcup V_{\partial_1}\sqcup V_{\partial_2}$, called bulk vertices and type 1 and 2 boundary vertices, such that
\begin{itemize}
    \item bulk vertices can have any valence,
    \item type 1 boundary vertices carry any number of incoming half-edges (and no outgoing half-edges),
    \item type 2 boundary vertices carry any number of outgoing half-edges (and no incoming half-edges),
    \item multiple edges and loose half-edges (leaves) are allowed.
\end{itemize}
\end{defn}
A labeling of a Feynman graph is a function from the set of half-edges to $\{1,\dots,\dim V\}$.

In our case our source manifold $\Sigma_3$ has boundary $\partial \Sigma_3=\partial_1 \Sigma_3 \sqcup \partial_2 \Sigma_3$, let $\Gamma$ be a Feynman graph and define
\begin{equation}
    \mathrm{Conf}_\Gamma(\Sigma_3)\coloneqq \mathrm{Conf}_{V_{\text{bulk}}}(\Sigma_3) \times \mathrm{Conf}_{V_{\partial_1}}(\partial_1 \Sigma_3) \times \mathrm{Conf}_{V_{\partial_2}}(\partial_2 \Sigma_3).
\end{equation}
The Feynman rules are given by a map associating to a Feynman graph $\Gamma$ a differential form $\omega_\Gamma \in \Omega^\bullet(\mathrm{Conf}_\Gamma(\Sigma_3))$.




\begin{defn}[Globalized split RW Feynman rules] Let $\Gamma$ be a labeled Feynman graph. We choose a configuration $\iota:V(\Gamma) \rightarrow \mathrm{Conf}(\Gamma)$, such that decompositions are respected. Then, we \textit{decorate} the graph according to the following rules, namely, the \textit{Feynman rules}:
\begin{itemize}
    \item Bulk vertices in $\Sigma_3$ decorated by ``globalized vertex tensors"
    \begin{equation}
        \begin{split}
             \Big(\hat{R}_k\Big)^{i_1\dots i_s}_{j;j_1\dots j_t}dx^j &\coloneqq \frac{\partial^{s+t}}{\partial \underline{\mathbf{A}}^{i_1}\dots\partial \underline{\mathbf{A}}^{i_s} \partial \underline{\mathbf{B}}_{j_1}\dots \partial \underline{\mathbf{B}}_{j_t}} \bigg|_{\underline{\mathbf{A}}=\underline{\mathbf{B}}=0} \Big(\hat{R}_k\Big)^i_j((\underline{\mathbf{A}}+\underline{\mathbf{B}})^{\otimes k})(\Omega_{il}\underline{\mathbf{A}}^l+\underline{\mathbf{B}}_i)dx^{j}\\
            \Big(\hat{R}_k\Big)^{i_1\dots i_s}_{{\Bar{j}};j_1\dots j_t} dx^{\Bar{j}}&\coloneqq \frac{\partial^{s+t}}{\partial \underline{\mathbf{A}}^{i_1}\dots\partial \underline{\mathbf{A}}^{i_s} \partial \underline{\mathbf{B}}_{j_1}\dots \partial \underline{\mathbf{B}}_{j_t}} \bigg|_{\underline{\mathbf{A}}=\underline{\mathbf{B}}=0} \Big(\hat{R}_k\Big)^i_{\Bar{j}}((\underline{\mathbf{A}}+\underline{\mathbf{B}})^{\otimes k})(\Omega_{il}\underline{\mathbf{A}}^l+\underline{\mathbf{B}}_i)dx^{\Bar{j}}
        \end{split}
    \end{equation}
    where $s, t$ are the out- and in-valencies of the vertex and $i_1, \dots, i_s$ and $j_1, \dots, j_t$ are the
labels of the out (respectively in-)oriented half-edges.

\item Boundary vertices $v \in V_{\partial_1}(\Gamma)$ with incoming half-edges labeled $i_1, \dots, i_k$ and no out-going  half-edges are decorated by a composite field $[\mathbb{A}^{i_1} \dots \mathbb{A}^{i_k}]$ evaluated at the point (vertex location) $\iota(v)$ on $\partial_1 \Sigma_3$.
\item Boundary vertices $v \in V_{\partial_2}$ on $\partial_2 M$ with outgoing half-edges labeled $j_1, \dots, j_l$ are decorated by $[\mathbb{B}_{j_1} \dots \mathbb{B}_{j_l}]$ evaluated at the point on $\partial_2 \Sigma_3$. 
\item Edges between vertices $v_1, v_2$ are decorated with the propagator $\eta (\iota(v_1),\iota(v_2))\cdot \delta^i_j$, with $\eta$ the propagator induced by $\mathcal{L} \subset \mathcal{Y}'$, the gauge-fixing Lagrangian.
\item Loose half-edges (leaves) attached to a vertex $v$ and labeled $i$ are decorated with the residual
fields $\mathsf{a}^i$ (for out-orientation), $\mathsf{b}_i$
(for in-orientation) evaluated at the point $\iota(v)$.
\end{itemize}
\end{defn}

The Feynman Rules are represented in Figs. \ref{fig:fr_glob_1}, \ref{fig:fr_glob_2} and \ref{fig:comp_fields_2}.

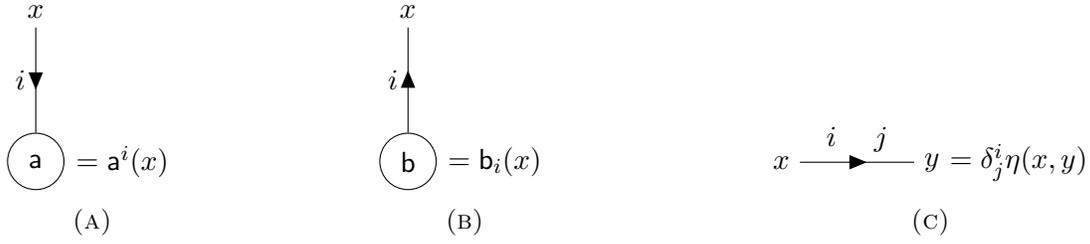
\begin{figure}[h]
    \centering
    \begin{subfigure}[b]{0.25\textwidth}
    \centering
\begin{tikzpicture}
  \begin{feynman}[every blob={/tikz/fill=white!30,/tikz/inner sep=1pt}]
  \vertex[blob] (m) at (0,-2) {$\mathsf{a}$};
  \vertex (a) at (0,0) {$x$} ;
 \diagram*{
      (a) -- [fermion, edge label' = $i$] (m)
    };
  \vertex [right=3em of m] {\(=\mathsf{a}^i(x)\)};
  \end{feynman}
\end{tikzpicture}
\caption{}
        \label{}
    \end{subfigure}%
    \hfill
\begin{subfigure}[b]{0.25\textwidth}
\centering
    \begin{tikzpicture}
  \begin{feynman}[every blob={/tikz/fill=white!30,/tikz/inner sep=0.5pt}]
  \vertex[blob] (m) at (0,-2) {$\mathsf{b}$};
  \vertex (a) at (0,0) {$x$} ;
 \diagram*{
      (m) -- [fermion, edge label = $i$] (a)
    };
     \vertex [right=3em of m] {\(=\mathsf{b}_i(x)\)};
  \end{feynman}
\end{tikzpicture}
  \caption{}
    \label{}
\end{subfigure}%
\hfill
\begin{subfigure}[b]{0.4\textwidth}
\centering
    \begin{tikzpicture}
  \begin{feynman}[every blob={/tikz/fill=white!30,/tikz/inner sep=0.5pt}]
  \vertex (a) at (0,0) {$x$};
  \vertex (b) at (2,0) {$y$} ;
 \diagram*{
      (a) -- [fermion, edge label = $i\hspace{5mm} j$] (b)
    };
     \vertex [right=3em of b] {\(=\delta^i_j\eta(x,y)\)};
  \end{feynman}
\end{tikzpicture}
  \caption{}
    \label{}
\end{subfigure}
    \caption{Feynman rules for residual fields and propagator}
    \label{fig:fr_glob_1}
\end{figure}

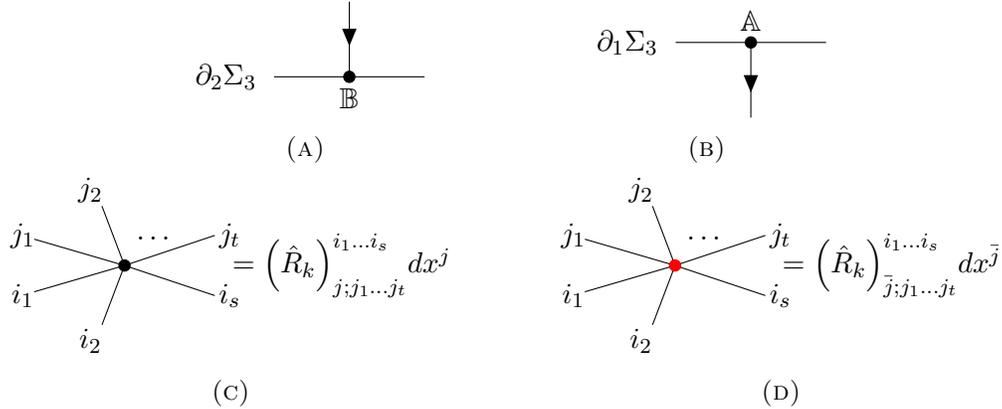
\begin{figure}[h]
    \centering
    \begin{subfigure}[b]{0.3\textwidth}
    \centering
       \begin{tikzpicture}
  \begin{feynman}[every blob={/tikz/fill=white!30,/tikz/inner sep=0.5pt}]
  \vertex (a) at (-1,0);
  \vertex (b) at (1,0);
  \vertex (m1) at (0, 1);
 
 \diagram*{
      (a) -- m [dot] -- (b),
      (m1) -- [fermion] m
    };
  \vertex [below=0.75em of m] {\(\mathbb{B}\)};
  \vertex [left=0.25em of a] {\(\partial_2 \Sigma_3\)};
  \end{feynman}
\end{tikzpicture}
\caption{}
        \label{}
    \end{subfigure}%
    \quad
\begin{subfigure}[b]{0.3\textwidth}
\centering
    \begin{tikzpicture}
  \begin{feynman}[every blob={/tikz/fill=white!30,/tikz/inner sep=0.5pt}]
  \vertex (a) at (-1,0);
  \vertex (b) at (1,0);
  \vertex (m1) at (0, -1);
 
 \diagram*{
      (a) -- m [dot] -- (b),
      m -- [fermion] (m1)
    };
  \vertex [above=0.75em of m] {\(\mathbb{A}\)};
  \vertex [left=0.25em of a] {\(\partial_1 \Sigma_3\)};
  \end{feynman}
\end{tikzpicture}
  \caption{}
    \label{}
\end{subfigure}%
\hfill
\begin{subfigure}[b]{0.42\textwidth}
\centering
        \begin{tikzpicture}
  \begin{feynman}[every blob={/tikz/fill=white!30,/tikz/inner sep=0.5pt}]
 \vertex (a) at (-1.2, 0.35);
  \vertex (b) at (-0.3, 0.79);
  \vertex (c) at (+1.2, 0.4);
  \vertex (e) at (-1.2, -0.35);
  \vertex (f) at (-0.3, -0.79);
  \vertex (g) at (+1.2, -0.4);
  \diagram*{
  (a) -- [] m [dot],
  (b) -- [] m [dot],
  (c) -- [] m [dot],
  (e) -- [] m [dot],
  (f) -- [] m [dot],
  (g) -- [] m [dot],
  }; 
 \vertex [right=7em of m] {\(\quad=\Big(\hat{R}_k\Big)^{i_1\dots i_s}_{j;j_1\dots j_t}dx^j\)};
 
 \vertex [above=0.2em of m, label=80:\(\dots\)] {}; 
 \vertex [] at (-1.35,0.4) {\(j_1\)};
  \vertex [] at (-0.45,1) {\(j_2\)};
 \vertex [] at (1.4, 0.4) {\(j_t\)};
  \vertex [] at (-1.35,-0.4) {\(i_1\)};
  \vertex [] at (-0.45,-1) {\(i_2\)};
 \vertex [] at (1.4, -0.4) {\(i_s\)};
  \end{feynman}
  \end{tikzpicture}
  \caption{}
  \label{}
    \end{subfigure}%
\quad
\begin{subfigure}[b]{0.42\textwidth}
\centering
        \begin{tikzpicture}
  \begin{feynman}[every blob={/tikz/fill=white!30,/tikz/inner sep=0.5pt}]
  \vertex (a) at (-1.2, 0.35);
  \vertex (b) at (-0.3, 0.79);
  \vertex (c) at (+1.2, 0.4);
  \vertex (e) at (-1.2, -0.35);
  \vertex (f) at (-0.3, -0.79);
  \vertex (g) at (+1.2, -0.4);
  \diagram*{
  (a) -- [] m [dot,red],
  (b) -- [] m [dot,red],
  (c) -- [] m [dot,red],
  (e) -- [] m [dot,red],
  (f) -- [] m [dot,red],
  (g) -- [] m [dot,red],
  }; 
 \vertex [right=7em of m] {\(\quad=\Big(\hat{R}_k\Big)^{i_1\dots i_s}_{{\Bar{j}};j_1\dots j_t}dx^{\Bar{j}}\)};
 
 \vertex [above=0.2em of m, label=80:\(\dots\)] {}; 
 \vertex [] at (-1.35,0.4) {\(j_1\)};
  \vertex [] at (-0.45,1) {\(j_2\)};
 \vertex [] at (1.4, 0.4) {\(j_t\)};
  \vertex [] at (-1.35,-0.4) {\(i_1\)};
  \vertex [] at (-0.45,-1) {\(i_2\)};
 \vertex [] at (1.4, -0.4) {\(i_s\)};
  \end{feynman}
  \end{tikzpicture}
  \caption{}
  \label{}
    \end{subfigure}
    \caption{Feynman rules for boundary fields and interaction vertices: we denote with a black dot the vertices arising from the $(2,0)$ part of the curvature (i.e. the terms corresponding to the term $\Sc_R$ in the action) and with a red dot the ones coming from the $(1,1)$ part (i.e. the terms corresponding to the term $\Sc_{\Bar{R}}$ in the action). Informally, we will call the first type of vertices ``black" vertex and the second one ``red" vertex.}
    \label{fig:fr_glob_2}
\end{figure}

\begin{figure}[hbt!]
    \centering
    \begin{subfigure}[b]{0.3\textwidth}
    \centering
       \begin{tikzpicture}
  \begin{feynman}[every blob={/tikz/fill=white!30,/tikz/inner sep=0.5pt}]
  \vertex (a) at (-2,0);
  \vertex (b) at (2,0);
  \vertex (d) at (0.4, 0.75);
  \vertex (e) at (1, 0.75);
  \vertex (f) at (-1, 0.75);
 \diagram*{
      (a) -- m [dot] -- (b),
      (d) -- [anti fermion] m,
      (e) -- [anti fermion] m,
      (f) -- [anti fermion] m
    };
  \vertex [below=0.75em of m] {\([\mathbb{A}^{i_1}\dots \mathbb{A}^{i_k}]\)};
  \vertex [left=0.25em of a] {\(\partial_1 \Sigma_3\)};
  \node at (-0.2, 0.5) {$\dots$};
  \node at (-1.1, 0.6) {$i_k$};
  \node at (0.5, 1) {$i_2$};
  \node at (1.2, 0.9) {$i_1$};
  \end{feynman}%
\end{tikzpicture}
\caption{}
        \label{fig:comp_fields}
    \end{subfigure}%
    \qquad \begin{subfigure}[b]{0.3\textwidth}
    \centering
       \begin{tikzpicture}
  \begin{feynman}[every blob={/tikz/fill=white!30,/tikz/inner sep=0.5pt}]
  \vertex (a) at (-2,0);
  \vertex (b) at (2,0);
  \vertex (d) at (0.4, 0.75);
  \vertex (e) at (1, 0.75);
  \vertex (f) at (-1, 0.75);
 \diagram*{
      (a) -- m [dot] -- (b),
      (d) -- [fermion] m,
      (e) -- [fermion] m,
      (f) -- [fermion] m
    };
  \vertex [below=0.75em of m] {\([\mathbb{B}_{j_1}\dots \mathbb{B}_{j_l}]\)};
  \vertex [left=0.25em of a] {\(\partial_2 \Sigma_3\)};
  \node at (-0.2, 0.5) {$\dots$};
  \node at (-1.1, 0.6) {$j_l$};
  \node at (0.5, 1) {$j_2$};
  \node at (1.2, 0.9) {$j_1$};
  \end{feynman}
\end{tikzpicture}
\caption{}
        \label{}
    \end{subfigure}%
    \caption{Feynman rules for the composite fields.}
    \label{fig:comp_fields_2}
\end{figure}
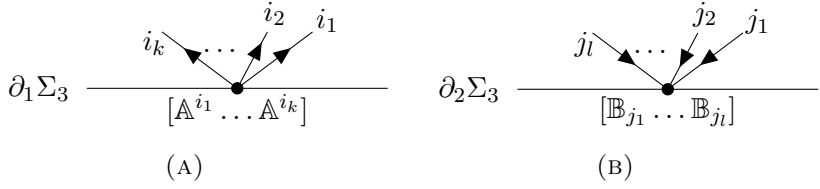




The full covariant quantum state for globalized split RW theory is defined analogously as in \cite{CMR17}.

\begin{defn}[Full quantum state for the globalized split RW theory]
\label{full_quantum_state_RW}
Let $\Sigma_3$ be a 3-dimensional manifold with boundary. Consider the data of a globalized split RW theory which consists of the globalized split space of fields $\tilde{\F}\surgS$ as in \eqref{space_of_fields_gs}, the globalized split space of boundary fields $\tilde{\F}\surgSB$ as in \eqref{space_boundary_fields_sg}, a polarization $\mathcal{P}$ on $\tilde{\F}\surgSB$  a good splitting $\tilde{\F}\surgS=\mathcal{B}_{\partial \Sigma_3}^{\mathcal{P}} \times \mathcal{V}_{\Sigma_3}^{\mathcal{P}} \times \mathcal{Y}'$ and $\mathcal{L} \subset \mathcal{Y}'$, the gauge-fixing Lagrangian. We can define the \textit{full quantum state for the globalized split RW theory} by the formal power series
\begin{equation}
    \boldsymbol{\hat{\psi}}\surgR(\mathbb{A},\mathbb{B};\mathsf{a},\mathsf{b})=T_{\Sigma_3}\exp\bigg(\frac{i}{\hbar}\sum_{\Gamma}\frac{(-i\hbar)^{\text{loops}(\Gamma)}}{|\text{Aut}(\Gamma)|}\varint_{\text{C}_\Gamma(\Sigma_3)}\omega_{\Gamma}(\mathbb{A}, \mathbb{B}; \mathsf{a}, \mathsf{b})
    \bigg).
\end{equation}
\end{defn}

\section{Proof of the modified differential Quantum Master Equation}
\label{sec:mdQME}
In the BV-BFV formalism on manifolds with boundary we expect the mQME to hold. This is a condition which requires the quantum state to be closed under a certain coboundary operator (see \cite{CMR17}). However, in the context of a globalized AKSZ theory, this condition becomes more complicated. The new condition is called \textit{modified differential Quantum Master Equation} (mdQME). We refer to \cite{BCM12,CMR14} for a discussion of the classical and quantum aspects of this condition. An extension for this discussion for manifolds with boundary was provided in\cite{CMoW17}. Finally, in \cite{CMoW19} the mdQME for anomaly-free, unimodular split AKSZ theories was proven, and later on in \cite{CMoW20} for the globalized PSM.

Our aim in this section is to prove the mdQME for the globalized split RW model, namely  
\begin{equation}
    \nabla_{\text{G}}\boldsymbol{\hat{\psi}}\surgR=0,
\end{equation}
where $\nabla_{\text{G}}$ is the \textit{quantum Grothendieck BFV (qGBFV) operator}  and $\boldsymbol{\hat{\psi}}\surgR$ is the full covariant quantum state for the globalized split RW theory.
As we will see, the proof follows almost verbatim from the proof of the mdQME in \cite{CMoW19}. Before addressing the proof, we focus on the qGBFV operator and we discuss the construction of the full BFV boundary operator.

\subsection{The quantum Grothendieck BFV operator}
\begin{defn}[qGBFV operator for the globalized split RW model]
Inspired by \cite{CMoW19}, we define the \textit{qGBFV operator for the globalized split RW model} as
\begin{equation}
    \nabla_\mathrm{G}\coloneqq\bigg(d_x+d_{\Bar{x}}-i\hbar \Delta_{\mathcal{V}_{\Sigma_3, x}}+\frac{i}{\hbar}\boldsymbol{\Omega}_{\partial \Sigma_3}\bigg),
\end{equation}
with $\boldsymbol{\Omega}_{\partial \Sigma_3}$ the full BFV boundary operator 
      \begin{equation}
        \boldsymbol{\Omega}_{\partial \Sigma_3}=\boldsymbol{\Omega}^{\mathbb{A}}_{\partial\Sigma_3}+\boldsymbol{\Omega}^{\mathbb{B}}_{\partial\Sigma_3}=\Omega^{\mathbb{A}}_{0}+\boldsymbol{\Omega}^{\mathbb{A}}_{\text{pert}}+\Omega^{\mathbb{B}}_0+\boldsymbol{\Omega}^{\mathbb{B}}_{\text{pert}},
    \end{equation}
    where
\begin{equation}
    \begin{split}
        \Omega^{\mathbb{A}}_0&=-i\hbar \varint_{\partial_1\Sigma_3}d\mathbb{A}^i\frac{\delta}{\delta \mathbb{A}^i},\\
        \Omega^{\mathbb{B}}_0&=-i\hbar \varint_{\partial_2\Sigma_3}d\mathbb{B}_i\frac{\delta}{\delta \mathbb{B}_i}.
    \end{split}
\end{equation}
and $\boldsymbol{\Omega}^{\mathbb{A}}_{\text{pert}}$ and $\boldsymbol{\Omega}^{\mathbb{B}}_{\text{pert}}$ are given by Feynman diagrams collapsing to the boundary in the $\mathbb{A}$-representation and $\mathbb{B}$-representation, respectively.
\end{defn}
\begin{rmk}
Note that $\nabla_{\text{G}}$ and $\Omega_{\partial \Sigma_3}$ are inhomogeneous forms on the holomorphic symplectic manifold $M$ since the globalized term in the action is a 1-form on $M$. Explicitly, for example in the $\mathbb{B}$-representation, we can decompose the $\boldsymbol{\Omega}^{\mathbb{B}}_{\mathrm{pert}}$ as
\begin{equation}
   \boldsymbol{\Omega}^{\mathbb{B}}_{\mathrm{pert}}=\underbrace{\boldsymbol{\Omega}^{\mathbb{B}}_{1,0}+\boldsymbol{\Omega}^{\mathbb{B}}_{0,1}}_{\coloneqq  \boldsymbol{\Omega}^{\mathbb{B}}_{(1)}}+\underbrace{\boldsymbol{\Omega}^{\mathbb{B}}_{2,0}+\boldsymbol{\Omega}^{\mathbb{B}}_{1,1}+\boldsymbol{\Omega}^{\mathbb{B}}_{0,2}}_{\coloneqq\boldsymbol{\Omega}^{\mathbb{B}}_{(2)}}+\dots.
\end{equation}
and similarly in the $\mathbb{A}$-representation.
\end{rmk}
In the next section, we proceed to give an explicit expression for the BFV boundary operator in the $\mathbb{B}$ and $\mathbb{A}$ representation. We start with the former. 

\subsection{BFV boundary operator in the $\mathbb{B}$-representation}
Let us remind the reader about the general form of the BFV boundary operator in the $\mathbb{B}$-representation for a split AKSZ theory \cite{CMoW19}:
\begin{equation}
    \boldsymbol{\Omega}_{\text{pert}}^{\mathbb{B}}\coloneqq \sum_{n,k \geq 0}\sum_{\Gamma}\frac{(i \hbar)^{\text{loops($\Gamma)$}}}{\mid \Aut(\Gamma)\mid} \varint_{\partial_2M}\bigg(\sigma_{\Gamma}\bigg)^{I_1\dots I_n}_{J_1\dots J_k}\wedge \mathbb{B}_{I_1}\wedge \dots \wedge \mathbb{B}_{I_n}\bigg((-1)^{kd} (i\hbar)^{k}\frac{\delta^{\mid J_1\mid + \dots +\mid J_k\mid}}{\delta [\mathbb{B}_{J_1} \dots \mathbb{B}_{J_k}]}\bigg)
\end{equation}
In order to find an explicit expression for the BFV boundary operator, we adopt the strategy in \cite{CMoW20} to find the BFV boundary operator in the $\mathbb{E}$-representation for the PSM. Their idea was to use the \textit{degree counting}. Indeed, in general, the form $\sigma_{\Gamma}$ is obtained as the integral over the compactification $\tilde{\mathrm{C}}_{\Gamma}(\mathbb{H}^d)$ of the  open configuration space modulo scaling and translation, with $\mathbb{H}^d$ the $d$-dimensional upper half-space:
\begin{equation}
\label{mdQME:integral}
    \sigma_{\Gamma}=\varint_{\tilde{\mathrm{C}}_{\Gamma(\mathbb{H}^d)}} \omega_{\Gamma},
\end{equation}
where $\omega_{\Gamma}$ is the product of limiting propagators at the point $p$ of collapse and vertex tensors. Note that in order for the integral \eqref{mdQME:integral} not to vanish the form degree of $\omega_\Gamma$ has to be the same as the dimension of $\tilde{\mathrm{C}}_{\Gamma}(\mathbb{H}^d)$. This gives constraints to the number of points in the bulk as well as points in the boundary admitted. We will apply this degree counting to our case, where, since we have $d=3$, the dimension of the compactified configuration space $\tilde{\mathrm{C}}_{\Gamma}(\mathbb{H}^3)$ is $\dim\tilde{\mathrm{C}}_{\Gamma}(\mathbb{H}^3)=3n+2m-3$, with $n$ the number of bulk vertices and $m$ the number of boundary vertices in $\Gamma$.

By using this procedure, in \cite{CMoW20} it was possible to find an explicit expression for the BFV boundary operator in the $\mathbb{E}$-representation for the PSM. As we will see, for us this is not possible. One could say that the cause is the nature of the RW model reflected in a dramatic increment in the number of Feynman rules as we go on in the $k$-index for the globalized terms in the action (see Eq. \eqref{class:coeff}). 
To see this in practice, let us show explicitly the Feynman rules for the terms in \eqref{class:R_1} and \eqref{class:R_2}, which we sum up in Table \ref{class:Tab_coeff_split_fr}. 

\begin{table}[!htb]
    \begin{minipage}{.5\linewidth}
      \centering
         \begin{tabular}[t]{|lccc|}
    \toprule
    Vertex     & Feynman rule &  Total degree & Name \\
         \midrule
    $(\hat{R}_0\big)^i_{j}$ &  \raisebox{\dimexpr15 pt-\totalheight\relax}{ \begin{tikzpicture}
  \begin{feynman}[every blob={/tikz/fill=white!30,/tikz/inner sep=0.5pt}]
  \vertex (d) at (0, -0.75);
  \diagram*{
  (d) -- [fermion] m [dot]
  }; 
  \end{feynman}
  \end{tikzpicture} }    &  0 & \textrm{I} \\[5mm]
    $(\hat{R}_0\big)^i_{j}$ &  \raisebox{\dimexpr15 pt-\totalheight\relax}{ \begin{tikzpicture}
  \begin{feynman}[every blob={/tikz/fill=white!30,/tikz/inner sep=0.5pt}]
  \vertex (d) at (0, -0.75);
  \diagram*{
  (d) -- [anti fermion] m [dot]
  }; 
  \end{feynman}
  \end{tikzpicture} }      &  2 & \textrm{II} \\[5mm]
    $\Big(\hat{R}_1\Big)^i_{j;s}$ &  \raisebox{\dimexpr20 pt-\totalheight\relax}{ \begin{tikzpicture}
  \begin{feynman}[every blob={/tikz/fill=white!30,/tikz/inner sep=0.5pt}]
  \vertex (d) at (0, -0.75);
  \vertex (c) at (0, 0.75);
  \diagram*{
  (d) -- [fermion] m [dot] -- [anti fermion] (c)
  }; 
  \end{feynman}
  \end{tikzpicture} }        &  0 & \textrm{III}\\[10mm]
    $\Big(\hat{R}_1\Big)^{is}_{j}$ & \raisebox{\dimexpr20 pt-\totalheight\relax}{ \begin{tikzpicture}
  \begin{feynman}[every blob={/tikz/fill=white!30,/tikz/inner sep=0.5pt}]
  \vertex (d) at (0, -0.75);
  \vertex (c) at (0, 0.75);
  \diagram*{
  (d) -- [anti fermion] m [dot] -- [anti fermion] (c)
  }; 
  \end{feynman}
  \end{tikzpicture} }     &  2 & \textrm{IV}\\[10mm]
    $\Big(\hat{R}_1\Big)^{is}_{j}$ & \raisebox{\dimexpr20 pt-\totalheight\relax}{ \begin{tikzpicture}
  \begin{feynman}[every blob={/tikz/fill=white!30,/tikz/inner sep=0.5pt}]
  \vertex (d) at (0, -0.75);
  \vertex (c) at (0, 0.75);
  \diagram*{
  (d) -- [anti fermion] m [dot] -- [fermion] (c)
  }; 
  \end{feynman}
  \end{tikzpicture} }     &  4 & \textrm{V}\\[10mm]
    $\Big(\hat{R}_2\Big)^i_{j;sm}$ &  \raisebox{\dimexpr20 pt-\totalheight\relax}{ \begin{tikzpicture}
  \begin{feynman}[every blob={/tikz/fill=white!30,/tikz/inner sep=0.5pt}]
  \vertex (d) at (0.5,  0.5);
  \vertex (c) at (-0.5, 0.5);
  \vertex (a) at (0, -0.7);
  \diagram*{
  (d) -- [fermion] m [dot] -- [anti fermion] (c);
  m -- [anti fermion] (a)
  }; 
  \end{feynman}
  \end{tikzpicture} }
     &  0 & \textrm{VI} \\[10mm]
    $\Big(\hat{R}_2\Big)^{im}_{j;s}$ &
    \raisebox{\dimexpr20 pt-\totalheight\relax}{ \begin{tikzpicture}
  \begin{feynman}[every blob={/tikz/fill=white!30,/tikz/inner sep=0.5pt}]
  \vertex (d) at (0.5,  0.5);
  \vertex (c) at (-0.5, 0.5);
  \vertex (a) at (0, -0.7);
  \diagram*{
  (d) -- [anti fermion] m [dot] -- [anti fermion] (c);
  m -- [anti fermion] (a)
  }; 
  \end{feynman}
  \end{tikzpicture} }&  2 & \textrm{VII}\\
  \bottomrule
    \end{tabular}
    \end{minipage}%
    \begin{minipage}{.5\linewidth}
      \centering
      \begin{tabular}[t]{|lccc|}
    \toprule
    Vertex     & Feynman rule &  Total degree & Name \\
         \midrule
    $\Big(\hat{R}_2\Big)^{ism}_{j}$ & \raisebox{\dimexpr20 pt-\totalheight\relax}{ \begin{tikzpicture}
  \begin{feynman}[every blob={/tikz/fill=white!30,/tikz/inner sep=0.5pt}]
  \vertex (d) at (0.5,  0.5);
  \vertex (c) at (-0.5, 0.5);
  \vertex (a) at (0, -0.7);
  \diagram*{
  (d) -- [anti fermion] m [dot] -- [fermion] (c);
  m -- [anti fermion] (a)
  }; 
  \end{feynman}
  \end{tikzpicture} }   &  4 & \textrm{VIII}\\[10mm]
    $\Big(\hat{R}_2\Big)^{ism}_{j}$ &
    \raisebox{\dimexpr20 pt-\totalheight\relax}{ \begin{tikzpicture}
  \begin{feynman}[every blob={/tikz/fill=white!30,/tikz/inner sep=0.5pt}]
  \vertex (d) at (0.5,  0.5);
  \vertex (c) at (-0.5, 0.5);
  \vertex (a) at (0, -0.7);
  \diagram*{
  (d) -- [anti fermion] m [dot] -- [fermion] (c);
  m -- [fermion] (a)
  }; 
  \end{feynman}
  \end{tikzpicture} }& 6 &\textrm{IX}\\[10mm]
    $\Big(\hat{R}_2\Big)^i_{\bar{j};sm}$ & \raisebox{\dimexpr20 pt-\totalheight\relax}{ \begin{tikzpicture}
  \begin{feynman}[every blob={/tikz/fill=white!30,/tikz/inner sep=0.5pt}]
  \vertex (d) at (0.5,  0.5);
  \vertex (c) at (-0.5, 0.5);
  \vertex (a) at (0, -0.7);
  \diagram*{
  (d) -- [fermion] m [dot, red] -- [anti fermion] (c);
  m -- [anti fermion] (a)
  }; 
  \end{feynman}
  \end{tikzpicture} }  &  0 &\textrm{X} \\[10mm]
    $\Big(\hat{R}_2\Big)^{im}_{\bar{j};s}$ & \raisebox{\dimexpr20 pt-\totalheight\relax}{ \begin{tikzpicture}
  \begin{feynman}[every blob={/tikz/fill=white!30,/tikz/inner sep=0.5pt}]
  \vertex (d) at (0.5,  0.5);
  \vertex (c) at (-0.5, 0.5);
  \vertex (a) at (0, -0.7);
  \diagram*{
  (d) -- [anti fermion] m [dot, red] -- [anti fermion] (c);
  m -- [anti fermion] (a)
  }; 
  \end{feynman}
  \end{tikzpicture} }   &  2 &\textrm{XI}\\[10mm]
    $\Big(\hat{R}_2\Big)^{ism}_{\bar{j}}$ &\raisebox{\dimexpr20 pt-\totalheight\relax}{ \begin{tikzpicture}
  \begin{feynman}[every blob={/tikz/fill=white!30,/tikz/inner sep=0.5pt}]
  \vertex (d) at (0.5,  0.5);
  \vertex (c) at (-0.5, 0.5);
  \vertex (a) at (0, -0.7);
  \diagram*{
  (d) -- [anti fermion] m [dot, red] -- [fermion] (c);
  m -- [anti fermion] (a)
  }; 
  \end{feynman}
  \end{tikzpicture} }   &  4 &\textrm{XII}\\[10mm]
    $\Big(\hat{R}_2\Big)^{ism}_{\bar{j}}$ &
    \raisebox{\dimexpr20 pt-\totalheight\relax}{ \begin{tikzpicture}
  \begin{feynman}[every blob={/tikz/fill=white!30,/tikz/inner sep=0.5pt}]
  \vertex (d) at (0.5,  0.5);
  \vertex (c) at (-0.5, 0.5);
  \vertex (a) at (0, -0.7);
  \diagram*{
  (d) -- [anti fermion] m [dot, red] -- [fermion] (c);
  m -- [fermion] (a)
  }; 
  \end{feynman}
  \end{tikzpicture} }& 6 &\textrm{XIII}\\[8.7mm]
    \bottomrule
    \end{tabular}
    \end{minipage} 
     \caption{Feynman rules for the globalization terms in the action \eqref{split_global_action}.}
    \label{class:Tab_coeff_split_fr}
\end{table}
Notice how the structure of the Feynman rules repeats similarly at each order (e.g. for $R_0$ we have 2 Feynman rules with degrees 0 and 2, respectively, while for $R_1$ we have 3 graphs with degrees 0, 2, 4). Hence, it is easy to understand how this works for higher order terms. 
From there, one can notice that we have two types of vertices: 
\begin{itemize}
    \item vertices which are 1-forms in $dx^{i}$: we will denote them by a black dot ($\bullet$) and refer to them as black vertices;
    \item vertices which are 1-forms in $dx^{{\bar{i}}}$: we will denote them by a red dot (\textcolor{red}{$\bullet$}) and refer to them as red vertices.
\end{itemize}

In our computations, we will limit ourselves to the Feynman rules in Table \ref{class:Tab_coeff_split_fr}, these are already enough to get a feeling about what is going on and even understand the behaviour of higher order terms, when possible. By using the names in the Table, since $n=\textrm{I}+\textrm{II}+\textrm{III}+\textrm{IV}+\textrm{V}+\textrm{VI}+\textrm{VII}+\textrm{VIII}+\textrm{IX}+\textrm{X}+\textrm{XI}+\textrm{XII}+\textrm{XIII}$ is the sum of the vertices, the degree counting produces the following equation
\begin{equation}
   \begin{split}
        \textrm{I}+\textrm{II}+\textrm{III}+\textrm{IV}+\textrm{V}+\textrm{VI}+\textrm{VII}+\textrm{VIII}&+\textrm{IX}+\textrm{X}+\textrm{XI}+\textrm{XII}+\textrm{XIII}+2m-3=\\
        &2\textrm{II}+2\textrm{IV}+4\textrm{V}+2\textrm{VII}+4\textrm{VIII}+6\textrm{IX}+2\textrm{XI}+4\textrm{XII}+6\textrm{XIII},
   \end{split}
\end{equation}
where on the right hand side we are taking into account that in the $\mathbb{B}$-representation, the arrows leaving the globalization vertex have to stay inside the collapsing subgraph. If this is not the case, by the boundary conditions on the propagator \cite{CMR17}, the result would be zero. 

First, let us focus on the black vertices (i.e. vertices \textrm{I}--\textrm{IX}). The equation reduces to
\begin{equation}
\label{degree counting eq}
    3\textrm{I}+\textrm{II}+3\textrm{III}+\textrm{IV}-\textrm{V}+3\textrm{VI}+\textrm{VII}-\textrm{VIII}-3\textrm{IX}+2m-3=0.
\end{equation}

The Feynman diagrams contributing to the BFV boundary operator are those whose vertices solve the equation \eqref{degree counting eq}. Hence, let us solve the equation case-by-case. 
Up to one bulk vertex, with the Feynman rules \textrm{I}--\textrm{IX} we have one diagram (see Fig. \ref{fig:first_diagram}).

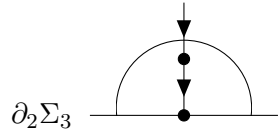
\begin{figure}[H]
    \centering
\begin{tikzpicture}
  \begin{feynman}[every blob={/tikz/fill=white!30,/tikz/inner sep=0.5pt}]
  \vertex (d) at (0, 0.75);
  \vertex (a) at (-1.25, -0.75);
  \vertex (b) at (1.25, -0.75);
  \node[dot] (x) at (0, -0.75);
  \diagram*{
  (d) -- [fermion] m [dot],
  (a) -- (b),
  m --[fermion] (x)
  }; 
  \vertex [left=0.25em of a] {\(\partial_2 \Sigma_3\)};
  \end{feynman}
  \clip (-0.9,-0.75) rectangle (0.9,0.9);
  \draw (0,-0.65) circle(0.9);
  \end{tikzpicture}
  \caption{First graph with a single black vertex contributing to the BFV boundary operator.}
  \label{fig:first_diagram}
  \end{figure}
  
From Fig. \ref{fig:first_diagram}, we notice that in order to have a degree 1 operator which satisfies the degree counting for higher order terms we need vertices with an even number of heads and tails. We show the first higher order contributions in Fig. \ref{fig:higher_diag}, while a general diagram contributing to the BFV operator
is exhibited in Fig. \ref{fig:gen_diagram}.

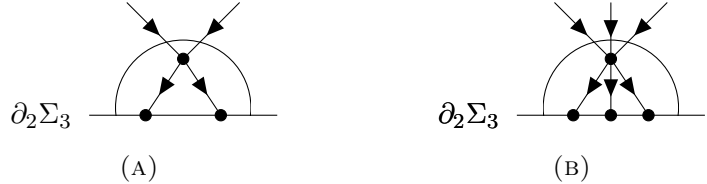
\begin{figure}[H]
\centering
  \begin{subfigure}[b]{0.4\textwidth}
\centering
\begin{tikzpicture}
  \begin{feynman}[every blob={/tikz/fill=white!30,/tikz/inner sep=0.5pt}]
  \vertex (d) at (-0.75, 0.75);
  \vertex (c) at (0.75, 0.75);
  \vertex (a) at (-1.25, -0.75);
  \vertex (b) at (1.25, -0.75);
  \node[dot] (x) at (-0.5, -0.75);
  \node[dot] (y) at (0.5, -0.75);
  \diagram*{
  (d) -- [fermion] m [dot],
  (c) -- [fermion] m,
  (a) -- (b),
  m --[fermion] (x),
  m --[fermion] (y)
  }; 
  \vertex [left=0.25em of a] {\(\partial_2 \Sigma_3\)};
  \end{feynman}
  \clip (-0.9,-0.75) rectangle (0.9,0.9);
  \draw (0,-0.65) circle(0.9);
  \end{tikzpicture}  
  \caption{}
    \label{}
\end{subfigure}\quad
\begin{subfigure}[b]{0.25\textwidth}
\centering
\begin{tikzpicture}
  \begin{feynman}[every blob={/tikz/fill=white!30,/tikz/inner sep=0.5pt}]
  \vertex (d) at (-0.75, 0.75);
  \vertex (c) at (0.75, 0.75);
  \vertex (e) at (0, 0.75);
  \vertex (a) at (-1.25, -0.75);
  \vertex (b) at (1.25, -0.75);
  \node[dot] (x) at (-0.5, -0.75);
  \node[dot] (y) at (+0.5, -0.75);
  \node[dot] (z) at (0, -0.75);
  \diagram*{
  (d) -- [fermion] m [dot],
  (c) -- [fermion] m,
  (e) -- [fermion] m,
  m  -- [fermion] (x),  
  m  -- [fermion] (y),
  m  -- [fermion] (z),
  (a) -- (b)
  }; 
  \vertex [left=0.25em of a] {\(\partial_2 \Sigma_3\)};
  \vertex [left=0.25em of a] {\(\partial_2 \Sigma_3\)};
  \end{feynman}
  \clip (-0.9,-0.75) rectangle (0.9,0.9);
  \draw (0,-0.65) circle(0.9);
  \end{tikzpicture} 
  \label{}
  \caption{}
\end{subfigure}
    \caption{Second and third graph with a single black vertex contributing to the BFV boundary operator.}
    \label{fig:higher_diag}
\end{figure}


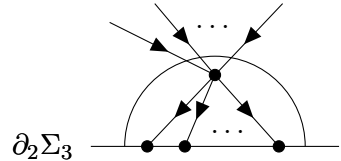
\begin{figure}[H]
\centering
\begin{tikzpicture}
  \begin{feynman}[every blob={/tikz/fill=white!30,/tikz/inner sep=0.5pt}]
  \vertex (d) at (-1.4, 0.7);
  \vertex (c) at (0.9, 0.95);
  \vertex (e) at (-0.8, 0.95);
  \vertex (a) at (-1.65, -0.95);
  \vertex (b) at (1.65, -0.95);
  \node[dot] (x) at (-0.9, -0.95);
  \node[dot] (y) at (+0.85, -0.95);
  \node[dot] (z) at (-0.4, -0.95);
  \diagram*{
  (d) -- [fermion] m [dot],
  (c) -- [fermion] m,
  (e) -- [fermion] m,
  m  -- [fermion] (x),
  m  -- [fermion] (y),
  m  -- [fermion] (z),
  (a) -- (b)
  }; 
  \vertex [left=0.25em of a] {\(\partial_2 \Sigma_3\)};
  \vertex [left=0.25em of a] {\(\partial_2 \Sigma_3\)};
  \end{feynman}
  \draw (-1.2,-0.95) arc (180:0:1.2);
  \node at (0,0.65) {\dots};
  \node at (0.2,-0.75) {\dots};
  \end{tikzpicture}  
    \caption{A general Feynman diagram contributing to the BFV operator in the $\mathbb{B}$-representation up to one black bulk vertex.}
    \label{fig:gen_diagram}
\end{figure}

Concerning the red vertices, the graphs contributing to the BFV operator up to one bulk vertex will start to appear from the vertices associated to the term $(R_3)_{\bar{j}}dx^{\bar{j}}$ (coming from the thid term in the action \ref{split_global_action}). Taking this into account, the general term of the diagrams with a red vertex is shown in Fig. \ref{fig:gen_diagram_red}.

\begin{figure}[H]
\centering
\begin{tikzpicture}
  \begin{feynman}[every blob={/tikz/fill=white!30,/tikz/inner sep=0.5pt}]
  \vertex (d) at (-1.4, 0.7);
  \vertex (c) at (0.9, 0.95);
  \vertex (e) at (-0.8, 0.95);
  \vertex (a) at (-1.65, -0.95);
  \vertex (b) at (1.65, -0.95);
  \node[dot] (x) at (-0.9, -0.95);
  \node[dot] (y) at (+0.85, -0.95);
  \node[dot] (z) at (-0.4, -0.95);
  \diagram*{
  (d) -- [fermion] m [red, dot],
  (c) -- [fermion] m,
  (e) -- [fermion] m,
  m  -- [fermion] (x),
  m  -- [fermion] (y),
  m  -- [fermion] (z),
  (a) -- (b)
  }; 
  \vertex [left=0.25em of a] {\(\partial_2 \Sigma_3\)};
  \vertex [left=0.25em of a] {\(\partial_2 \Sigma_3\)};
  \end{feynman}
  \draw (-1.2,-0.95) arc (180:0:1.2);
  \node at (0,0.65) {\dots};
  \node at (0.2,-0.75) {\dots};
  \end{tikzpicture}  
    \caption{A general Feynman diagram contributing to the BFV operator in the $\mathbb{B}$-representation up to one red bulk vertex. In particular, the graph with a total number of 4 arrows (2 entering and 2 leaving the red vertex) is the first non-zero contribution.}
    \label{fig:gen_diagram_red}
\end{figure}
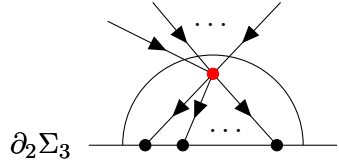

These considerations prove the following proposition.
\begin{prp}
\label{prp_op}
Consider the globalized split RW model, in the $\mathbb{B}$-representation, the first contribution to $\boldsymbol{\Omega}^{\mathbb{B}}_{\mathrm{pert}}$ is given by $\boldsymbol{\Omega}^{\mathbb{B}}_{(1)}= \boldsymbol{\Omega}^{\mathbb{B}}_{1,0}+ \boldsymbol{\Omega}^{\mathbb{B}}_{0,1}$ with
\begin{equation}
    \begin{split}
       \boldsymbol{\Omega}^{\mathbb{B}}_{1,0}&=\sum_{\substack{k\geq1,S_1, \dots S_k\\ i_1,\dots, i_k, j_1, \dots, j_k}}\frac{(-i\hbar)^k}{(k+S_1+\dots+S_k)!}\varint_{\partial_2\Sigma_3}\Big(\hat{R}_{2k-1}\Big)^{i_1\dots i_k}_{j;j_1\dots j_k} dx^j[\mathbb{B}_{i_1}\mathbb{B}_{S_1}]\dots \\
       &\hspace{8cm}\times[\mathbb{B}_{i_k}\mathbb{B}_{S_k}] \frac{\delta^{|j_1+\dots+j_k|+|S_1|+\dots+|S_k|}}{\delta [\mathbb{B}_{j_1}\dots \mathbb{B}_{j_k}][\delta \mathbb{B}_{S_1}]\dots [\delta \mathbb{B}_{S_k}]}\\
        \boldsymbol{\Omega}^{\mathbb{B}}_{0,1}&=\sum_{\substack{k\geq 2,S_1, \dots S_k\\ i_1,\dots, i_k, j_1, \dots, j_k}}\frac{(-i\hbar)^k}{(k+S_1+\dots+S_k)!}\varint_{\partial_2\Sigma_3}\Big(\hat{R}_{2k-1}\Big)^{i_1\dots i_k}_{\bar{j};j_1\dots j_k} dx^{\bar{j}}[\mathbb{B}_{i_1}\mathbb{B}_{S_1}]\dots\\
         &\hspace{8cm}\times[\mathbb{B}_{i_k}\mathbb{B}_{S_k}] \frac{\delta^{|j_1+\dots+j_k|+|S_1|+\dots+|S_k|}}{\delta [\mathbb{B}_{j_1}\dots \mathbb{B}_{j_k}][\delta \mathbb{B}_{S_1}]\dots [\delta \mathbb{B}_{S_k}]}.
    \end{split}
\end{equation}
\end{prp}

For $n>1$, the situation gets more complicated. We solve equation \eqref{degree counting eq} numerically. Empirically, for an even number of bulk vertices, we witness the absence of solutions. This implies immediately $\boldsymbol{\Omega}^{\mathbb{B}}_{(2)}=0$.

In the case $n=3$, the number of Feynman diagrams for the vertices \textrm{I}--\textrm{IX} increases dramatically with respect to the $n=1$ case. This increment is tamed since the necessity of having a degree 1 operator will decrease their number. However, we are not able to provide an explicit as well as general form for the BFV operator along the same lines as in Proposition \ref{prp_op}. We rely on examples which we show in Appendix \ref{app:feyn}. 

\begin{rmk}
Here we are assuming that the dimension of our target manifold $M$ is at least 4, if this would not be the case, then we would not have the 3 bulk vertices contribution to the BFV boundary operator. Hence, the number of bulk vertices allowed is bounded by the dimension of $M$. This was already noticed in \cite{CMoW20}. The difference here is that this reflects the ``odd Grassmanian nature" of the RW model with respect to CS theory (see Remark \ref{rmk_parameters}).
\end{rmk}

\subsection{BFV boundary operator in the $\mathbb{A}$-representation}
In the $\mathbb{A}$-representation, the arrows coming from the globalized vertices are allowed to leave the collapsing subgraph. Therefore, our arguments about the degree counting are not valid here. However since the coboundary operator has a total degree 1, while $\mathbb{A}^i$ has total degree 0, we can have at most 1 bulk vertex, i.e. $\boldsymbol{\Omega}^{\mathbb{A}}_{\mathrm{pert}}=\boldsymbol{\Omega}^{\mathbb{A}}_{1,0}+\boldsymbol{\Omega}^{\mathbb{A}}_{0,1}$ with
\begin{equation}
    \begin{split}
        \boldsymbol{\Omega}^{\mathbb{A}}_{1,0}&= \sum_{k\geq0}\,\,\varint_{\partial_1\Sigma_3}\sum_{J_1,\dots, J_r, I_1,\dots, I_s} \frac{(-i\hbar)^{|I_1|+\dots +|I_s|}}{(|I_1|+\dots +|I_s|)!}\Big(\hat{R}_k\Big)^{I_1\dots I_s}_{j;J_1\dots J_r}dx^j\prod^{r+s=k+1}_{r=1, s=1}[\mathbb{A}^{J_r}]\frac{\delta^{|I_s|}}{\delta [\mathbb{A}^{I_s}]},\\
        \boldsymbol{\Omega}^{\mathbb{A}}_{0,1}&= \sum_{k\geq3}\,\,\varint_{\partial_1\Sigma_3}\sum_{J_1,\dots, J_r, I_1,\dots, I_s} \frac{(-i\hbar)^{|I_1|+\dots +|I_s|}}{(|I_1|+\dots +|I_s|)!}\Big(\hat{R}_k\Big)^{I_1\dots I_s}_{\bar{j};J_1\dots J_r}dx^{\bar{j}}\prod^{r+s=k+1}_{r=1, s=1}[\mathbb{A}^{J_r}]\frac{\delta^{|I_s|}}{\delta [\mathbb{A}^{I_s}]},
    \end{split}
\end{equation}
where we label by the multiindex $J_r$ the arrows emanating from a boundary vertex towards the globalized vertex, by the multiindex $I_s$ the leaves emanating from the bulk vertex. The sum of $r$ and $s$ has to be $k+1$ since these are the total number of arrows leaving and arriving at a globalized vertex $(R_k)_jdx^j$ (or $(R_k)_{\bar{j}}dx^{\bar{j}}$).

\subsection{Flatness of the qGBFV operator for the globalized split RW model}
In this section, we prove that the qGBFV operator for the globalized split RW model squares to zero. The proof follows along the same lines as in \cite{CMoW19}, we will remark where there are differences and refer to their work when the procedure is identical.  Before entering into the details of the proof, we should mention that their proof (and the proof of the mdQME) depends on two assumptions: \textit{unimodularity} and \textit{absence of hidden faces} (\textit{anomaly-free} condition). The first means that tadpoles are not allowed. In the case of the globalized split RW model, we notice that this assumption is not needed since tadpoles vanish \cite{RW96}.

 
 \begin{assump}
 \label{ass_hidden_faces}
 We assume that the globalized split RW model is \emph{anomaly-free}, i.e. for every graph $\Gamma$, we have that
 \begin{equation}
     \varint_{F_{\geq 3}}\omega_\Gamma=0,
 \end{equation}
 where by $F_{\geq 3}$, we denote the union of the faces where at least three bulk vertices collapse in the bulk (also called \emph{hidden faces} \cite{BC}).
 \end{assump}
 
\begin{rmk}
It is well known that Chern--Simons theory is \emph{not} an anomaly-free theory \cite{AS91,AS94}. The construction of the quantum theory there depends on the choice of gauge-fixing. The appearance of anomalies can be resolve by choosing a framing and framing-dependent counter terms for the gauge-fixing. A famous example of an anomaly-free theory is given by the Poisson sigma model \cite{CF00} since by the result of Kontsevich \cite{Ko03} any 2-dimensional theory is actually anomaly-free. 
A general method for dealing with theories that do have anomalies is to add counter terms to the action. If the differential form $\omega_\Gamma$, which is integrated over the hidden faces, is exact, one can use the primitive form to cancel the anomalies by the additional vertices that appear. 
\end{rmk} 
 
Since the integrals we will consider are fiber integrals, we will apply of Stokes' theorem for integration along a compact fiber with corners, i.e.
\begin{equation}
    d\pi_*=\pi_*d-\pi^{\partial}_*,
\end{equation}
where $\pi_*$ denotes the fiber integration.
In particular, the application of Stokes' theorem to a fiber integral yiels
\begin{equation}
\label{stokes}
    (d_x+d_{\bar{x}})\varint_{\text{C}_{\Gamma}}\omega_\Gamma=\varint_{\text{C}_{\Gamma}}(d+d_{\bar{x}})\omega_\Gamma-\varint_{\partial \text{C}_{\Gamma}}\omega_\Gamma,
\end{equation}
where $d$ is the differential on $M\times C_\Gamma$.

\begin{thm}[Flatness of the qGBFV operator]\label{thm:flatness}
The qGBFV operator $\nabla_{\textup{G}}$ for the anomaly-free globalized split RW model squares to zero, i.e. 
\begin{equation}
\label{flatness_GBFV}
    (\nabla_{\textup{G}})^2\equiv0,
\end{equation}
where 
\begin{equation}
    \nabla_{\textup{G}}=d_{M}-i\hbar \Delta_{\mathcal{V}_{\Sigma_3, x}}+\frac{i}{\hbar}\boldsymbol{\Omega}_{\partial \Sigma_3}=d_x+d_{\Bar{x}}-i\hbar \Delta_{\mathcal{V}_{\Sigma_3, x}}+\frac{i}{\hbar}\boldsymbol{\Omega}_{\partial \Sigma_3}.
\end{equation}
\end{thm}
\begin{proof}
According to \cite{CMoW19}, the flatness of $\nabla_{\mathrm{G}}$ is equivalent to the equation 
\begin{equation}
\label{flatness_GBFV_3}
    i\hbar d_M\boldsymbol{\Omega}_{\partial \Sigma_3}-\frac12\bigg[\boldsymbol{\Omega}_{\partial \Sigma_3},\boldsymbol{\Omega}_{\partial \Sigma_3}\bigg]=0.
\end{equation}

This equation was proven for a globalized split AKSZ theory in \cite{CMoW19}, in which the $d_M$ is just the de Rham differential on the body of the target manifold. However, in our case, $d_M$ is the sum of the holomorphic and antiholomorphic Dolbeault differentials on $M$. 

We prove Eq. \eqref{flatness_GBFV_3} for $\boldsymbol{\Omega}^{\mathbb{B}}$. For $\boldsymbol{\Omega}^{\mathbb{A}}$, the proof is analogous as discussed in \cite{CMoW19}. Suppose we apply $d_M$ to a term of the form
\begin{equation}
    \boldsymbol{\Omega}^{\mathbb{B}}_\Gamma=\varint_{\partial_2\Sigma_3} \sigma_\Gamma \bigg(\Big(\hat{R}_k\Big)_jdx^j; \Big(\hat{R}_k\Big)_{\bar{j}}dx^{\bar{j}}\bigg)^I_{J_1\dots J_s}[\mathbb{B}^{J_1}]\dots [\mathbb{B}^{J_s}]\frac{\delta}{\delta [\mathbb{B}^{I}]},
\end{equation}
where $k$ could be any number greater than 0. Here, we chose the easiest term to express with more clarity what is going on. As in \cite{CMoW19}, we apply Stokes' theorem. However, this is different to the corresponding situation in \cite{CMoW19} since in our theory we have also red vertices\footnote{For the sake of clarity, we stress again that in \cite{CMoW19}, the vertices are only ``black" since $d_M$ is the de Rham differential on the body of the target manifold.}, which is portrayed by the fact that $\sigma_\Gamma$ depends also on $\Big(\hat{R}_k\Big)_{\bar{j}}dx^{\bar{j}}$. We obtain
 \begin{equation}
     \begin{split}
         (d_x+d_{\bar{x}})\boldsymbol{\Omega}^{\mathbb{B}}_\Gamma=\varint_{\partial_2\Sigma_3}\bigg\{(d_x+d_{\bar{x}})\sigma_{\Gamma}\bigg(\Big(\hat{R}_k\Big)_jdx^j; \Big(\hat{R}_k\Big)_{\bar{j}}dx^{\bar{j}}\bigg)\bigg\}^I_{J_1\dots J_s}[\mathbb{B}^{J_1}]\dots [\mathbb{B}^{J_s}]\frac{\delta}{\delta [\mathbb{B}^{I}]}  +[\Omega^\mathbb{B}_0,\boldsymbol{\Omega}^{\mathbb{B}}_\Gamma],
     \end{split}
 \end{equation}
 where the second term is produced when $d_x$ acts on the $\mathbb{B}$ fields (we do not have a corresponding term for $d_{\bar{x}}$ since we do not have fields $\mathbb{B}^{\bar{i}}$ terms to act on). By applying again Stokes' theorem, we have:
 \begin{equation}
     \begin{split}
         (d_x+d_{\bar{x}})\sigma_{\Gamma}\bigg(\Big(\hat{R}_k\Big)_jdx^j; \Big(\hat{R}_k\Big)_{\bar{j}}dx^{\bar{j}}\bigg)&=(d_x+d_{\bar{x}})\varint_{\tilde{\text{C}}_{\Gamma}}\omega_{\Gamma}\bigg(\Big(\hat{R}_k\Big)_jdx^j; \Big(\hat{R}_k\Big)_{\bar{j}}dx^{\bar{j}}\bigg)\\
     &=\varint_{\tilde{\text{C}}_{\Gamma}}(d+d_{\bar{x}})\omega_{\Gamma}\bigg(\Big(\hat{R}_k\Big)_jdx^j; \Big(\hat{R}_k\Big)_{\bar{j}}dx^{\bar{j}}\bigg)\\
     &\quad\pm\varint_{\partial\tilde{\text{C}}_{\Gamma}}\omega_{\Gamma}\bigg(\Big(\hat{R}_k\Big)_jdx^j; \Big(\hat{R}_k\Big)_{\bar{j}}dx^{\bar{j}}\bigg).
     \end{split}
 \end{equation}
\begin{rmk}
\label{rmk7.4.3}
  In principle, $d$ is the differential on $M\times \mathrm{C}_\Gamma$, hence it can be decomposed as $d=d_x+d_1+d_2$, where $d_1$ denotes the part of the differential acting on the propagator and $d_2$ the part acting on $\mathbb{B}$ fields (and, more generally, on $\mathbb{A}$ fields). We do not have a corresponding antiholomorphic differential on $M\times \mathrm{C}_\Gamma$ since the propagators and the fields are all holomorphic. This is different with respect to the case considered in \cite{CMoW19}.
\end{rmk}
 
 As in \cite{CMoW19}, we have $d\omega_{\Gamma} =d_x\omega_{\Gamma}$ and in the boundary integral we have three classes of faces. The first two types of faces, where more than two bulk points collapse and where a subgraph $\Gamma$ collapse at the boundary, can be proved as in \cite{CMoW19}. In particular, the former vanishes by our assumptions that the theory is anomaly-free (see Assumption \ref{ass_hidden_faces}), while the second produces exactly the term $\frac12\bigg[\boldsymbol{\Omega}^{\mathbb{B}}_{\mathrm{pert}},\boldsymbol{\Omega}^{\mathbb{B}}_{\mathrm{pert}}\bigg]$ by \cite[Lemma 4.9]{CMoW19}. On the other hand, the third case, when two bulk vertices collapse, has some differences with respect to the analogous situation in \cite{CMoW19} due to the already mentioned further presence of red vertices. Here we distinguish four cases:
 \begin{itemize}
         \item when a red vertex collapses with a black vertex, then these faces cancel out with\\ $d_x\omega_{\Gamma}\bigg(\Big(\hat{R}_k\Big)_{\bar{j}}dx^{\bar{j}}\bigg)$ by the dCME \eqref{dcme_2};
         \item when a black vertex collapses with a red vertex, then these faces cancel out with\\ $d_{\bar{x}}\omega_{\Gamma}\bigg(\Big(\hat{R}_k\Big)_jdx^j\bigg)$ by the dCME \eqref{dcme_4};
         \item when two black vertices collapse, then these faces cancel out with $d_x\omega_{\Gamma}\bigg(\Big(\hat{R}_k\Big)_{j}dx^{j}\bigg)$ by the dCME \eqref{dcme_1};
         \item when two red vertices collapse, then these faces cancel out with $d_{\bar{x}}\omega_{\Gamma}\bigg(\Big(\hat{R}_k\Big)_{\bar{j}}dx^{\bar{j}}\bigg)$ by the dCME \eqref{dcme_3}.
 \end{itemize}
By $\omega_{\Gamma}\bigg(\Big(\hat{R}_k\Big)_{\bar{j}}dx^{\bar{j}}\bigg)$ or $\omega_{\Gamma}\bigg(\Big(\hat{R}_k\Big)_{j}dx^{j}\bigg)$, we mean the part of the subgraph $\Gamma'$, which contains a red or black vertex. 

This proves (\ref{flatness_GBFV_3}), thus $(\nabla_{\textup{G}})^2\equiv 0$.
\end{proof}


 \subsection{Proof of the mdQME for the globalized split RW model}
 In this section, we are going to prove the mdQME for the globalized split RW model. The proof follows similarly as in \cite{CMoW19}. As before, we will refer to their work when the situation is identical and point out eventual differences. 
 
 \begin{thm}[mdQME for anomaly-free globalized split RW model]\label{thm:mdQME}
 Consider the full covariant perturbative state $\hat{\psi}_{\Sigma_3,x}$ as a quantization of the anomaly-free globalized split RW model. Then 
   \begin{equation}
   \label{mdqme_thm}
        \bigg(d_M-i\hbar \Delta_{\mathcal{V}_{\Sigma_3, x}}+\frac{i}{\hbar}\boldsymbol{\Omega}_{\partial \Sigma_3}\bigg)\boldsymbol{\hat{\psi}}\surgR=0.
    \end{equation}
 \end{thm}

\begin{proof}
Let $\mathcal{G}$ denote the set of Feynman graphs of the theory. Then, we can write the full covariant quantum state for the globalized split RW model as 
\begin{equation}
\label{proof_state}
   \boldsymbol{\hat{\psi}}\surgR=T_{\Sigma_3}\sum_{\Gamma\in\mathcal{G}}\varint_{\text{C}_\Gamma}\omega_\Gamma \Big(\hat{R}_jdx^j; \hat{R}_{\bar{j}}dx^{\bar{j}}\Big),
\end{equation}
where the combinatorial prefactor $\frac{(-i\hbar)^{\loops(\Gamma)}}{\vert\Aut(\Gamma)\vert}$ is included in $\omega_\Gamma$ (by $\loops$ we denote the number of loops of a graph $\Gamma$) and we denote the configuration space $\text{C}_\Gamma(\Sigma_3)$ by $\text{C}_\Gamma$ for simplicity. We note that $\omega_\Gamma$ is a ($\mathcal{V}_{\Sigma_3,x}$-dependent) differential form on $\text{C}_\Gamma\times M$. Again, following \cite{CMoW19}, we can apply Stokes' theorem \eqref{stokes} and we get
\begin{equation}
    d_M\varint_{\text{C}_\Gamma}\omega_\Gamma\Big(\hat{R}_jdx^j; \hat{R}_{\bar{j}}dx^{\bar{j}}\Big)=\varint_{\text{C}_\Gamma}(d+d_{\bar{x}})\omega_\Gamma\Big(\hat{R}_jdx^j; \hat{R}_{\bar{j}}dx^{\bar{j}}\Big)-\varint_{\partial\text{C}_\Gamma}\omega_\Gamma\Big(\hat{R}_jdx^j; \hat{R}_{\bar{j}}dx^{\bar{j}}\Big).
\end{equation} 

As mentioned in Remark \ref{rmk7.4.3}, the $d$ inside the integral is the total differential on $\text{C}_\Gamma(\Sigma_3)\times M$, and thus we can split it as
\begin{equation}
    d=d_x+d_1+d_2,
\end{equation}
where $d_1$ denotes the part of the differential acting on the propagators in $\omega_\Gamma$ and $d_2$ is the part acting on $\mathbb{B}$ and $\mathbb{A}$ fields.

With this setup, which is basically analogous to the one in \cite{CMoW19}, except for the presence of the red vertices and $d_{\bar{x}}$ already extensively discussed, Eq. \eqref{mdqme_thm} is verified by proving three relations
\begin{itemize}
    \item a relation between the application of $d_1$ and of $\Delta_{\mathcal{V}_{\Sigma_3, x}}$ to the quantum state,
    \item a relation between the application of $d_2$ and of $\Omega_0$ to the quantum state,
    \item a relation between the application of $d_M$ and of the boundary contributions to the quantum state.
\end{itemize}
The proofs of these relations can be carried from \cite{CMoW19} over to the globalized split RW model without any problem. The only difference is when they prove that the contributions in $\partial \text{C}_\Gamma$ consisting of diagrams with two bulk vertices collapsing vanish (which is needed for the third relation). In our case one should consider again three contributions: when two bulk black vertices collapse, when two bulk red vertices collapse, when a red vertex and a black one collapse. The vanishing of these terms follows from Eqs. \eqref{dcme_1}, \eqref{dcme_2}, \eqref{dcme_3}, \eqref{dcme_4}. The rest of the procedure is identical to \cite{CMoW19}.
\end{proof}

\section{Outlook and future direction}
\label{sec:outlook}
Our globalization construction leads to an interesting extension of some aspects in the program presented in \cite{ChanLeungLi2020} for manifolds with boundary and cutting-gluing techniques. In particular, it would be of interest to understand some relations to the deformation quantization of K\"ahler manifolds in the guise of \cite{ReshetikhinTakhtajan1999}, especially using the constructions of \cite{CMoW20}, and Berezin--Toeplitz quantization as presented in \cite{Schlichenmaier2010} (possibly for the noncompact case). 
It also leads to a more general globalization construction of an algebraic index theory formulation by using the BV formalism together with Fedosov's globalization approach as presented in \cite{GLL17}. Moreover, it might also be related to a case of twisted topological field theories, known as Chern--Simons--Rozansky--Witten TFTs, constructed by Kapustin and Saulina in \cite{KapustinSaulina2009}. In particular, they use the BRST formalism to produce interesting observables as Wilson loops and thus one might be able to combine it with ideas of \cite{AlekseevBarmazMnev2013,Mo20}. Another direction would be the study of the RW invariants through our construction for hyperK\"ahler manifolds. We guess that this would require studying observables of RW theory in the BV-BFV formulation, but the globalization procedure should tell something about these 3-manifold invariants. We hope that this might also be compatible with some generalizations of RW invariants in the non-hyperK\"ahler case as discussed in \cite{RS02}.

\begin{appendix}

\section{Topological quantum field theories}
\label{app:TQFT}

This appendix gives a brief introduction to perturbative and functorial constructions of topological (quantum) field theories, especially we recall Atiyah's TQFT axioms. 

\subsection{Brief introduction to perturbative quantum field theory}
On a spacetime manifold $\Sigma$, consider a space of fields\footnote{The space of fields is usually given by sections of some vector bundle over $\Sigma$.} $F_{\Sigma}$ and an action functional $S_\Sigma$ which is required to be \textit{local}. This means that the action is the integral of a density-valued Lagrangian $\mathscr{L}$, called \textit{Lagrangian density}, depending on the fields and on a finite number of higher derivatives. In particular, $S_\Sigma: F_{\Sigma}\rightarrow \mathbb{C}$, with     
\begin{equation}
   S_{\Sigma}(\phi)=\varint_{\Sigma}\mathscr{L}(\phi, \partial\phi, \dots), 
\end{equation}
where $\phi\in F_{\Sigma}$ is a field. The set of data consisting of $(\Sigma, F_{\Sigma}, S_{\Sigma})$ defines a classical Lagrangian field theory.

During the years, physicists have developed several approaches to quantum field theory. Roughly, we can split them into perturbative and non-perturbative methods. Here, we focus on the former. Note that by perturbative, we mean semiclassical: in physics jargon perturbation theory is the idea of expanding through a formal power series around the coupling constant of the action. In the perturbative setting, the protagonist of the story is the \textit{partition function} $Z$, which encodes all the information about the quantum theory it portraits. In general, we can express it through a path integral as
\begin{equation}
\label{ov.tft:part_funct_1}
    Z_{\Sigma}=\varint_{F_{\Sigma}} e^{\frac{i}{\hbar}S_{\Sigma}(\phi)}\mathscr{D}[\phi],
\end{equation}
where $\hbar$ is the reduced Planck constant. 

\begin{rmk}
\label{tft:rmk_measure}
In (\ref{ov.tft:part_funct_1}), $\mathscr{D}$ denotes a formal measure on $F_{\Sigma}$. Depending on the space of fields $F_\Sigma$, this measure is often mathematically ill-defined. Nevertheless, one can define \eqref{ov.tft:part_funct_1} by considering the methods of \emph{perturbative expansion} around critical points of $S_\Sigma$ in a formal power series in $\hbar$ with coefficients given by \emph{Feynman graphs} (see e.g. \cite{FeynmanHibbs1965,P}). 
\end{rmk}

Let us make the above discussion more precise. Consider $\Sigma$ to be a manifold with boundary $\partial \Sigma$ and $B_{\partial \Sigma}$ to be the space of boundary values of the fields on $M$. Since the boundary manifold is the boundary of $\Sigma$, we have a restriction map $F_{\Sigma}\xrightarrow[]{\pi}B_{\partial \Sigma}$. The partition function is thus a complex-valued function on $B_{\partial \Sigma}$ which can be written as 
\begin{equation}
\label{ov.tft:part_funct_2}
    Z_{\Sigma}(\phi_{\partial\Sigma}; \hbar)=\varint_{\{ \phi\in F_{\Sigma}\mid\  \phi\vert_{\partial \Sigma}:=\phi_{\partial \Sigma}\}} e^{\frac{i}{\hbar}S_{\Sigma}(\phi)}\mathscr{D}[\phi],
\end{equation}
where $\phi_{\partial \Sigma}$ is a point in $B_{\partial \Sigma}$. 

The manifold $\Sigma$ may be complicated and, as a result, the computation of $Z_{\Sigma}$ can become difficult. Therefore, it would be desirable to cut $\Sigma$ into smaller and, hopefully, easier pieces, compute the partition function there and then glue them together to get the overall state. Suppose, $\Sigma$ is closed and cut it in two disjoint manifolds $\Sigma_1$ and $\Sigma_2$ along a common boundary $\Sigma$, i.e. $\Sigma= \Sigma_1\sqcup_{\partial \Sigma}\Sigma_2$. If we paste them together, we expect the following condition to hold
\begin{equation}
    Z_{\Sigma}=\varint_{\phi_{\partial\Sigma}\in B_{\partial \Sigma}}Z_{\Sigma_1}(\phi_{\partial \Sigma})Z_{\Sigma_2}(\phi_{\partial \Sigma})\mathscr{D}[\phi_{\partial \Sigma}].
\end{equation}

\subsection{Brief introduction to functorial quantum field theory}
\label{func_qft}
The functorial approach to QFT was developed by Segal in the context of conformal field theory \cite{Se88} and by Atiyah for TQFT \cite{At88}. However, this description is general and it allows us to describe any QFT. 

According to Atiyah's axioms, an $n$-dimensional topological field quantum field theory consists of the following set of data: 
\begin{enumerate}
    \item A Hilbert space $\mathcal{H}(\Sigma)$, called the \textit{space of states}, associated to a closed oriented\footnote{The orientation endows the manifolds with symbols $\{in, out\}$ which denote \textit{incoming} or \textit{outgoing} orientation.} $(n-1)$-manifold $\Sigma$,
    \item A linear map of vector spaces $Z_M: \mathcal{H}_\mathrm{in}\rightarrow \mathcal{H}_\mathrm{out}$, called \textit{partition function}, associated to an oriented $n$-cobordism\footnote{See Example \ref{exm:cob} for a definition.} $M$ from $\Sigma_\mathrm{in}$ to $\Sigma_\mathrm{out}$ (i.e. the boundary of $M$ is assumed to be given as $\partial M=\Sigma_\mathrm{in}\sqcup \Sigma_\mathrm{out}$).
    \item Orientation-preserving diffeomorphisms $\phi: \Sigma_1\rightarrow \Sigma_2$ which act on $\mathcal{H}_\Sigma$ through unitary maps $\rho(\phi): \mathcal{H}_{\Sigma_\mathrm{in}}\rightarrow \mathcal{H}_{\Sigma_\mathrm{out}}$, with $\rho$ a representation. 
    \item Orientation-reversing identity diffeomorphisms $s_\Sigma: \Sigma\rightarrow \Bar{\Sigma}$, where we denote by $\Bar{\Sigma}$, the manifold with opposite orientation. These diffeomorphisms act by $\mathbb{C}$-anti-linear maps $\sigma_{\Sigma}\coloneqq\rho(s_\Sigma): \mathcal{H}_{\Sigma}\rightarrow \mathcal{H}_{\Bar{\Sigma}}$. 
\end{enumerate}

This set of data is required to satisfy the following axioms:
\begin{enumerate}
    \item[(i)](Multiplicativity) For two closed oriented $(n-1)$-manifolds $\Sigma$ and $\Sigma'$, the space of states is multiplicative, i.e. 
    \begin{equation}
        \mathcal{H}_{\Sigma \sqcup \Sigma'}=\mathcal{H}_\Sigma \otimes \mathcal{H}_{\Sigma'}.
    \end{equation}
    For two $n$-cobordisms $M: \Sigma_\mathrm{in}\rightarrow\Sigma_\mathrm{out}$ and $M': \Sigma'_\mathrm{in}\rightarrow\Sigma'_\mathrm{out}$, the partition function is multiplicative
    \begin{equation}
        Z_{M\sqcup M'}=Z_M\otimes Z_{M'}:\quad  \mathcal{\Sigma}_\mathrm{in}\otimes \mathcal{\Sigma'}_\mathrm{in}\rightarrow \mathcal{\Sigma}_\mathrm{out}\otimes \mathcal{\Sigma'}_\mathrm{out}
    \end{equation}
    \item[(ii)](Gluing) Let $M_1: \Sigma_1\rightarrow \Sigma_2$, $M_2: \Bar{\Sigma}_2\rightarrow \Sigma_3$ be two $n$-cobordisms, the glued cobordisms can be constructed by gluing along the common $\Sigma_2$-component as $M_1 \cup_{\Sigma_2}M_2: \Sigma_1 \rightarrow \Sigma_3$. The associated partition function is then obtained by composing the partition functions for $M_1$ and $M_2$ as linear maps:
    \begin{equation}
        Z_{M_1 \cup_{\Sigma_2}M_2}=Z_{M_2}\circ Z_{M_1}:\quad \mathcal{H}_{\Sigma_1}\rightarrow \mathcal{H}_{\Sigma_3}.
    \end{equation}
    \item[(iii)](Involutivity) $Z(\Bar{\Sigma})=Z(\Sigma)^\vee$, where $Z(\Sigma)^\vee$ is the dual vector space.
    \item[(iv)] $\mathcal{H}_{\emptyset}=\mathbb{C}$ and $Z_{\Sigma \times [0,1]}=\Id:\mathcal{\Sigma}\rightarrow\mathcal{\Sigma}$.
    \item[(v)] For $\phi: M\rightarrow M'$ a diffeomorphism, the following diagram commutes:
    \begin{center}
        \begin{tikzcd}[column sep=large, row sep=large]
    \mathcal{H}_{\Sigma_\mathrm{in}}\arrow[r, "Z_M"]\arrow[d, "\rho(\phi\mid_{\Sigma_\mathrm{in}})"']& \mathcal{H}_{\Sigma_\mathrm{out}}\arrow[d, "\rho(\phi\mid_{\Sigma_\mathrm{out}})"] \\
    \mathcal{H}_{\Sigma'_\mathrm{in}}\arrow[ r,"Z_{M'}"']&  \mathcal{H}_{\Sigma'_\mathrm{out}}
    \end{tikzcd}
    \end{center}
    
    It follows that $Z_M$ is invariant under diffeomorphisms of $M$ relative to its boundary components.
    \item[(vi)](Symmetry) The natural diffeomorphism $\Sigma \sqcup \Sigma'\rightarrow \Sigma' \sqcup \Sigma$ is sent by $\rho$ to the natural isomorphism $\mathcal{H}_{\Sigma}\otimes \mathcal{H}_{\Sigma'}\rightarrow\mathcal{H}_{\Sigma'}\otimes \mathcal{H}_{\Sigma}$.
    \item[(vii)] The partition function for the cylinder $\Sigma \times [0,1]$ viewed as a cobordism $\Sigma \times \Bar{\Sigma}\rightarrow\emptyset$ composed with the anti-linear map $(\sigma_\Sigma)^{-1}:\mathcal{H}_{\Sigma}\rightarrow \mathcal{H}_\Sigma$ yields the Hermitian inner product $\braket{-,-}: \mathcal{H}_\Sigma \times \mathcal{H}_\Sigma \rightarrow\mathbb{C}$. 
\end{enumerate}


Let $M$ be a closed $n$-manifold, which can be regarded as a cobordism $\emptyset \rightarrow \emptyset$. The associated partition function $Z_M\in \mathbb{C}$ is an invariant under orientation-preserving diffeomorphisms on $M$. 
In general, for a mapping torus
$\frac{\Sigma \times [0,1]}{\Sigma \times \{0\}\stackrel{\phi}{\sim} \Sigma \times \{1\}}$ with $\phi:\Sigma\rightarrow\Sigma$ a gluing diffeomorphism, the axioms above imply that $Z=\Tr_{\mathcal{H}_\Sigma}\rho(\phi)$. In particular, for the product manifold $\Sigma \times S^1$ formed by identifying the opposite ends of the cylinder:
\begin{equation}
\label{ov:tft.partfunc_withS1}
    Z_{\Sigma \times S^1}=\Tr_{\mathcal{H}_\Sigma}(\Id)=\dim \mathcal{H}_{\Sigma}\in \mathbb{Z}_{\geq 0}.
\end{equation}
This implies that the space of states is finite-dimensional.

\subsubsection{Atiyah's axioms for TQFTs}
Atiyah's axioms can be reformulated in the categorical language. We start with some prerequisites before arriving to the definition of a TQFT.

Let us consider a symmetric monoidal category $\mathbf{C}$: it is a category equipped with a bifunctor $\otimes: \mathbf{C}\times \mathbf{C}\rightarrow \mathbf{C}$, called \textit{monoidal product} which allows to, roughly speaking, ``multiply" objects. This product is well-defined because it is associative up to natural isomorphisms (in jargon it satisfies the \textit{pentagon equations} \cite{Ma71}). Moreover, a monoidal category is symmetric when for all the objects $A, B\in \mathbf{C}$ there are natural isomorphisms
\begin{equation}
    \beta_{A,B}: A\otimes B\rightarrow B\otimes A
\end{equation}
compatible with the associativity of the monoidal structure (they satisfy the \textit{hexagon equations} \cite{Ma71}).

\begin{exm}
\label{exm:cob}
We consider two examples, which we need for later:
\begin{enumerate}
    \item The category $\mathbf{Vect}_{\mathbb{K}}$ whose objects are $\mathbb{K}$-vector spaces for some field $\mathbb{K}$ and morphisms are $\mathbb{K}$-linear maps. It is monoidal with the usual tensor product as monoidal product ($\otimes :=\otimes_{\mathbb{K}}$) and with unit $\boldsymbol{1}:=\mathbb{K}$. Moreover, one can show that it is symmetric.
    \item The category of oriented \textit{cobordisms}$, \mathbf{Cob}^{\text{or}}_n$. The objects are oriented closed $(n-1)$- dimensional manifolds and morphisms are diffeomorphisms classes of bordisms. In a more down to Earth language, this means that the morphisms are given by the bulk of an oriented compact $n$-dimensional manifold with boundary, whose boundary components are the objects. We can compose a morphism with another morphism simply by gluing along the common boundaries. It has a monoidal structure where the monoidal product is given by the disjoint union and the unit object is the empty set $\emptyset$ viewed as an $(n-1)$-dimensional manifold. The objects are endowed with orientations labeled by symbols $\{\mathrm{in}, \mathrm{out}\}$.
\end{enumerate}
\end{exm}

Atiyah's axioms can be reformulated in a short way as:

\begin{defn}[Topological field theory]
Let $(\mathbf{C},\otimes)$ be a symmetric monoidal category. An $n$-dimensional oriented closed topological field theory (TFT) is a symmetric monoidal functor  
\begin{equation}
    Z: \mathbf{Cob}^{\text{or}}_n\rightarrow \mathbf{C}.
\end{equation}
\end{defn}

\begin{defn}[Topological quantum field theory]
\label{ov:tqft}
An $n$-dimensional oriented topological quantum field theory (TQFT) is a symmetric monoidal functor  
\begin{equation}
    Z: \mathbf{Cob}^{\text{or}}_n\rightarrow \mathbf{Vect}_{\mathbb{C}}.
\end{equation}
\end{defn}

\begin{rmk}
Note that the target category contains also infinite-dimensional vector spaces. However, an analogue of Eq. \eqref{ov:tft.partfunc_withS1}, implies that the state spaces are finite-dimensional.
\end{rmk}

\begin{rmk}
As seen in Definition \ref{ov:tqft}, the category of smooth oriented cobordisms is usually used to describe a TQFT. However, cobordisms may possess other geometric structures such as conformal structure, spin structure, framing, boundaries, etc. Consequently, the associated field theory will be conformal QFT, spin or framed TQFT, etc. For example, for Yang-Mills theories and sigma models, the source category is the category of smooth Riemannian manifolds with a collar at the boundary.
\end{rmk}

\begin{exm}
As first example, let us consider a cobordism represented by some pair of pants with genus 1 (see Fig. \ref{tqft:fig:pants}). The TQFT $F$ assigns to each boundary component a Hilbert space, i.e. $Z(\partial_k\Sigma)=\mathcal{H}_k$ for $k=1,2,3$. Since $Z$ is a symmetric monoidal functor, we have $Z(\partial_1\Sigma\sqcup \partial_2\Sigma\sqcup \partial_3\Sigma)=\mathcal{H}^{\vee}_1\otimes \mathcal{H}^{\vee}_2\otimes\mathcal{H}_3$. As said before, each cobordism comes with a certain orientation: $\partial_1\Sigma$ as well as $\partial_2\Sigma$ are incoming boundaries (which we denote in the figure by an incoming arrow), while $\partial_3\Sigma$ is an outgoing boundary (which we denote in the figure by an outgoing arrow). Associated to $\partial_1\Sigma$ and $\partial_2\Sigma$, we have an incoming Hilbert space $\mathcal{H}_{\text{in}}\coloneqq \mathcal{H}^{\vee}_1\otimes \mathcal{H}^{\vee}_2\cong \mathcal{H}_1\otimes \mathcal{H}_2$ and an outgoing Hilbert space $\mathcal{H}_{\text{out}}\coloneqq\mathcal{H}_3$ associated to $\partial_3\Sigma$. The state $\psi$ corresponding to this
cobordism and the given TQFT is given as the value of the morphism represented by the genus 1 pair of pants above (i.e. the bounding manifold $M$) under $F$.
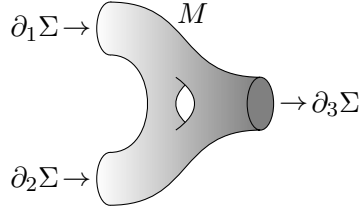
\begin{figure}[hbt!]
    \centering
    \begin{tikzpicture}[rotate=270,transform shape,tqft, view from=incoming,
   every incoming boundary component/.style={fill=gray} 
    ]
        \pic[draw,
    tqft/pair of pants,
  every lower boundary component/.style={draw},
   every incoming lower boundary component/.style={solid},
   every outgoing lower boundary component/.style={dashed},
   genus=1, 
   shade, shading angle=270,
  name=A
];
\node[label={below:$\uparrow$}] at (A-outgoing boundary 1) {};
\node[label={below:$\uparrow$}] at (A-outgoing boundary 2) {};
\node[label={above:$\uparrow$}] at (A-incoming boundary 1) {};
\node[rotate=90] at (A-outgoing boundary 1) {\hspace{-2cm $\partial_1\Sigma$}};
\node[rotate=90] at (A-outgoing boundary 2) {\hspace{-2cm $\partial_2\Sigma$}};
\node[rotate=90] at (A-incoming boundary 1) {\hspace{2cm $\partial_3\Sigma$}};
\node[rotate=90] at (-1.2,-0.9) {$M$};
\end{tikzpicture}
    \caption{Cobordism $M$ represented by pair of pants of genus 1 with boundary components $\partial_1\Sigma$, $\partial_2\Sigma$, $\partial_3\Sigma$.}
    \label{tqft:fig:pants}
\end{figure}
\end{exm}

\begin{exm}
As already mentioned in Section \ref{func_qft}, a closed manifold $\Sigma$ can be seen as a cobordism $\emptyset \rightarrow \emptyset$. We can cut it in two disjoint manifolds $\Sigma_1$ and $\Sigma_2$ along a common boundary $\partial\Sigma$, i.e. $\Sigma= \Sigma_1\sqcup_{\partial \Sigma}\Sigma_2$. Then we can assign an opposite orientation to $\partial_1\Sigma_1$ with respect to the orientation of $\partial_1\Sigma_2$. The same can be done to  $\partial_2\Sigma_1$ with respect to the orientation of $\partial_2\Sigma_2$. The two manifolds with boundary $\Sigma_1$ and $\Sigma_2$ can be glued back together to recover the partition function of the closed manifold $\Sigma$, see Fig. \ref{tft:fig:gluing}.

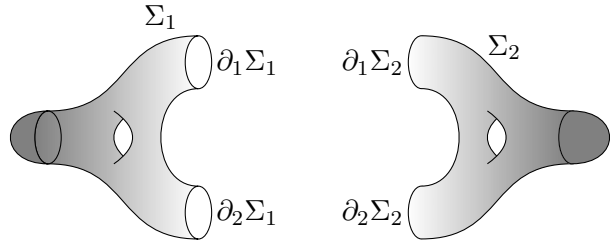
\begin{figure}[hbt!]
    \centering
  \begin{subfigure}[b]{0.3\textwidth}
    \centering
\begin{tikzpicture}[rotate=-90,transform shape,tqft, view from=incoming,
    ]
        \pic[draw,
   tqft/reverse pair of pants,
   every lower boundary component/.style={draw},
   every incoming boundary component/.style={solid, fill=white, draw=white},
   every outgoing boundary component/.style={solid, draw=white},
   every outgoing lower boundary component/.style={dashed},
   genus=1, 
   shade, shading angle=90,
  name=A
];
\pic[draw,
    tqft/cup,
  every lower boundary component/.style={draw},
   every incoming lower boundary component/.style={solid},
   every outgoing lower boundary component/.style={dashed},
   fill=gray,
  anchor=incoming boundary 1,name=B, at=(A-outgoing boundary 1)
];
 \node[rotate=90] at (A-incoming boundary 1) {\hspace{1.3cm $\partial_1\Sigma_1$}};
 \node[rotate=90] at (A-incoming boundary 2) {\hspace{1.3cm $\partial_2\Sigma_1$}};
 \node[rotate=90] at (-0.6,-0.5) {$\Sigma_1$};
\end{tikzpicture}
\end{subfigure}
 \begin{subfigure}[b]{0.3\textwidth}
    \centering
    \begin{tikzpicture}[rotate=270,transform shape,tqft, view from=incoming,
   every incoming boundary component/.style={fill=gray},
   every lower boundary component/.style={draw=gray}
    ]
        \pic[draw,
    tqft/pair of pants,
   every incoming lower boundary component/.style={solid,draw=gray},
   every outgoing lower boundary component/.style={dashed},
   genus=1, 
   shade, shading angle=270,
  name=A
];
\pic[draw,
    tqft/cap,
  every lower boundary component/.style={draw},
   fill=gray, every outgoing boundary component/.style={draw=white},
  anchor=outgoing boundary 1,name=B, at=(A-incoming boundary 1)
];

 \node[rotate=90] at (A-outgoing boundary 1) {\hspace{-1.3cm $\partial_1\Sigma_2$}};
 \node[rotate=90] at (A-outgoing boundary 2) {\hspace{-1.3cm $\partial_2\Sigma_2$}};
 \node[rotate=90] at (-1.2,-0.9) {$\Sigma_2$};
\end{tikzpicture}
\end{subfigure}
    \caption{Gluing of two manifolds $\Sigma_1$ and $\Sigma_2$ along a common boundary $\partial \Sigma$.}
    \label{tft:fig:gluing}
\end{figure}
\end{exm}


\begin{rmk}
It is important to highlight that the functorial approach to TQFT is not based on any perturbative framework, therefore, its nature is intrinsically non perturbative.  
\end{rmk}

Successively, in \cite{BD95}, Baez and Dolan suggested enhancing Atiyah's notion of TQFT to a functor from the $(\infty, n)$-extension of the cobordism category. Their idea is to allow gluing as well as cutting with higher codimension data. Moreover, they conjectured these TQFTs to be completely classifiable:  this conjecture is known as the \textit{Cobordism hypothesis}. In \cite{Lu09}, Lurie provided a complete classification result for fully extended TQFTs formulated in the language of $(\infty,n)$-categories, a generalization of the notion of a category.

\section{Elements of formal geometry}
\label{app:formal_geometry}
In this appendix, we want to explain some of the important notions of formal geometry developed in \cite{GK71,Bo11} which are used for the globalization procedure. 
\subsection{Formal power series on vector spaces}
Let $V$ be a finite dimensional vector space. The polynomial algebra on $V$ is given by
\begin{equation}
\Sym^\bullet (V^\vee) = \bigoplus_{k=0}^\infty \Sym^k(V^\vee).
\end{equation}
If we choose $\{e_1, \dots, e_n\}$ to be a basis of $V$, with dual basis $\{y_1, \dots , y_n\}$, then elements $f \in \Sym^\bullet (V^\vee)$ are given by
\begin{equation}
    f(y)=\sum_{i_1,\dots,i_n=1}^\infty f_{i_1,\dots, i_n}y_1^{i_1}\dots y_n^{i_n}=\sum_I f_Iy^{I},
\end{equation}
with only finitely many non-vanishing $f_I$. We have denoted by $I$ a multi-index and $y^I=y^{i_1}\dots y^{i_n}$, $y^{\emptyset} \coloneqq 1$.

This algebra can be completed to the algebra of formal power series $\reallywidehat{\Sym}^\bullet(V^\vee)$, with infinitely many nonzero coefficients $f_I$. Note that both algebras $\Sym^\bullet(V^\vee)$ and $\reallywidehat{\Sym}^\bullet (V^\vee)$ are commutative, with the multiplication of formal power series or polynomials, respectively, generated by $V^\vee$. One can specify derivations of these algebras by their value on these generators, therefore, the map
\begin{alignat}{1}
\label{iso}
    V \otimes \Sym^\bullet (V^\vee) &\rightarrow \Der(\Sym^\bullet (V^\vee))\\
    v \otimes f &\mapsto \Big(V^\vee \ni \alpha \mapsto \alpha(v) \cdot f\Big)
\end{alignat}
is an isomorphism with inverse
\begin{alignat}{1}
  \Der(\Sym^\bullet (V^\vee)) & \rightarrow V \otimes \Sym^\bullet (V^\vee) \\
    D & \mapsto \sum_{i=1}^n e_i \otimes D(y^i).
\end{alignat}
In coordinates this corresponds to sending $e_i \mapsto \frac{\partial}{\partial y^i}$.

\subsection{Formal exponential maps}
\label{subsec:form_exp_map}
Let $M$ be a manifold and let  $\varphi: U \rightarrow M$, with $U \subset TM$ an open neighbourhood of the zero section. For $x \in M$, $y \in T_xM \cap U$, we can write $\varphi(x,y)=\varphi_x(y)$. 
\begin{defn}[Generalized exponential map]
We call $\varphi$ a \textit{generalized exponential map} if for all $x \in M$ we have that  
\begin{enumerate}
    \item $\varphi_x(0)=x$, 
    \item $d\varphi_x(0)=\Id_{T_xM}$.
\end{enumerate}
\end{defn}
In local coordinates
\begin{equation}
    \varphi_x^i(y)=x^i+y^i+\frac{1}{2}\varphi_{x,jk}^iy^jy^k+\frac{1}{3!} \varphi_{x,jkl}^iy^jy^ky^l+\dots,
\end{equation}
where the $x_i$ and $y_i$ are  respectively the base and the fiber coordinates. Two generalized exponential maps are identified if their corresponding jets agree at all orders. 

\begin{defn}[Formal exponential map]
A \textit{formal exponential map} is an equivalence class of generalized exponential maps. A formal exponential map is completely specified by the sequence of functions $\Big (\varphi_{x,i_1,\dots,i_k}^i \Big)_{k=0}^\infty$.
\end{defn}

From now on, we will abuse notation and we will denote equivalence classes and their representatives by $\varphi$. One can produce a section $\sigma \in \Gamma(\reallywidehat{\Sym}^\bullet(T^\vee M))$ from a formal exponential map $\varphi$ and a function $f \in \mathcal{C}^\infty(M)$,  via $\sigma_x=\mathrm{T}\varphi_x^*f$, with $\mathrm{T}$ the Taylor expansion in the fiber coordinates around $y=0$ and the pullback defined by any representative of $\varphi$. We will denote this section by $\mathrm{T}\varphi^*f$, and note that it is independent of the choice of the representative since it only depends on the jets.

\subsection{Grothendieck connection}
\label{grotta}
\begin{defn}[Grothendieck connection]
On $\reallywidehat{\Sym}^\bullet(T^\vee M)$ we can define a flat connection $\Gr$, satisfying  the property that
\begin{equation}
    \Gr \sigma=0 \iff \sigma=\mathrm{T}\varphi^*f,
\end{equation}
for some $f\in \mathcal{C}^\infty(M)$. Namely, $\Gr=d+R$ with $R\in \Gamma\Big(T^\vee M \otimes TM \otimes \reallywidehat{\Sym}^\bullet(T^\vee M)\Big)$, a 1-form with values in derivations\footnote{We can use the isomorphism \eqref{iso} to identify derivations of $\reallywidehat{\Sym}^\bullet(T^\vee M)$ with $\Gamma\Big(TM \otimes \reallywidehat{\Sym}^\bullet(T^\vee M)\Big)$.} of $\reallywidehat{\Sym}^\bullet(T^\vee M)$. In local coordinates, $R$ is defined as $R_idx^i$ and 
\begin{equation}
    R_i(x;y)\coloneqq-\bigg[\bigg(\frac{\partial \varphi}{\partial y}\bigg)^{-1}\bigg]^k_j\frac{\partial \varphi^j}{\partial x^i}\frac{\partial}{\partial y^k}= Y_i^k(x;y)\frac{\partial}{\partial y^k}.
\end{equation}
Hence, we have
\begin{equation}
    R(x;y)=R_i(x;y)dx^i=Y_i^k \frac{\partial}{\partial y^k}dx^i.
\end{equation}
The connection $\Gr$ is called the \textit{Grothendieck connection}\footnote{In the setting of field theory we have called this the \emph{classical} Grothendieck connection in order to distinguish it from its quantum counterpart.}.
\end{defn}

For $\sigma \in \Gamma(\reallywidehat{\Sym}^\bullet(T^\vee M))$, $R(\sigma)$ is expressed via the Taylor expansion (in the $y$ coordinates) of
\begin{equation}
    -d_y\sigma \circ (d_y\varphi)^{-1} \circ d_x\varphi: \Gamma(TM) \rightarrow \Gamma(\reallywidehat{\Sym}^\bullet(T^\vee M)),
\end{equation}
and therefore, $R$ does not depend on the coordinate choice. For a vector field $\xi=\xi^i \frac{\partial}{\partial x^i}$, we have
\begin{equation}
    \Gr^\xi=\xi + \hat{\xi},
\end{equation}
with 
\begin{equation}
    \hat{\xi}(x;y)=\iota_\xi R(x;y)=\xi^i(x)Y_i^k(x;y) \frac{\partial}{\partial y^k}.
\end{equation}

\begin{rmk}
 The classical Grothendieck connection is flat (i.e. $D^2_G=0$). Moreover, the flatness condition translates into 
\begin{equation}
\label{MC_R}
    d_xR+\frac{1}{2}[R,R]=0,
\end{equation}
which is a \textit{Maurer--Cartan} (MC) equation for $R$.
\end{rmk}

\begin{rmk}
It can be proved that its cohomology is concentrated in degree 0 and is given by 
\begin{equation}
    H^0_{\Gr}\Big(\Gamma(\reallywidehat{\Sym}^\bullet(T^\vee M))\Big)=\mathrm{T}\varphi^* \mathcal{C}^\infty(M) \cong \mathcal{C}^\infty(M).
\end{equation}
\end{rmk}

\subsection{Formal vertical tensor fields}
Let $E \rightarrow M$ be any \textit{tensorial bundle}\footnote{A tensorial bundle is any bundle which is a  tensor product or antisymmetric or symmetric product of the tangent or cotangent bundle, or a direct sum thereof.}, for example $E=\bigwedge^k TM$. Its sections are called \textit{tensor fields of type E}.

\begin{defn}[Formal vertical bundle]
The associated \textit{formal vertical bundle} to $E$ is then $\hat{E} \coloneqq E \otimes \reallywidehat{\Sym}^\bullet(T^\vee M)$. Its sections are called \textit{formal vertical tensors of type $E$}.
\end{defn}

\begin{rmk}
These bundles can be thought of as tensors of the same type on $TM$ where the dependence on fiber directions is formal.
\end{rmk}

The formal exponential map defines an injective map 
\begin{equation}
    \mathrm{T}\varphi^*:E \rightarrow \hat{E}
\end{equation}
via the Taylor expansion of a tensor field pulled back\footnote{Note that $\varphi$ is a local diffeomorphism and hence we can define the pullback of contravariant tensors as the pushforward of the inverse.} to $U$ by $\varphi$. 

Furthermore, we can let $R$ act by formal derivatives and therefore, we get a Grothendieck connection $\Gr=d+R$ on any formal vertical tensor bundle. Similarly, as before, we have:
\begin{itemize}
    \item $\Gr$ is flat;
    \item flat sections of $\Gr$ are precisely the ones in the image of $\mathrm{T}\varphi^*$;
    \item the cohomology of $\Gr$ is concentrated in degree 0 and given by the flat sections, i.e. $\hat{E}$-valued 0-forms.
\end{itemize}

\subsection{Changing the formal exponential map}
\label{app5}
We will denote by $\varphi$ be a family of formal exponential maps depending on a parameter $t$ belonging to an open interval $I$. One can then associate to this family a formal exponential map $\psi$ for the manifold $M \times I$ by
\begin{equation}
     \psi(x,t,y,\tau) \coloneqq (\varphi_x(y),t+\tau),
\end{equation}
with $\tau$ the tangent variable to $t$. The corresponding connection $\tilde{R}$ is defined as follows. Let $\tilde{\sigma}$ be a section of $\reallywidehat{\Sym}^\bullet\Big(T^\vee(M\times I)\Big)$, by definition we have:
\begin{equation}
    \tilde{R}(\tilde{\sigma})=-(d_y\tilde{\sigma},d_\tau \tilde{\sigma}) \circ \begin{pmatrix} 
    (d_y\varphi)^{-1} & 0 \\
    0 & 1\end{pmatrix} \circ \begin{pmatrix}
    d_x\varphi & \Dot{\varphi} \\
    0 & 1
    \end{pmatrix}.
\end{equation}
Hence, $\tilde{R}=R+Cdt+T$, with $R$ defined as in Section \ref{grotta}, but with a $t$-dependence now, and $T=-dt\frac{\partial}{\partial \tau}$. 

The MC equation \eqref{MC_R} can be reformulated for $\tilde{R}$ observing that :
\begin{itemize}
    \item $d_x T=d_tT=0$,
    \item $T$ commutes with $R$ and $C$.
\end{itemize}

The $(2,0)$-form component of the MC equation over $M \times I$ yields again the MC equation for $R$, while the $(1,1)$- component reads
\begin{equation}
    \Dot{R}=d_xC+[R,C].
\end{equation}

\begin{rmk}
Under a change of formal exponential map, $R$ changes by a gauge transformation having as generator the section $C$ of $\hat{\mathfrak{X}}(TM) \coloneqq TM \otimes \reallywidehat{\Sym}^\bullet(T^\vee M)$. Finally, if $\sigma$ is a section in the image of $\mathrm{T}\varphi^*$, a simple computation yields
\begin{equation}
    \Dot{\sigma}=-L_C\sigma.
\end{equation}
One can think of it as the associated gauge transformation for sections.
\end{rmk}

\subsection{Extension to graded manifolds}
The previous results can be generalized to the category of graded manifolds exploiting the algebraic reformulation of formal exponential maps developed in \cite{LS17}.

More concretely, given a formal exponential map $\varphi$ on a smooth manifold $M$, one can construct a map
\begin{equation}
    \pbw: \Gamma(\reallywidehat{\Sym}^\bullet(TM)) \rightarrow \mathcal{D}(M)
\end{equation}
from sections of the completed symmetric algebra of the tangent bundle to the algebra of differential operators $\mathcal{D}$ by defining
\begin{equation}
    \pbw\Big(X_1 \odot \dots \odot X_n\Big)(f)=\frac{d}{dt_1}\Bigg|_{t_1=0} \dots \frac{d}{dt_n}\Bigg|_{t_n=0} f\Big(\varphi(t_1X_1+\dots+t_nX_n)\Big),
\end{equation}
where we denote by $\odot$ the symmetric product.
One can also define this map in the category of graded manifolds by choosing a torsion-free connection $\nabla$ on the tangent bundle of a graded manifold $M$ with Christoffel symbols $\Gamma^k_{ij}$. In particular, there still exists an element
$R^{\nabla} \in \Omega^1\Big(M,TM \otimes \reallywidehat{\Sym}^\bullet(T^\vee M)\Big)$ with the property that $\Gr=d_M+R^{\nabla}$ is a flat connection on $\reallywidehat{\Sym}^\bullet(T^\vee M)$, i.e.
\begin{equation}
    R^{\nabla}=-\delta+\Gamma+A^{\nabla}.
\end{equation}
In local coordinates $\{x^i\}$ on $M$ and $\{y^i\}$ on $TM$, we have   
\begin{equation}
\begin{split}
    \delta &=dx^i\frac{\partial}{\partial y^i},\\
    \Gamma &=-dx^i\Gamma^k_{ij}(x)y^j\frac{\partial}{\partial y^k}, \\
    A^{\nabla}&=dx^i \sum_{|J| \geq 2}A^k_{i,J}(x)y^J \frac{\partial}{\partial y^k}.
\end{split}
\end{equation}

We define $R_i \in \Gamma(M, \reallywidehat{\Sym}^\bullet(T^\vee M) \otimes TM)$ and $Y_i^k \in \Gamma(M, \reallywidehat{\Sym}^\bullet(T^\vee M))$ via
\begin{equation}
    R^{\nabla}=R_i(x;y)dx^i=Y_i^k(x;y)dx^i\frac{\partial}{\partial y^k}.
\end{equation}
In particular note that $\Gr$ extends to a differential on $\Omega^{\bullet}(M,\reallywidehat{\Sym}^\bullet(T^\vee M))$. T

The Taylor expansion of a function $f \in \mathcal{C}^\infty(M)$ can be defined as \cite{LS17}
\begin{equation}
\label{1}
    \mathrm{T}\varphi^* f \coloneqq \sum_I \frac{1}{I!}y^I\pbw \Big(\underset{\leftarrow}{\partial_x^I}\Big)(f),
\end{equation}
where
\begin{equation}
    \Big(\underset{\leftarrow}{\partial_x^I}\Big)=\underbrace{\partial_{x_1} \odot \dots \odot \partial_{x_1}}_{i_1} \odot \dots \odot \underbrace{\partial_{x_n} \odot \dots \odot \partial_{x_n}}_{i_n}.
\end{equation}
One can prove that \eqref{1} has still the same properties, i.e. the image of $\mathrm{T}\varphi^*$ consists precisely of the $\Gr$-closed sections of $\reallywidehat{\Sym}^\bullet(T^\vee M)$. 

We can describe how the exponential map varies under the choice of a connection mimicking the construction for the smooth case described in Section \ref{app5}. More concretely, assume we have a smooth family $\nabla^t$ of connections on $TM$, then we can associate to that family a connection $\tilde{\nabla}$ on $M \times I$. The associated $R^{\tilde{\nabla}}$ can be split as in Section \ref{app5}
\begin{equation}
R^{\tilde{\nabla}}=R^{\nabla^t}+C^{\nabla^t}dt+T,
\end{equation}
 where $C \in \Gamma(M,\reallywidehat{\Sym}^\bullet(T^\vee M))$. As previously, $\Gr^2=0$ means
\begin{equation}
    \Dot{R}^{\nabla^t}=d_MR^{\nabla^t}+[C^{\nabla^t},R ^{\nabla^t}],
\end{equation}
and for any section $\sigma$ in the image of $\mathrm{T}\varphi^*$ we have
\begin{equation}
    \Dot{\sigma}=-L_{C^{\nabla^t}} \sigma.
\end{equation}

\section{Elements of derived geometry}
\label{app:derived_geometry}
In Section \ref{sec:BV-BFV}, we have introduced the BV formalism as a way to deal with non-isolated critical points for the action of a gauge theory. In other words, this means that the critical locus of the action functional (i.e. the set of points such that $\delta S=0$) is singular. The BV formalism instructs us to resolve the singularities homologically by taking the \textit{derived critical locus} of the action functional, which is a smooth object in the category of derived spaces: this is done by the Koszul resolution of the critical locus. More generally, this procedure can be understood globally in the setting of \textit{derived algebraic geometry} (DAG) \cite{To14,PTTV13}. However, for the present work, we do not require the whole DAG language. For us it is sufficient to work with a ``tamed" version of DAG, namely the framework developed by Costello in \cite{Co11a,Co11b} to deal with formal mapping stacks which capture the geometry of derived critical loci in nonlinear sigma models. 

\subsection{Category of derived manifolds}
Here, we want to define the category of derived manifolds. Let us start with the objects.

Denote by $\Omega^\bullet(M)$ the de Rham algebra of a manifold $M$, which, in other words, is a sheaf of commutative differentially graded algebras.

\begin{defn}[Derived manifold, \cite{Co11a}]
A derived manifold (over $\mathbb{R}$) is a pair $(M, \mathcal{A})$, where $M$ is a smooth manifold and $\mathcal{A}$ is a sheaf of unital differentially graded $\Omega^\bullet(M)$, satisfying the conditions
\begin{enumerate}
    \item As a sheaf of $\mathcal{C}^\infty(M)$-algebras, $\mathcal{A}$ is locally free and of finite rank.
    \item There is a morphism $\A\rightarrow\mathcal{C}^\infty(M)$ of sheaves of $\Omega^\bullet(M)$-algebras and the kernel of this map is a sheaf of nilpotent ideals.
    \item The topology of $M$ has a basis such that the cohomology $\A(U)$ is concentrated in nonpositive degrees for each basis set of $U$.
\end{enumerate}
\end{defn}
\begin{exm}
         Trivially, any manifold $M$ with $\A=\cinfty$ is a derived manifold.
\end{exm}
\begin{exm}
    Let $M$ be a manifold, take $\A=\derham$ equipped with a de Rham differential. The pair $(M,\derham)$ is derived manifold which we denote by $M_{\,\text{dR}}$.
\end{exm}
\begin{exm}
    Let $M$ be a complex manifold, take $\A=\Omega^{0,\bullet}(M)$ equipped with a Dolbeault differential $\Bar{\partial}$. The pair $(M,\Omega^{0,\bullet}(M))$ is complex derived manifold.
\end{exm}

\begin{defn}[Morphisms of derived manifolds]
A morphism of derived manifolds $(M,\A)\rightarrow(N,\mathcal{B})$ is a smooth map $f:M\rightarrow N$ together with a morphism $\phi:f^{-1}\mathcal{B}\rightarrow\A$ of $f^{-1}\Omega^\bullet(N)$-algebras such that the diagram
\begin{center}
        \begin{tikzcd}[column sep=large, row sep=large]
    f^{-1}\mathcal{B}\arrow[r, "\phi"]\arrow[d]& \A\arrow[d] \\
    f^{-1}\mathcal{C}^\infty(N)\arrow[ r]& \mathcal{C}^\infty(M)
    \end{tikzcd}
    \end{center}
commutes.
\end{defn}

\begin{notat}
We denote by $\mathbf{DMan}$ the category with objects given by derived manifolds and morphisms given by the ones we have just defined.
\end{notat}

The notion of morphisms between derived manifolds is further enriched by introducing \textit{weak equivalences} between derived manifolds. For this purpose, we will use the nilpotent differential graded (dg) ideal $I$ of $(M,\A)$  defined as the kernel of the map $\A\rightarrow\mathcal{C}^\infty(M)$. Here, we have a filtration by powers of the nilpotent ideal. Let $\Grad\A$ denote the associated graded algebra with degree $k$ part $\Grad^k\A\coloneqq F^k\A/F^{k+1}\A$ and the induced differential.

\begin{defn}[Weak equivalence]
A morphism $(f,\phi):(M,\A)\rightarrow(N,\mathcal{B})$ of derived manifolds is a \textit{weak equivalence} if $f$ is a differomoprhisms and the induced map
\begin{equation}
    \Grad\phi:f^{-1}\Grad\mathcal{B}\rightarrow\Grad\A
\end{equation}
is a quasi-isomorphism.
\end{defn}

Having a filtration has also another aim: it should mirror the role of the tower of quotients of a local Artinian algebra in formal deformation theory. In that context, in many situations, it is useful to proceed by \textit{Artinian induction}: let $(A, \mathfrak{m})$ be a local Artinian algebra over $\mathbb{R}$, there is a tower 
\begin{equation}
    A=A/\mathfrak{m}^{n+1}\rightarrow A/\mathfrak{m}^{n}\rightarrow \dots \rightarrow A/\mathfrak{m}\cong \mathbb{R}.
\end{equation}
This tower is then used to prove some properties of $A$. Following these ideas, in fact, derived manifolds can be used to study derived deformation theory as Artinian algebras are used to study formal deformation theory. Now, let us define Artinian dg algebras and make these ideas more precise. We will be concise so we refer to \cite{Co11a,CG16} for a more detailed exposition. 

\begin{defn}[Artinian dg algebra]
An \textit{Artinian dg algebra} $R$ over a field $\mathbb{K}$ is a finite dimensional dg algebra over $\mathbb{K}$, concentrated in non positive degrees, with a unique nilpotent dg ideal $\mathfrak{m}$ such that $R/\mathfrak{m}\cong\mathbb{K}$.
\end{defn}
The relation between Artinian dg algebra and derived manifolds is explained by the following Proposition
\begin{prp}[\cite{GG14}]
\label{prp:art_dman}
There is a fully faithful embedding
\begin{align}
\begin{split}
\Spec: \mathbf{dgArt}^{op}_{\mathbb{K}}&\rightarrow\mathbf{DMan},\\
R&\mapsto \Spec R\coloneqq(\pt, R).
\end{split}
\end{align}
\end{prp}

The importance of Artinian dg algebras comes from being a sort of ``test object" in formal derived deformation theory.

\begin{defn}[Formal derived moduli problem, \cite{Lu11}]
\label{formal_mod_prob}
A formal derived moduli problem over $\mathbb{K}$ is a functor
\begin{equation}
    X:\mathbf{dgArt}_{\mathbb{K}}\rightarrow\mathbf{sSets},
\end{equation}
where $\mathbf{sSets}$ is the category of simplicial sets and $X$ is such that $X(\mathbb{K})$ is contractible and $X$ preserves certain homotopy limits.
\end{defn}

\begin{rmk}
Loosely speaking, Artinian dg algebras are points with nilpotent directions in derived manifolds. Hence, studying formal moduli problems corresponds to studying the formal neighbourhoods of such points.
\end{rmk}
\begin{rmk}
We can now combine Definition \ref{formal_mod_prob} and Definition \ref{prp:art_dman}. We generalize the formal derived moduli problems by extending the functor $X$ to a functor $\mathbf{DMan}^{op}\rightarrow \mathbf{sSets}$. In this way, we can study formal moduli problems parametrized by a smooth manifold $M$ (before they were parametrized by an Artinian algebra).
\end{rmk}

\subsection{Derived stacks}
In this section, we are going to introduce briefly the \textit{derived stacks}. These are the spaces studied in derived algebraic geometry. 

Recall the functor of points approach in algebraic geometry: a scheme can be defined as a functor from the category of commutative $\mathbb{K}$-algebras, i.e. $\mathbf{CAlg}_{\mathbb{K}}$, to the category of sets. Motivated by the study of moduli problems, since the focus was to classify objects with their isomorphisms, the target category was extended to the category of groupoids (a small category whose morphisms are invertible). These new functors were called stacks. A further generalization is called \textit{higher stacks}, where the interest is to classify objects up to a higher notion of equivalence rather than isomorphisms (e.g. quasi-isomorphisms). The target category in this case is extended to the category of simplicial sets. Finally derived stacks (or derived higher stacks) arrive when we enlarge the source category to $\mathbf{DCAlg}_{\mathbb{K}}$, i.e. the category of simplicial commutative $\mathbb{K}$-algebras. This category has a natural model category structure, which allows to do homotopy theory. Hence, derived stacks are defined as functors $\mathbf{DCAlg}_{\mathbb{K}}\rightarrow \mathbf{sSets}$ which send equivalences in the source category to weak homotopy equivalences on the target and satisfy a \textit{descent} condition \cite{To06} . 

The related definition in Costello's approach \cite{Co11a,Co11b} is similar, the only difference is for the source category which is the category of derived manifolds $\mathbf{DMan}^{op}$.




\begin{defn}[Derived stack]
A\textit{ derived stack} or \textit{derived space} is a functor:
\begin{equation}
    X:\mathbf{DMan}^{op}\rightarrow\mathbf{sSets}
\end{equation}
such that:
\begin{itemize}
    \item $X$ takes weak equivalences of derived manifolds to weak equivalences of simplicial sets.
    \item $X$ satisfies \v{C}ech descent.
\end{itemize}
\end{defn}

The notion of \v{C}ech descent is outside the scope of the present work, we refer to \cite{GG14,Ste17} for a definition. 

In the following, we will study a particular type of derived stack with a geometric interpretation, i.e. the derived stack represented by $\Linf$-spaces.

\subsection{\texorpdfstring{$L_{\infty}$}{L}-spaces}
The heart of the philosophy of deformation theory consists of the following statement: ``every formal derived moduli problem is represented by an $\Linf$-algebra". For the explicit statement see \cite{Lu10}. We will see how this works in our setting, but before, we need some definitions.

\begin{defn}[Curved $\Linf$-algebra over $A$]
Let $A$ be a commutative differential graded algebra (cdga) with a nilpotent dg ideal $I$ and $A^{\#}$ be the underlying graded algebra, with zero differential. A \textit{curved $\Linf$-algebra over \textup{$A$}} is a finitely generated projective $A^{\#}$-module $V$ together with a derivation of cohomological degree 1:
\begin{equation}
    d:\reallywidehat{\Sym}^\bullet(V[1]^\vee)\rightarrow\reallywidehat{\Sym}^\bullet(V[1]^\vee)
\end{equation}
such that:
\begin{itemize}
    \item $(\reallywidehat{\Sym}^\bullet(V[1]^\vee),d)$ is a cdga over $A$.
    \item $d$ preserves the ideal $\Sym^{>0}(V[1]^\vee)$ modulo the nilpotent ideal $I$.
\end{itemize}
\end{defn}

\begin{rmk}
If we take the Taylor components of $d$, we obtain the maps
\begin{equation}
    d_n: V^\vee\rightarrow\bigwedge\nolimits^n(V[n-2]^\vee),\quad n\geq 0
\end{equation}
which, upon dualization, become the $\Linf$-brackets
\begin{equation}
    \ell_n: \bigwedge\nolimits^n(V[n-2])\rightarrow V
\end{equation}
\end{rmk}

\begin{notat}
The completed symmetric algebra over $A^{\#}$ of the 1-shifted dual of $V$, i.e.  $\reallywidehat{\Sym}^\bullet(V[1]^\vee)$, is known as the \textit{Chevalley--Eilenberg complex} of $V$, and denoted in the following by $C^{\bullet}(V)$.
\end{notat}

\begin{defn}[Morphism of curved $\Linf$-algebras]
A \textit{morphism of curved $\Linf$-algebras} $\phi:V\rightarrow W$ over $A$ is a map $\pi^*:C^\bullet(W)\rightarrow C^\bullet(V)$ of cdga's over $A$ which respects the filtration by the ideal $I$.
\end{defn}

\begin{defn}[Maurer--Cartan element, Maurer--Cartan equation]
Let $\mathfrak{g}$ be the curved $\Linf$-algebra and $\alpha$ a degree 1 element of $\mathfrak{g}$. The \textit{Maurer--Cartan element} associated to $\alpha$ is
\begin{equation}
\label{MC}
    MC(\alpha)=\sum^{\infty}_{n=0}\frac{1}{n!}\ell_n(\alpha^{\otimes n}).
\end{equation}
The \textit{Maurer--Cartan equation} for $\alpha$ is $MC(\alpha)=0$.
\end{defn}

\begin{rmk}
To render \eqref{MC} well-defined, since it involves an infinite sum, we will only consider Maurer--Cartan elements in \textit{nilpotent} $\Linf$-algebras.
\end{rmk}

As we have preannounced before, formal moduli problems are described by $\Linf$-algebras. Explicitly, let $\mathfrak{g}$ be an $\Linf$-algebra and $R$ be an Artinian dg algebra over $\mathbb{R}$, with maximal ideal $\mathfrak{m}$. Then, we can associate a formal derived moduli problem to $\mathfrak{g}$ by sending $R$ to the simplicial set
\begin{equation}
    MC(\mathfrak{g}\otimes_{\mathbb{R}}\mathfrak{m})
\end{equation}
of solution of the Maurer--Cartan equation of the nilpotent $\Linf$-algebra $\mathfrak{g}\otimes_{\mathbb{R}}\mathfrak{m}$. More precisely, the Maurer--Cartan functor $MC_{\mathfrak{g}}$ is a formal derived moduli problem (see \ref{formal_mod_prob}). Moreover, if we send an $\Linf$-algebra to its associated Maurer--Cartan functor, we obtain an equivalence of categories \cite{Lu10}.

However, we are interested in $\Linf$-algebras parametrized by smooth manifolds and derived stacks, the global counterparts to the notion of formal moduli problems, which has a local nature.

\begin{defn}[curved $\Linf$-algebra over $\Omega^\bullet(X)$, $\Linf$-space]
Let X be a smooth manifold.
\begin{enumerate}
    \item A \textit{curved $\Linf$-algebra over $\Omega^\bullet(X)$} consists of a $\mathbb{Z}$-graded topological\footnote{By topological, we mean that the fibers are topological vector spaces and the transition maps are continuous.} vector bundle $\pi: V\rightarrow X$ and the structure of an $\Linf$-algebra on its sheaf of smooth sections, denoted by $\mathfrak{g}$, where the base algebra is over $\Omega^\bullet(X)$ with nilpotent ideal $I=\Omega^{\geq 1}(X)$.
    \item An $\Linf$-space is a pair $(X,\mathfrak{g})$, where $\mathfrak{g}$ is a curved $\Linf$-algebra over $\Omega^\bullet(X)$.
\end{enumerate}
\end{defn}

Now, we will explain how every $\Linf$-space defines a derived stack. It works in the same manner as $\Linf$-algebras determine formal moduli problems. Let $B\mathfrak{g}\coloneqq(X,\mathfrak{g})$ denote an $\Linf$-space. For a smooth map $f:Y\rightarrow X$, a curved $\Linf$-algebra is formed over $\Omega^\bullet(Y)$ as
\begin{equation}
    f^*\mathfrak{g}\coloneqq f^{-1}\mathfrak{g}\otimes_{f^{-1}\Omega^\bullet(X)}\Omega^\bullet(Y).
\end{equation}

\begin{defn}[$B\mathfrak{g}$ functor of points, \cite{Co11b,GG14}]
Let $\Delta^n$ be the standard $n$-simplex in $\mathbb{R}^n.$ For $B\mathfrak{g}$ an $\Linf$-space, its \textit{functor of points} is the functor
\begin{equation}
    MC_{B\mathfrak{g}}:\mathbf{DMan}^{op}\rightarrow\mathbf{sSets},
\end{equation}
which sends a derived manifold $\mathcal{M}=(M,\A)$ to the simplicial set $ MC_{B\mathfrak{g}}(\mathcal{M})$ whose $n$-simplices are given by pairs $(f,\alpha)$ where $f:M\rightarrow X$ is a smooth map and $\alpha$ a solution of the Maurer--Cartan equation in the nilpotent curved $\Linf$-algebra $f^*\mathfrak{g}\otimes_{\Omega^\bullet(M)}I_{\mathcal{M}}\otimes_{\mathbb{R}}\Omega^\bullet(\Delta^n)$, where $I_{\mathcal{M}}$ is the nilpotent ideal. In other words,
\begin{equation}
\label{mc_func_points}
    MC_{B\mathfrak{g}}(\mathcal{M})=\bigsqcup_{f\in \mathcal{C}^\infty(M,X)}MC(f^*\mathfrak{g}\otimes_{\Omega^\bullet(M)}I_{\mathcal{M}}).
\end{equation}
\end{defn}

\begin{thm}[\cite{GG14}]
For any $\Linf$-space, its functor of points is a derived stack.
\end{thm}

\begin{rmk}
Since in \eqref{mc_func_points} we are tensoring with the nilpotent ideal $I_{\mathcal{M}}$ instead of the whole algebra, the $\Linf$-algebra is nilpotent. This reflects the idea that the deformation functor $ MC_{B\mathfrak{g}}$ should deform only the nilpotent directions of a derived manifold. 
\end{rmk}

When $X$ is a complex manifold, the following theorem provides us with all the needed properties for its $\Linf$-space. 

\begin{thm}[\cite{Co11b}]
Let $X$ be a complex manifold and let $\Omega^{\#}(X)$ be the de Rham complex endowed with the zero differential. There exists an $\Linf$-space $X_{\Bar{\partial}}=(X,\mathfrak{g}_X)$ with the following properties:
\begin{enumerate}
    \item $\mathfrak{g}_X=\Omega^{\#}(X)\otimes_{\mathcal{C}^{\infty}(X)}T^{1,0}X[-1]$ as a $\Omega^{\#}(X)$-module.
    \item $C^\bullet(\mathfrak{g}_X)\cong \Omega^\bullet(X)\otimes_{\mathcal{C}^{\infty}(X)}\Jet^\mathrm{hol}_X$ as a $\Omega^{\bullet}(X)$-algebra.
    \item The jet prolongation map
    \begin{equation}
        \mathcal{C}^{\infty}(X)\hookrightarrow\Omega^\bullet(X)\otimes_{\mathcal{C}^{\infty}(X)}\Jet^\mathrm{hol}_X\cong C^\bullet(\mathfrak{g}_X)
    \end{equation}
    is a quasi-isomorphism of complexes of sheaves.
\end{enumerate}
\end{thm}

\subsection{Derived mapping spaces}
\label{der_map_spaces}
For an $\Linf$-space, we can think of its functor of points $MC_{B\mathfrak{g}}$ as the derived stack of maps into $B\mathfrak{g}$. With this idea in mind, in this section, we will see that if $(M,\A)$ is a derived manifold, a subset of the space of maps $(M,\A)\rightarrow(X,\mathfrak{g})$ is itself represented by an $\Linf$-space. 
Hence, let us define a new \textit{simplicial presheaf} (see \cite{GG14}) on the site of $\mathbf{DMan}$ given by
\begin{alignat}{1}
MC^{\mathcal{M}}_{B\mathfrak{g}}: \mathbf{DMan}^{op}&\rightarrow\mathbf{sSets},\\
\mathcal{N}&\mapsto MC_{B\mathfrak{g}}(\mathcal{M}\times \mathcal{N}).
\end{alignat}
This functor is again a derived stack since $MC_{B\mathfrak{g}}$ is a derived stack. In particular, this is the derived stack of maps from $\mathcal{M}$ to $B\mathfrak{g}$. In perturbation theory, when we have a space of fields given by the space of maps $\mathcal{M}\rightarrow B\mathfrak{g}$, we perturb around a subset of these maps, usually the constant maps. Here we do the same. We will consider the sub-simplicial presheaf
\begin{equation}
    \reallywidehat{MC}^{\mathcal{M}}_{B\mathfrak{g}}(\mathcal{N})\subset MC^{\mathcal{M}}_{B\mathfrak{g}}(\mathcal{N})
\end{equation}
of Maurer--Cartan solutions in which the underlying smooth map $M\rightarrow X$ is constant. More precisely, for an auxiliary derived manifold $(\mathcal{N}, \mathcal{B})$, consider
\begin{equation}
    \reallywidehat{MC}^{\mathcal{M}}_{B\mathfrak{g}}(\mathcal{N})\subset  MC(\mathcal{N}\times \mathcal{M})
\end{equation}
which consists of Maurer--Cartan elements $(f, \alpha)$ such that the underlying smooth map $f:N\times M\rightarrow X$ factors through the projection onto $M$. Costello showed that this space is itself an $\Linf$-space under certain conditions. This is useful for us since we would like $ \reallywidehat{MC}^{\mathcal{M}}_{B\mathfrak{g}}$ to represent the space of fields of a classical field theory.

\begin{prp}[\cite{Co11a,Co11b}]
Let $\mathcal{M}=(M,\A)$ be a derived manifold with nilpotent sheaf of ideals $I$ and with the property that, if $\A$ is filtered by the powers of the nilpotent ideal, then the
cohomology $\Gr^i\A$ f ( M) for $i\geq 1$ is concentrated in degrees $\geq 1$. \\
Let $(X, \mathfrak{g})$ be an $\Linf$-space such that the cohomology of the sheaf of $\Linf$-algebras $\mathfrak{g}_{\mathrm{red}}\coloneqq\mathfrak{g}/\Omega^{\geq 1}(X)$.\\
Then, the restricted Maurer--Cartan functor $\reallywidehat{\text{MC}}^{\mathcal{M}}_{B\mathfrak{g}}$ is weakly equivalent to the functor of points for the $\Linf$-space $(X, \mathfrak{g}\otimes \mathcal{A}(M))$.
\end{prp}

\begin{notat}
From now on, we will write $\reallywidehat{\Maps}(\mathcal{M}, B\mathfrak{g})$ for the $\Linf$-space $(X,\mathcal{A}(M)\otimes \mathfrak{g})$.
\end{notat}

\subsection{Shifted symplectic structures}
In \cite{Sch93}, Schwarz gave a definition of shifted symplectic structure on a dg manifold. Since in dg manifolds, all spaces of tensors are cochain complexes, the space of $i$-forms $\Omega^i(M)$ on a dg manifold is a cochain complex with a differential called the \textit{internal} differential. Moreover, we have also the de Rham differential $d_\mathrm{dR}:\Omega^{i}(M)\rightarrow\Omega^{i+1}(M)$. According to Schwarz a symplectic form is an element of the complex of 2-forms that is both, closed with respect to the de Rham differential and the internal differential. This element is also required to be non-degenerate. If the internal degree\footnote{The internal degree is also called ``ghost" degree.} of the form is $k$, it is said to be $k$-shifted. The physical relevant degrees are $-1$ for the space of bulk fields (BV formalism) and 0 for the boundary phase space (BFV formalism). 

In \cite{PTTV13}, Pantev et al. gave a definition of shifted symplectic structure in the context of derived geometry using the language of derived Artin stacks. Here, a closed $2$-form is a cocycle in the truncated de Rham complex
\begin{equation}
    \Omega^2_{cl}(X)\coloneqq \bigg(\bigoplus_{k\geq 2} \Omega^k(X)[-k+2], d_\mathrm{dR} \bigg)
\end{equation}
shifted in such a way that $\Omega^2(X)$ is in degree zero. 
\begin{rmk}
Note that a closed $2$-form is given by a sequence of forms $(\omega_2,\dots, \omega_k,\dots)$, with $\omega_k$ a form of degree $k$ and finitely many nonzero forms, such that $d_\mathrm{dR}\omega_k=\pm d_{\mathrm{int}}\omega_{k+1}$ for $k\geq 2$, where $d_\mathrm{int}$ is the internal differential. Therefore, to say that a 2-form is closed we need to specify more data than just a 2-form. Hence, being closed is a datum, it is not anymore a property like in the smooth case. 
\end{rmk}
In particular, a 2-form is symplectic when it is non-degenerate in a suitable sense (see \cite{PTTV13}). 

In \cite{CG16}, it is shown that a symplectic form of degree $k$ in the sense of Schwarz is the same as a degree $k-2$ non-degenerate invariant symmetric pairing on $\mathfrak{g}$. Moreover, we have the following lemma that closes the circles between all these apparently different notions of a symplectic form.
\begin{lmm}[\cite{CG16}]
Let $\mathfrak{g}$ be a finite dimensional $\Linf$-algebra. A $k$-shifted symplectic structure in the sense of \cite{PTTV13} on $B\mathfrak{g}$ is the same as a degree $k-2$ non-degenerate invariant symmetric pairing on $\mathfrak{g}$. 
\end{lmm}
Hence, a $k$-shifted symplectic structure on an $\Linf$-space $(X,\mathfrak{g})$ can be defined to be such a pairing on $\mathfrak{g}$. 

\begin{exm}
    Consider a complex manifold $X$ of dimension $2n$, by endowing $X$ with a holomorphic symplectic form (a non-degenerate 2-form on $T^{1,0}X$ which is closed under $d_X=\partial+\Bar{\partial}$, with $\partial$ the holomorphic differential and $\Bar{\partial}$ the antiholomorphic differential), the $\Linf$-space $X_{\Bar{\partial}}$ associated to $X$ becomes 0-shifted symplectic.
\end{exm}

\section{Examples of Feynman graphs for the BFV boundary operator in the $\mathbb{B}$-representation}
\label{app:feyn}
Here, we present the graphs appearing in the BFV boundary operator in the $\mathbb{B}$-representation up to three bulk vertices (black or red) and up to the Feynman rules in Table \ref{class:Tab_coeff_split_fr}. Let us consider $\boldsymbol{\Omega}^{\mathbb{B}}_3=\boldsymbol{\Omega}^{\mathbb{B}}_{3,0}+\boldsymbol{\Omega}^{\mathbb{B}}_{2,1}+\boldsymbol{\Omega}^{\mathbb{B}}_{1,2}+\boldsymbol{\Omega}^{\mathbb{B}}_{0,3}$. We present
\begin{itemize}
    \item the graphs appearing in $\boldsymbol{\Omega}^{\mathbb{B}}_{3,0}$ in Figs. \ref{fig:omega30_1}, \ref{fig:omega30_2} and \ref{fig:omega30_3};
   \item the graphs appearing in $\boldsymbol{\Omega}^{\mathbb{B}}_{2,1}$ in Fig. \ref{fig:omega21};
      \item the graphs appearing in $\boldsymbol{\Omega}^{\mathbb{B}}_{1,2}$ in Fig. \ref{fig:omega12};
         \item the graphs appearing in $\boldsymbol{\Omega}^{\mathbb{B}}_{0,3}$ in Fig. \ref{fig:omega03}.
\end{itemize}
We note that all the boundaries in the figures are assumed to be $\partial_2\Sigma_3$.

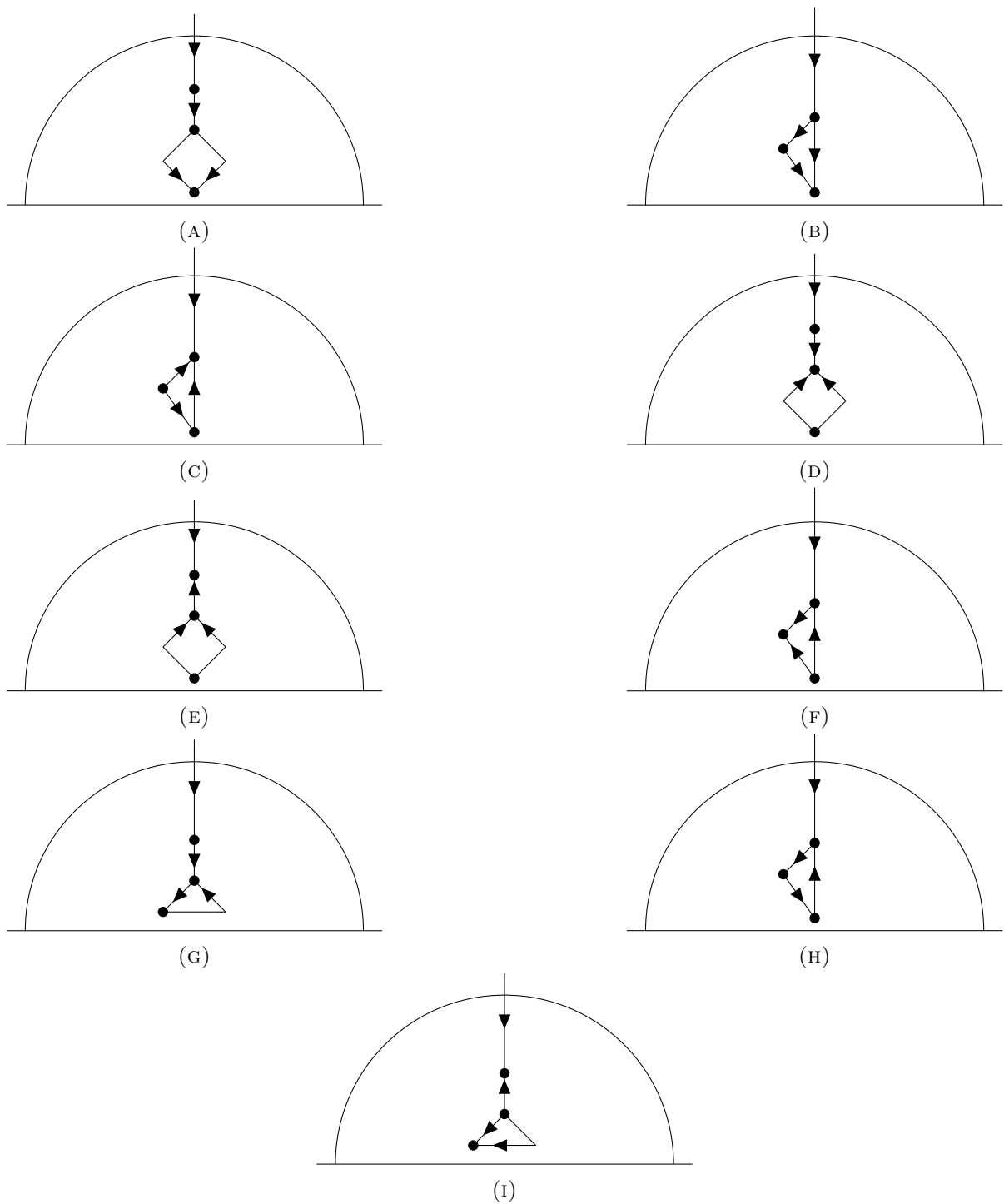
\begin{figure}[h!]
    \centering
     \begin{subfigure}[b]{0.4\textwidth}
     \centering
    \begin{tikzpicture}
  \begin{feynman}[every blob={/tikz/fill=white!30,/tikz/inner sep=0.5pt}]
  \node[dot] (d) at (0, 0.65);
  \vertex (f) at (0, 1.85);
  \vertex (a) at (-3, -1.2);
  \vertex (b) at (3, -1.2);
  \node[dot] (z) at (0, -1);
  \vertex (x) at (-0.5, -0.5);
  \vertex (y) at (0.5, -0.5);
  \diagram*{
  (f) -- [fermion] (d),
  (d) -- [fermion] m [dot],
  (a) -- (b),
   m -- (x),
   m -- (y),
   (x) -- [fermion] (z),
   (y) -- [fermion] (z),
  }; 
  \end{feynman}
  \draw (-2.7,-1.2) arc (180:0:2.7);
  \end{tikzpicture}
   \caption{}
  \label{}
    \end{subfigure}%
    \hfill
    \begin{subfigure}[b]{0.4\textwidth}
     \centering
      \begin{tikzpicture}
  \begin{feynman}[every blob={/tikz/fill=white!30,/tikz/inner sep=0.5pt}]
  \vertex (d) at (0, 1.75);
  \vertex (a) at (-3, -1.4);
  \vertex (b) at (3, -1.4);
  \node[dot] (z) at (0, -1.2);
  \node[dot] (x) at (-0.5, -0.5);
  \diagram*{
  (d) -- [fermion] m [dot],
  (a) -- (b),
   m -- [fermion] (x),
   (x) -- [fermion] (z),
   m -- [fermion] (z)
  }; 
  \end{feynman}
  \draw (-2.7,-1.4) arc (180:0:2.7);
  \end{tikzpicture}
  \caption{}
  \label{}
     \end{subfigure}%
     \hfill
     \begin{subfigure}[b]{0.4\textwidth}
     \centering
       \begin{tikzpicture}
  \begin{feynman}[every blob={/tikz/fill=white!30,/tikz/inner sep=0.5pt}]
  \vertex (d) at (0, 1.75);
  \vertex (a) at (-3, -1.4);
  \vertex (b) at (3, -1.4);
  \node[dot] (z) at (0, -1.2);
  \node[dot] (x) at (-0.5, -0.5);
  \diagram*{
  (d) -- [fermion] m [dot],
  (a) -- (b),
   m -- [anti fermion] (x),
   (x) -- [ fermion] (z),
   m -- [anti fermion] (z),
  }; 
  \end{feynman}
  \draw (-2.7,-1.4) arc (180:0:2.7);
  \end{tikzpicture}
     \caption{}
  \label{}
     \end{subfigure}%
     \hfill
     \begin{subfigure}[b]{0.4\textwidth}
     \centering
   \begin{tikzpicture}
  \begin{feynman}[every blob={/tikz/fill=white!30,/tikz/inner sep=0.5pt}]
  \node[dot] (d) at (0, 0.65);
  \vertex (f) at (0, 1.85);
  \vertex (a) at (-3, -1.2);
  \vertex (b) at (3, -1.2);
  \node[dot] (z) at (0, -1);
  \vertex (x) at (-0.5, -0.5);
  \vertex (y) at (0.5, -0.5);
  \diagram*{
  (f) -- [fermion] (d),
  (d) -- [fermion] m [dot],
  (a) -- (b),
   m -- [anti fermion] (x),
   m -- [anti fermion] (y),
   (x) --  (z),
   (y) --  (z),
  }; 
  \end{feynman}
  \draw (-2.7,-1.2) arc (180:0:2.7);
  \end{tikzpicture}
       \caption{}
  \label{}
     \end{subfigure}%
     \hfill
     \begin{subfigure}[b]{0.4\textwidth}
     \centering
   \begin{tikzpicture}
  \begin{feynman}[every blob={/tikz/fill=white!30,/tikz/inner sep=0.5pt}]
  \node[dot] (d) at (0, 0.65);
  \vertex (f) at (0, 1.85);
  \vertex (a) at (-3, -1.2);
  \vertex (b) at (3, -1.2);
  \node[dot] (z) at (0, -1);
  \vertex (x) at (-0.5, -0.5);
  \vertex (y) at (0.5, -0.5);
  \diagram*{
  (f) -- [fermion] (d),
  (d) -- [anti fermion] m [dot],
  (a) -- (b),
   m -- [anti fermion] (x),
   m -- [anti fermion] (y),
   (x) -- (z),
   (y) -- (z),
  }; 
  \end{feynman}
  \draw (-2.7,-1.2) arc (180:0:2.7);
  \end{tikzpicture}
            \caption{}
  \label{}
     \end{subfigure}%
     \hfill
     \begin{subfigure}[b]{0.4\textwidth}
     \centering
     \begin{tikzpicture}
  \begin{feynman}[every blob={/tikz/fill=white!30,/tikz/inner sep=0.5pt}]
  \vertex (d) at (0, 1.85);
  \vertex (a) at (-3, -1.4);
  \vertex (b) at (3, -1.4);
  \node[dot] (z) at (0, -1.2);
  \node[dot] (x) at (-0.5, -0.5);
  \diagram*{
  (d) -- [fermion] m [dot],
  (a) -- (b),
   m -- [ fermion] (x),
   (x) -- [ anti fermion] (z),
   m -- [ anti fermion] (z),
  }; 
  \end{feynman}
  \draw (-2.7,-1.4) arc (180:0:2.7);
  \end{tikzpicture}
           \caption{}
  \label{}
     \end{subfigure}%
     \hfill
     \begin{subfigure}[b]{0.4\textwidth}
     \centering
         \begin{tikzpicture}
  \begin{feynman}[every blob={/tikz/fill=white!30,/tikz/inner sep=0.5pt}]
  \node[dot] (d) at (0, 0.65);
  \vertex (f) at (0, 2.25);
  \vertex (a) at (-3, -0.8);
  \vertex (b) at (3, -0.8);
  \node[dot] (x) at (-0.5, -0.5);
  \vertex (y) at (0.5, -0.5);
  \diagram*{
  (f) -- [fermion] (d),
  (d) -- [fermion] m [dot],
  (a) -- (b),
   m -- [ fermion] (x),
   m -- [anti fermion] (y),
   (x) -- (y),
  }; 
  \end{feynman}
  \draw (-2.7,-0.8) arc (180:0:2.7);
  \end{tikzpicture}
     \caption{}
  \label{}
     \end{subfigure}%
     \hfill
     \begin{subfigure}[b]{0.4\textwidth}
     \centering
       \begin{tikzpicture}
  \begin{feynman}[every blob={/tikz/fill=white!30,/tikz/inner sep=0.5pt}]
  \vertex (d) at (0, 1.75);
  \vertex (a) at (-3, -1.4);
  \vertex (b) at (3, -1.4);
  \node[dot] (z) at (0, -1.2);
  \node[dot] (x) at (-0.5, -0.5);
  \diagram*{
  (d) -- [fermion] m [dot],
  (a) -- (b),
   m -- [ fermion] (x),
   (x) -- [ fermion] (z),
   m -- [ anti fermion] (z),
  }; 
  \end{feynman}
  \draw (-2.7,-1.4) arc (180:0:2.7);
  \end{tikzpicture}
\caption{}
  \label{}
     \end{subfigure}%
     \hfill
     \begin{subfigure}[b]{0.4\textwidth}
     \centering
    \begin{tikzpicture}
  \begin{feynman}[every blob={/tikz/fill=white!30,/tikz/inner sep=0.5pt}]
  \node[dot] (d) at (0, 0.65);
  \vertex (f) at (0, 2.25);
  \vertex (a) at (-3, -0.8);
  \vertex (b) at (3, -0.8);
  \node[dot] (x) at (-0.5, -0.5);
  \vertex (y) at (0.5, -0.5);
  \diagram*{
  (f) -- [fermion] (d),
  (d) -- [anti fermion] m [dot],
  (a) -- (b),
   m -- [ fermion] (x),
   m --  (y),
   (x) -- [ anti fermion] (y),
  }; 
  \end{feynman}
  \draw (-2.7,-0.8) arc (180:0:2.7);
  \end{tikzpicture}
     \caption{}
  \label{}
     \end{subfigure}%
    \caption{First part of the graphs appearing in $\boldsymbol{\Omega}^{\mathbb{B}}_{3,0}$. However, most of them do not give any contribution since by Kontsevich's vanishing lemma \cite{Ko03} all the graphs with vertices where exactly one arrow ends and starts vanish and also the graphs with double edges, i.e. when two edges connecting the same two vertices, vanish. This can be seen by using Kontsevich's angle form propagator on $\mathbb{H}^3$.}
    \label{fig:omega30_1}
\end{figure}
  

   \begin{figure}[h]
      \centering
      \begin{subfigure}[b]{0.325\textwidth}
      \centering
 \begin{tikzpicture}
  \begin{feynman}[]
  \node[dot] (d) at (0, 0.5);
  \vertex (e) at (0, 1.2);
  \vertex (a) at (-2.75, -1.4);
  \vertex (b) at (2.75, -1.4);
  \node[dot] (z) at (0, -1.2);
  \vertex (x) at (-0.5,0);
  \vertex (y) at (0.5,0);
  \node[dot] (q) at (0, -0.5);
  \diagram*{
  (x) --  (d),
  (y) --  (d),
  (a) -- (b),
  (e) --[fermion] (d),
   (q) --[anti fermion] (x),
   (q) --[anti fermion] (y),
   (q) -- [anti fermion] (z),
  }; 
  \end{feynman}
  \draw (-2.15,-1.4) arc (180:0:2.15);
  \end{tikzpicture}
       \caption{}
  \label{}
     \end{subfigure}%
     \hfill
     \begin{subfigure}[b]{0.325\textwidth}
     \centering
     \begin{tikzpicture}
  \begin{feynman}[]
  \node[dot] (d) at (0, 0.5);
  \vertex (e) at (0, 1.2);
  \vertex (a) at (-2.75, -1.4);
  \vertex (b) at (2.75, -1.4);
  \node[dot] (z) at (0, -1.2);
  \vertex (x) at (-0.5,0);
  \vertex (y) at (0.5,0);
  \node[dot] (q) at (0, -0.5);
  \diagram*{
  (x) --  (d),
  (y) --  (d),
  (a) -- (b),
  (e) --[fermion] (d),
   (q) --[anti fermion] (x),
   (q) --[anti fermion] (y),
   (q) -- [ fermion] (z),
  }; 
  \end{feynman}
  \draw (-2.15,-1.4) arc (180:0:2.15);
  \end{tikzpicture}
     \caption{}
  \label{}
     \end{subfigure}%
     \hfill
     \begin{subfigure}[b]{0.325\textwidth}
     \centering
     \begin{tikzpicture}
  \begin{feynman}[]
  \node[dot] (d) at (0, 0.5);
  \vertex (e) at (0, 1.2);
  \vertex (a) at (-2.75, -1.4);
  \vertex (b) at (2.75, -1.4);
  \node[dot] (z) at (0, -1.2);
  \vertex (x) at (-0.5,0);
  \vertex (y) at (0.5,0);
  \node[dot] (q) at (0, -0.5);
  \diagram*{
  (x) --  (d),
  (y) --  (d),
  (a) -- (b),
  (e) --[fermion] (d),
   (q) --[anti fermion] (x),
   (q) --[ fermion] (y),
   (q) -- [ fermion] (z),
  }; 
  \end{feynman}
  \draw (-2.15,-1.4) arc (180:0:2.15);
  \end{tikzpicture}
      \caption{}
  \label{}
     \end{subfigure}%
     \hfill
     \begin{subfigure}[b]{0.325\textwidth}
     \centering
         \begin{tikzpicture}
  \begin{feynman}[]
  \node[dot] (d) at (0, 0.5);
  \vertex (e) at (0, 1.2);
  \vertex (a) at (-2.75, -1.4);
  \vertex (b) at (2.75, -1.4);
  \node[dot] (z) at (0, -1.2);
  \vertex (x) at (-0.5,0);
  \vertex (y) at (0.5,0);
  \node[dot] (q) at (0, -0.5);
  \diagram*{
  (x) --  (d),
  (y) --  (d),
  (a) -- (b),
  (e) --[fermion] (d),
   (q) --[fermion] (x),
   (q) --[anti fermion] (y),
   (q) -- [anti fermion] (z),
  }; 
  \end{feynman}
  \draw (-2.15,-1.4) arc (180:0:2.15);
  \end{tikzpicture}
       \caption{}
  \label{}
     \end{subfigure}%
     \hspace{1cm}
     \begin{subfigure}[b]{0.325\textwidth}
     \centering
        \begin{tikzpicture}
  \begin{feynman}[]
  \node[dot] (d) at (0, 0.5);
  \vertex (e) at (0, 1.2);
  \vertex (a) at (-2.75, -1.4);
  \vertex (b) at (2.75, -1.4);
  \node[dot] (z) at (0, -1.2);
  \vertex (x) at (-0.5,0);
  \vertex (y) at (0.5,0);
  \node[dot] (q) at (0, -0.5);
  \diagram*{
  (x) --  (d),
  (y) --  (d),
  (a) -- (b),
  (e) --[fermion] (d),
   (q) --[fermion] (x),
   (q) --[fermion] (y),
   (q) -- [fermion] (z),
  }; 
  \end{feynman}
  \draw (-2.15,-1.4) arc (180:0:2.15);
  \end{tikzpicture}
     \caption{}
  \label{}
     \end{subfigure}%
     \caption{Second part of the graphs appearing in $\boldsymbol{\Omega}^{\mathbb{B}}_{3,0}$. All of them do not give any contribution since again by Kontsevich's vanishing lemma \cite{Ko03} all the graphs with double edges vanish.}
      \label{fig:omega30_2}
  \end{figure}
 \vfill
     \begin{figure}[h]
      \centering
      \begin{subfigure}[b]{0.325\textwidth}
      \centering
 \begin{tikzpicture}
  \begin{feynman}[]
  \node[dot,red] (d) at (0, 0.5);
  \vertex (e) at (0, 1.2);
  \vertex (a) at (-2.75, -1.4);
  \vertex (b) at (2.75, -1.4);
  \node[dot] (z) at (0, -1.2);
  \vertex (x) at (-0.5,0);
  \vertex (y) at (0.5,0);
  \node[dot,red] (q) at (0, -0.5);
  \diagram*{
  (x) --  (d),
  (y) --  (d),
  (a) -- (b),
  (e) --[fermion] (d),
   (q) --[anti fermion] (x),
   (q) --[anti fermion] (y),
   (q) -- [anti fermion] (z),
  }; 
  \end{feynman}
  \draw (-2.15,-1.4) arc (180:0:2.15);
  \end{tikzpicture}
       \caption{}
  \label{}
     \end{subfigure}%
     \hfill
     \begin{subfigure}[b]{0.325\textwidth}
     \centering
     \begin{tikzpicture}
  \begin{feynman}[]
  \node[dot,red] (d) at (0, 0.5);
  \vertex (e) at (0, 1.2);
  \vertex (a) at (-2.75, -1.4);
  \vertex (b) at (2.75, -1.4);
  \node[dot] (z) at (0, -1.2);
  \vertex (x) at (-0.5,0);
  \vertex (y) at (0.5,0);
  \node[dot,red] (q) at (0, -0.5);
  \diagram*{
  (x) --  (d),
  (y) --  (d),
  (a) -- (b),
  (e) --[fermion] (d),
   (q) --[anti fermion] (x),
   (q) --[anti fermion] (y),
   (q) -- [ fermion] (z),
  }; 
  \end{feynman}
  \draw (-2.15,-1.4) arc (180:0:2.15);
  \end{tikzpicture}
     \caption{}
  \label{}
     \end{subfigure}%
     \hfill
     \begin{subfigure}[b]{0.325\textwidth}
     \centering
     \begin{tikzpicture}
  \begin{feynman}[]
  \node[dot,red] (d) at (0, 0.5);
  \vertex (e) at (0, 1.2);
  \vertex (a) at (-2.75, -1.4);
  \vertex (b) at (2.75, -1.4);
  \node[dot] (z) at (0, -1.2);
  \vertex (x) at (-0.5,0);
  \vertex (y) at (0.5,0);
  \node[dot,red] (q) at (0, -0.5);
  \diagram*{
  (x) --  (d),
  (y) --  (d),
  (a) -- (b),
  (e) --[fermion] (d),
   (q) --[anti fermion] (x),
   (q) --[ fermion] (y),
   (q) -- [ fermion] (z),
  }; 
  \end{feynman}
  \draw (-2.15,-1.4) arc (180:0:2.15);
  \end{tikzpicture}
      \caption{}
  \label{}
     \end{subfigure}%
     \hfill
     \begin{subfigure}[b]{0.325\textwidth}
     \centering
         \begin{tikzpicture}
  \begin{feynman}[]
  \node[dot,red] (d) at (0, 0.5);
  \vertex (e) at (0, 1.2);
  \vertex (a) at (-2.75, -1.4);
  \vertex (b) at (2.75, -1.4);
  \node[dot] (z) at (0, -1.2);
  \vertex (x) at (-0.5,0);
  \vertex (y) at (0.5,0);
  \node[dot,red] (q) at (0, -0.5);
  \diagram*{
  (x) --  (d),
  (y) --  (d),
  (a) -- (b),
  (e) --[fermion] (d),
   (q) --[fermion] (x),
   (q) --[anti fermion] (y),
   (q) -- [anti fermion] (z),
  }; 
  \end{feynman}
  \draw (-2.15,-1.4) arc (180:0:2.15);
  \end{tikzpicture}
       \caption{}
  \label{}
     \end{subfigure}%
     \hspace{1cm}
     \begin{subfigure}[b]{0.325\textwidth}
     \centering
        \begin{tikzpicture}
  \begin{feynman}[]
  \node[dot,red] (d) at (0, 0.5);
  \vertex (e) at (0, 1.2);
  \vertex (a) at (-2.75, -1.4);
  \vertex (b) at (2.75, -1.4);
  \node[dot] (z) at (0, -1.2);
  \vertex (x) at (-0.5,0);
  \vertex (y) at (0.5,0);
  \node[dot,red] (q) at (0, -0.5);
  \diagram*{
  (x) --  (d),
  (y) --  (d),
  (a) -- (b),
  (e) --[fermion] (d),
   (q) --[fermion] (x),
   (q) --[fermion] (y),
   (q) -- [fermion] (z),
  }; 
  \end{feynman}
  \draw (-2.15,-1.4) arc (180:0:2.15);
  \end{tikzpicture}
     \caption{}
  \label{}
     \end{subfigure}%
     \caption{Graphs appearing in $\boldsymbol{\Omega}^{\mathbb{B}}_{1,2}$. All of them do not give any contribution since again by Kontsevich's vanishing lemma \cite{Ko03} all the graphs with double edges vanish.}
      \label{fig:omega12}
  \end{figure}


 \begin{figure}[h]
      \centering
      \begin{subfigure}[b]{0.4\textwidth}
  \begin{tikzpicture}
  \begin{feynman}[]
  \node[dot] (d) at (0, 1);
  \vertex (e) at (-0.5, 2);
  \vertex (f) at (0.5, 2);
  \vertex (a) at (-3, -1.4);
  \vertex (b) at (3, -1.4);
  \node[dot] (z) at (0, -0.7);
  \node[dot] (s) at (0, -1.4);
  \vertex (x) at (-0.5,-0.25);
  \vertex (q) at (0.5,-0.25);
  \node[dot] (y) at (0,0.25);
  \diagram*{
  (a) -- (b),
  (e) --[fermion] (d),
  (f) --[fermion] (d),
   (d) -- [anti fermion] (y),
   (y) --[anti fermion] (x),
   (y) --[anti fermion] (q),
   (z) -- (x),
   (z) -- (q),
   (z) --[fermion] (s)
  }; 
  \end{feynman}
  \draw (-2.7,-1.4) arc (180:0:2.7);
  \end{tikzpicture}
   \caption{}
  \label{}
     \end{subfigure}%
     \hfill
    \begin{subfigure}[b]{0.4\textwidth}
      \begin{tikzpicture}
  \begin{feynman}[]
  \node[dot] (d) at (0, 1);
  \vertex (e) at (-0.5, 2);
  \vertex (f) at (0.5, 2);
  \vertex (a) at (-3, -1.4);
  \vertex (b) at (3, -1.4);
  \node[dot] (z) at (0, -0.7);
  \node[dot] (s) at (0, -1.4);
  \vertex (x) at (-0.5,-0.25);
  \vertex (q) at (0.5,-0.25);
  \node[dot] (y) at (0,0.25);
  \diagram*{
  (a) -- (b),
  (e) --[fermion] (d),
  (f) --[fermion] (d),
   (d) -- [anti fermion] (y),
   (y) --[ fermion] (x),
   (y) --[ fermion] (q),
   (z) -- (x),
   (z) -- (q),
   (z) --[fermion] (s)
  }; 
  \end{feynman}
  \draw (-2.7,-1.4) arc (180:0:2.7);
  \end{tikzpicture}
        \caption{}
  \label{}
     \end{subfigure}%
     \hfill
   \begin{subfigure}[b]{0.4\textwidth}
      \begin{tikzpicture}
  \begin{feynman}[]
  \node[dot] (d) at (0, 1);
  \vertex (e) at (0, 2);
  \vertex (a) at (-3, -1.4);
  \vertex (b) at (3, -1.4);
  \node[dot] (z) at (0, -0.7);
  \node[dot] (s) at (0, -1.4);
  \vertex (x) at (0.5,0.5);
  \vertex (p) at (-0.5,0.5);
  \vertex (y) at (-2,1.75); 
  \diagram*{
  (a) -- (b),
  (e) --[fermion] (d),
  (p) -- (d),
  (x) -- (d),
  (p) --[fermion] m[dot],
  (x) --[fermion] m,
   m -- [anti fermion] (z),
   (y) -- [fermion] (z),
   (z) --[fermion] (s)
  }; 
  \end{feynman}
  \draw (-2.7,-1.4) arc (180:0:2.7);
  \end{tikzpicture}
       \caption{}
  \label{}
     \end{subfigure}%
      \hfill
   \begin{subfigure}[b]{0.4\textwidth}
   \begin{tikzpicture}
  \begin{feynman}[]
  \node[dot] (d) at (0, 1);
  \vertex (e) at (0, 2);
  \vertex (a) at (-3, -1.4);
  \vertex (b) at (3, -1.4);
  \node[dot] (z) at (0, -0.7);
  \node[dot] (s) at (0, -1.4);
  \vertex (x) at (0.5,0.5);
  \vertex (p) at (-0.5,0.5);
  \vertex (y) at (-2,1.75); 
  \diagram*{
  (a) -- (b),
  (e) --[fermion] (d),
  (p) -- (d),
  (x) -- (d),
  (p) --[anti fermion] m[dot],
  (x) --[anti fermion] m,
   m -- [anti fermion] (z),
   (y) -- [fermion] (z),
   (z) --[fermion] (s)
  }; 
  \end{feynman}
  \draw (-2.7,-1.4) arc (180:0:2.7);
  \end{tikzpicture}
        \caption{}
  \label{}
     \end{subfigure}%
      \hfill
     \begin{subfigure}[b]{0.4\textwidth}
       \begin{tikzpicture}
  \begin{feynman}[]
  \node[dot] (d) at (0, 0.25);
  \vertex (e) at (0, 1.75);
  \vertex (a) at (-3, -1.4);
  \vertex (b) at (3, -1.4);
  \node[dot] (z) at (0, -0.7);
  \node[dot] (s) at (0, -1.4);
  \node[dot] (x) at (-0.5,-0.25);
  \vertex (y) at (-1,1.75); 
  \diagram*{
  (a) -- (b),
  (e) --[fermion] (d),
  (x) -- (d),
  (x) --[anti fermion] (d),
  (x) --[anti fermion] (z),
   (d) -- [ fermion] (z),
   (y) -- [fermion] (x),
   (z) --[fermion] (s)
  }; 
  \end{feynman}
  \draw (-2.7,-1.4) arc (180:0:2.7);
  \end{tikzpicture}
          \caption{}
  \label{}
     \end{subfigure}%
     \hfill
       \begin{subfigure}[b]{0.4\textwidth}
        \begin{tikzpicture}
  \begin{feynman}[]
  \node[dot] (d) at (0, 0.25);
  \vertex (e) at (0, 1.75);
  \vertex (a) at (-3, -1.4);
  \vertex (b) at (3, -1.4);
  \node[dot] (z) at (0, -0.5);
  \node[dot] (s) at (0, -1.4);
  \node[dot] (x) at (-0.5,-0.25);
  \vertex (p) at (0.5,-0.9);
  \vertex (q) at (-0.5,-0.9);
  \vertex (y) at (-0.5,1.75);
  \vertex (n) at (-1,1.75);
  \diagram*{
  (a) -- (b),
  (e) --[fermion] (d),
  (x) -- (d),
  (x) --[anti fermion] (d),
   (d) -- [ fermion] (z),
   (y) -- [fermion] (x),
   (n) -- [fermion] (x),
   (p) --[fermion] (s),
   (p) -- (z),
   (q) --[fermion] (s),
   (q) -- (z)
  }; 
  \end{feynman}
  \draw (-2.7,-1.4) arc (180:0:2.7);
  \end{tikzpicture}
        \caption{}
  \label{}
     \end{subfigure}%
     \hfill
       \begin{subfigure}[b]{0.4\textwidth}
       \begin{tikzpicture}
  \begin{feynman}[]
  \node[dot] (d) at (0, 0.25);
  \vertex (e) at (0, 1.75);
  \vertex (a) at (-3, -1.4);
  \vertex (b) at (3, -1.4);
  \node[dot] (z) at (0, -0.5);
  \node[dot] (s) at (0, -1.4);
  \node[dot] (x) at (-0.5,-0.25);
  \vertex (p) at (0.75,-0.5);
  \vertex (y) at (-0.5,1.75);
  \vertex (n) at (-1,1.75);
  \diagram*{
  (a) -- (b),
  (e) --[fermion] (d),
  (x) --[anti fermion] (z),
   (d) -- [ fermion] (z),
   (y) -- [fermion] (x),
   (n) -- [fermion] (x),
   (z) --[fermion] (s),
   (d) -- (p),
   (p) --[fermion] (s)
  }; 
  \end{feynman}
  \draw (-2.7,-1.4) arc (180:0:2.7);
  \end{tikzpicture}
       \caption{}
  \label{}
     \end{subfigure}%
     \hfill
     \begin{subfigure}[b]{0.4\textwidth}
       \begin{tikzpicture}
  \begin{feynman}[]
  \node[dot] (d) at (0, 0.25);
  \vertex (e) at (-1, 1.75);
  \vertex (a) at (-3, -1.4);
  \vertex (b) at (3, -1.4);
  \node[dot] (z) at (0, -0.5);
  \node[dot] (s) at (0, -1.4);
  \node[dot] (x) at (0.75,-0.1);
  \vertex (p) at (0.75,-0.5);
  \vertex (y) at (+0.75,1.75);
  \vertex (n) at (+1.25,1.75);
  \diagram*{
  (a) -- (b),
  (e) --[fermion] (z),
  (x) -- (d),
  (x) --[anti fermion] (d),
   (d) -- [ anti fermion] (z),
   (y) -- [fermion] (x),
   (n) -- [fermion] (x),
   (z) --[fermion] (s),
   (p) --[fermion] (s),
   (p) -- (d)
  }; 
  \end{feynman}
  \draw (-2.7,-1.4) arc (180:0:2.7);
  \end{tikzpicture}
        \caption{}
  \label{}
     \end{subfigure}%
     \hfill
     \begin{subfigure}[b]{0.4\textwidth}
     \begin{tikzpicture}
  \begin{feynman}[]
  \node[dot] (d) at (0, 0.25);
  \vertex (e) at (-1, 1.75);
  \vertex (a) at (-3, -1.4);
  \vertex (b) at (3, -1.4);
  \node[dot] (z) at (0, -0.5);
  \node[dot] (s) at (0, -1.4);
  \node[dot] (x) at (0.75,-0.1);
  \vertex (y) at (0,1.75);
  \vertex (n) at (+1.25,1.75);
  \diagram*{
  (a) -- (b),
  (e) --[fermion] (z),
  (x) -- (d),
  (x) --[ fermion] (d),
   (d) -- [ anti fermion] (z),
   (y) -- [fermion] (d),
   (n) -- [fermion] (x),
   (z) --[fermion] (s),
   (x) --[fermion] (s)
  }; 
  \end{feynman}
  \draw (-2.7,-1.4) arc (180:0:2.7);
  \end{tikzpicture}
     \caption{}
  \label{}
     \end{subfigure}%
       \hfill
     \begin{subfigure}[b]{0.4\textwidth}
   \begin{tikzpicture}
  \begin{feynman}[]
  \node[dot] (d) at (0, 0.25);
  \vertex (e) at (0, 1.75);
  \vertex (a) at (-3, -1.4);
  \vertex (b) at (3, -1.4);
  \node[dot] (z) at (-0.5, -0.5);
  \node[dot] (s) at (0, -1.4);
  \node[dot] (x) at (0.75,-0.1);
  \vertex (y) at (+0.75,1.75);
  \vertex (n) at (+1.25,1.75);
  \vertex (q) at (-0.5,1.75);
  \vertex (r) at (-0.75, -0.9);
  \diagram*{
  (a) -- (b),
  (e) --[fermion] (d),
  (x) -- (d),
  (x) --[anti fermion] (d),
   (s) -- [anti fermion] (z),
   (y) -- [fermion] (x),
   (n) -- [fermion] (x),
   (r) --[fermion] (s),
   (z) -- (r),
   (d) --[fermion] (s),
   (q) --[fermion] (z)
  }; 
  \end{feynman}
  \draw (-2.7,-1.4) arc (180:0:2.7);
  \end{tikzpicture}
\caption{}
  \label{}
     \end{subfigure}%
     \caption{Third part of the graphs appearing in $\boldsymbol{\Omega}^{\mathbb{B}}_{3,0}$. Most of them do not give any contribution since again by Kontsevich's vanishing lemma \cite{Ko03} all the graphs with double edges vanish.}
      \label{fig:omega30_3}
  \end{figure}

\begin{figure}[h]
    \centering
     \begin{subfigure}[b]{0.4\textwidth}
     \centering
    \begin{tikzpicture}
  \begin{feynman}[every blob={/tikz/fill=white!30,/tikz/inner sep=0.5pt}]
  \node[dot] (d) at (0, 0.65);
  \vertex (f) at (0, 1.85);
  \vertex (a) at (-3, -1.2);
  \vertex (b) at (3, -1.2);
  \node[dot] (z) at (0, -1);
  \vertex (x) at (-0.5, -0.5);
  \vertex (y) at (0.5, -0.5);
  \diagram*{
  (f) -- [fermion] (d),
  (d) -- [fermion] m [dot, red],
  (a) -- (b),
   m -- (x),
   m -- (y),
   (x) -- [fermion] (z),
   (y) -- [fermion] (z),
  }; 
  \end{feynman}
  \draw (-2.7,-1.2) arc (180:0:2.7);
  \end{tikzpicture}
   \caption{}
  \label{}
    \end{subfigure}%
    \hfill
    \begin{subfigure}[b]{0.4\textwidth}
     \centering
      \begin{tikzpicture}
  \begin{feynman}[every blob={/tikz/fill=white!30,/tikz/inner sep=0.5pt}]
  \vertex (d) at (0, 1.75);
  \vertex (a) at (-3, -1.4);
  \vertex (b) at (3, -1.4);
  \node[dot] (z) at (0, -1.2);
  \node[dot] (x) at (-0.5, -0.5);
  \diagram*{
  (d) -- [fermion] m [dot, red],
  (a) -- (b),
   m -- [fermion] (x),
   (x) -- [fermion] (z),
   m -- [fermion] (z)
  }; 
  \end{feynman}
  \draw (-2.7,-1.4) arc (180:0:2.7);
  \end{tikzpicture}
  \caption{}
  \label{}
     \end{subfigure}%
     \hfill
     \begin{subfigure}[b]{0.4\textwidth}
     \centering
       \begin{tikzpicture}
  \begin{feynman}[every blob={/tikz/fill=white!30,/tikz/inner sep=0.5pt}]
  \vertex (d) at (0, 1.75);
  \vertex (a) at (-3, -1.4);
  \vertex (b) at (3, -1.4);
  \node[dot] (z) at (0, -1.2);
  \node[dot] (x) at (-0.5, -0.5);
  \diagram*{
  (d) -- [fermion] m [dot, red],
  (a) -- (b),
   m -- [anti fermion] (x),
   (x) -- [ fermion] (z),
   m -- [anti fermion] (z),
  }; 
  \end{feynman}
  \draw (-2.7,-1.4) arc (180:0:2.7);
  \end{tikzpicture}
     \caption{}
  \label{}
     \end{subfigure}%
     \hfill
     \begin{subfigure}[b]{0.4\textwidth}
     \centering
   \begin{tikzpicture}
  \begin{feynman}[every blob={/tikz/fill=white!30,/tikz/inner sep=0.5pt}]
  \node[dot] (d) at (0, 0.65);
  \vertex (f) at (0, 1.85);
  \vertex (a) at (-3, -1.2);
  \vertex (b) at (3, -1.2);
  \node[dot] (z) at (0, -1);
  \vertex (x) at (-0.5, -0.5);
  \vertex (y) at (0.5, -0.5);
  \diagram*{
  (f) -- [fermion] (d),
  (d) -- [fermion] m [dot, red],
  (a) -- (b),
   m -- [anti fermion] (x),
   m -- [anti fermion] (y),
   (x) --  (z),
   (y) --  (z),
  }; 
  \end{feynman}
  \draw (-2.7,-1.2) arc (180:0:2.7);
  \end{tikzpicture}
       \caption{}
  \label{}
     \end{subfigure}%
     \hfill
     \begin{subfigure}[b]{0.4\textwidth}
     \centering
   \begin{tikzpicture}
  \begin{feynman}[every blob={/tikz/fill=white!30,/tikz/inner sep=0.5pt}]
  \node[dot] (d) at (0, 0.65);
  \vertex (f) at (0, 1.85);
  \vertex (a) at (-3, -1.2);
  \vertex (b) at (3, -1.2);
  \node[dot] (z) at (0, -1);
  \vertex (x) at (-0.5, -0.5);
  \vertex (y) at (0.5, -0.5);
  \diagram*{
  (f) -- [fermion] (d),
  (d) -- [anti fermion] m [dot, red],
  (a) -- (b),
   m -- [anti fermion] (x),
   m -- [anti fermion] (y),
   (x) -- (z),
   (y) -- (z),
  }; 
  \end{feynman}
  \draw (-2.7,-1.2) arc (180:0:2.7);
  \end{tikzpicture}
            \caption{}
  \label{}
     \end{subfigure}%
     \hfill
     \begin{subfigure}[b]{0.4\textwidth}
     \centering
     \begin{tikzpicture}
  \begin{feynman}[every blob={/tikz/fill=white!30,/tikz/inner sep=0.5pt}]
  \vertex (d) at (0, 1.75);
  \vertex (a) at (-3, -1.4);
  \vertex (b) at (3, -1.4);
  \node[dot] (z) at (0, -1.2);
  \node[dot] (x) at (-0.5, -0.5);
  \diagram*{
  (d) -- [fermion] m [dot, red],
  (a) -- (b),
   m -- [ fermion] (x),
   (x) -- [ anti fermion] (z),
   m -- [ anti fermion] (z),
  }; 
  \end{feynman}
  \draw (-2.7,-1.4) arc (180:0:2.7);
  \end{tikzpicture}
           \caption{}
  \label{}
     \end{subfigure}%
     \hfill
     \begin{subfigure}[b]{0.4\textwidth}
     \centering
         \begin{tikzpicture}
  \begin{feynman}[every blob={/tikz/fill=white!30,/tikz/inner sep=0.5pt}]
  \node[dot] (d) at (0, 0.65);
  \vertex (f) at (0, 2.25);
  \vertex (a) at (-3, -0.8);
  \vertex (b) at (3, -0.8);
  \node[dot] (x) at (-0.5, -0.5);
  \vertex (y) at (0.5, -0.5);
  \diagram*{
  (f) -- [fermion] (d),
  (d) -- [fermion] m [dot, red],
  (a) -- (b),
   m -- [ fermion] (x),
   m -- [anti fermion] (y),
   (x) -- (y),
  }; 
  \end{feynman}
  \draw (-2.7,-0.8) arc (180:0:2.7);
  \end{tikzpicture}
     \caption{}
  \label{}
     \end{subfigure}%
     \hfill
     \begin{subfigure}[b]{0.4\textwidth}
     \centering
       \begin{tikzpicture}
  \begin{feynman}[every blob={/tikz/fill=white!30,/tikz/inner sep=0.5pt}]
  \vertex (d) at (0, 1.75);
  \vertex (a) at (-3, -1.4);
  \vertex (b) at (3, -1.4);
  \node[dot] (z) at (0, -1.2);
  \node[dot] (x) at (-0.5, -0.5);
  \diagram*{
  (d) -- [fermion] m [dot, red],
  (a) -- (b),
   m -- [ fermion] (x),
   (x) -- [ fermion] (z),
   m -- [ anti fermion] (z),
  }; 
  \end{feynman}
  \draw (-2.7,-1.4) arc (180:0:2.7);
  \end{tikzpicture}
\caption{}
  \label{}
     \end{subfigure}%
     \hfill
     \begin{subfigure}[b]{0.4\textwidth}
     \centering
    \begin{tikzpicture}
  \begin{feynman}[every blob={/tikz/fill=white!30,/tikz/inner sep=0.5pt}]
  \node[dot] (d) at (0, 0.65);
  \vertex (f) at (0, 2.25);
  \vertex (a) at (-3, -0.8);
  \vertex (b) at (3, -0.8);
  \node[dot] (x) at (-0.5, -0.5);
  \vertex (y) at (0.5, -0.5);
  \diagram*{
  (f) -- [fermion] (d),
  (d) -- [anti fermion] m [dot, red],
  (a) -- (b),
   m -- [ fermion] (x),
   m --  (y),
   (x) -- [ anti fermion] (y),
  }; 
  \end{feynman}
  \draw (-2.7,-0.8) arc (180:0:2.7);
  \end{tikzpicture}
     \caption{}
  \label{}
     \end{subfigure}%
    \caption{Graphs appearing in $\boldsymbol{\Omega}^{\mathbb{B}}_{2,1}$. Most of them do not give any contribution since again by Kontsevich's vanishing lemma \cite{Ko03} all the graphs with vertices where exactly one arrow ends and starts and graphs with double edges vanish.}
    \label{fig:omega21}
\end{figure}
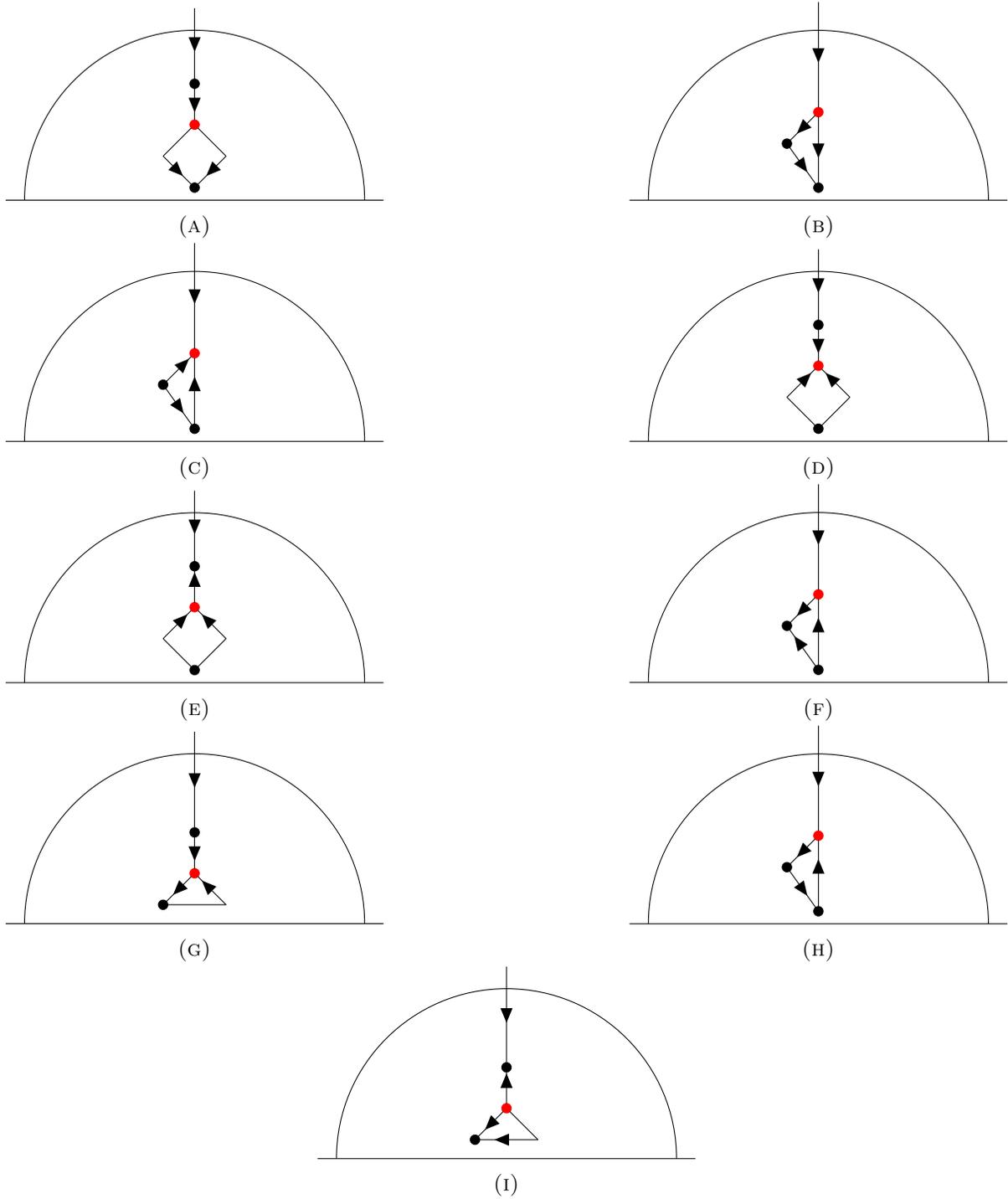

   \begin{figure}[h]
      \centering
      \begin{subfigure}[b]{0.4\textwidth}
  \begin{tikzpicture}
  \begin{feynman}[]
  \node[dot,red] (d) at (0, 1);
  \vertex (e) at (-0.5, 2);
  \vertex (f) at (0.5, 2);
  \vertex (a) at (-3, -1.4);
  \vertex (b) at (3, -1.4);
  \node[dot, red] (z) at (0, -0.7);
  \node[dot] (s) at (0, -1.4);
  \vertex (x) at (-0.5,-0.25);
   \vertex (q) at (0.5,-0.25);
  \node[dot,red] (y) at (0,0.25);
  \diagram*{
  (a) -- (b),
  (e) --[fermion] (d),
  (f) --[fermion] (d),
   (d) -- [anti fermion] (y),
   (y) --[anti fermion] (x),
   (z) -- (x),
   (y) --[anti fermion] (q),
   (z) -- (q),
   (z) --[fermion] (s)
  }; 
  \end{feynman}
  \draw (-2.7,-1.4) arc (180:0:2.7);
  \end{tikzpicture}
   \caption{}
  \label{}
     \end{subfigure}%
     \hfill
    \begin{subfigure}[b]{0.4\textwidth}
      \begin{tikzpicture}
  \begin{feynman}[]
  \node[dot,red] (d) at (0, 1);
  \vertex (e) at (-0.5, 2);
  \vertex (f) at (0.5, 2);
  \vertex (a) at (-3, -1.4);
  \vertex (b) at (3, -1.4);
  \node[dot, red] (z) at (0, -0.7);
  \node[dot] (s) at (0, -1.4);
  \vertex (x) at (-0.5,-0.25);
  \vertex (q) at (0.5,-0.25);
  \node[dot,red] (y) at (0,0.25);
  \diagram*{
  (a) -- (b),
  (e) --[fermion] (d),
  (f) --[fermion] (d),
   (d) -- [anti fermion] (y),
   (y) --[ fermion] (x),
   (z) -- (x),
   (y) --[ fermion] (q),
   (z) -- (q),
   (z) --[fermion] (s)
  }; 
  \end{feynman}
  \draw (-2.7,-1.4) arc (180:0:2.7);
  \end{tikzpicture}
        \caption{}
  \label{}
     \end{subfigure}%
     \hfill
   \begin{subfigure}[b]{0.4\textwidth}
      \begin{tikzpicture}
  \begin{feynman}[]
  \node[dot,red] (d) at (0, 1);
  \vertex (e) at (0, 2);
  \vertex (a) at (-3, -1.4);
  \vertex (b) at (3, -1.4);
  \node[dot, red] (z) at (0, -0.7);
  \node[dot] (s) at (0, -1.4);
  \vertex (x) at (0.5,0.5);
  \vertex (p) at (-0.5,0.5);
  \vertex (y) at (-2,1.75); 
  \diagram*{
  (a) -- (b),
  (e) --[fermion] (d),
  (p) -- (d),
  (x) -- (d),
  (p) --[fermion] m[dot,red],
  (x) --[fermion] m,
   m -- [anti fermion] (z),
   (y) -- [fermion] (z),
   (z) --[fermion] (s)
  }; 
  \end{feynman}
  \draw (-2.7,-1.4) arc (180:0:2.7);
  \end{tikzpicture}
       \caption{}
  \label{}
     \end{subfigure}%
      \hfill
   \begin{subfigure}[b]{0.4\textwidth}
   \begin{tikzpicture}
  \begin{feynman}[]
  \node[dot,red] (d) at (0, 1);
  \vertex (e) at (0, 2);
  \vertex (a) at (-3, -1.4);
  \vertex (b) at (3, -1.4);
  \node[dot, red] (z) at (0, -0.7);
  \node[dot] (s) at (0, -1.4);
  \vertex (x) at (0.5,0.5);
  \vertex (p) at (-0.5,0.5);
  \vertex (y) at (-2,1.75); 
  \diagram*{
  (a) -- (b),
  (e) --[fermion] (d),
  (p) -- (d),
  (x) -- (d),
  (p) --[anti fermion] m[dot,red],
  (x) --[anti fermion] m,
   m -- [anti fermion] (z),
   (y) -- [fermion] (z),
   (z) --[fermion] (s)
  }; 
  \end{feynman}
  \draw (-2.7,-1.4) arc (180:0:2.7);
  \end{tikzpicture}
        \caption{}
  \label{}
     \end{subfigure}%
      \hfill
     \begin{subfigure}[b]{0.4\textwidth}
       \begin{tikzpicture}
  \begin{feynman}[]
  \node[dot,red] (d) at (0, 0.25);
  \vertex (e) at (0, 1.75);
  \vertex (a) at (-3, -1.4);
  \vertex (b) at (3, -1.4);
  \node[dot, red] (z) at (0, -0.7);
  \node[dot] (s) at (0, -1.4);
  \node[dot,red] (x) at (-0.5,-0.25);
  \vertex (y) at (-1,1.75); 
  \diagram*{
  (a) -- (b),
  (e) --[fermion] (d),
  (x) -- (d),
  (x) --[anti fermion] (d),
  (x) --[anti fermion] (z),
   (d) -- [ fermion] (z),
   (y) -- [fermion] (x),
   (z) --[fermion] (s)
  }; 
  \end{feynman}
  \draw (-2.7,-1.4) arc (180:0:2.7);
  \end{tikzpicture}
          \caption{}
  \label{}
     \end{subfigure}%
     \hfill
       \begin{subfigure}[b]{0.4\textwidth}
        \begin{tikzpicture}
  \begin{feynman}[]
  \node[dot,red] (d) at (0, 0.25);
  \vertex (e) at (0, 1.75);
  \vertex (a) at (-3, -1.4);
  \vertex (b) at (3, -1.4);
  \node[dot, red] (z) at (0, -0.5);
  \node[dot] (s) at (0, -1.4);
  \node[dot,red] (x) at (-0.5,-0.25);
  \vertex (p) at (0.5,-0.9);
  \vertex (q) at (-0.5,-0.9);
  \vertex (y) at (-0.5,1.75);
  \vertex (n) at (-1,1.75);
  \diagram*{
  (a) -- (b),
  (e) --[fermion] (d),
  (x) -- (d),
  (x) --[anti fermion] (d),
   (d) -- [ fermion] (z),
   (y) -- [fermion] (x),
   (n) -- [fermion] (x),
   (p) --[fermion] (s),
   (p) -- (z),
   (q) --[fermion] (s),
   (q) -- (z)
  }; 
  \end{feynman}
  \draw (-2.7,-1.4) arc (180:0:2.7);
  \end{tikzpicture}
        \caption{}
  \label{}
     \end{subfigure}%
     \hfill
       \begin{subfigure}[b]{0.4\textwidth}
       \begin{tikzpicture}
  \begin{feynman}[]
  \node[dot,red] (d) at (0, 0.25);
  \vertex (e) at (0, 1.75);
  \vertex (a) at (-3, -1.4);
  \vertex (b) at (3, -1.4);
  \node[dot, red] (z) at (0, -0.5);
  \node[dot] (s) at (0, -1.4);
  \node[dot,red] (x) at (-0.5,-0.25);
  \vertex (p) at (0.75,-0.5);
  \vertex (y) at (-0.5,1.75);
  \vertex (n) at (-1,1.75);
  \diagram*{
  (a) -- (b),
  (e) --[fermion] (d),
  (x) --[anti fermion] (z),
   (d) -- [ fermion] (z),
   (y) -- [fermion] (x),
   (n) -- [fermion] (x),
   (z) --[fermion] (s),
   (d) -- (p),
   (p) --[fermion] (s)
  }; 
  \end{feynman}
  \draw (-2.7,-1.4) arc (180:0:2.7);
  \end{tikzpicture}
       \caption{}
  \label{}
     \end{subfigure}%
     \hfill
     \begin{subfigure}[b]{0.4\textwidth}
       \begin{tikzpicture}
  \begin{feynman}[]
  \node[dot,red] (d) at (0, 0.25);
  \vertex (e) at (-1, 1.75);
  \vertex (a) at (-3, -1.4);
  \vertex (b) at (3, -1.4);
  \node[dot, red] (z) at (0, -0.5);
  \node[dot] (s) at (0, -1.4);
  \node[dot,red] (x) at (0.75,-0.1);
  \vertex (p) at (0.75,-0.5);
  \vertex (y) at (+0.75,1.75);
  \vertex (n) at (+1.25,1.75);
  \diagram*{
  (a) -- (b),
  (e) --[fermion] (z),
  (x) -- (d),
  (x) --[anti fermion] (d),
   (d) -- [ anti fermion] (z),
   (y) -- [fermion] (x),
   (n) -- [fermion] (x),
   (z) --[fermion] (s),
   (p) --[fermion] (s),
   (p) -- (d)
  }; 
  \end{feynman}
  \draw (-2.7,-1.4) arc (180:0:2.7);
  \end{tikzpicture}
        \caption{}
  \label{}
     \end{subfigure}%
     \hfill
     \begin{subfigure}[b]{0.4\textwidth}
     \begin{tikzpicture}
  \begin{feynman}[]
  \node[dot,red] (d) at (0, 0.25);
  \vertex (e) at (-1, 1.75);
  \vertex (a) at (-3, -1.4);
  \vertex (b) at (3, -1.4);
  \node[dot, red] (z) at (0, -0.5);
  \node[dot] (s) at (0, -1.4);
  \node[dot,red] (x) at (0.75,-0.1);
  \vertex (y) at (0,1.75);
  \vertex (n) at (+1.25,1.75);
  \diagram*{
  (a) -- (b),
  (e) --[fermion] (z),
  (x) -- (d),
  (x) --[ fermion] (d),
   (d) -- [ anti fermion] (z),
   (y) -- [fermion] (d),
   (n) -- [fermion] (x),
   (z) --[fermion] (s),
   (x) --[fermion] (s)
  }; 
  \end{feynman}
  \draw (-2.7,-1.4) arc (180:0:2.7);
  \end{tikzpicture}
     \caption{}
  \label{}
     \end{subfigure}%
       \hfill
     \begin{subfigure}[b]{0.4\textwidth}
   \begin{tikzpicture}
  \begin{feynman}[]
  \node[dot,red] (d) at (0, 0.25);
  \vertex (e) at (0, 1.75);
  \vertex (a) at (-3, -1.4);
  \vertex (b) at (3, -1.4);
  \node[dot, red] (z) at (-0.5, -0.5);
  \node[dot] (s) at (0, -1.4);
  \node[dot,red] (x) at (0.75,-0.1);
  \vertex (y) at (+0.75,1.75);
  \vertex (n) at (+1.25,1.75);
  \vertex (q) at (-0.5,1.75);
  \vertex (r) at (-0.75, -0.9);
  \diagram*{
  (a) -- (b),
  (e) --[fermion] (d),
  (x) -- (d),
  (x) --[anti fermion] (d),
   (s) -- [anti fermion] (z),
   (y) -- [fermion] (x),
   (n) -- [fermion] (x),
   (r) --[fermion] (s),
   (z) -- (r),
   (d) --[fermion] (s),
   (q) --[fermion] (z)
  }; 
  \end{feynman}
  \draw (-2.7,-1.4) arc (180:0:2.7);
  \end{tikzpicture}
\caption{}
  \label{}
     \end{subfigure}%
     \caption{Graphs appearing in $\boldsymbol{\Omega}^{\mathbb{B}}_{0,3}$. Most of them do not give any contribution since again by Kontsevich's vanishing lemma \cite{Ko03} all the graphs with double edges vanish.}
      \label{fig:omega03}
  \end{figure}

\end{appendix}
\clearpage

\printbibliography
\end{document}